\pdfoutput=1
\documentclass[a4paper,11pt,twoside]{report}

\usepackage{amsmath}
\usepackage{amstext}
\usepackage{amsthm}
\usepackage{braket}
\usepackage{graphicx}
\usepackage[percent]{overpic}
\usepackage{bbm}
\usepackage{bm}
\usepackage{threeparttable}
\usepackage{tabularx}
\usepackage{multirow}
\usepackage{booktabs}
\usepackage[english]{babel}
\usepackage[nohyperlinks]{acronym}
\usepackage{microtype}
\usepackage{pifont}

\makeatletter
\renewcommand\@pnumwidth{1.8em}
\makeatother

\usepackage[T1]{fontenc}
\usepackage{textcomp}
\usepackage{fix-cm}
\usepackage[charter,cal=cmcal,sfscaled=false]{mathdesign}

\usepackage[paper=a4paper,innermargin=4cm,outermargin=3cm,top=3.5cm,bottom=3.5cm]{geometry}
\usepackage{setspace}
\usepackage[PetersLenny]{fncychap}
\usepackage{fancyhdr}

\usepackage[square,comma,numbers,sort&compress]{natbib}

\usepackage[plainpages=false]{hyperref}
\hypersetup{pdfauthor={Martin Aulbach},
  pdftitle={Classification of Entanglement in Symmetric States},
  pdfsubject={Ph.D. thesis, July 2011},
  pdfkeywords={symmetric} {Majorana representation} {geometric measure}
  {maximally entangled} {maximal entanglement} {SLOCC} {M\"{o}bius transformation}
  {entanglement families} {global entanglement} {maximal n-tangles}
  {Platonic solids} {dual polyhedra} {Martin Aulbach},
  colorlinks,
  citecolor=blue,
  linkcolor=blue,
  urlcolor=blue}
\usepackage[figure]{hypcap}

\newcommand{\Permsymm}{Per\-mu\-ta\-tion-sym\-met\-ric }
\newcommand{\permsymm}{per\-mu\-ta\-tion-sym\-met\-ric }
\newcommand{\permsymmnospace}{per\-mu\-ta\-tion-sym\-met\-ric}
\newcommand{\permantisymm}{per\-mu\-ta\-tion-an\-ti\-sym\-met\-ric }
\newcommand{\compactlist}{
  \setlength{\itemsep}{0pt}
  \setlength{\leftskip}{-1.0em}
}
\newcommand{\verycompactlist}{
  \setlength{\itemsep}{0pt}
  \setlength{\leftskip}{-1.3em}
}
\newcommand{\yestick}{\ding{51}}
\newcommand{\notick}{\ding{55}}
\newcommand{\toths}{T\'{o}th's }
\newcommand{\toth}{T\'{o}th }
\newcommand{\totx}{T\'{o}}
\newcommand{\tothx}{T\'{o}th}
\newcommand{\mob}{M\"{o}bius }
\newcommand{\ens}{\enspace}
\newcommand{\bmr}[1]{\bm{#1}}
\newcommand{\cc}[1]{{#1}^{*}}
\newcommand{\dash}{\nobreakdash-\hspace{0pt}}
\newcommand{\matha}{\mathcal{A}}
\newcommand{\mathb}{\mathcal{B}}
\newcommand{\mathd}{\mathcal{D}}
\newcommand{\mathm}{\mathcal{M}}
\newcommand{\mathh}{\mathcal{H}}
\newcommand{\Order}{\mathcal{O}}
\newcommand{\suc}{\text{SU}(2)}
\newcommand{\slc}{\text{SL}(2,\mathbb{C})}
\newcommand{\pslc}{\text{PSL}(2,\mathbb{C})}
\newcommand{\co}{\text{c}}
\newcommand{\si}{\text{s}}

\newcommand{\RE}{\operatorname{Re}}
\newcommand{\IM}{\operatorname{Im}}
\newcommand{\spa}{\operatorname{span}}
\newcommand{\rotx}{\operatorname{R_{x}}}
\newcommand{\rotxs}{\operatorname{R_{x}^{s}}}
\newcommand{\roty}{\operatorname{R_{y}}}

\newcommand{\rotz}{\operatorname{R_{z}}}
\newcommand{\rotzs}{\operatorname{R_{z}^{s}}}
\newcommand{\one}{\mathbbm{1}}
\newcommand{\mbbn}{\mathbb{N}}
\newcommand{\mbbz}{\mathbb{Z}}
\newcommand{\mbbr}{\mathbb{R}}
\newcommand{\mbbrp}{\mathbb{R}^{+}}
\newcommand{\mbbrr}{\mathbb{R}^{3}}
\newcommand{\mbbc}{\mathbb{C}}
\newcommand{\cext}{\overline{\mathbb{C}}}
\newcommand{\D}{\text{d}}
\newcommand{\E}{\text{e}}
\newcommand{\Trace}{\text{Tr}}
\newcommand{\etal}{\textit{et al. }}
\newcommand{\I}{\text{i}}
\newcommand{\Eg}{E_{\text{g}}}
\newcommand{\Egt}{\widetilde{E}_{\text{g}}}
\newcommand{\EG}{E_{\text{G}}}
\newcommand{\bracket}[2]{\braket{ #1 | #2 }}

\newcommand{\pure}[1]{\ket{ #1 } \! \bra{ #1 }}
\newcommand{\sym}[1]{\ket{\text{S}_{#1}}}
\newcommand{\psis}{\ket{\psi^{\text{s}}}}
\newcommand{\Psis}{\ket{\Psi^{\text{s}}}}
\newcommand{\phis}{\ket{\phi^{\text{s}}}}
\newcommand{\psisn}{\ket{\psi_{\!n}^{\text{s}}}}
\newcommand{\varphis}{\ket{\varphi^{\text{s}}}}
\providecommand{\Abs}[1]{\left\lvert#1\right\rvert}
\providecommand{\abs}[1]{\lvert#1\rvert}
\providecommand{\norm}[1]{\lVert#1\rVert}
\providecommand{\tfra}[2]{\tfrac{#1}{#2}}
\providecommand{\quo}[1]{{``{#1}''}}
\providecommand{\eq}[1]{Equation~\eqref{#1}}
\providecommand{\Eq}[1]{Equation~\eqref{#1}}
\providecommand{\fig}[1]{Figure~\ref{#1}}
\providecommand{\Fig}[1]{Figure~\ref{#1}}
\providecommand{\tabref}[1]{Table~\ref{#1}}
\providecommand{\Tabref}[1]{Table~\ref{#1}}
\providecommand{\chap}[1]{Chapter~\ref{#1}}
\providecommand{\Chap}[1]{Chapter~\ref{#1}}
\providecommand{\sect}[1]{Section~\ref{#1}}
\providecommand{\Sect}[1]{Section~\ref{#1}}
\providecommand{\theoref}[1]{Theorem~\ref{#1}}
\providecommand{\Theoref}[1]{Theorem~\ref{#1}}
\providecommand{\lemref}[1]{Lemma~\ref{#1}}
\providecommand{\Lemref}[1]{Lemma~\ref{#1}}
\providecommand{\corref}[1]{Corollary~\ref{#1}}
\providecommand{\Corref}[1]{Corollary~\ref{#1}}
\providecommand{\conref}[1]{Conjecture~\ref{#1}}

\acrodef{ICA}{independent component analysis}

\newtheorem{theorem}{Theorem}
\newtheorem{lemma}[theorem]{Lemma}
\newtheorem{corollary}[theorem]{Corollary}
\newtheorem{conjecture}[theorem]{Conjecture}

\newcommand{\blue}[1]{{\bf \textcolor[rgb]{0.0,0.39,1.0}{#1}}}

\begin{document}

\pagenumbering{alph}
\capstartfalse
\begin{titlepage}
  \begin{center}
    {\small Submitted in accordance with the requirements for the
      degree of\\Doctor of Philosophy}
    \vspace{0.5in}
    
    {\sffamily
      \begin{figure}[h]
        \centering
        \includegraphics[scale=1.25]{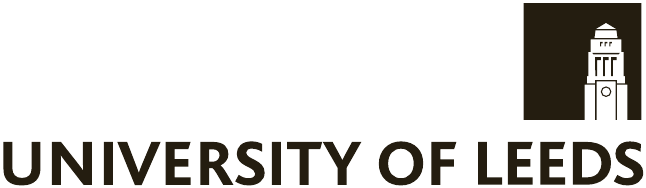}
      \end{figure}
      {\large \bfseries \textsf{ SCHOOL OF PHYSICS \& ASTRONOMY}}}
    \vspace{2.0in}
    
    {\huge \textbf{Classification of Entanglement in \\ Symmetric
        States}}
    \vspace{0.5in}
    
    {\Large Martin Aulbach}
    \vspace{0.5in}
    
    {\Large July 2011}
    \vfill
    
    {\small The candidate confirms that the work submitted is his own
      and that appropriate credit has been given where reference has
      been made to the work of others.
      \vspace{2mm}
      
      This copy has been supplied on the understanding that it is
      copyright material and that no quotation from the thesis may be
      published without proper acknowledgement.}
  \end{center}
\end{titlepage}

\capstarttrue
\pagestyle{empty}
\cleardoublepage

\pagenumbering{roman}
\pagestyle{plain}
\onehalfspacing

\phantomsection
\addcontentsline{toc}{chapter}{Acknowledgements}

\chapter*{Acknowledgements}

Firstly, I would like to thank my supervisor Vlatko Vedral with whom I
share not only a love for physics, but also for Guinness and a good
cigar.  He provided me with all the support I needed during my PhD,
and at the same time he allowed me to pursue my own research
interests. I greatly benefited from his extensive knowledge and from
his encouragement to visit conferences and research groups at other
universities.  Great thanks also go to my second supervisor, Jacob
Dunningham, who was extremely helpful whenever I needed advice. I
could not have wished for better supervisors.

Although a friend and colleague rather than a supervisor, Damian
Markham took the role of my informal mentor. He shared his broad
knowledge with me and provided me with guidance throughout my PhD, for
which I am very grateful.

Many thanks go to Viv Kendon and Andreas Winter for agreeing to be my
examiners, and to Almut Beige and Jiannis Pachos for acting as my
research assessment panel.

For stimulating discussions I am particularly grateful to Mark
Williamson, Jacob Biamonte, Lin Chen, Wonmin Son, Dagmar Bru\ss{},
Pedro Ribeiro, R{\'e}my Mosseri, Christopher Hadley, Andreas Osterloh
and of course Mio Murao, who always warmly welcomed me as a visitor in
her research group in Tokyo.

For financial assistance throughout my research degree I am very
grateful for the William Wright Smith Scholarship, provided to me by
the University of Leeds.

It was a pleasure to be part of the Quantum Information group in
Leeds, thanks to its members. Both during and (especially!) outside
office hours we had an almost indecent amount of fun. Thank you, Mark,
Michal, Cristhian, Jenny, Libby, Fran, Jess, Neil, Bruno, David, Tony,
Andreas, Rob, Katherine, Matt, Martin, Melody, Joe, Nick, Luke,
Jonathan, Abbas, Veiko, Adam, Elica, Giovanni and Mireia!

Finally, I would like to thank my parents for everything they did for
me. This thesis is dedicated to them.

\cleardoublepage

\onehalfspacing

\phantomsection
\addcontentsline{toc}{chapter}{Abstract}
\begin{center}
  {\Large \bfseries Classification of Entanglement in Symmetric
    States}
  
  \vspace{0.8cm}
  Martin Aulbach
  
  \vspace{0.3cm}
  Ph.D. thesis, July 2011
  
  \vspace{2.0cm}
  {\bfseries Abstract}
\end{center}

\noindent Quantum states that are symmetric with respect to
permutations of their subsystems appear in a wide range of physical
settings, and they have a variety of promising applications in quantum
information science.  In this thesis the entanglement of symmetric
multipartite states is categorised, with a particular focus on the
pure multi-qubit case and the geometric measure of entanglement.  An
essential tool for this analysis is the Majorana representation, a
generalisation of the single-qubit Bloch sphere representation, which
allows for a unique representation of symmetric $n$ qubit states by
$n$ points on the surface of a sphere.  Here this representation is
employed to search for the maximally entangled symmetric states of up
to $12$ qubits in terms of the geometric measure, and an intuitive
visual understanding of the upper bound on the maximal symmetric
entanglement is given.  Furthermore, it will be seen that the Majorana
representation facilitates the characterisation of entanglement
equivalence classes such as \ac{SLOCC} and the \ac{DC}. It is found
that \ac{SLOCC} operations between symmetric states can be described
by the \mob transformations of complex analysis, which allows for a
clear visualisation of the \ac{SLOCC} freedoms and facilitates the
understanding of \ac{SLOCC} invariants and equivalence classes.  In
particular, explicit forms of representative states for all symmetric
\ac{SLOCC} classes of up to 5 qubits are derived.  Well-known
entanglement classification schemes such as the 4 qubit entanglement
families or polynomial invariants are reviewed in the light of the
results gathered here, which leads to sometimes surprising
connections.  Some interesting links and applications of the Majorana
representation to related fields of mathematics and physics are also
discussed.

\cleardoublepage

\acresetall
\onehalfspacing
\pagestyle{fancy} \fancyhead[LO,RE]{}
\fancyhead[LE]{\nouppercase{\bfseries \leftmark}}
\fancyhead[RO]{\nouppercase{\bfseries \rightmark}}

\tableofcontents
\listoffigures
\listoftables
\clearpage

\pagestyle{plain}
\onehalfspacing

\phantomsection
\addcontentsline{toc}{chapter}{Acronyms}

\chapter*{List of Acronyms}

\begin{acronym}[SLOCC]
  \acro{LU}{Local Unitary}
  \acrodefplural{LU}{Local Unitaries}
  \acro{LO}{Local Operation}
  \acro{ILO}{Invertible Local Operation}
  \acro{LOCC}{Local Operations and Classical Communication}
  \acro{SLOCC}{Stochastic Local Operations and Classical Communication}
  \acro{DC}{Degeneracy Configuration}
  \acro{EF}{Entanglement Family}
  \acrodefplural{EF}{Entanglement Families}
  \acro{GHZ}{Greenberger-Horne-Zeilinger}
  \acro{MBQC}{measurement-based quantum computation}
  \acrodefplural{MBQC}[MBQC]{Measurement-based quantum computation}
  \acro{GM}{geometric measure of entanglement}
  \acro{LMG}{Lipkin-Meshkov-Glick}
  \acro{MP}{Majorana point}
  \acro{CPS}{closest product state}
  \acro{CPP}{closest product point}
  \acro{iff}{if and only if}
  \acro{d.f.}{degrees of freedom}
\end{acronym}

\cleardoublepage

\phantomsection
\addcontentsline{toc}{chapter}{Publications}

\chapter*{List of Publications}

\noindent
M.~Aulbach, D.~Markham, and M.~Murao.
\newblock
\textbf{The maximally entangled symmetric state in terms of the
  geometric measure}.
\newblock
\href{http://dx.doi.org/10.1088/1367-2630/12/7/073025}{New
  J. Phys. \textbf{12}, 073025 (2010)}.
\newblock
eprint: \href{http://arxiv.org/abs/1003.5643}{arXiv:1003.5643}.

\vspace{0.2cm}

\hfill (contains results presented in \chap{geometric_measure},
\ref{majorana_representation} and \ref{solutions})

\vspace{1.3cm}

\noindent
M.~Aulbach, D.~Markham, and M.~Murao.
\newblock
\textbf{Geometric Entanglement of Symmetric States and the Majorana
  Representation}.
\newblock
\emph{Proceedings of the 5th Conference on Theory of Quantum
  Computation, Communication and Cryptography}, edited by W.~van Dam,
V.~M.  Kendon, and S.~Severini,
\newblock
\href{http://dx.doi.org/10.1007/978-3-642-18073-6}{pp. 141--158 (LNCS,
  Berlin, 2010)}.
\newblock
ISBN 978-3-642-18072-9.
\newblock
eprint: \href{http://arxiv.org/abs/1010.4777}{arXiv:1010.4777}.

\vspace{0.2cm}

\hfill (contains results presented in \chap{geometric_measure},
\ref{majorana_representation}, \ref{solutions} and \ref{connections})

\vspace{1.3cm}

\noindent
M.~Aulbach.
\newblock
\textbf{Symmetric entanglement classes for n qubits}.
\newblock
\textit{in submission} (2011).
\newblock
preprint: \href{http://arxiv.org/abs/1103.0271}{arXiv:1103.0271}.

\vspace{0.2cm}

\hfill (contains results presented in \chap{classification})

\cleardoublepage

\pagestyle{fancy} \fancyhead[LO,RE]{}
\fancyhead[LE]{\nouppercase{\bfseries \leftmark}}
\fancyhead[RO]{\nouppercase{\slshape \rightmark}}
\pagenumbering{arabic}

\chapter{Introduction}\label{introduction}

\begin{quotation}
  In this preliminary chapter the subject of the present thesis is
  motivated and its objectives are formulated. This is followed by a
  brief review of some basic concepts of quantum information science,
  with a particular focus on entanglement theory and \permsymm states,
  the two topics that form the main focus of this work.  An overview
  of the subsequent chapters and the main results presented therein
  can be found at the end of this chapter.
\end{quotation}

\section{Motivation}\label{motivation}

Symmetry principles hold a special place in physics, and it is easy to
undervalue their significance for the historical development of many
important physical theories.  Newton himself did not consciously
formulate his revolutionary equations of motion for any particular
frame of reference, thus implicitly considering all directions and
points in space to be equivalent \cite{Wigner}.  Nearly two centuries
later the symmetries of electrodynamics were encapsulated into
Maxwell's equations, taking into account both Lorentz and gauge
invariance \cite{Gross96}, but it was not before Einstein that it was
realised that Maxwell's equations are merely a \emph{consequence} of
the relativistic invariance, and thus symmetry, of space-time itself.
In the standard model of modern particle physics the CPT\dash symmetry
postulates that our universe is indistinguishable from one with
inverted particle charges $C: q \mapsto -q$ (C\dash symmetry), parity
inversion $P: \bmr{r} \mapsto -\bmr{r}$ (P\dash symmetry) as well as
time reversal $T: t \mapsto -t$ (T\dash symmetry).  And going beyond
the standard model, the theory of supersymmetry postulates a further
physical symmetry between bosons and fermions, thus leading to the
postulation of yet-to-be-observed superpartners of the existing
elementary particles.

Noether's theorem outlines how continuous symmetries of physical
systems give rise to conserved quantities.  For example, the
conservation of energy arises from translations in time, and the
conversation of linear and angular momentum arises from translations
and rotations in space, respectively.  In quantum mechanics the
corresponding conservation laws follow directly from the kinematics of
the underlying theory, with physical quantities such as position and
momentum being expressed by operators on vectors of a Hilbert space
\cite{Wigner}. Many other important consequences of symmetry can be
observed in quantum mechanics: The selection rules governing atomic
spectra are the consequence of rotational symmetry, the different
aggregation behaviour of bosons and fermions is due to the invariance
or sign-change of the wave function under exchange of identical
particles, and in relativistic quantum mechanics the representations
of the full symmetry group -- the Poincar{\'e} group -- allows for a
complete classification of the elementary particles
\cite{Gross96}.\footnote{Slightly ironically, many phenomena in the
  world around us are due to \emph{symmetry breaking}. The more
  fundamental kind of symmetry breaking, spontaneous symmetry
  breaking, gives rise to non-symmetric states despite the laws of
  physics being symmetric themselves. Examples of such manifestations
  are crystals (broken translational invariance), magnetism (broken
  rotational invariance) and superconductivity (broken phase
  invariance) \protect\cite{Gross96}. Phase transitions between
  symmetric and non-symmetric states appear everywhere in physics,
  from down-to-earth occurrences in condensed matter physics to the
  unification of the fundamental forces of nature during the first
  moments after the big bang.}

The ground state of a quantum mechanical system with a finite number
of degrees of freedom is symmetric \cite{Gross96}, i.e. the state
remains invariant under permutations of the system's parts, and no
part is in any way different from any other.  This is a first
indication that symmetric quantum states play a particular role in
quantum physics.  Recently it has become possible to implement certain
symmetric states \cite{Prevedel09,Wieczorek09,Kiesel07} or even
arbitrary symmetric states \cite{Bastin09b} actively in experiments,
so it is only natural to gauge their possible applications in various
areas of physics. In this thesis the \permsymm quantum states will be
investigated from the perspective of quantum information theory
\cite{Nielsen}, a young, vibrant and highly interdisciplinary research
field that combines aspects of physics, mathematics, computer science,
chemistry and recently even biology \cite{Sarovar10,Gauger11}.  The
realisation that information is physical has lead to a revision of our
understanding of how nature works, and it has given rise to a
multitude of fascinating new applications. The most famous among these
is probably the \textbf{quantum computer}, initially suggested by
Feynman for efficient simulations of quantum systems
\cite{Feynman82}. Since then theorists have unearthed several
intriguing algorithms where a computer operating with qubits (quantum
mechanical spin-$\tfrac{1}{2}$ systems) rather than ordinary bits
would provide an exponential speedup (such as Shor's algorithm for
factorisation \cite{Shor94}), or at least a quadratic speedup
(Grover's algorithm for database searches \cite{Grover96}).  Other
exciting applications of quantum information are the teleportation of
quantum states over large distances via \textbf{quantum teleportation}
\cite{Bennett93}, and in principle unconditionally secure
communication between remote parties via \textbf{quantum cryptography}
\cite{Bennett84,Ekert91}. While the experimental realisation of
quantum computation and teleportation is still in its infancy, the
technically more mature status of quantum cryptography has allowed the
first commercial enterprises (e.g. ID Quantique) to enter the market.

Along with the superposition principle, the non-local property of
entanglement is considered to be one of the most striking consequences
of quantum physics.  Entanglement describes quantum correlations
between separate parts of a system that cannot be explained in terms
of classical physics, and these correlations are of particular
importance in quantum information science.  Entanglement is an
essential ingredient for quantum teleportation \cite{Bennett93},
superdense coding \cite{Bennett92}, \ac{MBQC} \cite{Nest07} and some
quantum cryptography protocols \cite{Ekert91}.  It can therefore be
considered as a \quo{standard currency} in many applications, and it
is desirable to know which states of a given Hilbert space are the
most entangled ones.  Unfortunately, for systems consisting of more
than two parts the quantification of entanglement is difficult due to
the existence of different types of entanglement, each of which may
capture a different desirable quality of a state as a resource
\cite{Dur00}.  It is therefore unsurprising that many different
entanglement measures have been proposed in order to quantify the
amount of entanglement of multipartite quantum states
\cite{Horodecki09}.  Some entanglement measures are not useful for the
analysis of larger systems, due to their bipartite definition, and
most measures are notoriously difficult to compute. For these reasons
the present thesis focuses on the \ac{GM} \cite{Shimony95, Wei04}, an
inherently multipartite entanglement measure that is not too difficult
to compute.

Returning to the concept of symmetry in physics, we recall that
\permsymm quantum states appear naturally in some systems
\cite{Ribeiro08,Orus08}, that it is possible to prepare them
experimentally \cite{Prevedel09,Wieczorek09,Kiesel07,Bastin09b}, and
that they have found some applications
\cite{Dhondt06,Korbicz05,Korbicz06,Ivanov10}.  Many canonical states
that appear in quantum information science are symmetric, e.g. Bell
diagonal states, \ac{GHZ} states \cite{Greenberger90}, W and Dicke
states \cite{Dicke54}, and the Smolin state \cite{Smolin01}.  These
aspects make it worthwhile to investigate the theoretical properties
as well as the practical usefulness of symmetric states for specific
quantum information tasks.  In particular, not much is known so far
about how to categorise the entanglement present in symmetric states,
and which symmetric states exhibit a high degree of entanglement. New
operational implications (in terms of usefulness for certain tasks) or
visualisations of symmetric states and their entanglement would also
be highly desirable.  With this we formulate the following goals for
the thesis:

\begin{itemize}
\item How can the entanglement of symmetric states be classified?
  
\item Which symmetric states are maximally entangled?
  
\item What operational meaning do symmetric states (or their
  entanglement) have?
  
\item How can symmetric states (or their entanglement) be visualised?
  
\item What kind of links exist between symmetric states and other
  fields of physics and mathematics?
\end{itemize}

A central tool for our analysis of symmetric entanglement will be the
Majorana representation \cite{Majorana32}, a generalisation of the
Bloch sphere representation of single qubits.  This will not only
provide us with a very useful visual representation of symmetric
states, but also allows us to classify the different types of
entanglement present in symmetric states, and to simplify the search
for maximal entanglement.  The Majorana representation will be
introduced, along with other elementary concepts of quantum
information theory, during the remainder of this introductory chapter.

\section{Quantum entanglement}\label{quantumentanglement}

In this section we will review some elementary concepts from the
theory of quantum entanglement and quantum information.  This is by no
means a comprehensive overview, but rather a selection of those
aspects that will be of particular importance in this thesis.  For a
comprehensive and recent review of quantum entanglement it is
suggested to consult the review article composed by the Horodecki
family \cite{Horodecki09}.

\subsection{Qubit and Bloch sphere}\label{blochsphere}

In analogy to the \textbf{bit} from classical information theory the
smallest unit of information in quantum information theory is called
the \textbf{qubit}, an abbreviation of \quo{quantum bit}. In contrast
to the classical bit which either takes the value $0$ or $1$, a qubit
can be in any superposition of the two basis vectors $\ket{0}$ and
$\ket{1}$, known as the computational basis.  Physically a qubit can
be realised by any quantum 2-level system, such as the spin of an
electron or the polarisation of a photon.  The state of a pure qubit
system can be written as $\ket{\phi} = a \ket{0} + b \ket{1}$, with
complex coefficients $a$ and $b$ that satisfy the normalisation
condition $\abs{a}^2 + \abs{b}^2 = 1$. By means of an unphysical
global phase the complex phase of the first coefficient can be
eliminated without restricting generality, which allows one to employ
the notation $\ket{\phi} = \cos \tfra{\theta}{2} \ket{0} + \E^{\I
  \varphi} \sin \tfra{\theta}{2} \ket{1}$ with two real parameters
$\theta \in [0, \pi]$ and $\varphi \in [0, 2 \pi)$.  Because of the
frequent use of this notation throughout the thesis, the trigonometric
expressions will be abbreviated as $\co_{\theta} := \cos
\tfra{\theta}{2}$ and $\si_{\theta} := \sin \tfra{\theta}{2}$.  The
famous \textbf{Bloch sphere} representation \cite{Nielsen} employs
this parameterisation to uniquely identify any pure qubit state with a
unit vector in $\mbbrr$, as shown in \fig{blochspherepicture}.  In
this picture the two basis vectors $\ket{0}$ and $\ket{1}$, which
correspond to the possible values of a classical bit, are represented
by the north pole and south pole of the Bloch sphere, respectively.
Any other point on the surface of the sphere represents a state
$\ket{\phi} = \co_{\theta} \ket{0} + \E^{\I \varphi} \si_{\theta}
\ket{1}$ that is in a superposition of the two basis states $\ket{0}$
and $\ket{1}$.  The \textbf{measurement} of such a state in the
computational basis $\{ \ket{0} , \ket{1} \}$ yields the outcome
$\ket{0}$ with probability $\abs{\co_{\theta}}^{2}$ and the outcome
$\ket{1}$ with probability $\abs{\si_{\theta}}^{2}$.  The natural
metric on the Bloch sphere is given by the Fubini-Study metric
\cite{Chruscinski06}, and the distance between two normalised qubits,
$\ket{\phi_{1}}$ and $\ket{\phi_{2}}$, in this metric is $\gamma (
\phi_{1} , \phi_{2} ) = \arccos \abs{\bracket{\phi_{1}}{\phi_{2}}}$,
i.e. the geometrical angle between the two corresponding points on the
Bloch sphere.

\begin{figure}
  \centering
  \begin{overpic}[scale=.5]{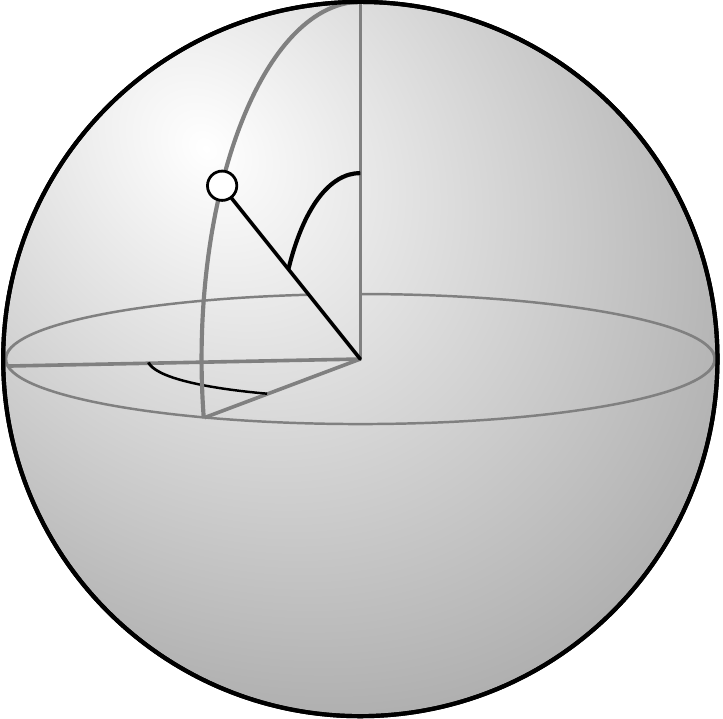}
    \put(14,74){$\ket{\phi}$}
    \put(44,61){$\theta$}
    \put(33,46.5){{\footnotesize $\varphi$}}
  \end{overpic}
  \caption[Bloch sphere representation of a
  qubit]{\label{blochspherepicture} Every pure state of a single qubit
    $\ket{\phi} = \co_{\theta} \ket{0} + \E^{\I \varphi} \si_{\theta}
    \ket{1}$ can be parameterised by two angles, the inclination
    $\theta \in [0, \pi]$ and the azimuth $\varphi \in [0, 2
    \pi)$. They give rise to the Bloch sphere representation on the
    surface of a unit sphere, with the Cartesian coordinates given by
    $(\sin \theta \cos \varphi , \sin \theta \sin \varphi , \cos
    \theta )$.}
\end{figure}

Pure qubit states are mathematically expressed as vectors of the
two-dimensional Hilbert space $\mathh = \mbbc^{2}$, but they are
unique only up to normalisation and an unphysical global phase, which
results in the two real degrees of freedom that manifest themselves as
the surface of the Bloch sphere.  If only partial information is known
about a quantum state, it has to be treated as a \textbf{mixed}
state\footnote{Mixed states are mathematically expressed as density
  matrices acting on the Hilbert space $\mathh$.  Any mixed state can
  be cast as a probability distribution of pure states, $\rho =
  \sum_{i=1}^{n} p_{i} \pure{\psi_{i}}$, and in general there exists
  an infinite number of such decompositions.  Every mixed state $\rho$
  must fulfil the following: 1.)  self-adjoint: $\rho =
  \rho^{\dagger}$, 2.)  semi-positive: $\rho \geq 0$
  (i.e. non-negative probabilities), and 3.)  unit trace: $\Trace [
  \rho ] = 1$ (i.e. probabilities sum up to one).  The set of mixed
  states is called the \emph{state space} $\mathcal{S}( \mathh )$, and
  a state $\rho \in \mathcal{S}( \mathh )$ is pure \acl{iff} $\rho^2 =
  \rho$.}.  While pure qubit states correspond to points on the
surface of the Bloch sphere, mixed qubit states correspond to the
interior of the sphere by means of the Pauli matrix representation of
the density matrix
\begin{equation}\label{mixed_state}
  \rho = \frac{1}{2} ( \one + x \sigma_{x} + y \sigma_{y} + z
  \sigma_{z} ) = \frac{1}{2} ( \one + \bmr{r} \bmr{\sigma} ) \ens ,
\end{equation}
with $\abs{\bmr{r}}^2 = \abs{x}^2 + \abs{y}^2 + \abs{z}^2 \leq 1$, and
where $\bmr{r} = (x,y,z) \in \mbbrr$ is the corresponding Bloch vector
within the unit sphere. The more mixed a state is, the closer it lies
to the centre of the Bloch sphere, with the maximally mixed state
$\rho = \one$ lying at the origin of the sphere.  The Pauli matrices
$\sigma_{x}$, $\sigma_{y}$ and $\sigma_{z}$ give rise to the rotation
operators which rotate Bloch vectors around the $X$-, $Y$- or $Z$-axis
by an angle $\vartheta$:

\begin{subequations}\label{rotationmatrices}
  \begin{align}
    \rotx (\vartheta)& = \E^{- \I \frac{\vartheta}{2} \sigma_{x}} =
    \begin{pmatrix}
      \phantom{-} \cos \tfra{\vartheta}{2}& - \I \sin
      \tfra{\vartheta}{2}
      \vspace*{0.5mm} \\
      - \I \sin \tfra{\vartheta}{2}& \phantom{-} \cos
      \tfra{\vartheta}{2}
    \end{pmatrix} \ens , \label{x_rotationmatrix} \\
    \roty (\vartheta)& = \E^{- \I \frac{\vartheta}{2} \sigma_{y}} =
    \begin{pmatrix}
      \cos \tfra{\vartheta}{2}& - \sin \tfra{\vartheta}{2}
      \vspace*{0.5mm} \\
      \sin \tfra{\vartheta}{2}& \phantom{-} \cos
      \tfra{\vartheta}{2}
    \end{pmatrix} \ens , \label{y_rotationmatrix} \\
    \rotz (\vartheta)& = \E^{- \I \frac{\vartheta}{2} \sigma_{z}} =
    \begin{pmatrix}
      \E^{- \I \frac{\vartheta}{2}}& 0 \\
      0& \E^{\I \frac{\vartheta}{2}}
    \end{pmatrix} \ens . \label{z_rotationmatrix}
  \end{align}
\end{subequations}

A rotation around an arbitrary axis $\bmr{n} = (n_{x} , n_{y} ,
n_{z})$, with $\abs{\bmr{n}} = 1$, that runs though the origin of the
Bloch sphere is given by $\operatorname{R}_{\bmr{n}} (\vartheta) =
\E^{- \I \frac{\vartheta}{2} \bmr{n} \bmr{\sigma}}$ and can be
straightforwardly calculated with the equations above.  In
mathematical terms the \textbf{unitary operations} are elements of
$\suc$, and in general they do not keep the coefficient of the
$\ket{0}$ vector of a pure state real and non-negative, so a
multiplication with a suitable global phase may be necessary after
rotation in order to return to the standard qubit notation $\ket{\phi}
= \co_{\theta} \ket{0} + \E^{\I \varphi} \si_{\theta} \ket{1}$. For
$Z$-axis rotations $\rotz (\vartheta)$ this global phase is simply
$\E^{\I \frac{\vartheta}{2}}$, independent of the Bloch vector
$\ket{\phi}$ that is being rotated.

While measurements destroy the state of an unknown qubit, this is not
the case with the rotation operators described above. Applying such a
unitary operation on an unknown qubit in the laboratory has the effect
of a rotation of its Bloch vector around an axis on the Bloch sphere,
without measuring or destroying the state unknown to the experimenter.

\subsection{Bipartite and multipartite systems}\label{multipartite}

Quantum systems that consist of two subsystems (e.g. two qubits) are
commonly known as \textbf{bipartite} systems, while systems with three
or more subsystems are referred to as \textbf{multipartite}
systems\footnote{Note that bipartite and multipartite quantum systems
  do not need to manifest themselves as an accumulation of distinct
  physical objects such as electrons, photons, etc., each of which
  gives rise to the Hilbert space of one subsystem. Instead,
  entanglement can exist between different degrees of freedom of a
  single physical particle, or even between different particle
  numbers, although the latter may lead to a violation of
  superselection rules \protect\cite{Wick52,Ashhab07}.}.  This
seemingly arbitrary distinction will become more meaningful later when
considering the qualitatively different manifestations of entanglement
in bipartite and multipartite systems.  Coined by Einstein as
\textit{spukhafte Fernwirkung} (\quo{spooky action at a distance}),
entanglement describes an inherently nonlocal correlation between
detached quantum systems that is predicted by quantum theory, and
which cannot be adequately described or explained in the language of
classical physics, at least without making assumptions about
\emph{hidden variables} \cite{Einstein35}.  The nonexistence of such
hidden variables in nature has been sufficiently validated
experimentally over the last few decades \cite{Aspect81}, thanks to
the ingenious Bell inequalities \cite{Bell64}.

In the language of quantum mechanics, an entangled quantum system is
one whose state vector cannot be expressed as the tensor product of
vectors of its subsystems.  In the simplest case of a bipartite
quantum system this is the case if $\ket{\psi} \neq \ket{\phi_1}
\otimes \ket{\phi_2}$, i.e.  no description of one part is complete
without information about the other.  One example for two qubits is
the Bell state $\ket{\psi^{+}} = \frac{1}{\sqrt{2}} \big( \ket{0}_{1}
\ket{0}_{2} + \ket{1}_{1} \ket{1}_{2} \big)$ where the two subsystems
are perfectly correlated with each other in the sense that the
measurement of one part of the system in a suitably chosen measurement
basis (for $\ket{\psi^{+}}$ the computational basis $\{ \ket{0},
\ket{1} \}$) mediates a \quo{collapse} of the other part into the same
state. For example, if a measurement of part $1$ in the computational
basis yields $\ket{0}_{1}$, then part 2 collapses to $\ket{0}_{2}$, so
any subsequent measurement of part 2 yields $\ket{0}_{2}$.  In other
words, the measurement turns the initially entangled state into one of
the two product states $\ket{\psi_{1}} = \ket{0}_{1} \ket{0}_{2}$ or
$\ket{\psi_{2}} = \ket{1}_{1} \ket{1}_{2}$ with equal probability.
The most striking aspect of this is that the measurement of one part
instantaneously affects the other part, regardless of the spatial
distance between the subsystems.  This cannot be used for superluminal
communication, however, because the randomness of the measurement
outcomes prevents the transmission of information by quantum
measurements alone, thus preserving a central tenet of special
relativity.

The Hilbert space of a multipartite system is given by the tensor
product of the subsystems' Hilbert spaces, i.e. $\mathh = \mathh_{1}
\otimes \cdots \otimes \mathh_{N}$, where $\mathh_{i}$ is the Hilbert
space of the $i$-th subsystem.  Since quantum states are uniquely
described only up to normalisation and a global phase, it makes sense
to introduce the projective Hilbert space \textbf{P}$\mathh$ as the
set of all unique pure quantum states.  The standard metric on
\textbf{P}$\mathh$ is the Fubini-Study metric \cite{Brody01}.

A multipartite pure quantum state $\ket{\psi} \in \mathh$ is
\textbf{separable} \ac{iff} it can be written as a tensor product of
states from the individual subsystems:
\begin{equation}\label{multipartite_product}
  \ket{\psi} = \ket{\phi_{1}} \otimes \cdots \otimes \ket{\phi_{N}}
  \ens , \quad \text{with} \quad \ket{\phi_{i}} \in \mathh_{i}
  \,\, \forall \, i \ens .
\end{equation}
States that are not separable are \textbf{entangled}.  For mixed
quantum states $\rho \in S(\mathh)$ of a multipartite system
separability is defined by the existence of a (non-unique) convex sum
of product states
\begin{equation}\label{multipartite_separable}
  \rho = \sum_{j} p_{j} \, \rho_{1}^{j} \otimes \cdots
  \otimes \rho_{N}^{j} \ens , \quad \text{with} \quad
  \rho_{i}^{j} \in S(\mathh_{i}) \,\, \forall \, i, j \ens .
\end{equation}
For finite-dimensional subsystems an orthonormalised basis $\{ \ket{0}
, \ldots , \ket{d_{i} - 1} \}_{i}$ can be chosen for each subsystem,
with $d_{i}$ denoting the dimension of $\mathh_{i}$. A pure quantum
state of the composite system can then be cast as

\begin{equation}\label{multipartite_state}
  \ket{\psi} = \sum_{i_{1}=0}^{d_{1} - 1} \cdots
  \sum_{i_{n}=0}^{d_{n} - 1} a_{ i_{1} , \ldots , i_{n} }
  \ket{i_{1}}_{1} \otimes \cdots \otimes \ket{i_{n}}_{n} \ens ,
\end{equation}

where the $a_{ i_{1} , \ldots , i_{n} }$ are complex coefficients, and
$_{j}\bracket{i_{A}}{i_{B}}_{j} = \delta_{AB}$ for all $j \in \{1,
\ldots , n \}$ and all $A,B \in \{ 0, \ldots , d_{j} - 1 \}$.  For
brevity the basis states $\ket{i_{1}}_{1} \otimes \cdots \otimes
\ket{i_{n}}_{n}$ of the composite system will be abbreviated as
$\ket{i_{1}} \ket{i_{2}} \cdots \ket{i_{n}}$, or simply $\ket{i_{1}
  i_{2} \cdots i_{n}}$. The normalisation $\bracket{\psi}{\psi} = 1$
will be implied throughout this thesis, except for a few cases where
states are easier to represent in unnormalised form and where the
normalisation does not matter.

Of particular interest in this thesis will be states whose
coefficients are all real or positive.  We call a quantum state
$\ket{\psi}$ of the form \eqref{multipartite_state} \textbf{real} if
$a_{ i_{1} , \ldots , i_{n} } \in \mbbr$ for all $i_{1}, \ldots ,
i_{n}$, and \textbf{positive} if $a_{ i_{1} , \ldots , i_{n} } \geq 0$
for all $i_{1}, \ldots , i_{n}$.  It should be noted that these
properties intrinsically depend on the chosen basis, and that states
which are real or positive in one \emph{computational basis} (a basis
made up of tensors of local bases) generally do not exhibit this
property in another basis.  In turn, a state that is not real or
positive in one basis may be recast as a real or positive state by
choosing a different basis, although this is in general not possible.
Only for bipartite states it is always possible to find
orthonormalised bases for the subsystems so that a given state can be
expressed as a positive state in the form
\eqref{multipartite_state}. This is possible thanks to the
\textbf{Schmidt decomposition} of linear algebra which -- applied to
quantum information -- states that any pure state of a bipartite
system $\ket{\psi} \in \mathh_{A} \otimes \mathh_{B}$ with $d = \min
\{ \dim ( \mathh_{A} ) , \dim ( \mathh_{B} ) \}$ can be expressed in
the form
\begin{equation}\label{schmidt_decomp}
  \ket{\psi} = \sum_{i=0}^{d-1} \alpha_{i} \ket{i} \ket{i} \ens ,
  \quad \text{with} \quad \alpha_{0} \geq \ldots \geq \alpha_{d-1}
  \geq 0 \ens ,
\end{equation}
where the non-negative numbers $\alpha_{i}$ are called the
\textbf{Schmidt coefficients} \cite{Nielsen}. The minimum number of
nonvanishing coefficients required for the Schmidt decomposition is
known as the \textbf{Schmidt rank}.  The Schmidt decomposition and the
Schmidt rank are important tools for the analysis of bipartite states,
which will become clear in the next section.

Unfortunately, the elegant Schmidt decomposition
\eqref{schmidt_decomp} does not exist in the multipartite setting, a
first indication that the bipartite case is qualitatively different
from the multipartite case. Several attempts have been made to find a
\textbf{generalised Schmidt decomposition}, a standard form for the
multipartite setting which imposes certain restrictions on the
coefficients of a given state by choosing suitable orthonormal bases
for all subsystems
\cite{Acin00b,Carteret00,Verstraete03,Kraus10,Kraus10b}.  Here we
mention the generalised Schmidt decomposition of Carteret \etal
\cite{Carteret00} which is defined for arbitrary finite-dimensional
multipartite states.  For the sake of simplicity, we consider $n$
equal subsystems, each of dimension $d$. In analogy to the 2-level
qubit, we refer to such a $d$-level quantum system as a
\textbf{qudit}.  The standard form imposes the following conditions on
the coefficients in \eq{multipartite_state} for the $n$ qudit case:
\begin{subequations}\label{generalised_schchmidt}
  \begin{align}
    a_{\underbrace{\scriptstyle i \ldots ik \!\!}_{j \text{ indices}} i
      \ldots i}& = 0&
    \forall& \, j, i \ens
    \forall \, i < k \leq d - 1 \ens ,
    \label{generalised_schmidt1} \\
    a_{\underbrace{\scriptstyle (d-1) \ldots (d-1)i}_{j \text{ indices}}
      (d-1) \ldots (d-1)}& \geq 0&
    \forall& \, j, i \ens , \label{generalised_schmidt2} \\
    \abs{a_{\scriptstyle ii \ldots ii}}& \geq
    \abs{a_{\scriptstyle j_{1} \ldots j_{n}}}&
    \forall& \, i \ens
    \forall \, i \leq j_{r} \,\, (r=1, \dots , n) \ens .
    \label{generalised_schmidt3}
  \end{align}
\end{subequations}
These conditions clearly do not necessarily result in a real or
positive state in general, and most multipartite states do not allow
for a real or positive representation.  If the subsystems do not have
equal dimensions, then the conditions are of a more complicated
form. However, \eq{generalised_schmidt1} straightforwardly generalises
to
\begin{equation}\label{generalised_schmidt_zero}
  a_{\underbrace{\scriptstyle i \ldots ik \!\!}_{j \text{ indices}} i
    \ldots i} = 0 \qquad \forall \, j, i \ens
  \forall \, i < k \leq d_{j} - 1 \ens ,
\end{equation}
where $d_{j}$ is the dimension of subsystem $j$.

A multipartite generalisation of the Schmidt rank has also been put
forward, and is commonly referred to as the \textbf{tensor rank}
\cite{Dur00,Eisert01}.  This quantity is given by the minimum number
of product states needed to expand a given state.  The tensor rank has
featured prominently in some recent works, and has been employed to
find further evidence of a qualitative difference between the
bipartite and multipartite setting \cite{Chen10b}.

Two well-known multipartite states with interesting entanglement
features are the \ac{GHZ} state \cite{Greenberger90} and the W state.
In the general case of $n$ qubits their form is
\begin{align}\label{ghz_w_def}
  \ket{\text{GHZ}_{n}}& = \tfrac{1}{\sqrt{2}}
  \left( \ket{00 \ldots 00} + \ket{11 \ldots 11} \right)
  \ens , \\
  \ket{\text{W}_{n}}& = \tfrac{1}{\sqrt{n}}
  \left( \ket{10 \ldots 0} + \ket{010 \ldots 0} +
    \ldots + \ket{00 \ldots 01} \right) \ens .
\end{align}
These two states have found a broad range of uses in quantum
information science \cite{Horodecki09}.  For example, the 3 qubit
\ac{GHZ} state has been employed to tell Bell's theorem without
inequalities \cite{Greenberger90}, and the $n$-qubit \ac{GHZ} state
can be considered the most non-local with respect to all possible
two-output, two-setting Bell inequalities \cite{Werner01}.  However,
the \ac{GHZ} state loses all its entanglement if a particle is lost,
because its one-particle reduced density matrix $\Trace_{i}
(\pure{\text{GHZ}_{n}}) = \tfra{1}{2} ( \pure{00 \ldots 00} + \pure{11
  \ldots 11} )$ is a separable state.  On the other hand, W states
still retain a considerable amount of entanglement after the removal
of an arbitrary particle, and it has been shown that the $n$-qubit W
state is the optimal state for leader election \cite{Dhondt06}.  This
shows that the entanglement of \ac{GHZ} and W states is of a different
nature, and such qualitative aspects of entanglement and their
characterisation will be investigated in the next section.

\subsection{Entanglement classes}\label{entanglement_classes}

In order to categorise different types of entanglement, it makes sense
to partition the given Hilbert space into equivalence
classes\footnote{In mathematical terms $\mathcal{F} = \{
  \mathcal{F}_{1}, \ldots , \mathcal{F}_{k} \}$ is a partition of a
  set $\mathcal{G}$ if it satisfies the conditions $\mathcal{F}_{i}
  \neq \varnothing$ for all $i$, $\mathcal{F}_{i} \cap \mathcal{F}_{j}
  = \varnothing$ for all $i \neq j$, and $\bigcup\limits_{i=1}^{k}
  \mathcal{F}_{i} = \mathcal{G}$. The $\mathcal{F}_{i}$ are the
  equivalence classes of $\mathcal{F}$.}, with an operationally
motivated definition of equivalence.  The most intuitive
classification scheme is that of \textbf{\ac{LU}} equivalence.  In
\sect{blochsphere} the effect of unitary operations on a single qubit
was outlined.  When generalising this concept to an arbitrary number
of quantum particles distributed among spatially separated
experimenters, then the local application of unitary operations on
each particle is referred to as an \ac{LU} operation. Such operations
are both deterministic and reversible, and -- from a mathematical
viewpoint -- equivalent to selecting a different orthonormalised basis
for the computational representation of a given state.  Therefore two
\ac{LU}-equivalent states $\rho_{\psi}
\stackrel{\text{LU}}{\longleftrightarrow} \rho_{\phi}$ are expected to
have precisely the same physical properties, in particular the same
entanglement.  A comprehensive analysis of the equivalence classes of
$n$ qubit pure states under \ac{LU} operations has recently been
achieved by Kraus \cite{Kraus10}, and subsequently employed to find
the different \ac{LU} equivalence classes of up to five qubits
\cite{Kraus10b}.

In order to perform quantum information tasks, it is necessary for the
experimenters to manipulate the states of their quantum particles in
more ways than by \ac{LU} operations alone.  The different types of
\textbf{quantum operations} \cite{Vidal00} that can be performed on a
given state $\rho$ are the following:
\begin{itemize}
\compactlist
\item \emph{unitary transformation:} \hspace{27.2mm}
  $\rho \longmapsto U \rho U^{\dagger} \ens ,$

  \hspace{4mm} where $U$ is a unitary operator.

\item \emph{selective projective measurement:} \hspace{13.2mm}
  $\rho \longmapsto \{ p_{i} , \sigma_{i} \} \ens ,$

  \hspace{4mm} i.e. the measurement outcome $\sigma_{i}$ is observed
  with probability $p_{i}$.

\item \emph{non-selective projective measurement:} \qquad
  $\rho \longmapsto \sum\limits_{i} p_{i} \sigma_{i} \ens ,$

  \hspace{4mm} i.e. discarding the measurement outcome yields a
  mixture of all possible outcomes.

\item \emph{addition of an ancilla system:} \hspace{18.5mm}
  $\rho \longmapsto  \rho \otimes \omega \ens ,$
  
  \hspace{4mm} where $\omega$ is an auxiliary quantum system
  (\quo{ancilla}) added to the system.

\item \emph{removal of a subsystem:} \hspace{27.3mm}
  $\rho \longmapsto \Trace_{A} [ \rho ] \ens ,$

  \hspace{4mm} where subsystem $A$ is removed from the quantum system
  by a partial trace.
\end{itemize}
These different kinds of quantum evolution can all be subsumed under
linear completely positive maps $\mathcal{E} : S ( \mathh ) \to S (
\mathh )$ with the help of Kraus operators \cite{Nielsen}.

As it is typically not possible in practice to perform joint
operations on spatially separated particles, the quantum operations
act locally.  The experimenters are however able to coordinate their
actions by communicating with each other over a classical channel,
e.g. by telephone.  This leads to the paradigm of \textbf{\ac{LOCC}}
whereby quantum states are modified by performing \acp{LO} on the
subsystems and allowing the transmission of classical communication
between the parties.  As seen from the list of quantum operations
above, such \ac{LOCC} operations are in general irreversible.  For the
case of pure states, however, it has been shown (see Corollary 1 of
\cite{Bennett00} or Theorem 4 of \cite{Vidal00}) that two states are
\ac{LOCC}-equivalent \ac{iff} they are \ac{LU}-equivalent.  This
defines a partition of the pure Hilbert space which is equivalent to
the partition generated by \ac{LU} equivalence. For two pure $n$ qudit
states \textbf{\ac{LOCC} equivalence} is mathematically expressed as
\begin{equation}\label{gen_locc_cond}
  \ket{\psi} \stackrel{\text{LOCC}}{\longleftrightarrow} \ket{\phi}
  \quad \Longleftrightarrow \quad
  \exists \, \matha_{1} , \ldots , \matha_{n} \in \text{SU}(d) :
  \ket{\psi} = \matha_{1} \otimes \cdots
  \otimes \matha_{n} \ket{\phi} \ens ,
\end{equation}
and by definition the \ac{LOCC} equivalence of two general states
(denoted as $\rho_{\psi} \stackrel{\text{LOCC}}{\longleftrightarrow}
\rho_{\phi}$) requires that a \emph{deterministic} conversion is
possible in \emph{both directions}. This is a much more stringent
requirement than deterministic \ac{LOCC} conversion in only one
direction (denoted as $\rho_{\psi} \stackrel{\text{LOCC}}{\longmapsto}
\rho_{\phi}$).  For the pure bipartite case the latter conversions are
fully characterised by the theory of \textbf{majorisation}
\cite{Nielsen99}, which induces a partial order with the help of the
Schmidt decomposition \eqref{schmidt_decomp}.  More precisely, the
necessary and sufficient conditions for deterministically converting a
pure two-qudit state into another one are
\begin{equation}\label{majorisation}
  \ket{\psi} \stackrel{\text{LOCC}}{\longmapsto} \ket{\phi}
  \quad \Longleftrightarrow \quad \forall \, 0 \leq j \leq d-1 :
  \sum_{i=0}^{j} \alpha_{i}^{2} \leq
  \sum_{i=0}^{j} {\alpha '}_{i}^{2} 
  \ens ,
\end{equation}
where the $\{ \alpha_{i} \}$ and $\{ {\alpha '}_{i} \}$ are the
Schmidt coefficients of $\ket{\psi}$ and $\ket{\phi}$, respectively.
From this is can be seen that pure bipartite states are
\ac{LOCC}-equivalent (or \ac{LU}-equivalent) to each other \ac{iff}
they have the same Schmidt coefficients:
\begin{equation}\label{locc_schmidt}
  \ket{\psi} \stackrel{\text{LU}}{\longleftrightarrow} \ket{\phi}
  \quad \Longleftrightarrow \quad
  \ket{\psi} \stackrel{\text{LOCC}}{\longleftrightarrow} \ket{\phi}
  \quad \Longleftrightarrow \quad
  \alpha_{i} = {\alpha '}_{i} \quad \forall \, i \ens .
\end{equation}

The conditions \eqref{majorisation} give rise to
\ac{LOCC}-incomparable states which cannot be converted into each
other either way.  On the other hand, there are \textbf{maximally
  entangled} states from which all other states, pure or mixed, can be
generated with certainty using only \ac{LOCC} operations. For two
$d$-level systems the maximally entangled states are those that are
\ac{LU}-equivalent to
\begin{equation}\label{max_ent_states}
  \ket{\Psi_{d}} = \frac{1}{\sqrt{d}} \sum_{i=0}^{d-1}
  \ket{i i} \ens .
\end{equation}
The non-existence of an analogous result for multipartite systems --
due to the absence of the Schmidt decomposition -- is one of the
reasons for the qualitative difference between the bipartite and
multipartite case.

As useful as the concept of \ac{LOCC} equivalence is from an
operational point of view, it is of little help to categorise the
wealth of inequivalent entanglement types found in multipartite
Hilbert spaces.  Many attempts have been made to find further
operationally motivated classifications, and the most prominent one
among these is the equivalence under \textbf{\ac{SLOCC}}, which is
identical to \ac{LOCC} equivalence except that the interconversion of
two states need not be deterministic.  Instead the success probability
of a conversion only needs to be non-zero.  The concept of stochastic
interconvertibility was first introduced by Bennett \etal
\cite{Bennett00} and later formalised by D\"{u}r \etal \cite{Dur00}.
\ac{SLOCC} operations are mathematically expressed as \acp{ILO}
\cite{Dur00}, and are also known as local filtering operations.  In
the case of pure $n$ qudit states the \ac{SLOCC}-equivalence reads
\begin{equation}\label{gen_slocc_cond}
  \ket{\psi} \stackrel{\text{SLOCC}}{\longleftrightarrow} \ket{\phi}
  \quad \Longleftrightarrow \quad
  \exists \, \mathb_{1} , \ldots , \mathb_{n} \in \text{SL}(d,\mbbc) :
  \ket{\psi} = \mathb_{1} \otimes \cdots
  \otimes \mathb_{n} \ket{\phi} \ens .
\end{equation}
It is clear that \ac{SLOCC}-equivalence implies \ac{LOCC}-equivalence,
and therefore the partition of the Hilbert space into \ac{LOCC}
equivalence classes is a \emph{refinement}\footnote{In the language of
  set theory, if $A$ and $B$ are two partitions of a set $M$, then the
  partition $A$ is a \textbf{refinement} of $B$ ($A \leq B$) if every
  element of $A$ is a subset of some element of $B$. For the
  entanglement classification schemes introduced here this means that
  $\text{LOCC} \leq \text{SLOCC}$.}  of the partition into \ac{SLOCC}
classes.

\ac{SLOCC} operations have a clear operational interpretation in the
sense that \emph{on average} they cannot increase the amount of
entanglement, although it is possible to obtain more entangled states
with a certain probability.  The latter is of importance for
experimentalists, because joint operations on multiple copies of a
state are often unfeasible, in which case \ac{SLOCC} operations on a
single copy are the best available entanglement distillation strategy
\cite{Verstraete03}.  While \ac{SLOCC} operations have the power to
dramatically increase or decrease the amount of entanglement shared
between parties, they cannot create entanglement out of nowhere or
completely destroy it, due to their local nature.  In particular, it
is not possible to generate entangled states from separable states by
\ac{SLOCC}, even probabilistically, something that is clear from the
definition of separability.

Multiqubit entanglement has been well studied in terms of \ac{SLOCC}
equivalence, in particular for a single copy of a pure $n$ qubit
state.  In the 2 qubit case there exist only two \ac{SLOCC}
equivalence classes, the class of separable states, and the generic
class of entangled states to which almost all states belong.  In
particular, any pure entangled state can be converted into any other
pure entangled state with non-zero probability.

For 3 qubits there exist six \ac{SLOCC} classes \cite{Dur00}, namely
the separable class, the three biseparable classes AB-C, AC-B, BC-A,
the class with W-type entanglement and the generic class with
\ac{GHZ}-type entanglement.  The canonical example of
\ac{SLOCC}-inequivalent entangled states are the three qubit
$\ket{\text{GHZ}_{3}}$ and $\ket{\text{W}_{3}}$ state.  Their tensor
rank is 2 and 3, respectively \cite{Dur00}, and the tensor rank has
been shown to be an \ac{SLOCC} invariant \cite{Dur00,Eisert01}.
Another way to distinguish between \ac{GHZ}-type and W-type states is
the \textbf{3-tangle}, an entanglement measure for three qubits
\cite{Coffman00,Dur00}.  The 3-tangle is zero not only for all states
that are separable under any bipartite cut, but even for states where
this is not the case, e.g. the $\ket{\text{W}_{3}}$ state.  The only
\ac{SLOCC} class with nonzero 3-tangle is that with \ac{GHZ}-type
entanglement, and in this sense $\ket{\text{GHZ}_{3}}$ is said to
contain \emph{genuine}\footnote{There is no universally accepted
  definition for the concept of \protect\quo{genuine} (or
  \protect\quo{true}) entanglement, but a common theme is that most or
  all of the local density matrices should be maximally mixed.  The
  \ac{GHZ} states exhibit this property, but the W states do not.
  Although W states are entangled over all parties, their multipartite
  entanglement is of a pairwise nature, i.e. within parts of the
  system.} tripartite entanglement \cite{Verstraete03}.

For 4 qubits the number of \ac{SLOCC} classes becomes infinite, and
there is no generic class to which almost all states belong.  Because
of this, various attempts have been made to find alternative and
physically meaningful classification schemes tailored for the 4 qubit
case. Techniques employed for this include Lie group theory
\cite{Verstraete02}, the hyperdeterminant \cite{Miyake03}, an
inductive approach \cite{Lamata07}, polynomial invariants
\cite{Viehmann11} and string theory \cite{Borsten10}.  For example,
Verstraete \etal \cite{Verstraete02} introduced the concept of
\textbf{\acp{EF}} with the help of normal forms, and found nine
different \acp{EF}. Since every \ac{SLOCC} equivalence class belongs
to exactly one \ac{EF}, the \ac{SLOCC} classes are a refinement of the
\acp{EF}.

An important tool for the study of entanglement equivalence classes
are quantities that do not change under a set of local operations such
as \ac{LU} or \ac{SLOCC} operations. Such quantities are known as
\textbf{invariants}, and they can provide information about the type
of entanglement present in a system.  Examples are the Schmidt rank
and the tensor rank, which are known to be invariant under \ac{SLOCC}
operations.

One popular approach to find \ac{SLOCC} invariants is to study
\textbf{polynomial invariants}. These are polynomials in the
coefficients of pure states that remain invariant under \ac{SLOCC}
operations.  Such polynomial invariants are entanglement monotones
with respect to \ac{SLOCC} operations \cite{Verstraete03}, and they
allow one to construct entanglement measures
\cite{Wootters98,Coffman00,Luque03,Osterloh05,Osterloh06}. In the two
and three qubit case the well-known concurrence (also known as
2-tangle) \cite{Wootters98} and 3-tangle \cite{Coffman00} are special
cases of polynomial invariants \cite{Miyake03}.  For four and five
qubits polynomial invariants have also been constructed
\cite{Luque03,Luque06,Osterloh05,Osterloh06,Levay06,Dokovic09}.  In
\cite{Luque03,Luque06} the polynomial invariants were constructed from
classical invariant theory, and the values of these invariants for the
\acp{EF} \cite{Verstraete02} were derived.  An alternative approach is
to employ the expectation values of antilinear operators, with an
emphasis on the permutational invariance of the \textbf{global
  entanglement measure} \cite{Osterloh05,Osterloh06}.

\subsection{Entanglement measures}\label{entanglement_measures}

In the previous section entanglement was characterised qualitatively
in the form of equivalence classes.  Now entanglement will be analysed
from a quantitative viewpoint by means of \textbf{entanglement
  measures}.  These help to assess the usefulness of given states as
resources for certain quantum informational tasks, and different
entanglement measures may capture different desirable qualities of a
state.  For bipartite, pure quantum states it is known that all
entanglement measures are essentially equivalent
\cite{Horodecki00,Horodecki09,Nielsen}, and one can find a unique
total order in the asymptotic regime of many copies.  For mixed states
of bipartite systems, as well as in the multipartite case, however, no
total ordering and therefore no unique entanglement measure exists
\cite{Morikoshi04,Wei03b,Nielsen}.

An entanglement measure $E : \mathcal{S} ( \mathcal{H} ) \to \mbbrp$
is a functional which maps the set of density matrices acting on a
Hilbert space $\mathcal{H}$ to non-negative numbers, $\rho \mapsto E (
\rho ) \geq 0$, and which satisfies certain axioms. Some of the most
common axioms are the following \cite{Vidal00,Horodecki00,Vedral98}:
\begin{enumerate}
\item \emph{Separable states:} \hspace{18.2mm} $E ( \rho ) = 0 \ens ,$
  \hspace{28.05mm} if $\rho$ is separable.
  
\item \emph{Invariance under \ac{LU}:} \hspace{11.45mm} $E ( \rho ) = E
  ( \sigma ) \ens ,$ \hspace{21.9mm} if $\rho
  \stackrel{\text{LU}}{\longleftrightarrow} \sigma$.

\item \emph{Monotonicity under \ac{LOCC}:} \quad $E ( \rho )
  \geq \sum\limits_{i} p_{i} E ( \sigma_{i} ) \ens ,$ \hspace{13.55mm}
  if $\rho \stackrel{\text{LOCC}}{\longmapsto} \left\{
    \begin{array}{l}
      \{ p_{i} , \sigma_{i} \}  \\
      \sum_{i} p_{i} \sigma_{i} \\
    \end{array} \right.$

\item \emph{Convexity:} \hspace{27.9mm} $E ( \rho ) \leq
  \sum\limits_{i} p_{i} E ( \rho_{i} ) \ens ,$ \hspace{13.8mm} where
  $\rho = \sum\limits_{i} p_{i} \rho_{i}$.

\item \emph{Additivity:} \hspace{27.9mm} $E ( \rho^{\otimes n} ) = n E
  ( \rho ) \ens ,$ \hspace{16.2mm} for all $n \in \mbbn$.

\item \emph{Strong Additivity:} \hspace{16.5mm} $E ( \rho \otimes
  \sigma ) = E ( \rho ) + E ( \sigma ) \ens ,$ \quad for all $\sigma
  \in \mathcal{S} ( \mathcal{H} )$.
\end{enumerate}
It is natural to require that an entanglement measure be zero for
non-entangled states, and from the previous section it is clear that
the measure should remain invariant under \ac{LU}.  Axiom 3 is the
most fundamental one, as the non-increase of entanglement under local
transformations (i.e. \ac{LOCC}) \cite{Vidal00,Horodecki00} lies at
the heart of our understanding of entanglement as a non-local resource
shared between parties. The natural extension of this axiom to
\ac{SLOCC} operations is that the entanglement shall not increase
\emph{on average}. The fourth axiom guarantees that entanglement
cannot be increased by mixing, something that can be understood as the
information loss encountered when going from a selection of
identifiable states to a mixture of those states. Since mixing is a
local operation, Axiom 4 is automatically fulfilled if Axiom 3
holds\footnote{At first glance the mathematical forms of Axiom 3 and
  Axiom 4 seem to contradict each other, so we stress the difference
  between their physical motivations: Axiom 3 describes a
  (non-)selective projective measurement of a given state $\rho$
  (l.h.s.), resulting in a random measurement outcome $\sigma_{i}$ or
  a superposition thereof (r.h.s.).  Axiom 4 starts with a selection
  of \emph{identifiable} states $\rho_{i}$ (r.h.s.) which are
  transformed into a mixture $\rho$ (l.h.s.), something that can be
  physically realised if an ancilla system (with orthonormal basis $\{
  \ket{i} \}$) attached to the initial state is lost: $\sum_{i} p_{i}
  \pure{i} \otimes \rho_{i} \mapsto \sum_{i} p_{i} \rho_{i}$.}.

Axioms 1 to 4 are regarded as the most important criteria for any
entanglement measure, and they coincide with the necessary properties
of \textbf{entanglement monotones}, as defined by Vidal
\cite{Vidal00}. Indeed, entanglement monotones derive their name from
the crucial requirement of monotonicity under \ac{LOCC}.

Axioms 5 and 6 are only two of many further properties that could be
required from any well-behaved entanglement measure. Even though
additivity looks like a natural requirement for entanglement measures
and is closely related to various operational meanings
\cite{Christandl04,Hastings09,Shor04}, many measures lack this
property. The strong additivity, also known as full additivity, is an
even more elusive property which is featured only by very few
measures, e.g. the squashed entanglement \cite{Christandl04}.  From
the definition it is clear that strong additivity implies regular
additivity.  The property of (strong) additivity can not only be
defined for entanglement measures, but also for individual states: A
state $\rho$ is additive with respect to an entanglement measure $E$
if $E ( \rho^{\otimes n} ) = n E ( \rho )$ holds for all $n \in
\mbbn$, and strongly additive if $E ( \rho \otimes \sigma ) = E ( \rho
) + E ( \sigma )$ for any state $\sigma$.

Many different entanglement measures have been defined, with either
operational or abstract advantages in mind. We will refrain from
providing an overview here, and instead refer to the review articles
\cite{Plenio07,Horodecki09}.  The single most important entanglement
measure for this thesis, the \acl{GM}, will be comprehensively
reviewed in \chap{geometric_measure}.

\section{Symmetric states}\label{symmetric_states}

\Permsymm quantum states are states that are invariant under any
permutation of their subsystems.  For an $n$-partite state
$\ket{\psi}$ this is the case \ac{iff} $P \ket{\psi} = \ket{\psi}$ for
all $P \in S_{n}$, where $S_{n}$ is the symmetric group of $n$
elements.  In the $n$ qubit case the symmetric sector
$\mathh_{\text{s}} \subset \mathh$ of the Hilbert space is spanned by
the $n+1$ Dicke states $\sym{n,k}$, $0 \leq k \leq n$, the equally
weighted sums of all permutations of computational basis states with
$n-k$ qubits being $\ket{0}$ and $k$ being $\ket{1}$
\cite{Dicke54,Toth07}:
\begin{equation}\label{dicke_def}
  \sym{{n,k}} = {\binom{n}{k}}^{- 1/2} \sum_{\text{perm}}
  \underbrace{ \ket{0} \ket{0} \cdots \ket{0} }_{n-k}
  \underbrace{ \ket{1} \ket{1} \cdots \ket{1} }_{k} \ens .
\end{equation}
From a physical point of view the Dicke states are the simultaneous
eigenstates of the total angular momentum $J$ and its $z$-component
$J_z$ \cite{Dicke54,Stockton03,Toth07}.  Dicke states were recently
produced in several experiments
\cite{Kiesel07,Prevedel09,Wieczorek09,Chiuri10,Thiel07}, they can be
detected experimentally \cite{Toth07,Krammer09,Thiel07,Hume09}, and
they have been proposed for certain tasks \cite{Ivanov10}.  We will
abbreviate the above notation of the Dicke states to $\sym{k}$
whenever the total number of qubits is clear.

A general pure symmetric state of $n$ qubits $\psis$ is a linear
combination of the $n+1$ orthonormalised Dicke states,
\begin{equation}\label{gen_symm_state}
  \psis = \sum_{k=0}^{n} a_k \sym{n,k} \ens ,
\end{equation}
with $a_{k} \in \mbbc$.  A generalisation to the qudit case is
straightforward \cite{Wei04}, with a general symmetric state of an $n$
qudit system being a linear combination of the Dicke states,
\begin{equation}\label{dicke_def_qudit}
  \sym{{n,\bmr{k}}} = \sqrt{\frac{k_{0}! k_{1}!
      \cdots k_{d-1}! }{n!}}
  \sum_{\text{perm}}
  \underbrace{ \ket{0} \cdots \ket{0} }_{k_{0}}
  \underbrace{ \ket{1} \cdots \ket{1} }_{k_{1}} \cdots
  \underbrace{ \ket{d-1} \cdots \ket{d-1} }_{k_{d-1}} \ens ,
\end{equation}
with $\bmr{k} = ( k_{0} , k_{1} , \cdots , k_{d-1} )$, and
$\sum_{i=0}^{d-1} k_{i} = n$. The main focus of this thesis will
however be symmetric states of $n$ qubits, as defined in
\eq{gen_symm_state}.

The theoretical and experimental analysis of symmetric states, e.g. as
entanglement witnesses or in experimental setups
\cite{Korbicz05,Korbicz06,Prevedel09,Wieczorek09,Kiesel07,Bastin09b,Bastin11},
is valuable for a variety of reasons.  Symmetric states have found use
in quantum information tasks such as leader election \cite{Dhondt06}
or as the initial state in Grover's algorithm \cite{Ivanov10}, and
they could possibly be useful for \acf{MBQC} \cite{Nest06} because
they are not too entangled for being computationally universal
\cite{Gross09}.  Symmetric states are known to appear in the Dicke
model \cite{Schneider02}, as eigenstates in various models of solid
states physics such as the \ac{LMG} model \cite{Orus08,Ribeiro08}, and
in the study of macroscopic entanglement of $\eta$-paired high $T_c$
superconductivity \cite{Vedral04}.  Furthermore, symmetric states have
been actively implemented experimentally
\cite{Prevedel09,Wieczorek09,Kiesel07,Bastin09b}, and their symmetric
properties facilitate the analysis of their entanglement properties
\cite{Bastin09,Hayashi08,Hubener09,Markham11,Mathonet10,Toth09}.  In
experiments with many qubits, it is often not possible to access
single qubits individually, necessitating a fully symmetrical
treatment of the initial state and the system dynamics \cite{Toth07}.

For these reasons symmetric states have featured prominently in recent
studies of entanglement theory, such as the characterisation of their
entanglement classes under \ac{SLOCC}
\cite{Markham11,Bastin09,Mathonet10,Aulbach11}, or the determination
of their maximal entanglement in terms of the \acl{GM}
\cite{Aulbach10,Martin10,Aulbach10lncs}.

\subsection{Majorana representation}\label{majorana_definition_sect}

In classical physics, the angular momentum $\bmr{J}$ of a system can
be represented by a point on the surface of the 3D unit sphere $S^2$,
which corresponds to the direction of $\bmr{J}$.  No such simple
representation is possible in quantum mechanics, but Ettore Majorana
\cite{Majorana32} pointed out that a pure state of spin-$j$ (in units
of $\hbar$) can be uniquely represented by $2j$ not necessarily
distinct points on $S^2$.  Given that $S^2$ can be associated with the
Bloch sphere, it is clear that this is a generalisation of the
spin-$\tfra{1}{2}$ (qubit) case, where the 2D Hilbert space is
isomorphic to the unit vectors on the Bloch sphere.  As seen in
\fig{majorana_dicke}, the three eigenstates $\ket{1, -1}$, $\ket{1,
  0}$ and $\ket{1, 1}$ of a spin-$1$ particle correspond to two points
being at the north pole, one at the north pole and the other at the
south pole and both of them at the south pole, respectively.

An equivalent representation can be shown to exist for \permsymm
states of $n$ spin-$\tfra{1}{2}$ particles \cite{Majorana32,Bacry74},
with an isomorphism mediating between all states of a spin-$j$
particle and the symmetric states of $2j$ qubits.  For a system of $n$
spin-$\tfra{1}{2}$ particles the eigenbasis of the square of the total
spin operator $\bmr{S}^2$ and its $z$ component $S_z$ can be
represented in the form $\ket{S,m}$, where $S (S+1) \hbar^2$ and $m
\hbar$ are the corresponding eigenvalues. It is the $n+1$ states from
the maximum spin sector $S = \tfra{n}{2}$ that are fully
\permsymmnospace, and it is those states that are identified as the
symmetric basis states, the Dicke states $\sym{n,k} \equiv
\ket{\tfra{n}{2},k - \tfra{n}{2}}$, with $k= 0 , \ldots , n$.  A
general state belonging to the maximum spin sector $\sum\limits_{m = -
  n/2}^{n/2} a_{m} \ket{\tfra{n}{2},m}$ is therefore equivalent to the
previous definition \eqref{gen_symm_state} of symmetric states.

\begin{figure}
  \centering
  \begin{overpic}[scale=1.0]{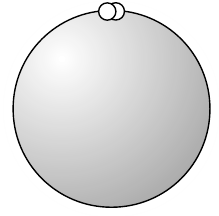}
    \put(-9,0){(a)}
  \end{overpic}
  \hspace{6mm}
  \begin{overpic}[scale=1.0]{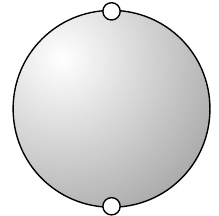}
    \put(-9,0){(b)}
  \end{overpic}
  \hspace{6mm}
  \begin{overpic}[scale=1.0]{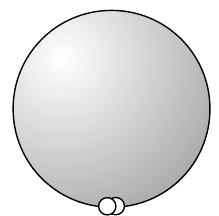}
    \put(-9,0){(c)}
  \end{overpic}
  \caption[Majorana representations of the Dicke states of two
  qubits]{\label{majorana_dicke} The Majorana representations of the
    three eigenstates (a) $\ket{1, -1}$, (b) $\ket{1, 0}$ and (c)
    $\ket{1, 1}$ of a spin-$1$ particle are shown, with the \aclp{MP}
    indicated by white circles.  By means of the isomorphism between
    spin-$j$ states and the symmetric states of $2j$ qubits, these are
    also the Majorana representations of the three symmetric basis
    states of two qubits, the Dicke states (a) $\sym{0} = \ket{00}$,
    (b) $\sym{1} = \tfra{1}{\sqrt{2}} ( \ket{01} + \ket{10} )$, and
    (c) $\sym{2} = \ket{11}$.}
\end{figure}

By means of the \textbf{Majorana representation} any symmetric state
of $n$ qubits $\psis$ can be uniquely composed, up to an unphysical
global phase, from a sum over all permutations $P \in S_{n}$ of $n$
indistinguishable single qubit states $\{ \ket{\phi_{1}} , \dots ,
\ket{\phi_{n}} \}$, with $S_n$ being the symmetric group of $n$
elements.
\begin{gather}
  \psis = \frac{\E^{\I \delta}}{\sqrt{K}} \sum_{ \text{perm} }
  \ket{\phi_{P(1)}} \otimes \ket{\phi_{P(2)}} \otimes \cdots \otimes
  \ket{\phi_{P(n)}} \ens ,
  \label{majorana_definition} \\
  \text{with} \quad \ket{\phi_i} = \cos \tfrac{\theta_i}{2} \ket{0} +
  \E^{\I \varphi_i} \sin \tfrac{\theta_i}{2} \ket{1} \ens , \quad
  \text{and} \quad K = n! \sum_{\text{perm}} \prod_{i = 1}^{n}
  \bracket{\phi_{i}}{\phi_{P(i)}} \ens . \nonumber
\end{gather}
Here $\E^{\I \delta}$ is a global phase, and the normalisation factor
$K$ is in general different for different $\psis$.  The qubits
$\ket{\phi_{i}}$ are uniquely determined by the choice of $\psis$ and
they determine the normalisation factor $K$.  By means of
\eq{majorana_definition}, any $n$ qubit state $\psis$ can be
unambiguously visualised by a multiset of $n$ points (each of which
has a Bloch vector pointing in its direction) on the surface of
$S^{2}$.  We call these points the \textbf{\acp{MP}}, and the sphere
on which they lie the \textbf{Majorana sphere}.

One nice property of the Majorana representation is that the \ac{MP}
distribution rotates rigidly on the Majorana sphere under the effect
of \ac{LU} operations on the subsystems. We have already seen in
\sect{blochsphere} that unitary operations $U \in \suc$ acting on a
single qubit, $U \ket{\phi} = \ket{\vartheta}$, correspond to
rotations of the Bloch vector around an axis on the Bloch sphere.
Applying the same single-qubit unitary operation $U$ on each of the
$n$ subsystems of a symmetric state $\psis$ yields another symmetric
state $\varphis$ by means of the map
\begin{equation}\label{symmetric_lu}
  \psis \longmapsto \varphis =
  U \otimes \cdots \otimes U \psis \ens ,
\end{equation}
and from \eq{majorana_definition} it follows that
\begin{equation}\label{lu_rotation}
  \varphis = \frac{\E^{\I \delta}}{\sqrt{K}} \sum_{\text{perm}}
  \ket{\vartheta_{P(1)}} \ket{\vartheta_{P(2)}} \cdots
  \ket{\vartheta_{P(n)}} \ens , \quad
  \text{with} \quad \ket{\vartheta_{i}} = U \ket{\phi_{i}} \,\,
  \forall \, i \ens .
\end{equation}
In other words, the \ac{MP} distribution of $\varphis$ is obtained by
a joint rotation of the \ac{MP} distribution of $\psis$ along a common
axis on the Majorana sphere.  Therefore the two \ac{LOCC}-equivalent
states $\psis$ and $\varphis$ have different \acp{MP}, but the same
\emph{relative} distribution (i.e., unchanged distances and angles) of
the \acp{MP} \cite{Markham11}.

To present some examples of \ac{MP} distributions, we consider the
three symmetric basis states of two qubits, the Dicke states $\sym{0}
= \ket{00}$, $\sym{1} = \tfra{1}{\sqrt{2}} ( \ket{01} + \ket{10} )$,
and $\sym{2} = \ket{11}$.  Their Majorana representations, shown in
\fig{majorana_dicke}, are two points on the north pole
($\ket{\phi_{1}} = \ket{\phi_{2}} = \ket{0}$), one on the north pole
and the other on the south pole ($\ket{\phi_{1}} = \ket{0},
\ket{\phi_{2}} = \ket{1}$), and two points on the south pole
($\ket{\phi_{1}} = \ket{\phi_{2}} = \ket{1}$), respectively.  While
$\sym{0}$ and $\sym{2}$ are separable states with zero entanglement,
$\sym{1}$ is the Bell state $\ket{\psi^{+}} = \tfra{1}{\sqrt{2}} (
\ket{01} + \ket{10} )$, a maximally entangled state of two qubits.
This state is represented by an antipodal pair of \acp{MP}, and it is
easy to verify that the amount of bipartite entanglement directly
increases with the distance between the two \acp{MP}.  For symmetric
states of three and more qubits this picture is not as clear, but one
would expect that symmetric states with a high degree of entanglement
are represented by \ac{MP} distributions that are well spread out over
the sphere.  We will use this idea along with other symmetry arguments
to look for the most entangled symmetric states in \chap{solutions}.

It is important to realise that the \ac{MP} states $\ket{\phi_{i}}$
that make up the Majorana representation \eqref{majorana_definition}
do not belong to a particular subsystem of the underlying physical
system.  Instead, the \acp{MP} should be viewed as abstract qubit
states from which the symmetric state of a physical system can be
reconstructed. In the next section we will see that the relationship
between a symmetric state and its \acp{MP} is equivalent to the
relationship between the coefficients and the zeroes of a complex
polynomial.

If the \acp{MP} of a symmetric state are known, then the explicit form
of the composite state can be directly calculated from
\eq{majorana_definition}.  On the other hand, if the \acp{MP} of a
given symmetric state $\psis = \sum^{n}_{k=0} a_k \sym{k}$ are
unknown, they can be determined by solving a system of $n+1$ equations
equivalent to Vieta's formulas \cite{WeissteinVieta}:
\begin{gather}
  a_k = {\binom{n}{k}}^{\frac{1}{2}} \sum_{\text{perm}}
  \text{S}_{P(1)} \cdots \text{S}_{P(k)} \text{C}_{P(k+1)} \cdots
  \text{C}_{P(n)} \ens , \label{state_to_mp} \\
  \text{with} \quad \text{C}_{i} = \cos \tfrac{\theta_i}{2} \ens ,
  \quad \text{S}_{i} = \E^{\I \varphi_{i}} \sin \tfrac{\theta_i}{2}
  \ens . \nonumber
\end{gather}

The Majorana representation has been rediscovered several times
\cite{Leboeuf91,PenroseRindler}, and has been put to many different
uses across physics. In relation to the foundations of quantum
mechanics, it has been used to find efficient proofs of the
Kochen-Specker theorem \cite{Zimba93,PenroseRindler} and to study the
\quo{quantumness} of pure quantum states in several respects
\cite{Zimba06,Giraud10}, as well as the approach to classicality in
terms of the discriminability of states \cite{Markham03}. It has also
been used to study Berry phases in high spin systems \cite{Hannay98}
and quantum chaos \cite{Hannay96,Leboeuf91}, and it has been put into
relation to geometrically motivated \ac{SLOCC} invariants
\cite{Levay06}.  Within many-body physics it has been used for finding
solutions to the \acf{LMG} model \cite{Ribeiro08}, and for studying
and identifying phases in spinor Bose-Einstein-condensates
\cite{Barnett06,Barnett07,Barnett08,Makela07}.  It has also been used
to look for optimal resources for reference frame alignment
\cite{Kolenderski08}, for phase estimation, and in quantum optics for
the multi-photon states generated by spontaneous parametric
down-conversion \cite{Kolenderski10}.  Furthermore, the Majorana
representation has been employed for finding a new proof of
Sylvester's theorem on Maxwell multiples \cite{Dennis04}, and for
analysing the relationship between spherical designs \cite{Crann10}
and anticoherent spin states \cite{Zimba06}.  The Majorana
representation has recently become a useful tool in studying and
characterising the entanglement of permutation-symmetric states
\cite{Aulbach11,Bastin09,Markham11,Mathonet10}, which has interesting
mirrors in the classification of phases in spinor condensates
\cite{Markham11,Barnett07}. Very recently further operational
interpretations of the \ac{MP} distribution have been discovered with
respect to additivity \cite{Chen11} and the equivalence of different
entanglement measures \cite{Markham11}.

\subsection{Stereographic projection}\label{stereo_proj}

The stereographic projection, a well-known concept from complex
analysis \cite{Forst}, describes an isomorphism between the points on
the surface of the $\mathcal{S}^{2}$ sphere and the points of the
extended complex plane $\cext = \mbbc \cup \{ \infty \}$.  As seen in
\fig{sterproj_examples}, the projection is mediated by rays
originating from the north pole of the Riemann sphere, thus projecting
points from the surface of the sphere along rays onto the complex
plane. By definition, the north pole is projected onto the \quo{point
  at infinity}.  The inverse projection from the plane onto the sphere
is also possible, and if the centre of the Riemann sphere coincides
with the origin of the complex plane, as shown in
\fig{sterproj_examples}(a), the inverse stereographic projection $v:
\cext \to \mbbrr$ has the form
\begin{equation}\label{sterproj_coord1}
  v(z) = \left\{ 
    \begin{array}{l l}
      \frac{1}{\abs{z}^2 + 1}
      \big( 2\RE (z) , 2 \IM (z) , \abs{z}^2 -1 \big)&
      \quad \text{for } z \in \mbbc
      \vspace{0.2cm} \\
      (0,0,1)&
      \quad \text{for } z = \infty \\
    \end{array} \right.
\end{equation}

\begin{figure}[ht]
  \centering
  \begin{overpic}[scale=0.92]{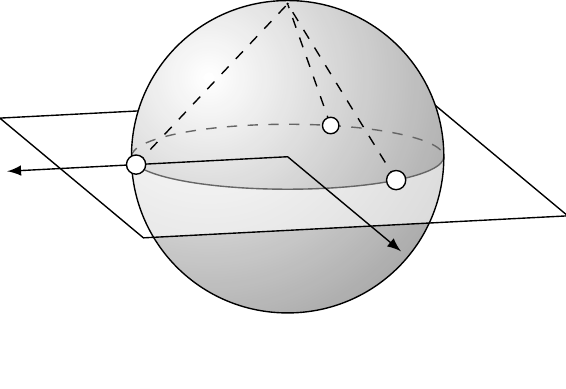}
    \put(5,0){(a)}
    \put(4,31){Re}
    \put(59,21){Im}
    \put(17,34.5){$z_1$}
    \put(70,31.5){$z_2$}
    \put(56,40.5){$z_3$}
  \end{overpic}
  \hfill
  \begin{overpic}[scale=0.92]{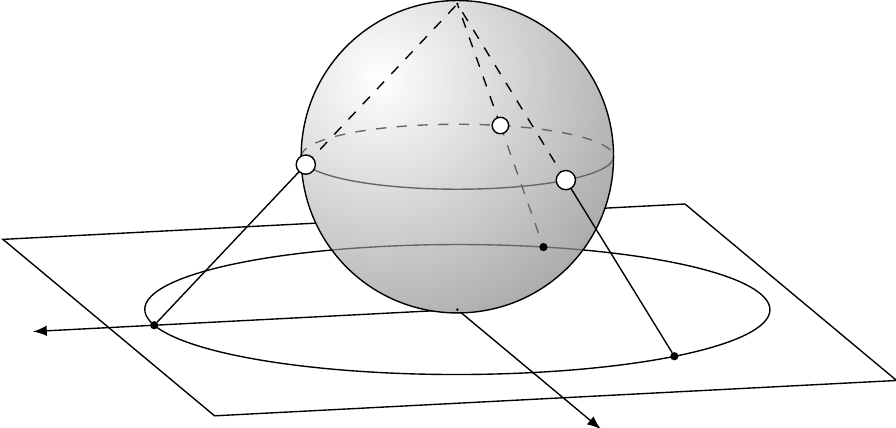}
    \put(-2,0){(b)}
    \put(7,6){Re}
    \put(58,-1){Im}
    \put(12,13.5){$z_1$}
    \put(75,10){$z_2$}
    \put(56,17){$z_3$}
  \end{overpic}
  \caption[Stereographic projection]{\label{sterproj_examples} An
    example of a stereographic projection of points on the Riemann
    sphere onto the complex plane is shown for two different positions
    of the sphere.}
\end{figure}
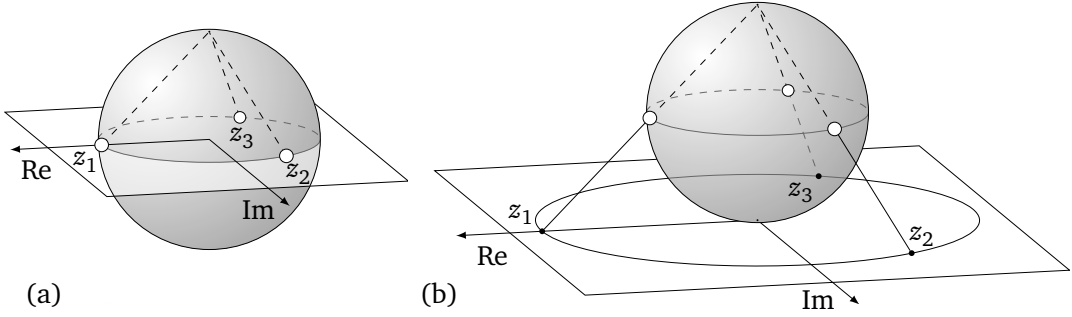

The stereographic projection is well-defined as long as the sphere's
north pole lies above the complex plane, and a frequently used
alternative position for the Riemann sphere is shown in
\fig{sterproj_examples}(b). Here the sphere rests on the plane, and
the inverse stereographic projection reads
\begin{equation}\label{sterproj_coord2}
  \widetilde{v}(z) = \left\{ 
    \begin{array}{l l}
      \frac{1}{ \frac{1}{2} \abs{z}^2 + 2}
      \big( 2\RE (z) , 2 \IM (z) , \abs{z}^2 \big)&
      \quad \text{for } z \in \mbbc
      \vspace{0.2cm} \\
      (0,0,2)&
      \quad \text{for } z = \infty \\
    \end{array} \right.
\end{equation}
The stereographic projection is of interest to us, because it is
closely linked to the Majorana representation of symmetric states.
For a given symmetric state $\psis = \sum_{k=0}^{n} a_k \sym{n,k}$ the
coefficients $a_k$ uniquely define a function $\psi (z) : \mbbc \to
\mbbc$ known as the Majorana polynomial, or alternatively the
characteristic polynomial, amplitude function \cite{Radcliffe71}, or
coherent state decomposition \cite{Leboeuf91}:
\begin{equation}\label{majpoly}
  \psi (z) = \sum_{k=0}^{n} ( -1 )^{k} \binom{n}{k}^{\frac{1}{2}}
  a_{k} \, z^{k} = A \prod_{k=1}^{n} ( z - z_{k} ) \ens .
\end{equation}
The Majorana polynomial represents symmetric states in terms of spin
coherent states \cite{Kolenderski08}, which can be seen from its
definition as $\psi (z) = \braket{ \sigma (z) |^{\otimes n} |
  \psi^{\text{s}} }$, where $z := \E^{\I \varphi} \tan
\frac{\theta}{2}$ uniquely parameterises the single qubit states
$\ket{\sigma (z)} = \ket{0} - \cc{z} \ket{1}$.  The right-hand side of
the equation follows from the fundamental theorem of algebra which
states that every polynomial of degree $n$ has $n$ not necessarily
distinct complex roots, and can be uniquely factorised up to a
prefactor $A$. We will call the $\{ z_{k} \}_{k}$ the Majorana roots,
and from the preceding discussion it is clear that there exist
one-to-one correspondences between the unordered set of \acp{MP} of a
symmetric state, its coefficients and the Majorana roots
\begin{equation}\label{isomrophisms}
  \text{\acp{MP} } \{ \ket{\phi_{k}} \}_{k} \hspace{2mm}
  \Longleftrightarrow \hspace{2mm}
  \text{coefficients } \{ a_{k} \}_{k} \hspace{2mm}
  \Longleftrightarrow \hspace{2mm}
  \text{Majorana roots } \{ z_{k} \}_{k} \ens .
\end{equation}
Intriguingly, the isomorphism between the \acp{MP} and the Majorana
roots is precisely described by the (inverse) stereographic projection
if the Riemann sphere is considered to be the Majorana sphere. The
\acp{MP} $\ket{\phi_{k}}$, represented on the sphere by the end points
of their Bloch vectors, are then projected onto the Majorana roots
$z_{k} \in \mbbc$ lying in the complex plane.  If any \acp{MP} lie at
the north pole, they are associated with the \quo{point at infinity},
and in this case the sum and product in \eq{majpoly} only run up to $n
- r$, where $r$ is the number of \acp{MP} being $\ket{0}$.

\section{Overview of the thesis}\label{overview}

With the recapitulation of some basic concepts of quantum information
theory behind us, we can now shift our focus towards new results.
\Chap{geometric_measure} through \chap{connections} are all research
chapters with original results.  Nevertheless, most of these chapters
are interspersed with introductory notes on non-elementary topics in
quantum information and related fields. Among these are the
introduction of the geometric measure of entanglement (\sect{gm_def}),
\acf{MBQC} (\sect{resources_for_mbqc}), an overview of spherical
optimisation problems (\sect{extremal_point}), the review of symmetric
entanglement classification schemes (\sect{overview_entclass}), the
\mob transformations of complex analysis (\sect{mobius_def}),
symmetric \ac{SLOCC} invariants (\sect{symm_inv}) and global
entanglement measures (\sect{global_ent}).  In \chap{connections}
several topics from mathematics and physics that are viewed in light
of results obtained in this thesis are introduced.

An overview of the contents presented in this thesis is given in the
following, sorted by chapters.  At the end of each summary of a
chapter reference is made where work has been published or is the
result of a collaboration.

\subsubsection*{Chapter 2: Geometric Measure of Entanglement}

The geometric measure of entanglement, an entanglement measure
particularly suited for the analysis of multipartite states, is
discussed in \chap{geometric_measure}.  After a comprehensive review
of this measure in \sect{gm_def}, it is applied to the general case of
arbitrary quantum systems in \sect{gen_results}, where a high number
of distinct closest product states is conjectured for maximally
entangled states.  In \sect{standard_form_coeff} a standard form is
derived for arbitrary $n$ qubit states with the help of the geometric
measure.  This is followed by an examination of states with positive
coefficients in \sect{pos_results}, with the conclusion that in
general the addition of complex phases to the coefficients of a
positive state leads to a considerable increase of the
entanglement. The case of symmetric $n$ qubit states is considered in
\sect{symm_results}, where a new proof for the upper bound on the
maximal symmetric entanglement is presented in connection with
\theoref{const_integral}.  This proof has the advantage of an
intuitive visualisation by means of a constant integration volume of a
spherical function, something that will be valuable for later
chapters.  In \sect{resources_for_mbqc} arguments are presented that
symmetric $n$ qubit states are not useful as resources for \ac{MBQC},
even in the context of stochastic approximate \ac{MBQC}.

\Sect{standard_form_coeff} is based on unpublished work with Seiji
Miyashita, Mio Murao and Damian Markham.  The results of
\sect{upper_bound} and \ref{resources_for_mbqc} were published in
\cite{Aulbach10,Aulbach10lncs}.

\subsubsection*{Chapter 3: Majorana Representation and Geometric
  Entanglement}

In \chap{majorana_representation} the Majorana representation is
applied to analyse the geometric entanglement of $n$ qubit symmetric
states. The first section combines the review of some known aspects
with the presentation of new results or methods.  After a discussion
in \sect{visualisation} about the visualisation of all the information
about symmetric states and their entanglement, the well-understood
properties of two and three qubit symmetric states are reviewed from
the perspective of our methodology in \sect{two_three}.  This is
followed by an introduction of the concept of totally invariant states
in \sect{invariant_and_additivity}, where it is shown that totally
invariant positive symmetric states are additive with respect to three
distance-like entanglement measures.  By means of the Majorana
representation the search for the maximally entangled symmetric states
can be understood as a spherical optimisation problem, and because of
this, \sect{extremal_point} reviews two classical point distribution
problems, \toths problem and Thomson's problem, and puts them in
contrast to the \quo{Majorana problem}.  This is followed in
\sect{analytic} by the derivation of several analytical results which
connect the coefficients of symmetric states to their Majorana
representation.  In particular, it will be seen that the Majorana
representation of states with real coefficients exhibits a reflective
symmetry, and that particularly strong restrictions are imposed on the
Majorana representation of positive states.
\Theoref{theo_general_maj_rep} presents a generalisation of the
Majorana representation which is useful to simplify the analysis of
many symmetric states.

The contents of \sect{two_three} were published in \cite{Aulbach10},
and most of the results presented in \sect{extremal_point} and
\ref{analytic} were published in \cite{Aulbach10,Aulbach10lncs}.

\subsubsection*{Chapter 4: Maximally Entangled Symmetric States}

In \chap{solutions} the conjectured maximally entangled symmetric
quantum states of up to $12$ qubits in terms of the geometric measure
of entanglement are derived by a combination of numerical and
analytical methods.  First, the methodology employed for the search is
outlined in \sect{methodology}, and then a comprehensive discussion of
all the solutions, accompanied with visualisations, is given in
\sect{maxent_results}.  Along the way the obtained solutions are
compared to those of the classical point distributions of \toth and
Thomson.  In \sect{discussion} the results obtained are summarised and
interpreted from various points of view, such as entanglement scaling,
positive versus general states, operational implications and
distribution patterns in the Majorana representation.

Parts of \sect{methodology} and \ref{discussion} were published in
\cite{Aulbach10}, and the majority of the results presented in
\sect{maxent_results} were published in
\cite{Aulbach10,Aulbach10lncs}.

\subsubsection*{Chapter 5: Classification of Symmetric Entanglement}

While the preceding chapter focused on the quantitative
characterisation of the entanglement of symmetric states,
\chap{classification} shifts the focus towards qualitative aspects.
Three different entanglement classification schemes, namely \ac{LOCC},
\ac{SLOCC} and the Degeneracy Configuration, are reviewed for
symmetric states in \sect{overview_entclass}. It is found that the
\mob transformations of complex analysis, reviewed in \sect{mobius},
accurately describe \ac{SLOCC} transformations between symmetric
states, and that they provide a straightforward visualisation of the
innate \ac{SLOCC} freedoms.  The insights gained from this
relationship motivate the subsequent sections.  In
\sect{representative} representative states with simple Majorana
representations are derived for all symmetric \ac{SLOCC} classes of up
to 5 qubits, and in \sect{families} the results gathered for the 4
qubit case are put into relation to the concept of Entanglement
Families introduced in \cite{Verstraete02}.  In
\sect{determiningslocc} examples are given how known properties of the
\mob transformations can be of practical value to determine whether
two symmetric states are \ac{SLOCC}-equivalent or not, and in
\sect{symm_inv} a connection is made between \ac{SLOCC} invariants of
4 qubit symmetric states and areas on the Majorana sphere.  Finally,
\sect{global_ent} compares the maximally entangled symmetric states in
terms of the geometric measure with the extremal states of so-called
\quo{global entanglement measures}, such as those that detect
\quo{genuine} $n$-party entanglement.

The results presented in this chapter have not been published yet, but
most of the contents of \sect{overview_entclass} through
\sect{determiningslocc} can be found in the preprint \cite{Aulbach11}.

\subsubsection*{Chapter 6: Links and Connections}

In \chap{connections} several smaller findings are outlined.  First,
in \sect{anticoherent_queens} our results about maximally entangled
symmetric $n$ qubit states are compared to two different concepts of
\quo{maximally non-classical} spin-$\frac{n}{2}$ states, namely the
\quo{anticoherent} spin states \cite{Zimba06} and the \quo{queens of
  quantum} \cite{Giraud10}.  In \sect{dualpoly} a quantum analogue to
the concept of the Platonic duals from classical geometry is
unearthed, and in \sect{LMG} the ground states of the \ac{LMG} model
\cite{Lipkin65,Lipkin65b,Lipkin65c}, a spin model, are discussed and
investigated in light of the Majorana representation.

The topic of \sect{anticoherent_queens} was briefly touched on in
\cite{Aulbach10} and presented in detail in \cite{Aulbach10lncs}.  The
results of \sect{dualpoly} were published in \cite{Aulbach10lncs}.

\subsubsection*{Chapter 7: Conclusions}

The thesis concludes with \chap{conclusion}. First a summary of the
main results obtained in the previous chapters is given in
\sect{summary_results}.  This is followed in \sect{outlook} by an
outlook on some open questions, as well as new ideas or research
directions that are worthy of being tracked further.

\cleardoublepage

\chapter{Geometric Measure of Entanglement}
\label{geometric_measure}

\begin{quotation}
  The first research chapter starts with an introduction to the
  geometric measure of entanglement, an entanglement measure
  particularly suited for multipartite states.  The properties of this
  measure are analysed for a variety of systems, starting with
  arbitrary finite-dimensional multipartite systems, and then becoming
  more specific by considering $n$ qubit systems, positive states, and
  symmetric states.

  Among the results found is the observation that in general the
  maximally entangled states are expected to have a large number of
  closest product states, and that positive states are less entangled
  than non-positive states. A new proof with the advantage of a
  straightforward geometric interpretation is found for the upper
  bound on maximal symmetric $n$ qubit entanglement, and arguments are
  brought forward that symmetric quantum states cannot be used as
  resources for \acf{MBQC}, even in the setting of approximate
  \ac{MBQC}.
\end{quotation}

\section{Introduction and motivation}\label{gm_def}

The \acf{GM} is an entanglement measure which satisfies all the
desired properties of an entanglement monotone \cite{Wei03}.  It was
initially proposed for pure bipartite states by Shimony
\cite{Shimony95}, and was subsequently generalised by Barnum \etal
\cite{Barnum01} as well as Wei \etal \cite{Wei03}.  Unlike many other
entanglement measures, the \ac{GM} explicitly accommodates
multipartite systems. Such a holistic characterisation of many-body
entanglement instead of considering bipartite splits of the system
(e.g. by means of the concurrence of reduced density matrices) will be
particularly valuable for the analysis of symmetric states where no
part of the system is distinguished from any other.

Furthermore, many other entanglement measures, such as the relative
entropy of entanglement \cite{Plenio01,Vedral97,Vedral98}, are
notoriously difficult to compute in the multipartite setting even for
pure states, in part because of the absence of the Schmidt
decomposition.  In contrast to this, the \ac{GM} allows for a
comparatively easy calculation, because the variational problem runs
only over pure product states.  It will be seen that for symmetric
states the computational complexity is further reduced.

The \ac{GM} has found applications in several fields, including signal
processing, particularly in the fields of multi-way data analysis,
high order statistics and \ac{ICA}, where it is known under the name
\emph{rank one approximation to high order tensors}
\cite{Lathauwer00,Zhang01,Kofidis02,Wang09,Ni07,Silva08}.  In the area
of quantum phase transitions the \ac{GM} has been used to analyse the
\acf{LMG} model \cite{Orus08} as well as other spin models
\cite{Orus08b,Wei05,Nakata09}.  The survival of entanglement in
thermal states was studied with the \ac{GM} \cite{Markham08}, and in
quantum information theory the measure has been employed to derive the
generalised Schmidt decomposition of Carteret \etal \cite{Carteret00}
and for the study of entanglement witnesses \cite{Wei03,Hayashi08}.
On top of this, the measure has a variety of operational
interpretations, including the usability of initial states for
Grover's algorithm \cite{Biham02,Shimoni04}, additivity of channel
capacities \cite{Werner02} and classification of states as resources
for \ac{MBQC} \cite{Gross09,Mora10,Nest07}.  In state discrimination
under \ac{LOCC} the role of entanglement in blocking the ability to
access information locally is strictly monotonic -- the higher the
geometric entanglement, the harder it is to access information locally
\cite{Hayashi06}. The reverse does not hold, i.e. less entanglement
does not necessarily make discrimination easier.

The \ac{GM} is a distance-like entanglement measure, which means that
it assesses the entanglement in terms of the \quo{remoteness} of the
given state from the set of separable states.  In the case of the
\ac{GM} this remoteness is expressed by the maximal overlap of a given
pure multipartite state $\ket{\psi}$ with all pure product states
\cite{Shimony95,Wei03,Barnum01}, which can also be defined as the
geodesic distance with respect to the Fubini-Study metric
\cite{Brody01}.  Here we present the \ac{GM} in the inverse
logarithmic form\footnote{There are different definitions of the
  geometric measure in the scientific literature, with the two most
  common ones being $\EG ( \ket{\psi} ) = 1 -
  \abs{\bracket{\psi}{\Lambda}}^2$, as defined in
  \protect\cite{Wei03}, and $\Eg ( \ket{\psi} ) = - \log_2
  \abs{\bracket{\psi}{\Lambda}}^2$, introduced in
  \protect\cite{Wei04}.  With the exception of
  \protect\sect{resources_for_mbqc}, where $\EG$ is more useful for
  comparison with the literature, we will use $\Eg$ throughout this
  thesis.}, because this allows for an easier comparison with related
entanglement measures and because it has stronger operational
implications e.g. for channel capacity additivity \cite{Werner02} or
the (strong) additivity \cite{Wei03,Zhu10,Chen11}.
\begin{equation}\label{geo_def}
  \Eg (\ket{\psi} ) = \min_{\ket{\lambda} \in
    \mathh_{\text{SEP}} } - \log_2 
  \abs{ \bracket{\psi}{\lambda} }^2 = - \log_2 
  \abs{ \bracket{\psi}{\Lambda} }^2 \ens .
\end{equation}
This entanglement measure satisfies Axioms 1 to 4 introduced in
\sect{entanglement_measures}, and additionally the values of $\Eg$ are
strictly positive for all entangled states. Although not additive in
general, it is known that for some classes of states this measure is
additive or even strongly additive.  The definition of the \ac{GM} can
be viewed as an optimisation problem in the sense that one looks for
the best approximation of an entangled state $\ket{\psi}$ by a product
state $\ket{\lambda}$, i.e. a state with zero entanglement.  The
product state which has maximal overlap with $\ket{\psi}$ is denoted
by $\ket{\Lambda} \in \mathh_{\text{SEP}}$, and will be referred to as
the \textbf{\ac{CPS}}.  It should be noted that a given $\ket{\psi}$
can have more than one \ac{CPS}.  Indeed, it will follow from
\theoref{numberofcps} that some states are likely to have a large
number of distinct \acp{CPS}.

For bipartite systems the optimisation problem \eqref{geo_def} is
trivial if the given state $\ket{\psi}$ is provided in its Schmidt
decomposition \eqref{schmidt_decomp}, because $\ket{00}$ is a \ac{CPS}
\cite{Carteret00,Hilling10}, yielding the geometric entanglement $\Eg
( \ket{\psi} ) = - \log_2 \alpha_{0}^{2}$. For the maximally entangled
two qudit states \eqref{max_ent_states} this gives $\Eg ( \ket{\Psi} )
= \log_2 d$.

Although defined for pure states, the \ac{GM} can be extended to mixed
states by means of a convex roof construction \cite{Plenio07},
\begin{equation}\label{geo_mixed}
  \Eg ( \rho ) = \min_{ \{ p_{i}, \ket{\psi_{i}} \} } \sum_{i} p_{i}
  \Eg \left( \ket{\psi_{i}} \right) \ens ,
\end{equation}
over all decompositions of $\rho$ into pure states $\rho = \sum_{i}
p_{i} \pure{\psi_{i}}$.  This minor deficiency of the \ac{GM} -- the
absence of a generic definition for mixed states -- does not need to
concern us, because we will focus on the entanglement of pure
states\footnote{Pure states usually carry more entanglement than mixed
  states, and it is believed that the maximally entangled states can
  be found among pure states.  At least for the subset of symmetric
  states the search for the maximally entangled state in terms of the
  \ac{GM} can be restricted to pure states, because the maximally
  entangled symmetric state is pure \protect\cite{Martin10}.}.

Due to its compactness, the pure Hilbert space of a finite-dimensional
system (e.g. $n$ qudits) always contains at least one maximally
entangled state $\ket{\Psi}$ with respect to the \ac{GM}, and to each
such state relates at least one \ac{CPS}.  The task of determining
maximal entanglement can therefore be formulated as a max-min problem,
with the two extrema not necessarily being unambiguous:
\begin{equation}\label{geo_max_min}
  \begin{split}
    \Eg^{\text{max}}& = \max_{\ket{\psi} \in \mathh}
    \min_{\ket{\lambda} \in \mathh_{\text{SEP}} }
    - \log_2  \abs{\bracket{\psi}{\lambda}}^2 \\
    {}& = \max_{\ket{\psi} \in \mathh} - \log_2
    \abs{\bracket{\psi}{\Lambda ( \psi )}}^2 = - \log_2
    \abs{\bracket{\Psi}{\Lambda ( \Psi ) }}^2 \ens .
  \end{split}
\end{equation}
Werner \etal \cite{Werner02} have defined the function $G(\ket{\psi} )
= \max\limits_{\ket{\lambda} \in \mathh_{\text{SEP}}}
\abs{\bracket{\psi}{\lambda}} = \abs{\bracket{\psi}{\Lambda}}$ as the
\textbf{injective tensor norm}, a quantity that is known as the
maximal probability of success in Grover's search algorithm
\cite{Grover96}, and which has been used to define an operational
entanglement measure, the Groverian entanglement\footnote{The
  Groverian measure is in fact identical to $\EG = 1 -
  \abs{\bracket{\psi}{\Lambda}}^2$, up to a square operation.}
\cite{Biham02,Shimoni04}.  Note that $G^2$ is simply the fidelity
between the states $\ket{\psi}$ and $\ket{\Lambda}$, so $\Eg$ can be
viewed as the negative logarithm of a fidelity
\cite{Uhlmann76,Jozsa94,Nielsen}.  Because of the relationship $\Eg =
- \log_2 G^2$, and because $f(x) = - \log x^2$ is a strictly monotonic
function, the task of finding the maximally entangled state is
equivalent to solving the min-max problem
\begin{equation}\label{minmax}
  \min_{\ket{\psi} \in \mathh} G(\ket{\psi}) =
  \min_{\ket{\psi} \in \mathh}
  \max_{\ket{\lambda} \in \mathh_{\text{SEP}} }
  \abs{\bracket{\psi}{\lambda}} \ens .
\end{equation}

The geometric measure $\Eg$ has close links to other distance-like
entanglement measures, namely the relative entropy of entanglement
$E_{\text{R}}$ \cite{Vedral98,Vedral97} and the logarithmic robustness
of entanglement $E_{\text{Rob}} = \log_2 (1+ R)$, where $R$ is the
usual global robustness of entanglement \cite{Cavalcanti06,Vidal99}.
Between these measures the inequalities
\begin{equation}\label{meas_pure_ineq}
  \Eg ( \ket{\psi} ) \leq E_{\text{R}} ( \ket{\psi} ) \leq
  E_{\text{Rob}} ( \ket{\psi} )
\end{equation}
hold for all pure states
\cite{Wei04,Hayashi06,Hayashi08,Cavalcanti06}. These inequalities do
not hold for mixed states\footnote{A counterexample is the Smolin
  state \cite{Smolin01}, a bound entangled mixed positive symmetric
  state, which has $\Eg = 3$ \cite{Wei04b,Zhu10}, but $E_{\text{R}} =
  E_{\text{Rob}} = 1$ \cite{Murao01,Zhu10}.  Its von Neumann entropy
  is $S = 2$, yielding $\Egt = \Eg - S = 1$ \cite{Zhu10}.}, but a
generalisation is possible by defining $\Egt ( \rho ) := \Eg ( \rho )
- S ( \rho )$, where $S ( \rho ) = - \Trace ( \rho \log \rho )$ is the
von Neumann entropy, which is zero for all pure states:
\begin{equation}\label{meas_ineq}
  \Egt ( \rho ) \leq E_{\text{R}} ( \rho ) \leq
  E_{\text{Rob}} ( \rho ) \ens .
\end{equation}
For pure states the relationship \eqref{meas_pure_ineq} implies that
the \ac{GM} is a lower bound for both the relative entropy of
entanglement and the logarithmic robustness of entanglement.  For
stabiliser states (e.g. \ac{GHZ} state), Dicke states (e.g. W state),
\permantisymm basis states \cite{Hayashi06,Hayashi08,Markham07} and
symmetric states with \emph{totally invariant} \ac{MP} distributions
\cite{Markham11} (which will be discussed in
\sect{invariant_and_additivity}) the three distance-like entanglement
measures coincide:
\begin{equation}\label{meas_eq}
  \Eg = E_{\text{R}} = E_{\text{Rob}} \ens .
\end{equation}
This equivalence is intriguing because the three measures have
different interpretations.  As an entropic quantity, $E_{\text{R}}$
has information theoretic implications, while $E_{\text{Rob}}$
measures the resistance of entanglement against arbitrary noise.

Next we consider the geometric entanglement of the two paradigmatic
$n$ qubit states of \eq{ghz_w_def}, the \ac{GHZ} state and W
state. The set of their \acp{CPS} are
\begin{align}
  \ket{\Lambda_{\text{GHZ}}}& = \{ \ket{00 \ldots 00} ,
  \ket{11 \ldots 11} \} \ens , \label{ghz_w_cps_1} \\
  \ket{\Lambda_{\text{W}}}& = \big\{ \big(
  \tfrac{\sqrt{n-1}}{\sqrt{n}} \ket{0} + \E^{\I \varphi}
  \tfrac{1}{\sqrt{n}} \ket{1} \big)^{\otimes n} \: \vert \: \varphi
  \in [0, 2 \pi ) \big\} \ens , \label{ghz_w_cps_2}
\end{align}
From this it can be seen that the \ac{GHZ} state has two different
\acp{CPS}, while the W state has a one-parametric continuum of
\acp{CPS}.  The amount of geometric entanglement follows as
\begin{align}\label{ghz_w_ent}
  \Eg (\ket{\text{GHZ}_{n}})& = 1 \ens , \\
  \Eg (\ket{\text{W}_{n}})& = \log_2 \big( \tfra{n}{n-1} \big)^{n-1}
  \ens .
\end{align}
For the \ac{GHZ} state the amount of geometric entanglement is 1,
regardless of the number of qubits. On the other hand, the
entanglement of the W state goes asymptotically towards $\log_2 (e)$
as $n \to \infty$. For $n \geq 3$ the \ac{GHZ} state has less
geometric entanglement than the W state, a property not exhibited by
many other entanglement measures.

Next we will briefly review the known upper and lower bounds on the
maximal possible amount of geometric entanglement for $n$ qubit
states.  It should be kept in mind, however, that the maximally
entangled state and its amount of entanglement depends on the chosen
entanglement measure \cite{Plenio07}, and therefore different
entanglement measures may not only yield different values for the
maximal entanglement, but also different maximally entangled states.

For the general case of pure $n$ qubit states the upper bound $\Eg
(\ket{\psi}) \leq n-1$ on the geometric entanglement has been derived
in \cite{Jung08}.  Although no states of more than two qubits reach
this bound \cite{Jung08}, most $n$ qubit states come close.  For $n >
10$ qubits the inequality $\Eg > n - 2 \log_2 (n) - 3$ holds for
almost all states, something that makes the overwhelming majority of
states too entangled to be useful for \ac{MBQC} \cite{Gross09}.  A
similar result that holds for arbitrary dimensions of the parties was
derived by Zhu \etal \cite{Zhu10}, and for $n$ qubits their
Proposition 25 yields $\Eg > n - 2 \log_2 (n) - \log_2 (9 \ln 2)$.
Resources for \ac{MBQC} must be considerably less entangled than most
states (although this is by no means a sufficient criterion,
cf. Bremner \etal \cite{Bremner09}). For example, the entanglement of
2D cluster states consisting of $n$ qubits, a well-known \ac{MBQC}
resource, was found to be $\Eg = \tfra{n}{2}$ \cite{Markham07}.

\subsection{Symmetric  states}\label{gm_symmetric_states}

Here we will briefly review some known results about
permutation-symmetric states with respect to the \ac{GM}. Firstly, the
definition of the \ac{GM} \eqref{geo_def} suggests that the overlap of
a symmetric state $\psis$ with a product state will be maximal if the
product state is also symmetric.  This straightforward conjecture has
been actively investigated \cite{Wei04,Hayashi08}, but a proof is far
from trivial.  After some special cases were proven
\cite{Hayashi09,Wei10}, H\"{u}bener \etal \cite{Hubener09} were able
to give a proof for the general case of pure symmetric
states\footnote{One could ask whether this result also holds for
  translationally invariant states (which appear in spin models), but
  this is not the case. A trivial counterexample is the state
  $\ket{\psi} = \tfrac{1}{\sqrt{2}} \left( \ket{0101} + \ket{1010}
  \right)$, which is \ac{LU}-equivalent to the \ac{GHZ} state and
  which has the two non-symmetric closest product states $\ket{0101}$
  and $\ket{1010}$ \cite{Hubener09}.}.  They showed that for $n \geq
3$ qudits the \acp{CPS} of a pure symmetric state are
\emph{necessarily} symmetric, thus greatly reducing the complexity of
finding the \acp{CPS} and the entanglement of symmetric states. A
generalisation of this result to mixed symmetric states was recently
achieved by Zhu \etal \cite{Zhu10}.  Pure symmetric product states of
$n$ qubits can be written as $\ket{\Lambda^{\text{s}}} =
\ket{\sigma}^{\otimes n}$ with only one single-qubit state
$\ket{\sigma} \in \mbbc^{2}$.  Therefore every \ac{CPS}
$\ket{\Lambda^{\text{s}}} = \ket{\sigma}^{\otimes n}$ of a multi-qubit
symmetric state $\psis$ can be visualised on the Majorana sphere by
the Bloch vector of $\ket{\sigma}$, and in analogy to the \acp{MP} we
refer to $\ket{\sigma}$ as a \textbf{\ac{CPP}} of $\psis$.

For positive symmetric states, i.e. states that are symmetric as well
as positive, it is known that they have at least one \ac{CPS} that is
positive symmetric itself \cite{Hayashi09,Wei10}.  However, while each
\ac{CPS} of a positive symmetric state is necessarily symmetric for $n
\geq 3$ qudits \cite{Hubener09}, it need not be positive, and
counterexamples for this will appear in \chap{solutions}.

Upper and lower bounds on the maximal geometric entanglement of $n$
qubit states were already reviewed, with the observation that $\Eg$
scales linearly with $n$.  We will now look at the same question for
symmetric states, i.e. how does the entanglement of the maximally
entangled symmetric $n$ qubit state scale?

In order to derive a simple lower bound, consider the Dicke states
introduced in \eq{dicke_def}.  For a given Dicke state $\sym{n,k}$
with $0 \leq k \leq n$ it is known \cite{Wei03,Wei04} that any of the
states
\begin{equation}\label{dicke_cps}
  \ket{\Lambda} = \Big( \sqrt{ \tfrac{n-k}{n} } \ket{0} +
  \E^{\I \varphi}  \sqrt{ \tfrac{k}{n} } \ket{1} \Big)^{\otimes n}
  \ens ,
\end{equation}
with $\varphi \in [0, 2 \pi )$, is a \ac{CPS}.  With this the
geometric entanglement of $\sym{n,k}$ can be calculated to be
\begin{equation}\label{dicke_ent}
  \Eg ( \sym{{n,k}} ) = \log_2 \left( \frac{
      \big( \tfrac{n}{k} \big)^k \big( \tfrac{n}{n-k} \big)^{n-k}}
    {\binom{n}{k}} \right) \ens .
\end{equation}
From this formula it can be seen that the maximally entangled Dicke
state is $\sym{{n,\frac{n}{2}}}$ for even $n$ and the two equivalent
states $\sym{{n, \lfloor \frac{n}{2} \rfloor }}$ and $\sym{{n, \lceil
    \frac{n}{2} \rceil }}$ for odd $n$.  Using the Stirling
approximation $n ! \sim \sqrt{2 \pi n} (\frac{n}{e})^{n}$, the
asymptotic amount of entanglement of the maximally entangled $n$ qubit
Dicke state for large $n$ is found to be
\begin{equation}\label{dicke_ent_scaling}
  \Eg^{\text{Dicke}} \approx \log_2 \sqrt{\tfrac{n \pi}{2}} \ens .
\end{equation}
In general the maximally entangled symmetric state of $n$ qubits is a
superposition of Dicke states, so \eq{dicke_ent_scaling} is a lower
bound on the maximal symmetric entanglement.

An upper bound on the \ac{GM} for symmetric $n$ qubit states has been
derived from the separable decomposition of the identity on the
symmetric subspace (denoted $\one_{\text{Symm}}$), see
e.g. \cite{Renner},
\begin{equation}\label{sep_decomposition}
  \int\limits_{\mathcal{S}^{2}}
  (  \pure{\theta} )^{\otimes n}\omega(\theta)
  = \frac{1}{n+1}\one_{\text{Symm}} \ens ,
\end{equation}
where $\omega$ denotes the uniform probability measure over the unit
sphere $\mathcal{S}^2$ of normalised single qubit vectors.  It is easy
to see that $G (\ket{\psi})^2 = \max\limits_{\omega \in
  \mathh_{\text{SEP}} } \text{Tr} (\omega \pure{\psi}) \geq
\tfrac{1}{n+1}$. Hence, the entanglement of a symmetric $n$ qubit
state $\psis$ is bounded from above by
\begin{equation}\label{upper_bound_eq}
  \Eg ( \psis ) \leq \log_2 (n+1) \ens .
\end{equation}
From \eq{dicke_ent_scaling} and \eqref{upper_bound_eq} one can see
that the maximal symmetric entanglement scales logarithmically with
the number of qubits.  This is a qualitative departure from the linear
scaling behaviour observed in general $n$ qubit states.

\section{Results for general states}\label{gen_results}

\subsection{Closest product states of the maximally entangled
  state}\label{cps_number}

We will now show that for systems with arbitrary dimensions and an
arbitrary number of parties the maximally entangled states can be cast
as superpositions of their \acp{CPS}.  In other words, if $\ket{\Psi}$
is maximally entangled, then the span of its \acp{CPS} contains
$\ket{\Psi}$ itself.  Furthermore, $\ket{\Psi}$ has at least two
linearly independent \acp{CPS}.  These results are obtained without
any knowledge about which states are the maximally entangled ones, and
the set of \acp{CPS} itself does not form a vector space in general,
because linear combinations of product states do not need to be
product states themselves.  The idea of the proof is that for any
state not lying in the span of its \acp{CPS} it is possible to find an
explicit variation which increases the geometric entanglement of the
state.  The main ingredient of the proof is the multipartite Schmidt
decomposition of Carteret \etal \cite{Carteret00} which was already
introduced in \sect{multipartite}.
\begin{theorem}\label{numberofcps}
  Let $\ket{\Psi} \in \mathh = \mathh_{1} \otimes \cdots \otimes
  \mathh_{n}$ be a normalised pure state of an $n$-partite system with
  finite-dimensional subspaces $\dim ( \mathh_{i} ) = d_{i} \geq 2$,
  and let $\Lambda \subset \mathh$ be the set of \acp{CPS} of
  $\ket{\Psi}$.  If $\ket{\Psi}$ is maximally entangled with respect
  to the \ac{GM}, then $\ket{\Psi} \in \spa ( \Lambda )$ and there
  exist at least two linearly independent \acp{CPS}.
\end{theorem}

\begin{proof}
  Let us assume that $\ket{\Psi}$ is maximally entangled, but
  $\ket{\Psi} \notin U := \spa ( \Lambda )$. This implies $U \neq
  \mathh$, and one can use the orthogonal complement $V := U^{\perp} =
  \{ \ket{v} \in \mathh : \bracket{v}{u} = 0 \,\, \forall \, u \in U
  \}$ of $U$, with $0 < \dim ( V ) < \dim ( \mathh )$, to write
  $\mathh$ as an internal direct sum of two complex vector spaces:
  $\mathh = U \oplus V$.  Because of $\ket{\Psi} \notin U$ there
  exists a $\ket{\zeta} \in V$ so that $\bracket{\Psi}{\zeta} \neq
  0$. We can then define the variation
  \begin{equation}\label{variation}
    \ket{\psi ( \epsilon )} := (1 - \epsilon) \ket{\Psi} +
    \epsilon \ket{\xi} \ens , \quad \text{with } \:
    \epsilon > 0 \ens \text{and } \: \ket{\xi} :=
    \tfrac{2}{\bracket{\Psi}{\zeta} +
      \bracket{\zeta}{\Psi} } \ket{\zeta} \ens .
  \end{equation}
  Obviously $\lim\limits_{\epsilon \to 0} \ket{\psi ( \epsilon )} =
  \ket{\Psi}$, and $\bracket{\psi ( \epsilon )}{\psi ( \epsilon )} = 1
  + \Order{( \epsilon^{2} )}$. In the following $\ket{\psi ( \epsilon
    )}$ can be considered to be normalised, because second order
  variations play no role in subsequent calculations and can thus be
  ignored.  Since $\ket{\xi} \propto \ket{\zeta}$ it follows that
  $\ket{\xi} \in V$ and thus $\bracket{\xi}{\Lambda_{i}} =0$ for all
  $i$.  Writing $f( \epsilon ) := \max\limits_{\ket{\lambda} \in
    \text{SEP}} \abs{\bracket{\psi ( \epsilon )}{\lambda}}$, we will
  show that $f( \epsilon ) < f (0)$ for sufficiently small, but
  nonzero $\epsilon$ and therefore $\ket{\Psi}$ cannot be maximally
  entangled.  Because we consider infinitesimal variations, it
  suffices to investigate $g_{\epsilon} ( \lambda ) :=
  \abs{\bracket{\psi ( \epsilon )}{\lambda}}$ near the global maxima
  $\ket{\lambda} = \ket{\Lambda_{i}}$ of $g_{0} ( \lambda )$.  It will
  turn out that the value of $g_{\epsilon}$ consistently decreases in
  the neighbourhood of each $\ket{\Lambda_{i}}$ as $\epsilon$ is
  turned on.  Note that the value of $g_{\epsilon} ( \lambda )$ may
  increase near its non-global maxima, but this is not of concern to
  us, because the variation can be chosen sufficiently small, as seen
  in \fig{variations_diagram}.
  
  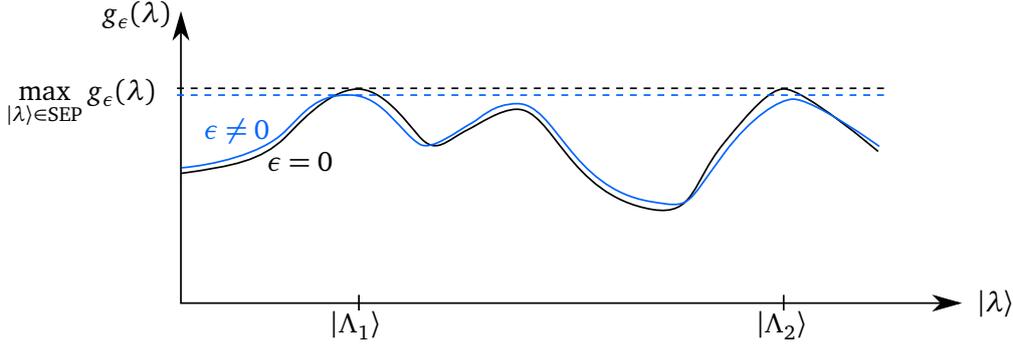
\begin{figure}[ht]
    \centering
    \begin{overpic}[scale=.8]{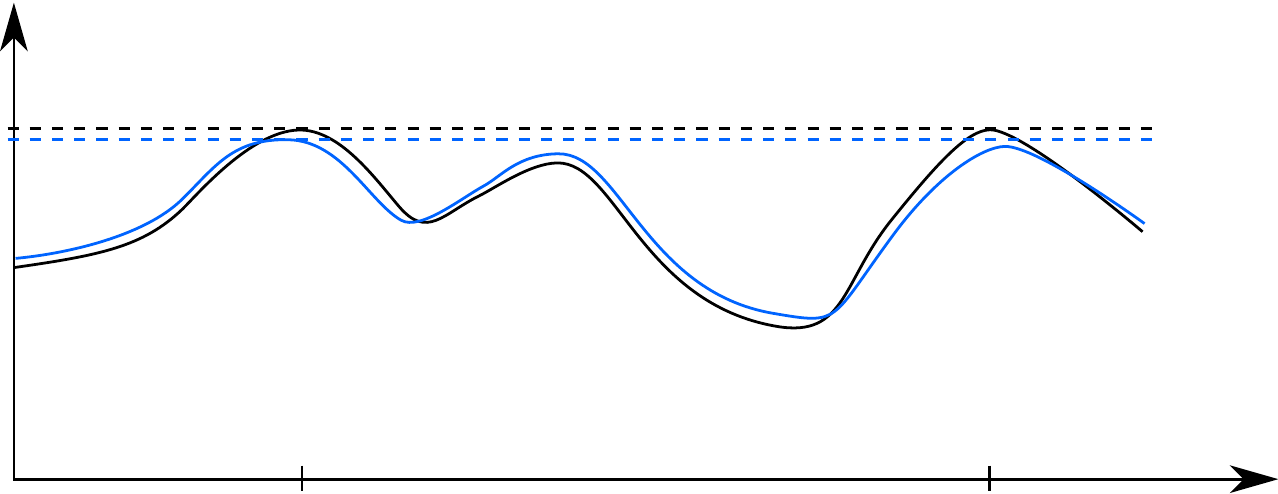}
      \put(102,0){$\ket{\lambda}$}
      \put(20,-2.7){$\ket{\Lambda_{1}}$}
      \put(74,-2.7){$\ket{\Lambda_{2}}$}
      \put(-9,37){$g_{\epsilon} ( \lambda )$}
      \put(-21,27){$\max\limits_{\ket{\lambda} \in \text{SEP}}
        g_{\epsilon} ( \lambda )$}
      \put(12,18){$\epsilon = 0$}
      \put(4,22){\blue{$\epsilon \neq 0$}}
    \end{overpic}
    \caption[Infinitesimal variations of the amplitude
    function]{\label{variations_diagram} Schematic representation of
      the change in $g_{\epsilon} ( \lambda )$ as $\epsilon$ is turned
      on.  For sufficiently small $\epsilon$ only the areas near the
      \acp{CPS} of $\ket{\Psi}$ need to be considered in order to
      determine the largest value of $g_{\epsilon} ( \lambda )$. The
      blue curve representing $\epsilon \neq 0$ attains only one
      global maximum, which lies in the vicinity of
      $\ket{\Lambda_{1}}$.}
  \end{figure}
  
  In the following we will choose an arbitrary $\ket{\Lambda_{i}}$ --
  denoted as $\ket{\Lambda}$ -- and show that $g_{\epsilon} ( \lambda
  ) = \abs{\bracket{\psi ( \epsilon )}{\lambda}}$ decreases near
  $\ket{\lambda} = \ket{\Lambda}$. This procedure can be performed for
  each $\ket{\Lambda_{i}}$, thus proving that $\ket{\psi ( \epsilon
    )}$ is more entangled than $\ket{\Psi}$. Note that even though the
  following calculations rely on a basis that depends on the chosen
  $\ket{\Lambda_{i}}$, the variation $\ket{\psi ( \epsilon )}$ of
  \eq{variation} is independent of any basis, and thus $\ket{\psi (
    \epsilon )}$ is the same for each $\ket{\Lambda_{i}}$.
  
  In the proof of Theorem 2 of \cite{Carteret00} the factorisable
  orthonormal basis was chosen in a way so that the state
  $\ket{\lambda} = \ket{00 \cdots 00}$ is a maximum of the overlap
  function $g( \lambda ) = \abs{\bracket{\psi}{\lambda}}$. Since the
  choice of this maximum is arbitrary, this means that there exists a
  basis so that $\ket{\Lambda} = \ket{00 \cdots 00}$ is a \ac{CPS},
  and that the coefficients $a_{ i_{1} , \ldots , i_{n} }$ of the
  state $\ket{\Psi}$ (cf. \eq{multipartite_state}) satisfy the
  conditions outlined in Theorem 2 of \cite{Carteret00}. In
  particular, $\bracket{\Psi}{\Lambda} = a_{00 \cdots 00}$, and the
  following special case of \eq{generalised_schmidt_zero} holds:
  \begin{equation}\label{gen_schmidt}
    a_{\underbrace{\scriptstyle 00
        \ldots 0k}_{j \text{ indices}} 0 \ldots 00} = 0 \qquad
    \forall \, 1 \leq j \leq n \ens \forall \, 1 \leq k
    \leq d_{j} - 1 \ens .
  \end{equation}
  Arbitrary variations of $\ket{\Lambda} = \ket{00 \cdots 00}$ can be
  defined as follows:
  \begin{equation}\label{variation_delta}
    \begin{split}
      \ket{\lambda ( \bmr{\delta} )}& = \ket{\delta^{1}} \otimes
      \ket{\delta^{2}} \otimes \ldots \otimes \ket{\delta^{n}}
      \ens , \quad \text{with} \\
      \ket{\delta^{j}}& = (1- \delta_{0}^{j}) \ket{0} + \delta_{1}^{j}
      \ket{1} + \ldots + \delta_{d_{j} -1}^{j}
      \ket{d_{j} - 1} \quad \forall \, j \ens .
    \end{split}
  \end{equation}
  Here the $\delta_{i}^{j}$ are small complex-valued variations that
  are independent from each other with the only restriction being the
  $n$ normalisation conditions $\bracket{\delta^{j}}{\delta^{j}} =
  1$. The variation $\ket{\lambda ( \bmr{\delta} )}$ remains a product
  state satisfying $\lim\limits_{\bmr{\delta} \to 0} \ket{\lambda (
    \bmr{\delta} )} = \ket{\Lambda}$ as well as $\bracket{\lambda (
    \bmr{\delta} )}{\lambda ( \bmr{\delta} )} = 1$.
  
  We will proceed to show that $\abs{\bracket{\psi ( \epsilon
      )}{\lambda ( \bmr{\delta} )}} < \abs{\bracket{\Psi}{\Lambda}}$
  in the entire neighbourhood of $\bmr{\delta} = \bmr{0}$ for small
  but nonzero values of $\epsilon$. For this purpose we can ignore any
  terms of order $\Order ( \epsilon^{2} )$, $\Order ( \bmr{\delta}^2
  )$, $\Order ( \epsilon \bmr{\delta} )$, and higher.  From
  \eq{variation}, \eqref{gen_schmidt}, \eqref{variation_delta} and
  $\bracket{\xi}{\Lambda} = \bracket{\xi}{00 \cdots 00} = 0$ it
  follows that
  \begin{subequations}\label{variation_equations}
    \begin{align}
      \bracket{\psi ( \epsilon )}{\lambda ( \bmr{\delta} )}& = (1
      - \epsilon ) \bracket{\Psi}{\lambda ( \bmr{\delta} )} +
      \epsilon \bracket{\xi}{\lambda ( \bmr{\delta} )} \label{var_eq1} \\
      {}& = (1- \epsilon ) \left[ a_{00 \ldots 00} - \left( \sum_{j=1}^{n}
          \sum_{i=1}^{d_{j} -1} \delta_{i}^{j}
          a_{\underbrace{\scriptstyle 00 \ldots 0i}_{j \text{ indices}}
            0 \ldots 00}
        \right) + \Order ( \bmr{\delta}^{2} ) \right] \nonumber \\
      {}& \quad + \epsilon \bra{\xi} \left[ (1 - \delta_{0}^{1} ) \ket{0}
        \otimes \ldots \otimes (1 - \delta_{0}^{n} ) \ket{0} \right]
      + \Order ( \epsilon \bmr{\delta} ) \label{var_eq2} \\
      {}& \approx (1- \epsilon ) a_{00 \ldots 00} + \epsilon
      \bracket{\xi}{00 \cdots 00} = (1- \epsilon )
      \bracket{\Psi}{\Lambda} \ens , \label{var_eq3}
    \end{align}
  \end{subequations}
  and therefore $\abs{\bracket{\psi ( \epsilon )}{\lambda (
      \bmr{\delta} )}} \approx (1 - \epsilon)
  \abs{\bracket{\Psi}{\Lambda}} < \abs{\bracket{\Psi}{\Lambda}}$.
  
  The existence of at least two linearly independent \acp{CPS} for the
  maximally entangled state $\ket{\Psi}$ immediately follows from the
  observation that the span of a single product state cannot contain
  entangled states.
\end{proof}
For the special case of qubit systems ($d_{1} = \ldots = d_{n} = 2$)
\eq{gen_schmidt} directly follows from \theoref{coeff_theo} (which
will be introduced in the following section) without the need to
invoke the generalised Schmidt decomposition of Carteret \etal
\cite{Carteret00}.  Furthermore, it is straightforward to adapt
\theoref{numberofcps} to the case of symmetric states. This will be
done in \corref{numberofcpp} in \sect{symm_results}.

Given a maximally entangled state, is the set of distinct \acp{CPS}
discrete or continuous?  And if it is continuous, can it be
parameterised in some way?  Tamaryan \etal
\cite{Tamaryan08,Tamaryan10} noticed that some highly entangled $n$
qubit W-type states have a continuous one-parametric range of closest
separable states, and for the W-states themselves this continuous
range was already given in \eq{ghz_w_cps_2}.  For three qubits
$\ket{\text{W}_{3}}$ is known to be the maximally entangled state
\cite{Chen10}, and with the parameterisation $\ket{\Lambda ( \varphi )
} = \big( \sqrt{\tfrac{2}{3}} \ket{0} + \E^{\I \varphi}
\sqrt{\tfrac{1}{3}} \ket{1} \big)^{\otimes 3}$ for its \acp{CPS}
\theoref{numberofcps} can be verified by considering the relation
$\ket{\text{W}_{3}} = \tfrac{3}{8} \big( \ket{\Lambda{(0)}} - \I
\ket{\Lambda{(\tfrac{\pi}{2})}} - \ket{\Lambda{( \pi )}} + \I
\ket{\Lambda{(\tfrac{3 \pi}{2})}} \big)$.  It remains interesting to
see whether continuous ranges of \acp{CPS} also exist for larger
number of particles. This question will be reviewed in
\sect{number_cpp} in light of results gained in later chapters.

Given a maximally entangled state $\ket{\Psi}$, how large is $U = \spa
( \Lambda )$?  \Theoref{numberofcps} only tells us that $\dim U \geq
2$.  For three qubits the maximally entangled state
$\ket{\text{W}_{3}}$ is symmetric, from which it follows that all its
\acp{CPS} must be symmetric \cite{Hubener09}, which implies $U \subset
\mathh_{\text{s}}$. Because the symmetric subspace $\mathh_{\text{s}}$
is strictly smaller than $\mathh$ (consider e.g. biseparable states),
$U$ is strictly smaller than $\mathh$.  Nevertheless, there is reason
to believe that $\dim U$ is in general high: For sufficiently large
$n$ qubit systems ($n > 10$) the maximal geometric entanglement scales
as $\Eg \approx n - \Order ( \log_2 (n) )$ or higher \cite{Gross09},
and the relationship $\Eg \leq \log_2 (r) + \Order (1)$ has been found
to hold with high probability for random states with tensor rank $r$
\cite{Bremner09}.  Although the maximally entangled state is by no
means a \quo{random state}, the property of the Schmidt measure $P =
\log_2 (r)$ being an entanglement measure \cite{Eisert01} makes it
reasonable to expect that maximally entangled states in terms of the
geometric measure have a tensor rank $r > 2^{n - \Order ( \log_2 (n)
  )}$.  Since $\ket{\Psi} \in \spa ( \Lambda )$, this means that every
expansion of $\ket{\Psi}$ in terms of linearly independent \acp{CPS}
consists of at least $r$ terms, which in turn implies the existence of
at least $r$ linearly independent \acp{CPS}.  Therefore we conjecture
that $\dim U > 2^{n - \Order ( \log_2 (n) )}$, which is close to $\dim
\mathh = 2^{n}$.

\subsection{Standard form of coefficients}\label{standard_form_coeff}

The generalised Schmidt decomposition of Carteret \etal
\cite{Carteret00} for multipartite states was already mentioned in the
introductory \sect{multipartite}.  Here I present a similar standard
form for the coefficients of $n$ qubit states that I derived in
collaboration with Seiji Miyashita and Mio Murao.  I was unaware of
the former work in \cite{Carteret00} while doing so, and there are
similarities between the two forms.  The following
\theoref{coeff_theo} can be understood as a special case of the
standard form in \cite{Carteret00} with weaker implications on the
coefficients. It is nevertheless interesting, because our proof is
different, and because we make the connection to the \ac{GM} more
explicit.

Consider an $n$ qubit state $\ket{\psi}$ written in the notation of
\eqref{multipartite_state}.  The state has at least one \ac{CPS}, and
by choosing the computational basis accordingly, we can set
$\ket{\Lambda} = \ket{00 \ldots 0}$ to be a \ac{CPS}.  The injective
tensor norm (which determines the amount of geometric entanglement) is
then $G(\ket{\psi}) = \abs{\bracket{\psi}{00 \ldots 0}} = \abs{a_{00
    \dotsc 0}}$, i.e. the amount of entanglement of $\ket{\psi}$ is
given by the first coefficient $a_{00 \dotsc 0}$. By means of the
global phase this coefficient can be taken to be positive.

\begin{theorem}\label{coeff_theo}
  For every pure $n$ qubit state $\ket{\psi}$ one can choose a
  computational basis with the notation of \eqref{multipartite_state}
  in which $\ket{\Lambda} = \ket{0}^{\otimes n}$ is a \ac{CPS}, the
  coefficient $a_{00 \dotsc 0}$ is positive, and the following
  conditions hold:
  \vspace{1mm}
  
  For 2 qubits:
  \begin{subequations}\label{2_qubit_coeff}
    \begin{gather}
      a_{10} = a_{01} = 0 \label{2qc_1} \ens , \\
      a_{00}^2 \geq \abs{a_{11}}^2 \ens . \label{2qc_2}
    \end{gather}
  \end{subequations}
  
  For 3 qubits:
  \begin{subequations}\label{3_qubit_coeff}
    \begin{gather}
      a_{100} = a_{010} = a_{001} = 0 \ens , \label{3qc_1} \\
      a_{000}^2 \geq \abs{a_{110}}^2 \ens , \label{3qc_2} \\
      a_{000}^2 \geq \abs{a_{101}}^2 \ens , \label{3qc_3} \\
      a_{000}^2 \geq \abs{a_{011}}^2 \ens , \label{3qc_4} \\
      a_{000}^3 - 2 \abs{ a_{110} a_{101} a_{011} }
      - a_{000} \big( \abs{a_{110}}^2
      + \abs{a_{101}}^2 + \abs{a_{011}}^2 \big) \geq 0 \ens .
      \label{3qc_5}
    \end{gather}
  \end{subequations}
  
  For $n$ qubits:
  \begin{subequations}\label{n_qubit_coeff}
    \begin{gather}
      a_{ \{ 1 \} } = 0 \ens , \label{nqc_1} \\
      a_{00 \ldots 0}^2 \geq \abs{a_{ \{ 11 \} }}^2
      \ens , \label{nqc_2} \\
      a_{00 \ldots 0}^3 - 2 \, \abs{ a_{ \{110\} } a_{ \{101\} }
        a_{ \{011\} } } -
      a_{00 \ldots 0} \big( a_{ \{110\} } + a_{ \{101\} }
      + a_{ \{011\} } \big) \geq 0 \ens , \label{nqc_3}
    \end{gather}
  \end{subequations}
  where $a_{ \{ 1 \} }$ stands for any of the $n$ coefficients $a_{10
    \ldots 0} \, , \, a_{01 \ldots 0} \, , \, \ldots \, , \, a_{0
    \ldots 01}$ ,
  
  \noindent $a_{ \{ 1 1 \} }$ stands for any of the $\binom{n}{2}$
  coefficients $ a_{110 \ldots 0} \, , \, a_{101 \ldots 0} \, , \,
  \ldots \, , \, a_{0 \ldots 011}$ ,
  
  \noindent and the tuple $\{ a_{ \{110\} } \, , \, a_{ \{101\} } \, ,
  \,a_{ \{011\} } \}$ can be any of the $\binom{n}{3}$ different
  tuples
  \begin{equation*}
    \begin{split}
      \{ a_{\bmr{110}0 \ldots 0}& \, , \, a_{\bmr{101}0 \ldots
        0} \, , \, a_{\bmr{011}0 \ldots 0} \} \ens , \\
      \{ a_{\bmr{11}0\bmr{0}0 \ldots 0}& \, , \,
      a_{\bmr{10}0\bmr{1}0 \ldots 0} \, , \,
      a_{\bmr{01}0\bmr{1}0 \ldots 0} \}
      \ens , \\
      {}& \vdots \\
      \{ a_{0 \ldots 0\mathbf{110}}& \, , \, a_{0 \ldots
        0\mathbf{101}} \, , \, a_{0 \ldots 0\mathbf{011}} \} \ens .
    \end{split}
  \end{equation*}
\end{theorem}

\begin{proof}
  The possibility of finding a computational basis in which
  $\ket{\Lambda} = \ket{0}^{\otimes n}$ is a \ac{CPS} and $a_{00
    \dotsc 0}$ is positive was already explained, so we only need to
  verify the conditions on the other coefficients.  For this we
  consider the first and second partial derivatives of the overlap
  function $g ( \lambda ) = \abs{\bracket{\psi}{\lambda}}$ around the
  point of the maximum $\ket{\Lambda} = \ket{0}^{\otimes n}$.
  
  We start with the 2 qubit case. Let $\ket{\psi} = ( a_{00} \ens
  a_{01} \ens a_{10} \ens a_{11})^{T}$ be the given state in an
  appropriate basis.  A general product state (up to a global phase)
  can be written as
  \begin{equation*}
    \ket{\lambda} =
    \begin{pmatrix}
      \sqrt{1 - b_1^2} \\
      b_1 \, \E^{\I \beta_1}
    \end{pmatrix}
    \otimes
    \begin{pmatrix}
      \sqrt{1 - b_2^2} \\
      b_2 \, \E^{\I \beta_2}
    \end{pmatrix}
    =
    \begin{pmatrix}
      \sqrt{1 - b_1^2}    \sqrt{1 - b_2^2} \\
      b_1 \sqrt{1 - b_2^2} \, \E^{\I \beta_1} \\
      b_2 \sqrt{1 - b_1^2} \, \E^{\I \beta_2} \\
      b_1 b_2 \, \E^{\I (\beta_1 + \beta_2 )}
    \end{pmatrix} \ens ,
  \end{equation*}
  with $b_1 , b_2 \in [0,1]$ and $\beta_1 , \beta_2 \in [0,2
  \pi)$. Expanding $b_1, b_2$ by a Taylor series around $\ket{\Lambda}
  = \ket{00}$ gives
  \begin{equation*}
    \ket{00} + \delta \ket{00} =
    \begin{pmatrix}
      1 - \frac{1}{2} ( \delta b_1^2 + \delta b_2^2 ) \\
      \delta b_1 \E^{\I \beta_1} \\
      \delta b_2 \E^{\I \beta_2} \\
      \delta b_1 \delta b_2 \E^{\I (\beta_1 + \beta_2 )}
    \end{pmatrix} \ens , \quad \text{and}
  \end{equation*}
  \begin{multline}
    \Abs{\bra{\psi} \big( \ket{00} + \delta \ket{00} \big) } =
    \big| a_{00} - \tfrac{1}{2} a_{00}(\delta b_1^2 + \delta b_2^2 ) \\
    + \cc{a}_{01} \delta b_1 \E^{\I \beta_1} + \cc{a}_{10} \delta b_2
    \E^{\I \beta_2} + \cc{a}_{11} \delta b_1 \delta b_2 \E^{\I
      (\beta_1 + \beta_2 )} \big| \ens .
    \label{variation_expr}
  \end{multline}
  $\abs{\bracket{\psi}{\lambda}}$ must have a maximum at
  $\ket{\Lambda} = \ket{00}$, so the first partial derivatives of
  \eq{variation_expr} with respect to $b_1$ and $b_2$ must be
  zero. This yields $a_{01} = a_{10} = 0$.  With the freely variable
  $\beta_1, \beta_2$ chosen s.t.  $\cc{a}_{11} \E^{\I (\beta_1 +
    \beta_2 )} \in \mbbr$, \eq{variation_expr} becomes
  \begin{equation*}
    a_{00} - \tfrac{1}{2} a_{00}(\delta b_1^2 + \delta b_2^2 ) +
    \abs{a_{11}} \delta b_1 \delta b_2 \ens .
  \end{equation*}
  The Hessian Matrix of the second partial derivatives with respect to
  $b_1$ and $b_2$ is then
  \begin{equation*}
    H =
    \begin{pmatrix}
      - a_{00}&     \abs{a_{11}} \\
      \abs{a_{11}}& - a_{00}
    \end{pmatrix} \ens .
  \end{equation*}
  At the maximum $\ket{\Lambda} = \ket{00}$ the Hessian Matrix must be
  negative semidefinite \cite{Horn}. This is equivalent to the
  conditions
  \begin{equation*}
    a_{00} > 0 \quad \text{and} \quad a_{00}^2 \geq \abs{a_{11}}^2 \ens .
  \end{equation*}
  Calculations for higher ($n \geq 3$) qubit numbers run
  analogously. For the general $n$ qubit case the Equations
  \eqref{nqc_1} are obtained by setting the first partial derivatives
  to zero.  The second partial derivatives give rise to an $n \times
  n$ Hessian Matrix
  \begin{equation*}
    H =
    \begin{pmatrix}
      - a_{00 \ldots 0} & \widetilde{a}_{110 \ldots 0} &
      \widetilde{a}_{101 \ldots 0} & \cdots & \widetilde{a}_{100 \ldots 1} \\
      \widetilde{a}_{110 \ldots 0} & - a_{00 \ldots 0} &
      \widetilde{a}_{011 \ldots 0} & \cdots & \widetilde{a}_{010 \ldots 1} \\
      \widetilde{a}_{101 \ldots 0} & \widetilde{a}_{011 \ldots 0} &
      - a_{00 \ldots 0} & \cdots & \widetilde{a}_{001 \ldots 1} \\
      \vdots  & \vdots  & \vdots & \ddots & \vdots  \\
      \widetilde{a}_{100 \ldots 1} & \widetilde{a}_{010 \ldots 1} &
      \widetilde{a}_{001 \ldots 1} & \cdots & - a_{00 \ldots 0}
    \end{pmatrix}
    \ens ,
  \end{equation*}
  with $\widetilde{a}_{110 \ldots 0} = \RE \left[ a_{110 \ldots 0}
    \E^{\I \beta_1 + \beta_2} \right]$, $\widetilde{a}_{101 \ldots 0}
  = \RE \left[ a_{101 \ldots 0} \E^{\I \beta_1 + \beta_3} \right]$,
  and so on. Considering only the $3 \times 3$ leading principal minor
  (the top left $3 \times 3$ submatrix) of $H$, we find that by
  suitably choosing the three variables $\beta_1, \beta_2, \beta_3$ as
  \begin{gather*}
    \beta_1 = \tfrac{1}{2} ( - \alpha_{110 \ldots 0} -
    \alpha_{101 \ldots 0} + \alpha_{011 \ldots 0} ) \ens ,
    \\
    \beta_2 = \tfrac{1}{2} ( - \alpha_{110 \ldots 0} +
    \alpha_{101 \ldots 0} - \alpha_{011 \ldots 0} ) \ens ,
    \\
    \beta_3 = \tfrac{1}{2} ( + \alpha_{110 \ldots 0} -
    \alpha_{101 \ldots 0} - \alpha_{011 \ldots 0} ) \ens ,
  \end{gather*}
  where ${\alpha}_{ijk \ldots}$ is the phase of $a_{ijk \ldots}$
  (i.e. $a_{ijk \ldots} = \abs{a_{ijk \ldots}} \E^{\I \alpha_{ijk
      \ldots}}$), we obtain
  \begin{equation}\label{33matrix}
    H_{3 \times 3} =
    \begin{pmatrix}
      - a_{00 \ldots 0} & \abs{{a}_{110 \ldots 0}} &
      \abs{{a}_{101 \ldots 0}} \\
      \abs{{a}_{110 \ldots 0}} & - a_{00 \ldots 0} &
      \abs{{a}_{011 \ldots 0}} \\
      \abs{{a}_{101 \ldots 0}} & \abs{{a}_{011 \ldots 0}} &
      - a_{00 \ldots 0}
    \end{pmatrix}
    \ens .
  \end{equation}
  The negative semidefinity of $H$ results in necessary conditions for
  all leading principal minors.  The $2 \times 2$ and $3 \times 3$
  leading principal minors can be taken from $H_{3 \times 3}$ of
  \eq{33matrix}, and they yield the first inequality in \eqref{nqc_2}
  and \eqref{nqc_3}, respectively.  Since $\ket{\Lambda} =
  \ket{0}^{\otimes n}$ is symmetric, the indices of the qubits are
  interchangeable, thus giving rise to all the permutations
  incorporated in \eqref{nqc_2} and \eqref{nqc_3}.
\end{proof}

For two qubits the conditions \eqref{2_qubit_coeff} directly lead to a
set of maximally entangled states $\ket{\phi} = \tfra{1}{\sqrt{2}}
\left( \ket{00} + \E^{\I \varphi} \ket{11} \right)$ with $a_{00} =
\tfrac{1}{\sqrt{2}}$, and hence $\Eg = 1$.  For three and more qubits,
however, it is not easy to locate the maximally entangled states.
This is because the function $g ( \lambda ) =
\abs{\bracket{\psi}{\lambda}}$ has in general several maxima, and
because the conditions of \eq{n_qubit_coeff} were derived only from
the property that $\ket{\Lambda} = \ket{0}^{\otimes n}$ is a
\emph{local} maximum.  Therefore, given a state $\ket{\psi}$ that
satisfies the conditions \eqref{n_qubit_coeff}, we cannot be sure that
$\ket{\Lambda} = \ket{0}^{\otimes n}$ is a \ac{CPS} (i.e. global
maximum).  For example, the arbitrarily weakly entangled state
$\sqrt{\epsilon} \ket{000} + \sqrt{1 - \epsilon} \ket{111}$, $\epsilon
\rightarrow 0$ satisfies the conditions \eqref{3_qubit_coeff}, because
$\ket{\lambda} = \ket{000}$ is a local maximum of $g ( \lambda )$,
even though the global maximum is $\ket{\Lambda} = \ket{111}$.  This
simple example shows that already the pure 3 qubit case exhibits a
much more diverse structure of the overlap function $g ( \lambda )$
than the bipartite one.

\Theoref{coeff_theo} provides necessary conditions for $\ket{\Lambda}
= \ket{0}^{\otimes n}$ being a \ac{CPS}, and can therefore be
considered as a special case of the generalised Schmidt decomposition
of Carteret \etal \cite{Carteret00}.  This can be most easily seen by
comparing \eq{nqc_1} and \eqref{nqc_2} to Theorem 1 of
\cite{Carteret00}.  One difference between the theorems is that we
make the connection to the \ac{CPS} and thus the \ac{GM} explicitly
clear. In \cite{Carteret00} the fact that $\ket{\Lambda} =
\ket{0}^{\otimes n}$ is a maximum of $g ( \lambda )$ is touched upon
only in the proof, and this maximum is not required to be global.

\section{Results for positive states}\label{pos_results}

Positive states are particularly easy to treat with respect to the
\ac{GM} due to the absence of complex phases in their coefficients.
Perhaps the most intriguing result in this respect is that all
positive states are strongly additive with respect to the
\ac{GM}\footnote{The geometric measure can be additive only in the
  logarithmic form $\Eg$ defined in \protect\eq{geo_def}. The
  alternative definition $\EG$ used in \protect\cite{Wei03} does not
  exhibit additivity properties.}, whereas almost all other states
lack this property, as shown by Zhu \etal \cite{Zhu10}.

Here we prove two results, namely that positive states have positive
\acp{CPS}, and that positive states generally have less entanglement
than non-positive states with the same weightings of their basis
states.  The proofs of these two findings are similar to each other.

\begin{lemma}\label{lem_pos_cps}
  Every pure state $\ket{\psi}$ of a finite-dimensional system that is
  positive with respect to some computational basis has at least one
  positive \ac{CPS} in that basis.
\end{lemma}
\begin{proof}
  Picking any computational basis in which the coefficients of
  $\ket{\psi}$ are all positive, we denote the orthonormal basis of
  subsystem $j$ with $\{ \ket{i_{j}} \}$, $i_{j} = 0, \ldots , d_{j} -
  1$, and can write the state as $\ket{\psi} = \sum_{\bmr{i}}
  a_{\bmr{i}} \ket{i_1} \cdots \ket{i_n}$, with $\bmr{i} = ( i_1 ,
  \dots , i_n)$ and $a_{\bmr{i}} \geq 0$ for all $\bmr{i}$. We pick
  one \ac{CPS} of $\ket{\psi}$ and write it as $\ket{\Lambda} =
  \bigotimes_{j} \ket{\sigma_{j}}$, where $\ket{\sigma_{j}} =
  \sum_{i_j} b^{j}_{i_j} \ket{i_j}$ (with $b^{j}_{i_j} \in \mbbc$) is
  the state of subsystem $j$. Now define a new normalised product
  state as $\ket{\Lambda^{+}} = \bigotimes_{j} \ket{\sigma_{j}^{+}}$,
  where $\ket{\sigma_{j}^{+}} = \sum_{i_j} \Abs{b^{j}_{i_j}}
  \ket{i_j}$.  Because of $\abs{\bracket{\psi}{\Lambda^{+}}} =
  \sum_{\bmr{i}} a_{\bmr{i}} \prod_{j} \Abs{b^{j}_{i_j}} \geq
  \Abs{\sum_{\bmr{i}} a_{\bmr{i}} \prod_{j} b^{j}_{i_j}} =
  \abs{\bracket{\psi}{\Lambda}}$, the positive state
  $\ket{\Lambda^{+}}$ is also a \ac{CPS} of $\ket{\psi}$.
\end{proof}

This result, which I published together with Damian Markham and Mio
Murao in \cite{Aulbach10}, was independently found by Zhu \etal
\cite{Zhu10}.  \Lemref{lem_pos_cps} asserts that positive states have
at least one positive \ac{CPS}, but the existence of non-positive
\acp{CPS} is not ruled out.  Most positive states have only one
\ac{CPS} (which is necessarily positive), but it is not difficult to
find examples of positive states with non-positive \acp{CPS}.  For
example, one of the \acp{CPS} of the Bell state $\ket{\psi} = \sym{1}
= \tfrac{1}{\sqrt{2}} \left (\ket{01} + \ket{10} \right)$ is
$\ket{\Lambda} = \tfrac{1}{2} \left( \ket{0} + \I \ket{1}
\right)^{\otimes 2}$, and more examples will appear in
\chap{solutions}.

A statement analogous to \lemref{lem_pos_cps} does not hold for real
states, i.e. there exist real states that have no real \ac{CPS}.  A
trivial example are the rotated $n$ qubit \ac{GHZ} states $\rotxs (
\tfra{\pi}{2} ) \ket{\text{GHZ}_{n}}$ for which it follows from
\eq{ghz_w_cps_1} and \eq{x_rotationmatrix} that they only have the two
\acp{CPS} $\ket{\Lambda_{1}} = \tfrac{1}{\sqrt{2^n}} \left( \ket{0} +
  \I \ket{1} \right)^{\otimes n}$ and $\ket{\Lambda_{2}} =
\tfrac{1}{\sqrt{2^n}} \left( \ket{0} - \I \ket{1} \right)^{\otimes
  n}$.

The next theorem asserts that the amount of geometric entanglement of
multipartite states of any dimension is in general higher for
\textbf{phased} states, i.e. states whose coefficients are not
restricted to positive values in a given computational basis.  For
this purpose we define the corresponding positive state
$\ket{\psi^{+}} = \sum_{i} \abs{a_{i}} \ket{i}$ of a given state
$\ket{\psi} = \sum_{i} a_{i} \ket{i}$ to be the state that is obtained
from $\ket{\psi}$ by removing the complex phases from all
coefficients, and we call $\ket{\psi^{+}}$ the corresponding
\textbf{non-phased}\footnote{The term \quo{non-phased} has been chosen
  over \quo{dephased} to avoid confusion with the physical process of
  phase coherence loss.}  state. Note that $\ket{\psi^{+}}$
automatically inherits the normalisation of $\ket{\psi}$, and that
$\ket{\psi^{+}}$ subtly depends on the basis in which $\ket{\psi}$ is
represented.

\begin{theorem}\label{theo_positive_entanglement}
  Every pure state $\ket{\psi}$ of a finite-dimensional system
  contains at least the same amount of geometric entanglement as the
  corresponding non-phased state $\ket{\psi^{+}}$, i.e.  $\Eg (
  \ket{\psi} ) \geq \Eg ( \ket{\psi^{+}} )$.
\end{theorem}
\begin{proof}
  Using the same notation as in the proof of \lemref{lem_pos_cps}, we
  write the given state as $\ket{\psi} = \sum_{\bmr{i}} a_{\bmr{i}}
  \ket{i_1} \cdots \ket{i_n}$, with $\bmr{i} = ( i_1 , \dots , i_n)$,
  and the corresponding non-phased state as $\ket{\psi^{+}} =
  \sum_{\bmr{i}} \abs{a_{\bmr{i}}} \ket{i_1} \cdots \ket{i_n}$.  We
  take one of the \acp{CPS} of $\ket{\psi}$ and denote it as
  $\ket{\Lambda} = \bigotimes_{j} \ket{\sigma_{j}}$, where
  $\ket{\sigma_{j}} = \sum_{i_j} b^{j}_{i_j} \ket{i_j}$ is the state
  of subsystem $j$.  The corresponding non-phased state of
  $\ket{\Lambda}$ is a normalised product state with positive
  coefficients $\ket{\Lambda^{+}} = \bigotimes_{j}
  \ket{\sigma_{j}^{+}}$, with $\ket{\sigma_{j}^{+}} = \sum_{i_j}
  \Abs{b^{j}_{i_j}} \ket{i_j}$.  Using the following inequality
  \begin{equation*}
    \abs{\bracket{\psi^{+}}{\Lambda^{+}}} =
    \sum_{\bmr{i}} \abs{ a_{\bmr{i}} }
    \prod_{j} \Abs{ b^{j}_{i_j} } \geq
    \bigg| \sum_{\bmr{i}} a_{\bmr{i}} \prod_{j} b^{j}_{i_j} \bigg|
    = \abs{\bracket{\psi}{\Lambda}} \ens ,
  \end{equation*}
  it follows that $\Eg ( \ket{\psi} ) = - \log_2
  \abs{\bracket{\psi}{\Lambda}}^2 \geq - \log_2
  \abs{\bracket{\psi^{+}}{\Lambda^{+}}}^2 \geq \Eg ( \ket{\psi^{+}}
  )$.
\end{proof}

It should be noted that -- unlike \lemref{lem_pos_cps} --
\theoref{theo_positive_entanglement} is valid for arbitrary choices of
the computational basis, and to each such basis relates a non-phased
state $\ket{\psi^{+}}$.  These non-phased states are in general all
different from each other and thus carry different amounts of
entanglement, but their unifying feature is that they do not carry
more entanglement than $\ket{\psi}$.

\Theoref{theo_positive_entanglement} has not been published yet, but
it has been cited in the form of a private communication with Dagmar
Bru\ss{} in \cite{Chiuri10}. There they propose the experimental
creation and detection of \quo{phased Dicke states} \cite{Krammer09}
by means of hyperentangled photons \cite{Barreiro05}, and from
\theoref{theo_positive_entanglement} it is clear that these states
carry at least the same amount of geometric entanglement as regular
Dicke states.  It is reasonable to expect that for sufficiently large
systems the phased states are much higher entangled than their
non-phased counterparts, and indeed Zhu \etal \cite{Zhu10} provide all
the ingredients necessary to prove the following corollary:

\begin{corollary}\label{max_pos_ent}
  Every positive $n$ qubit state $\ket{\psi}$ has entanglement $\Eg
  (\ket{\psi}) \leq \frac{n}{2}$, and this bound is strict for even
  $n$.
\end{corollary}

\begin{proof}
  According to Theorem 5 of \cite{Zhu10}, every positive $n$ qubit
  state is strongly additive with respect to \ac{GM}.  On the other
  hand, from Proposition 23 of the same paper it follows that $\Eg
  \leq \frac{n}{2}$ holds for all strongly additive $n$ qubit
  states. A trivial example of a positive $n$ qubit state with $\Eg =
  \frac{n}{2}$ for even $n$ is $\ket{\psi} = \left( \ket{00} +
    \ket{11} \right)^{\otimes \frac{n}{2}}$, i.e. $\frac{n}{2}$ Bell
  pairs.
\end{proof}

As outlined in \sect{gm_def}, the known bounds on the maximal
entanglement for general $n$ qubit states $\ket{\psi}$ are $n - 2
\log_2 (n) - \Order (1) \leq \Eg ( \ket{\psi} ) \leq n-1$, and most
states surpass the lower bound \cite{Jung08,Gross09,Zhu10}.  This
implies that the overwhelming majority of states is significantly
higher entangled than the corresponding non-phased states.

Does a similar result hold for real states?  Surprisingly, the
restriction to real coefficients has only little impact on the maximal
entanglement, as outlined in Proposition 25 of \cite{Zhu10}.  More
specifically, for real states the lower bound found for \ac{GM} is
only 1 ebit less than for general states, regardless of the number of
parties and their dimensions.  The property of real states being
closer in spirit to general states than to positive states can also be
seen from that fact that antisymmetric basis states are much higher
entangled than symmetric basis states, even though the two differ only
by the phase factors that induce a sign change under odd permutations
of the parties \cite{Hayashi08}.

\section{Results for symmetric states}\label{symm_results}

\subsection{Closest product states of the maximally entangled state}\label{cps_symm}

As mentioned in \sect{cps_number}, it is possible to derive a
symmetrised version of \theoref{numberofcps}, where the maximally
entangled state among all symmetric multi-qudit states is
considered\footnote{For Lemma 4 of \protect\cite{Aulbach10} I
  previously provided an explicit proof for the existence of at least
  two \acp{CPP} for maximally entangled symmetric $n$ qubit states.
  This proof will not be reiterated here, because
  \protect\corref{numberofcpp} is a strictly stronger result.}.

\begin{corollary}\label{numberofcpp}
  Let $\Psis \in \mathh_{\text{s}} \subset \mathh = \mathh_{0} \otimes
  \cdots \otimes \mathh_{0} $ be a normalised pure symmetric state of
  an $n$-partite ($n \geq 3$) qudit system with $\dim (\mathh_{0}) =
  d$.  The set of \acp{CPS} of $\Psis$ is denoted by $\Lambda \subset
  \mathh_{\text{s}}$ and the set of \acp{CPP} is denoted by $\Sigma
  \subset \mathh_{0}$.  If $\Psis$ is maximally entangled with respect
  to the \ac{GM}, then $\Psis \in \spa ( \Lambda )$ and there exist at
  least two linearly independent \acp{CPS}.
\end{corollary}

\begin{proof}
  One can restrict to fully symmetric \acp{CPS}, because symmetric
  states of three or more qudits have no nonsymmetric \acp{CPS}
  \cite{Hubener09}. In particular, this means that there exists a
  one-to-one correspondence between the \acp{CPS} and \acp{CPP},
  i.e. if $\ket{\sigma_{i}} \in \Sigma$ then
  $\ket{\Lambda_{i}^{\text{s}}} = \ket{\sigma_{i}}^{\otimes n} \in
  \Lambda$, and vice versa.  The proof is performed analogously to the
  one of \theoref{numberofcps}, so we focus only on aspects that are
  not immediately clear from comparison.  Firstly, the restriction to
  symmetric states $\Psis \in \mathh_{\text{s}}$ does not prevent one
  from using the generalised Schmidt decomposition \cite{Carteret00},
  because the symmetric basis states $\sym{n, \bmr{k}}$ can be
  trivially expanded with the computational basis states $\ket{ i_{1}
    i_{2} \cdots i_{n} }$, and the same is done with the symmetric
  \acp{CPS} $\ket{\Lambda_{i}^{\text{s}}}= \ket{\sigma_{i}}^{\otimes
    n}$.

  Secondly, the variation $\ket{\psi^{\text{s}} ( \epsilon )}$ of
  $\Psis$ according to \eq{variation} is performed with a symmetric
  product state $\ket{\zeta^{\text{s}}} = \ket{\kappa}^{\otimes n} \in
  \mathh_{\text{s}}$ with $\ket{\kappa} \in V \subset \mathh_{0}$,
  where $V$ is the orthogonal complement of $U := \spa ( \Sigma ) \neq
  \mathh_{0}$.  Equivalently, the variation $\ket{\lambda^{\text{s}} (
    \bmr{\delta} ) }$ of $\ket{\Lambda^{\text{s}}} = \ket{0}^{\otimes
    n}$ should be fully symmetric: $\ket{\lambda^{\text{s}} (
    \bmr{\delta} )} = \ket{ \bmr{\delta} }^{\otimes n}$, with $\ket{
    \bmr{\delta} } = (1- \delta_{0}) \ket{0} + \delta_{1} \ket{1} +
  \ldots + \delta_{d -1} \ket{d - 1}$.
\end{proof}

\Corref{numberofcpp} provides a useful necessary condition for
checking whether candidates for maximal symmetric $n$ qubit
entanglement (which will be studied in \chap{solutions}) are indeed
maximally entangled, namely that the states must lie in the span of
their \acp{CPS}.

\subsection{Upper bound on symmetric entanglement}\label{upper_bound}

An upper bound on the maximal symmetric entanglement was already
introduced with \eq{upper_bound_eq}.  Here I present an alternative
proof for this bound with the advantage of having an intuitive
geometric meaning.  The same proof as mine which I published in
\cite{Aulbach10} was independently found by Martin \etal
\cite{Martin10}.  The known lower bounds on maximal symmetric
entanglement will be reviewed later in \sect{entanglement_scaling}.
\begin{theorem}\label{const_integral}
  For every symmetric $n$ qubit state $\psis$ the following equality
  holds:
  \begin{equation}\label{intvolume}
    \int\limits_{0}^{2 \pi} \int\limits_{0}^{\pi}
    \abs{\bracket{\psi^{\text{s}}}{\lambda(\theta , \varphi) }}^2
    \sin \theta \, \emph{d} \theta \emph{d} \varphi =
    \frac{4 \pi}{n+1} \ens ,
  \end{equation}
  where $\ket{\lambda(\theta , \varphi)} = \left( \co_{\theta} \ket{0}
    + \E^{\I \varphi} \si_{\theta} \ket{1} \right)^{\otimes n}$.
\end{theorem}
\begin{proof}
  A symmetric $n$ qubit state can be written as $\psis = \sum_{k =
    0}^{n} a_k \E^{\I \alpha_k} \sym{k}$, with $a_k \in \mbbr$,
  $\alpha_k \in [0,2 \pi)$ and the normalisation condition $\sum_{k}
  a_k^2 = 1$.  Writing the \ac{CPS} as $\ket{\lambda} =
  \ket{\sigma}^{\otimes n}$ with $\ket{\sigma} = \co_{\theta} \ket{0}
  + \E^{\I \varphi} \si_{\theta} \ket{1}$, we obtain
\begin{equation}\label{scalarproduct_expr}
  \bracket{\psi^{\text{s}}}{\lambda} = \sum_{k = 0}^{n}
  \E^{\I ( k \varphi - \alpha_{k} )} a_k \co_{\theta}^{n-k}
  \si_{\theta}^{k} \sqrt{{\binom{n}{k}}} \ens .
\end{equation}
Using the set of qubit unit vectors $\mathcal{S}^{2}$ and the uniform
measure over the unit sphere $d \mathcal{B}$, the squared norm of
\eq{scalarproduct_expr} can be integrated over the unit sphere:
\begin{equation}
  \int\limits_{\ket{\sigma} \in \mathcal{S}^{2}}
  \abs{\bracket{\psi^{\text{s}}}{\lambda}}^2 d \mathcal{B} =
  \int\limits_{0}^{2 \pi} \int\limits_{0}^{\pi}
  \abs{\bracket{\psi^{\text{s}}}{\lambda (\theta , \varphi ) }}^2
  \sin \theta \, \D \theta \D \varphi \ens .
\end{equation}
Taking into account that $\int_{0}^{2 \pi} \E^{\I m \varphi} \D
\varphi = 0$ for any integer $m \neq 0$, one obtains
\begin{subequations}\label{calc_mean}
  \begin{align}
    \int\limits_{0}^{2 \pi}& \int\limits_{0}^{\pi}
    \left[ \: \sum_{k = 0}^{n} a_k^2 \co_{\theta}^{2(n-k)}
      \si_{\theta}^{2k} {\binom{n}{k}} \right]
    \sin \theta \, \D \theta \D \varphi
    \label{calc_mean1} \\
    {}& = 2 \pi \sum_{k = 0}^{n}
    a_k^2 {\binom{n}{k}} \int\limits_{0}^{\pi}
    \co_{\theta}^{2(n-k)}
    \si_{\theta}^{2k}
    \sin \theta \, \D \theta
    \label{calc_mean2} \\
    {}& = 4 \pi \sum_{k = 0}^{n} a_k^2 {\binom{n}{k}}
    \frac{\Gamma (k+1) \Gamma (n-k+1)}{\Gamma (n+2)}
    = 4 \pi \sum_{k = 0}^{n} a_k^2 \frac{1}{n+1} =
    \frac{4 \pi}{n+1} \ens .
    \label{calc_mean3}
  \end{align}
\end{subequations}
The equivalence of \eq{calc_mean2} and \eqref{calc_mean3} follows from
the different definitions of the Beta function \cite{Abramowitz}.
\end{proof}

Since the mean value of $\abs{\bracket{\psi^{\text{s}}}{\lambda
    (\theta , \varphi ) }}^2$ over the Bloch sphere is $\tfrac{4
  \pi}{n+1}$, it follows that $G^{2} ( \psis ) =
\abs{\bracket{\psi^{\text{s}}}{\Lambda^{\text{s}}}}^2 =
\max\limits_{\ket{\lambda} \in \mathh_{\text{SEP}} }
\abs{\bracket{\psi^{\text{s}}}{\lambda}}^2$ must be at least
$\tfrac{1}{n+1}$. This leads to the upper bound $\Eg ( \psis ) \leq
\log_2 (n+1)$ for any symmetric $n$ qubit state.

The integral in \eq{intvolume} is the same for all symmetric $n$ qubit
states, and in \sect{preliminaries} it will be seen that this allows
for an intuitive visualisation of the geometric entanglement of
symmetric states by means of spherical distributions with constant
volume in $\mbbrr$.  From a mathematical perspective, the constant
integral is a consequence of Schur's Lemma which implies that the
uniform mixture of the product states $\pure{\lambda} = (
\pure{\sigma} )^{\otimes n}$ equals the identity $\one$, up to the
prefactor $\binom{n+d-1}{n}^{-1}$, where $d$ is the dimension of the
subsystems (see e.g. \cite{Renner,Perelomov}).  In particular, this
implies that \theoref{const_integral} can be readily generalised to
the qudit case, yielding the upper bound $\Eg ( \psis ) \leq \log
\binom{n+d-1}{n}$ for any symmetric $n$ qudit state.

\subsection{Measurement-based quantum computation}
\label{resources_for_mbqc}

One of the leading schemes for the physical implementation of quantum
computing is \textbf{\acf{MBQC}}, also known as one-way quantum
computation. Before the start of the computation an entangled resource
state, for example a two-dimensional cluster state \cite{Briegel01},
is prepared. A proper choice of subsequent single qubit measurements
then allows one to deterministically create any desired state on the
unmeasured qubits, as long as the initially prepared state is
sufficiently large and a \textbf{universal resource} state.  A family
of $n$ qubit states $\Psi = \{ \ket{\psi_{n}} \}_{n}$ is said to be a
universal resource if any given state can be prepared
deterministically and exactly by \ac{MBQC} from a state
$\ket{\psi_{n}}$ with sufficiently large $n$ \cite{Nest07}.  The
resource character of the initial state is evident from the fact that
only local operations are performed, and the consequently irreversible
reduction of entanglement explains the term \quo{one-way quantum
  computation}.

In order for states to be useful for \ac{MBQC} their entanglement must
come in the \quo{right dose}.  On the one hand, universal resources
for \ac{MBQC} must be maximally entangled in a certain sense
\cite{Mora10,Nest07}, and if the entanglement of a set of \ac{MBQC}
resource states scales anything below logarithmically with the number
of parties, it cannot be an efficient resource for deterministic
universal \ac{MBQC} \cite{Nest07}.  On the other hand, somewhat
surprisingly, if the entanglement is too large, it is also not a good
resource for \ac{MBQC}: If the geometric entanglement of an $n$ qubit
system scales as $n - \Order{( \log_2 n )}$, then a computation
performed with such a resource can be simulated efficiently on a
classical computer \cite{Gross09}.  Indeed, most quantum states are
too entangled for being computationally universal \cite{Gross09},
although the right amount of geometric entanglement is by no means a
sufficient condition \cite{Bremner09}.

It was seen that the maximal geometric entanglement of symmetric
states scales much slower than that of general states, namely
logarithmically rather than linearly.  One could therefore ask whether
symmetric states are useful for \ac{MBQC}, because they are not too
entangled.  Unfortunately, \eq{upper_bound_eq} implies that symmetric
states scale at most logarithmically, whereas the entanglement of
exact, deterministic \ac{MBQC} resources must scale
faster-than-logarithmically \cite{Nest07}. Somewhat weaker
requirements are imposed upon approximate, stochastic \ac{MBQC}
resources \cite{Mora10}, although this generally leads only to a small
extension of the class of suitable resources in the vicinity of exact,
deterministic resources (e.g. 2D cluster states with holes).  One can
therefore expect that symmetric states cannot be used even for
approximate, stochastic \ac{MBQC}.

To underline this conjecture, we will show that Dicke states with a
fixed number of excitations cannot be useful for
$\epsilon$-approximate, deterministic \ac{MBQC} \cite{Mora10}.  To be
consistent with the notation used in \cite{Mora10}, we temporarily
switch to the alternative definition of the geometric measure $\EG
(\ket{\psi}) = 1 - \abs{\bracket{\Lambda_{\psi}}{\psi}}^2$ for the
duration of this section.  Roughly speaking, $\epsilon$-approximate
universal resource states can be converted into any other state by
\ac{LOCC} with an inaccuracy of at most $\epsilon$.  The
$\epsilon$-version of the \ac{GM} is defined as \cite{Mora10}
\begin{equation}\label{epsilon_entanglement_definition}
  \EG^{\epsilon} (\rho) = \min \{ \EG (\sigma)
  \, \vert \, D(\rho,\sigma) \leq \epsilon \} \ens ,
\end{equation}
where $D$ is a distance that is \quo{strictly related to the
  fidelity}, meaning that for any two states $\rho$ and $\sigma$,
$D(\rho,\sigma) \leq \epsilon \Rightarrow F(\rho,\sigma) \geq 1 - \eta
(\epsilon)$, where $0 \leq \eta (\epsilon) \leq 1$ is a strictly
monotonically increasing function with $\eta (0) = 0$.
$\EG^{\epsilon} (\rho)$ can be understood as the guaranteed
entanglement obtained from a preparation of $\rho$ with inaccuracy
$\epsilon$. One possible choice of $D$ is the trace distance, which
for pure states reads $D_{\text{t}} (\ket{\psi},\ket{\phi}) = \sqrt{1
  - \abs{\bracket{\psi}{\phi}}^2 } = \sqrt{1 - F}$, where $F$ is the
fidelity. In this case one can choose $\eta(\epsilon) = \epsilon^{2}$.

As shown in Example 1 of \cite{Mora10}, the family of W states
$\Psi_{\text{W}} = \{ \ket{\text{W}}_{n} \}_{n}$, with
$\ket{\text{W}_{n}} \equiv \sym{{n,1}}$, is not an
$\epsilon$-approximate universal resource for $\eta(\epsilon) \lesssim
0.001$. In the following we generalise this result to families of
Dicke states $\Psi_{{\text{S}}_{k}} = \{ \sym{{n,k}} \}_{n}$ with an
arbitrary, but fixed number of excitations $k \in \mbbn$.

\begin{theorem}\label{mbqc_example}
  For any fixed $k \in \mbbn$ the family of Dicke states
  $\Psi_{{\text{S}}_{k}} = \{ \ket{S_{n,k} } \}_{n}$ cannot be an
  $\epsilon$-approximate universal \ac{MBQC} resource for
  $\eta(\epsilon) \lesssim 0.001 \, k^{-3/2}$.
\end{theorem}

\begin{proof}
  Using \eq{dicke_ent} and the Stirling approximation for high $n$,
  the asymptotic geometric entanglement of the family
  $\Psi_{{\text{S}}_{k}}$ is found to be
  \begin{equation*}
    \EG ( \Psi_{{\text{S}}_{k}} ) = 1 - \frac{k^k}{\E^{k} k!} \ens .
  \end{equation*}
  Specifically, the amount of geometric entanglement remains finite
  for arbitrary values of $n$, allowing us to apply Proposition 3 and
  Theorem 1 of \cite{Mora10} to show that the necessary condition for
  $\epsilon$-approximate deterministic universality,
  \begin{equation*}
    \EG ( \Psi_{{\text{S}}_{k}} ) > 1-4 \eta^{\frac{1}{3}} +
    3.4 \eta^{\frac{2}{3}} \ens ,
  \end{equation*}
  is violated for $\eta(\epsilon) \lesssim 0.001 \, k^{-\frac{3}{2}}$.
\end{proof}

Of course, many other quantum information tasks are not restricted by
the requirements of \ac{MBQC}-universality, and thus highly entangled
symmetric states can be valuable resources for tasks such as the
leader election problem \cite{Dhondt06} or \ac{LOCC} discrimination
\cite{Hayashi06}.

It should be noted that the uselessness of $n$ qubit symmetric states
for quantum computation can also be inferred from the observation that
the tensor rank of such states scales only polynomially\footnote{The
  $(n+1)$-dimensional space of $n$ qubit symmetric states can be
  spanned by the continuous set of spin coherent states
  \protect\cite{Perelomov}, i.e.  $\mathh_{\text{s}} = \spa \{
  \ket{\sigma}^{\otimes n} , \ket{\sigma} \in \mbbc^{2} \}$. This
  implies that symmetric $n$ qubit states have a tensor rank of at
  most $n+1$.}.  Superpositions of a polynomial number of product
states can be simulated efficiently classically, because these states
and all subsequent states arising during the computation have an
efficient classical description \cite{Bremner09}, due to the tensor
rank being an entanglement monotone.

\cleardoublepage

\chapter{Majorana Representation and Geometric Entanglement}
\label{majorana_representation}

\begin{quotation}
  In this chapter the Majorana representation of symmetric states will
  be employed to investigate symmetric $n$ qubit states with respect
  to the geometric measure of entanglement. Starting with a discussion
  of different visualisation techniques and a review of two and three
  qubit symmetric states, we move on to the concepts of totally
  invariant states and spherical optimisation problems. This is
  followed by the derivation of a variety of analytical results about
  the Majorana representation. These results will be helpful for the
  study of the geometric entanglement of symmetric states in later
  chapters.
\end{quotation}

\section{Preliminaries}\label{preliminaries}

\subsection{Visualisation of symmetric states}\label{visualisation}

In \sect{majorana_definition_sect} the Majorana representation was
introduced for symmetric $n$ qubit states, and a visualisation by
means of the \acfp{MP} on the Majorana sphere was shown in
\fig{majorana_dicke}.  This approach is now extended to encompass
information about the geometric entanglement of states.

By means of its definition \eqref{geo_def} the geometric entanglement
of a state is determined by its \acfp{CPS}.  As mentioned in
\sect{gm_def}, for symmetric states there exists at least one
symmetric \ac{CPS} $\ket{\Lambda^{\text{s}}} = \ket{\sigma}^{\otimes
  n}$, and for $n \geq 3$ qudits all \acp{CPS} are symmetric
\cite{Hubener09}. Here we consider the case of $n$ qubit symmetric
states. The single-qubit states $\ket{\sigma}$ are then called
\acfp{CPP}, because they can be represented on the Majorana sphere by
their Bloch vectors.  In this way the set of \acp{CPP} can be visually
represented alongside the \acp{MP}. \Fig{ghz_w_picture} shows such
Majorana representations for the three qubit \ac{GHZ} and W state, two
states that are symmetric as well as positive, and whose entanglement
may be considered extremal (see e.g. \cite{Tamaryan09}).
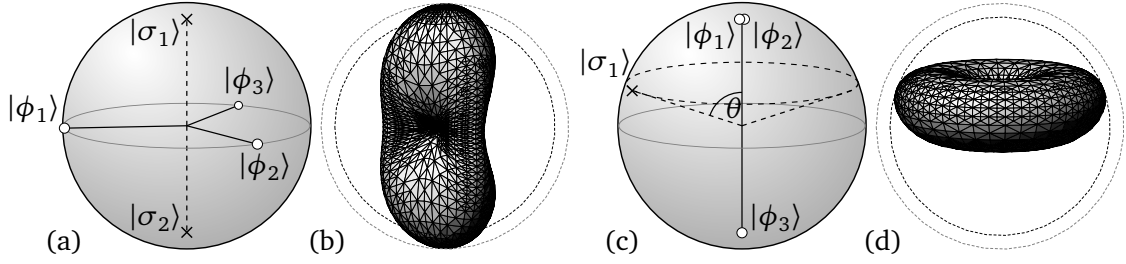
\begin{figure}
  \begin{overpic}[scale=.45]{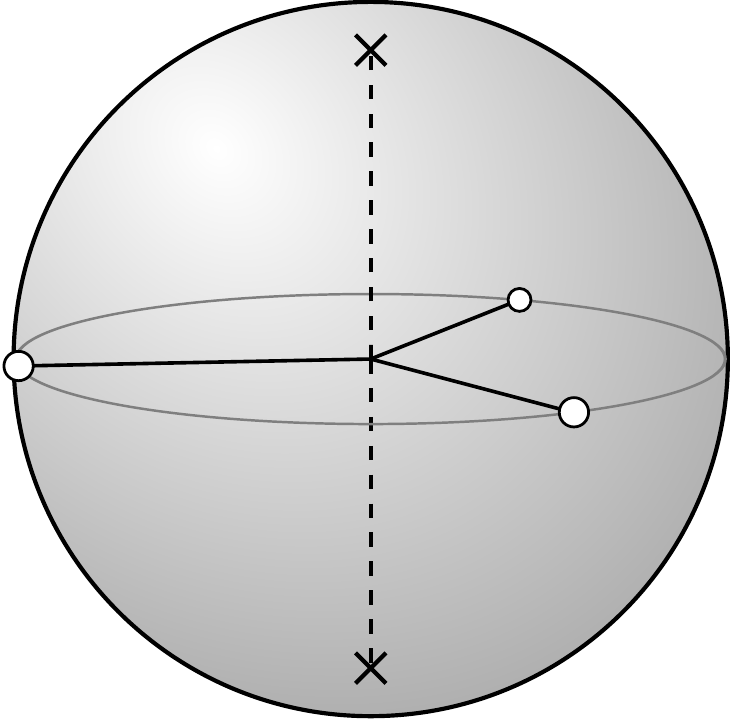}
    \put(-5,0){(a)}
    \put(-20,53){$\ket{\phi_{1}}$}
    \put(71,30){$\ket{\phi_{2}}$}
    \put(65,64){$\ket{\phi_{3}}$}
    \put(28,83){$\ket{\sigma_{1}}$}
    \put(28,10){$\ket{\sigma_{2}}$}
  \end{overpic}
  \begin{overpic}[scale=.22]{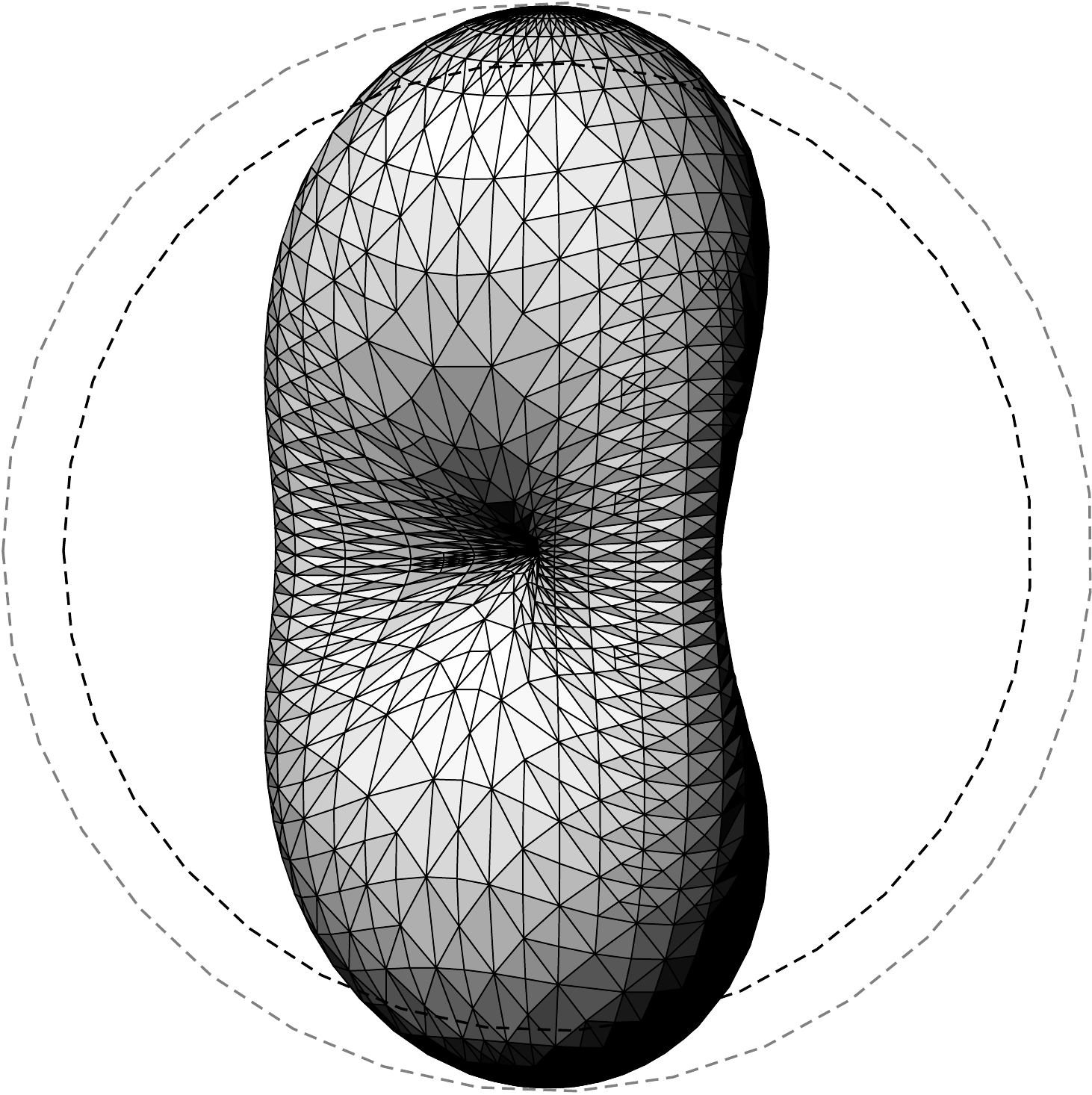}
    \put(-5,0){(b)}
  \end{overpic}
  \hfill
  \begin{overpic}[scale=.45]{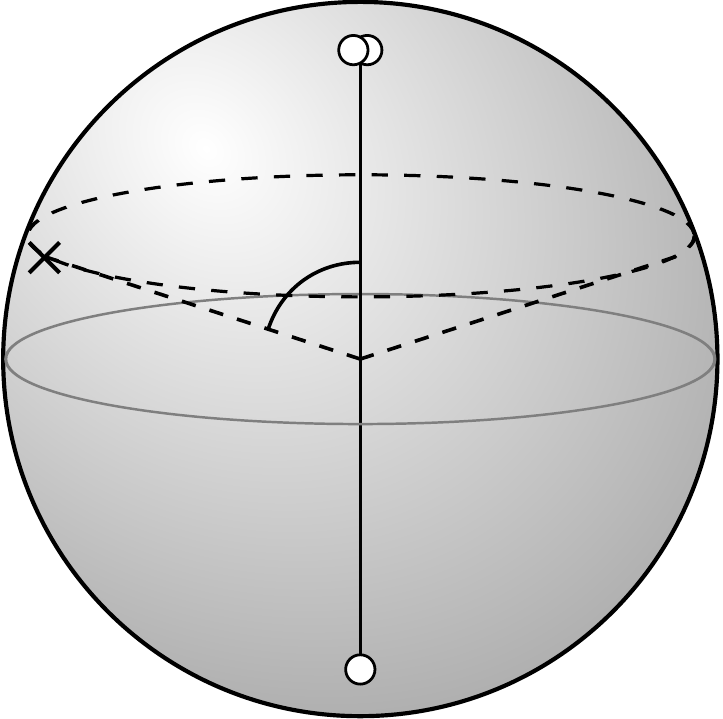}
    \put(-5,0){(c)}
    \put(27,83){$\ket{\phi_{1}}$}
    \put(54,83){$\ket{\phi_{2}}$}
    \put(53,10){$\ket{\phi_{3}}$}
    \put(-15,72){$\ket{\sigma_{1}}$}
    \put(43,53){$\theta$}
  \end{overpic}
  \begin{overpic}[scale=.22]{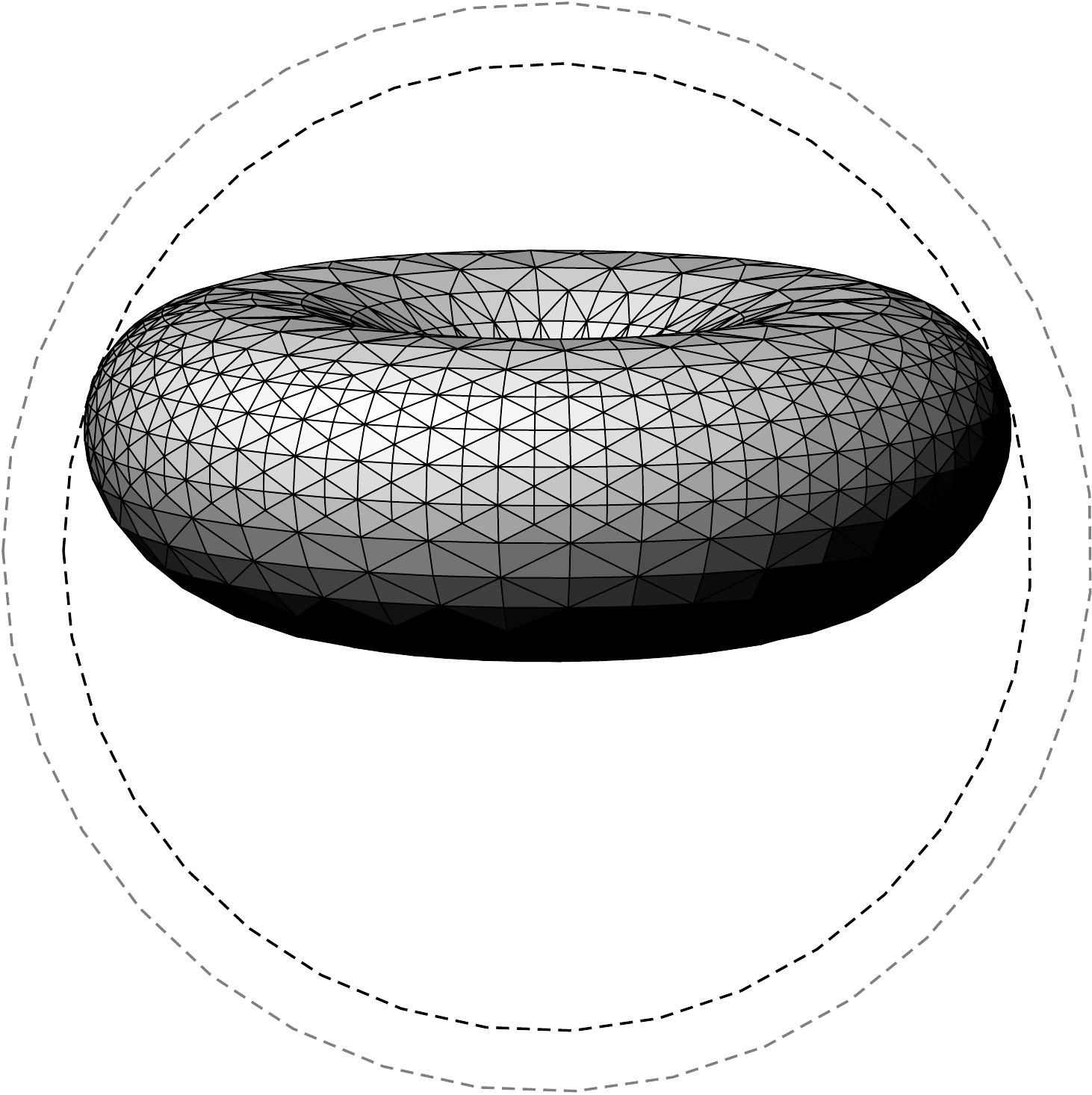}
    \put(-5,0){(d)}
  \end{overpic}
  \caption[Visualisations of the 3 qubit GHZ and W
  state]{\label{ghz_w_picture} Different visualisations of the three
    qubit \ac{GHZ} state and W state are shown.  The \acp{MP} (white
    circles) and \acp{CPP} (crosses) of $\ket{\text{GHZ}_{3}}$ are
    shown in (a), and those of $\ket{\text{W}_{3}}$ in (c).  Plots of
    the corresponding spherical amplitude functions $g^2 (\theta ,
    \varphi )$ are shown in (b) and (d). The global maxima and zeros
    of $g^2$ coincide with the \acp{CPP} and the antipodes of the
    \acp{MP}, respectively.  The maximal values of $g^2$ are indicated
    by circles with radii $G^2 ( \ket{\text{GHZ}_{3}} ) = \frac{1}{2}$
    and $G^2 ( \ket{\text{W}_{3}} ) = \frac{4}{9}$, respectively.}
\end{figure}
The state $\ket{\text{GHZ}_{3}} = \tfra{1}{\sqrt{2}} ( \ket{000} +
\ket{111} )$ has the \acp{MP}
\begin{equation}\label{ghz_mps}
  \ket{\phi_{1}} = \tfrac{1}{\sqrt{2}} \big( \ket{0} + \ket{1} \big)
  \ens ,
  \hspace{2.5mm}
  \ket{\phi_{2}} = \tfrac{1}{\sqrt{2}} \big( \ket{0} +
  \E^{\I \tfra{2 \pi}{3}} \ket{1} \big) \ens ,
  \hspace{2.5mm}
  \ket{\phi_{3}} = \tfrac{1}{\sqrt{2}} \big( \ket{0} +
  \E^{\I \tfra{4 \pi}{3}} \ket{1} \big) \ens ,
\end{equation}
and its two \acp{CPP} are
\begin{equation}\label{ghz_cpps}
  \ket{\sigma_{1}} = \ket{0}
  \ens , \quad
  \ket{\sigma_{2}} = \ket{1} \ens .
\end{equation}
From this it follows that $G^2 =
\abs{\bracket{\text{GHZ}_{3}}{\sigma}^{\otimes 3}}^2 = \frac{1}{2}$
and hence the geometric entanglement is $\Eg ( \ket{\text{GHZ}_{3}} )
= 1$.  \Fig{ghz_w_picture}(a) shows the distribution of \acp{MP} and
\acp{CPP} on the Majorana sphere.  The three \acp{MP} form an
equilateral triangle on the equator, and the two \acp{CPP} lie at the
north and south pole, respectively.  The general $n$ qubit \ac{GHZ}
state \eqref{ghz_w_def} is represented by $n$ equidistant \acp{MP} on
the equator, and the \acp{CPP} are the same as in \eqref{ghz_cpps}
\cite{Markham11}.

The W state $\ket{\text{W}_{3}} = \sym{3,1} = \tfra{1}{\sqrt{3}} (
\ket{001} + \ket{010} + \ket{100} )$ is a Dicke state with the
\acp{MP} $\ket{\phi_{1}} = \ket{\phi_{2}} = \ket{0}$ and
$\ket{\phi_{3}} = \ket{1}$, and due to the azimuthal symmetry the set
of \acp{CPP} is formed by the continuum $\ket{\sigma} =
\sqrt{\tfra{2}{3}} \ket{0} + \E^{\I \varphi} \sqrt{\tfra{1}{3}}
\ket{1}$, with $\varphi \in [0,2 \pi)$.  \Fig{ghz_w_picture}(c) shows
the \acp{MP} and the circle of \acp{CPP}, with the positive \ac{CPP}
denoted by a cross.  The entanglement is $\Eg ( \ket{\text{W}_{3}} ) =
\log_2 \left( \tfra{9}{4} \right) \approx 1.17$, which is higher than
that of the \ac{GHZ} state.  It was recently shown that in terms of
the \ac{GM}, the W state is the maximally entangled three qubit state
\cite{Chen10}, and therefore it is also the maximally entangled
symmetric one.

Mediated by the stereographic projection, the amplitude function $\psi
(z) : \mbbc \to \mbbc$ defined on the complex plane in \eq{majpoly}
corresponds to the function $f ( \theta , \varphi ) =
\bracket{\psi^{\text{s}}}{\sigma (\theta , \varphi ) }^{\otimes n}$
with $\ket{\sigma(\theta , \varphi)} = \co_{\theta} \ket{0} + \E^{\I
  \varphi} \si_{\theta} \ket{1}$ defined on the Majorana sphere.
Taking the absolute value of this function, we define a real-valued
function $g : \mathcal{S}^2 \to [0,1]$ as follows
\begin{equation}\label{spher_amp_funct}
  g ( \theta , \varphi ) = \abs{\bracket{\psi^{\text{s}}}{\sigma
      (\theta , \varphi ) }^{\otimes n}} \ens ,
\end{equation}
and call it the \textbf{spherical amplitude function}. This function
has already played an important role in \theoref{const_integral} for
the derivation of the upper bound on the maximal symmetric
entanglement.  Comparing \eq{spher_amp_funct} to the definition of the
\ac{GM} \eqref{geo_def}, it is seen that the global maxima of $g$ are
the \acp{CPP} of $\psis$.  Furthermore, from the definition of the
Majorana representation \eqref{majorana_definition} it is clear that
the zeros of $g$ are the antipodes\footnote{In mathematical terms, the
  antipode of a Bloch vector $\ket{\phi} = \co_{\theta} \ket{0} +
  \E^{\I \varphi} \si_{\theta} \ket{1}$ is the unique Bloch vector
  $\ket{\phi}^{\perp} = \si_{\theta} \ket{0} - \E^{\I \varphi}
  \co_{\theta} \ket{1}$ which is orthogonal to the first one:
  $\bracket{\phi}{\phi}^{\perp} = 0$.}, i.e. the diametrically
opposite points of the \acp{MP}. The plots of $g^2$ shown for the
\ac{GHZ} and W state in \fig{ghz_w_picture} demonstrate that the
spherical amplitude function allows for an intuitive and powerful
visualisation of the entire information about a symmetric state and
its geometric entanglement.  For a given $\psis$ it is often not easy
to calculate the \acp{MP} and \acp{CPP} analytically, but $g$ makes it
very easy to do so numerically. This makes the spherical amplitude
function a powerful tool for the numerical search for high and maximal
geometric entanglement.  The amount of entanglement present in a
symmetric state decreases with increasing values of the injective
tensor norm, which is simply the maximum value of the spherical
amplitude function: $G = \max_{\ket{\sigma}} g ( \theta , \varphi
)$. Circles indicating the value of $G^2$ are shown in
\fig{ghz_w_picture}, and they provide a visual representation of the
difference in entanglement between the \ac{GHZ} and W state.

As another example we present the \quo{tetrahedron state}, the
symmetric state of four qubits whose \acp{MP} are uniquely defined (up
to \ac{LU}) by the vertices of a regular tetrahedron inscribed in the
Majorana sphere.  In the orientation shown in
\fig{tetrahedron_visual}(a) the \acp{MP} are
\begin{equation}\label{tetrahedron_mps}
  \ket{\phi_{1}} = \ket{0} \ens , \quad
  \ket{\phi_{2,3,4}} = \sqrt{\tfrac{1}{3}} \ket{0} +
  \E^{\I \kappa} \sqrt{\tfrac{2}{3}} \ket{1} \ens ,
\end{equation}
with $\kappa = 0 , \frac{2 \pi}{3} , \frac{4 \pi}{3}$.
From this the tetrahedron state follows as
\begin{equation}\label{tetrahedron_state}
  \ket{\Psi_{4}} =
  \sqrt{\tfra{1}{3}} \sym{0} + \sqrt{\tfra{2}{3}} \sym{3} \ens ,
\end{equation}
and its entanglement is $\Eg (\ket{\Psi_{4}}) = \log_2 3 \approx
1.585$.  \Fig{tetrahedron_visual}(b) shows the spherical amplitude
function $g^2 ( \theta , \varphi )$ from which it is clear that there
exist four \acp{CPP}, which coincide with the locations of the
\acp{MP}.

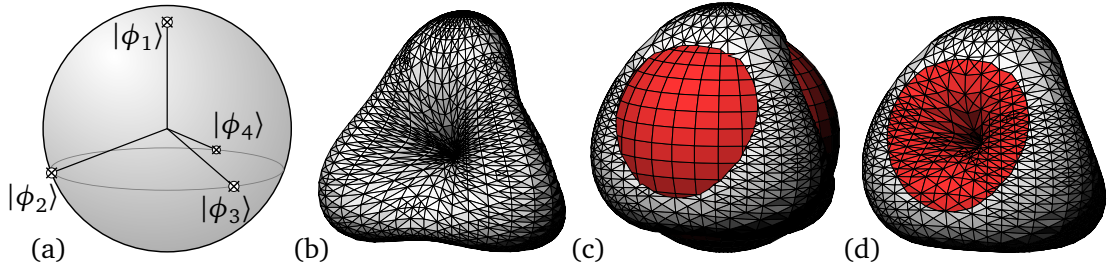
\begin{figure}
  \centering
  \begin{overpic}[scale=.45]{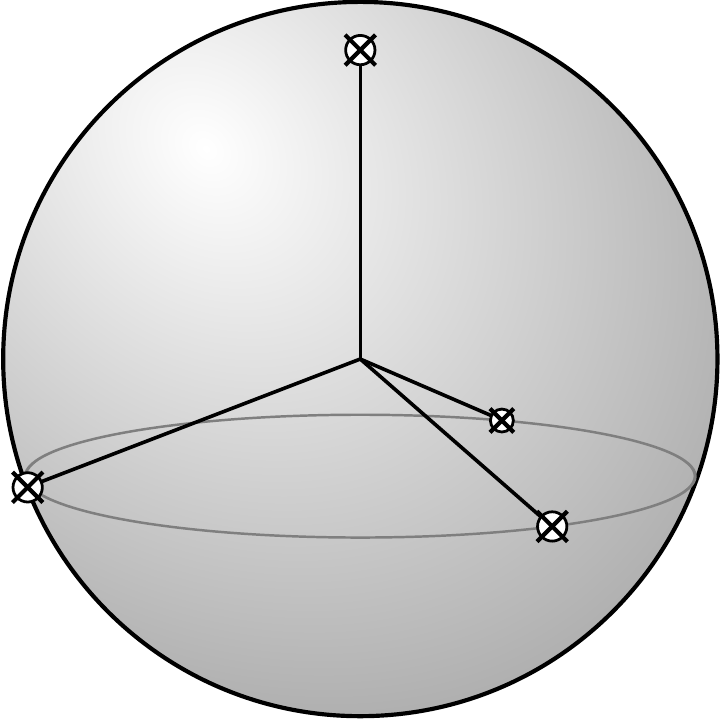}
    \put(-5,-2){(a)}
    \put(28,84){$\ket{\phi_{1}}$}
    \put(-14,19){$\ket{\phi_{2}}$}
    \put(63,14.5){$\ket{\phi_{3}}$}
    \put(68,48){$\ket{\phi_{4}}$}
  \end{overpic}
  \hspace{1.2mm}
  \begin{overpic}[scale=.4]{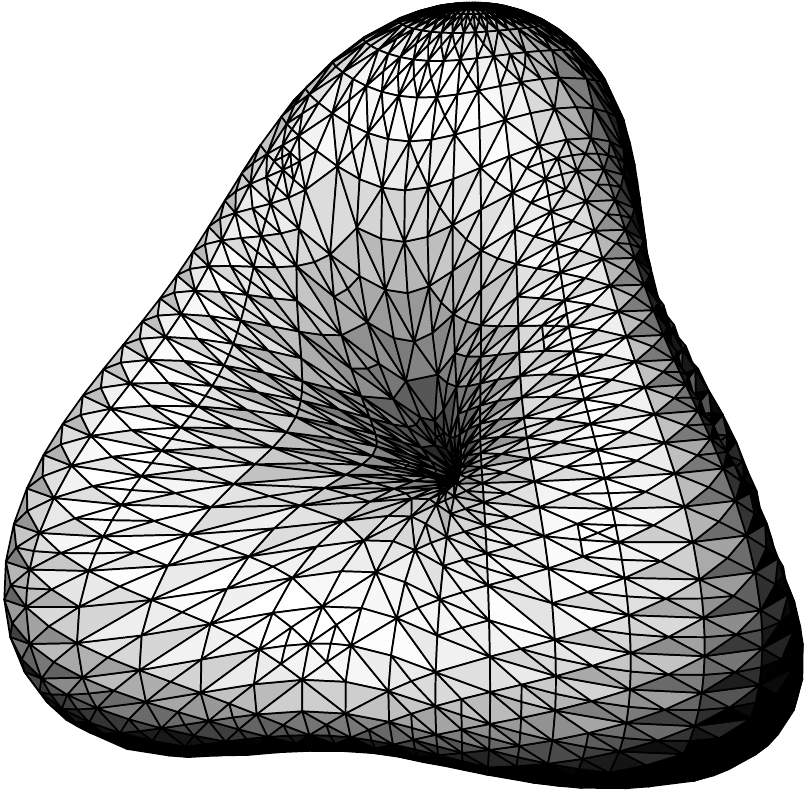}
    \put(-10,-2){(b)}
  \end{overpic}
  \hspace{1.2mm}
  \begin{overpic}[scale=.175]{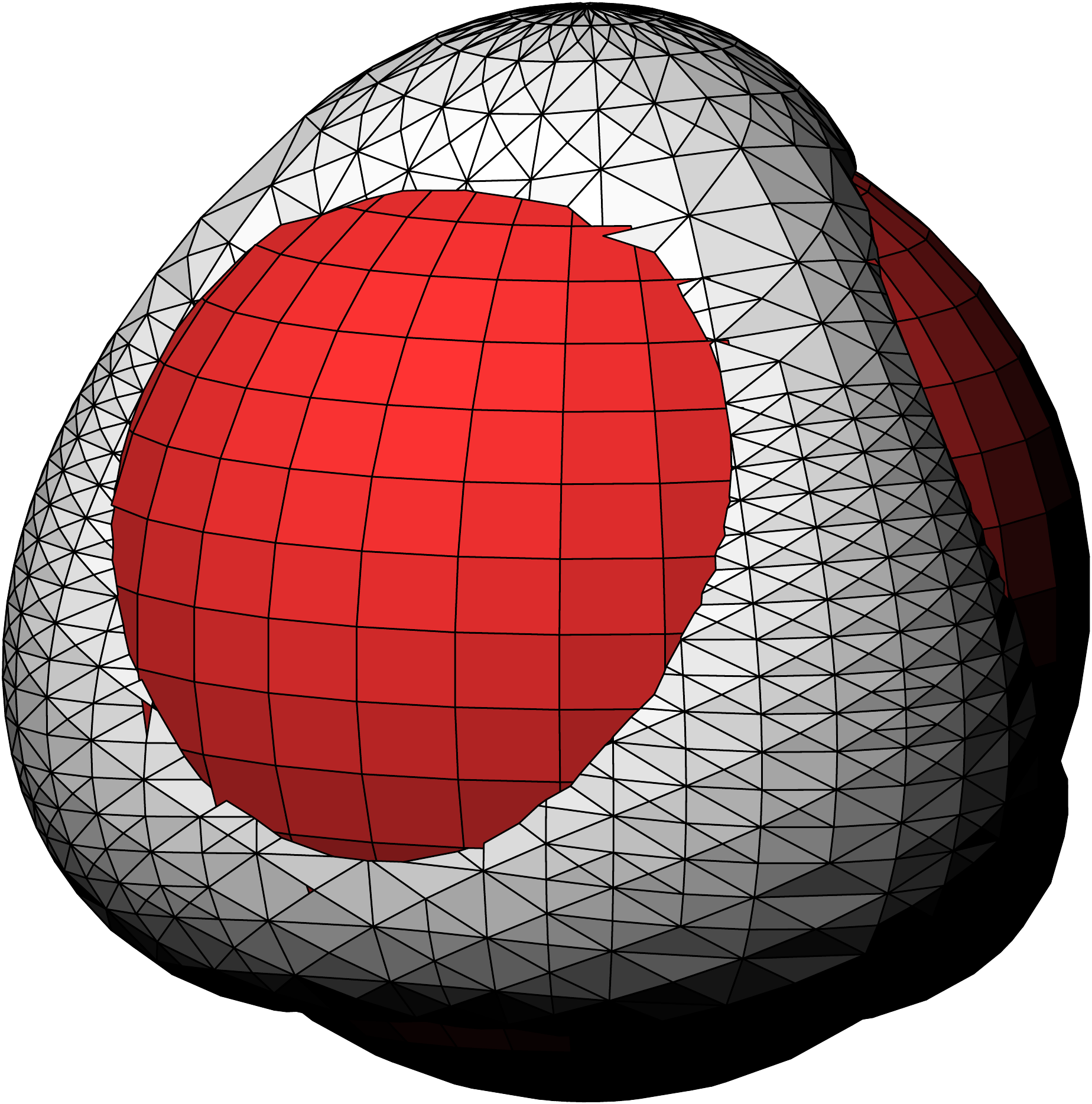}
    \put(-8,-2){(c)}
  \end{overpic}
  \hspace{1.2mm}
  \begin{overpic}[scale=.175]{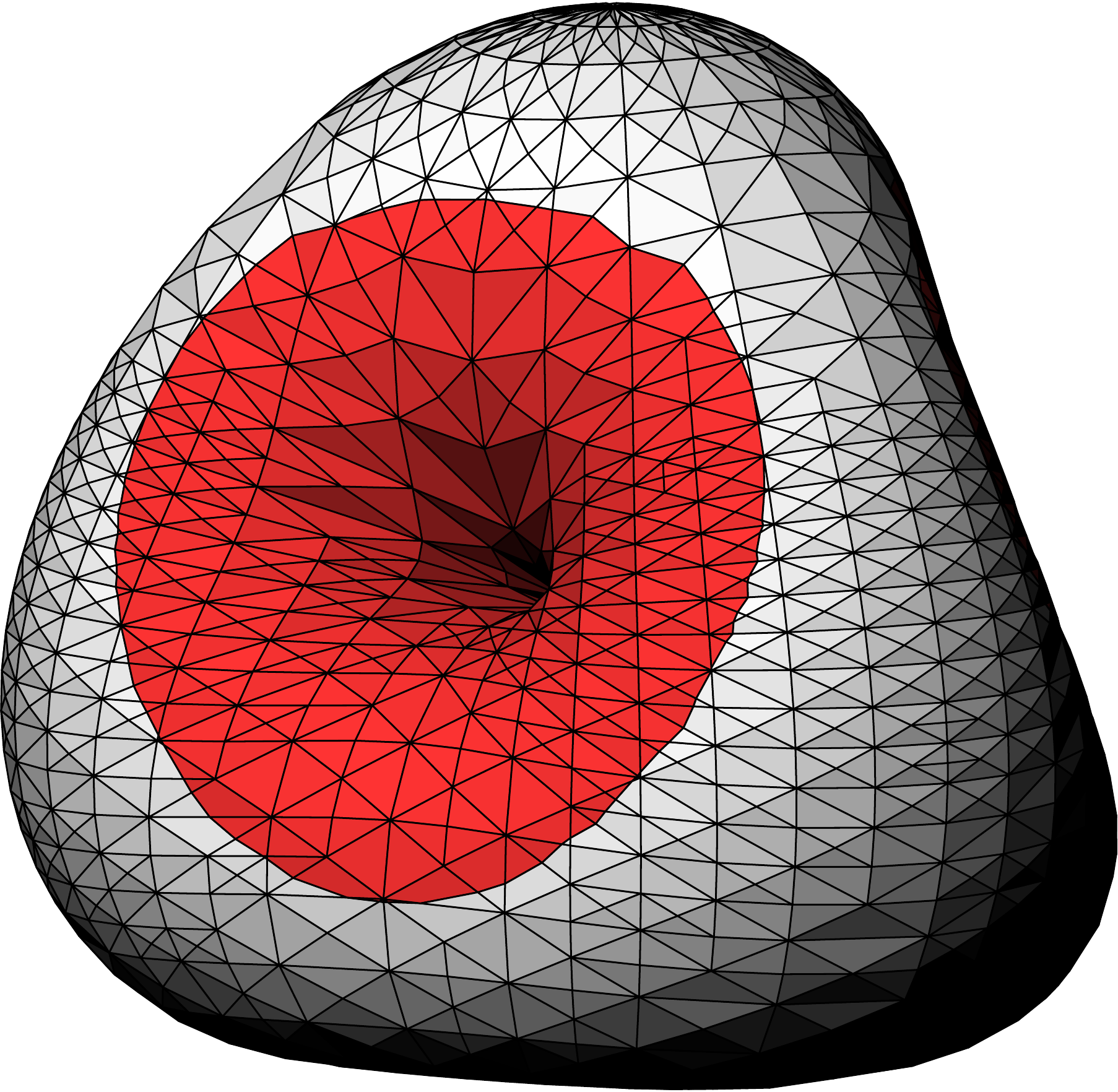}
    \put(-8,-2){(d)}
  \end{overpic}
  \caption[Visualisations of the 4 qubit tetrahedron
  state]{\label{tetrahedron_visual} Different visualisations of the
    \quo{tetrahedron state} of 4 qubits. The \acp{MP} and \acp{CPP}
    are shown in (a), with each vertex of the regular tetrahedron
    occupied by one \ac{MP} and one \ac{CPP}.  The spherical amplitude
    function $g^2 ( \theta , \varphi )$ is shown in (b), and the
    spherical volume function $g^{\frac{2}{3}} ( \theta , \varphi )$
    in (c) and (d).  The volume described by $g^{\frac{2}{3}}$ is $V =
    \frac{4 \pi}{15}$, and a red sphere with the same volume is
    inscribed in (c).  The area where the values of $g^{\frac{2}{3}}$
    are smaller than the radius $r = \sqrt[3]{1/5}$ of the sphere is
    coloured red in (d).  The largest value of $g^{\frac{2}{3}}$ is
    $G^{\frac{2}{3}} = \sqrt[3]{1/3}$.}
\end{figure}

Considering the form of the integral appearing in
\theoref{const_integral}, one could suspect that the three-dimensional
volume described by the values of the function $g^2 ( \theta , \varphi
)$ is the same for all $n$ qubit states.  This is however not the
case, because \eq{intvolume} does not describe a volume integration.
Fortunately, as stated by the following corollary, this can be easily
remedied by considering $g^{\frac{2}{3}} ( \theta , \varphi )$
instead.

\begin{corollary}\label{cor_volume}
  For every $n$ qubit symmetric state the three-dimensional volume
  bordered by the values of the function $g^{\frac{2}{3}} (\theta,
  \varphi)$ is $V = \frac{4 \pi}{3 (n+1)}$.
\end{corollary}
\begin{proof}
  The volume of an object can be determined mathematically by
  integrating the constant function 1 over the interior:
  \begin{equation}\label{calc_volume}
    \begin{split}
      V& = \int\limits_{0}^{2 \pi} \int\limits_{0}^{\pi}
      \int\limits_{0}^{R( \theta , \varphi )}
      r^2 \sin \theta \, \D r \D \theta \D \varphi =
      \int\limits_{0}^{2 \pi} \int\limits_{0}^{\pi}
      \frac{R^{3} ( \theta, \varphi )}{3}
      \sin \theta \, \D \theta \D \varphi
      \\
      {}& = \frac{1}{3} \int\limits_{0}^{2 \pi} \int\limits_{0}^{\pi}
      g^{2} ( \theta, \varphi ) \sin \theta \, \D \theta \D \varphi
      \stackrel{\eqref{intvolume}}{=}
      \frac{4 \pi}{3 (n+1)} \ens ,
    \end{split}
  \end{equation}
  where $R( \theta , \varphi ) \equiv g^{\frac{2}{3}} ( \theta ,
  \varphi )$ are the radial values in spherical coordinates.
\end{proof}
This corollary implies that the object described by the contour of
$g^{\frac{2}{3}} ( \theta , \varphi )$ has the same volume as a sphere
with radius $r = \frac{1}{\sqrt[3]{n+1}}$.  We will call
$g^{\frac{2}{3}} ( \theta , \varphi )$ the \textbf{spherical volume
  function}, and it is depicted in \fig{tetrahedron_visual} for the
tetrahedron state.

Let us make the difference between the integrals of \eq{intvolume} and
\eqref{calc_volume} more explicit.  The integral appearing in
\eq{intvolume} is a two-dimensional spherical integral over the unit
sphere with $g^2 ( \theta , \varphi )$ as its integrand.  This
integral has the value $\frac{4 \pi}{n+1}$, and because the surface
area of the unit sphere is $4 \pi$, this means that the average value
of $g^2$ is $\frac{1}{n+1}$, implying $G^2 \geq \frac{1}{n+1}$ and
thus $\Eg ( \psis ) \leq \log_2 (n+1)$.  In contrast to this,
\eq{calc_volume} describes a three-dimensional volume integral over a
shape bordered by the values of the function $g^{\frac{2}{3}} ( \theta
, \varphi )$.  This integral has the value $\frac{4 \pi}{3 (n + 1)}$,
which is equal to the volume of a sphere with radius $r =
\frac{1}{\sqrt[3]{n+1}}$.  Because of this, the largest value of
$g^{\frac{2}{3}}$ satisfies $G^{\frac{2}{3}} \geq
\frac{1}{\sqrt[3]{n+1}}$, which results in the same upper bound on
$\Eg$ as derived from the spherical integral.

Because an exponent on $g ( \theta , \varphi )$ does not change the
qualitative properties of this function, it is mostly a matter of
taste with which power to work with.  The merit of the spherical
volume function $g^{\frac{2}{3}}$ is that (for fixed $n$) its volume
is the same for all $n$ qubit symmetric states.  This allows for a
clear operational interpretation as a constant volume which has to be
moulded as uniformly as possible to obtain the highest symmetric
entanglement.  Nevertheless, most of the plots of the spherical
amplitude function displayed in this thesis for $n$ qubit symmetric
states rely on $g^2$, because the shape of the volume is then more
pronounced for lower $n$, and thus easier to interpret (see e.g.
\fig{tetrahedron_visual}(b) and (d)).

\subsection{Two and three qubit symmetric states}\label{two_three}

As a first application of the Majorana representation we review the
well-studied cases of two and three qubits. Remarkably, for two qubits
the Schmidt decomposition \eqref{schmidt_decomp} allows one to cast
every pure state as a positive symmetric state of the form $\psis =
\alpha_{0} \ket{00} + \alpha_{1} \ket{11}$, with $\alpha_{0} \geq
\alpha_{1}$ and with $\ket{\Lambda} = \ket{00}$ being a \ac{CPP}.  The
\acp{MP} are then, up to normalisation, $\ket{\phi_{1}} =
\sqrt{\alpha_{0}} \ket{0} + \I \sqrt{\alpha_{1}} \ket{1}$ and
$\ket{\phi_{2}} = \sqrt{\alpha_{0}} \ket{0} - \I \sqrt{\alpha_{1}}
\ket{1}$, and the geometric entanglement is $\Eg = - \log_2
\alpha_{0}^{2}$. The spherical distance between the two \acp{MP} is
the only degree of freedom present in this Majorana
representation\footnote{The definition $\EG = 1 -
  \abs{\bracket{\psi}{\Lambda}}^2$ makes the relationship between the
  spherical distance and the entanglement explicitly clear: With an
  angle $2 \theta$ between the two \acp{MP}, it follows that $\EG =
  \sin^2 \theta$ \protect\cite{Wei03}.}. Coinciding \acp{MP}
correspond to separable states (here $\psis = \ket{00}$), and
antipodal \acp{MP} correspond to maximally entangled states (here the
Bell state $\psis = \tfra{1}{\sqrt{2}} ( \ket{00} + \ket{11} )$) with
$\Eg = 1$. Since the \ac{CPP} is fixed at $\ket{\Lambda} = \ket{00}$,
the geometric entanglement increases monotonically with the spherical
distance between the \acp{MP}.  This unambiguous characterisation by
means of a two-point-distance is a signature of the existence of a
total order for the entanglement of pure two qubit states.

\begin{figure}
  \centering
  \begin{minipage}{138mm}
    \hfill
    \begin{overpic}[scale=.134]{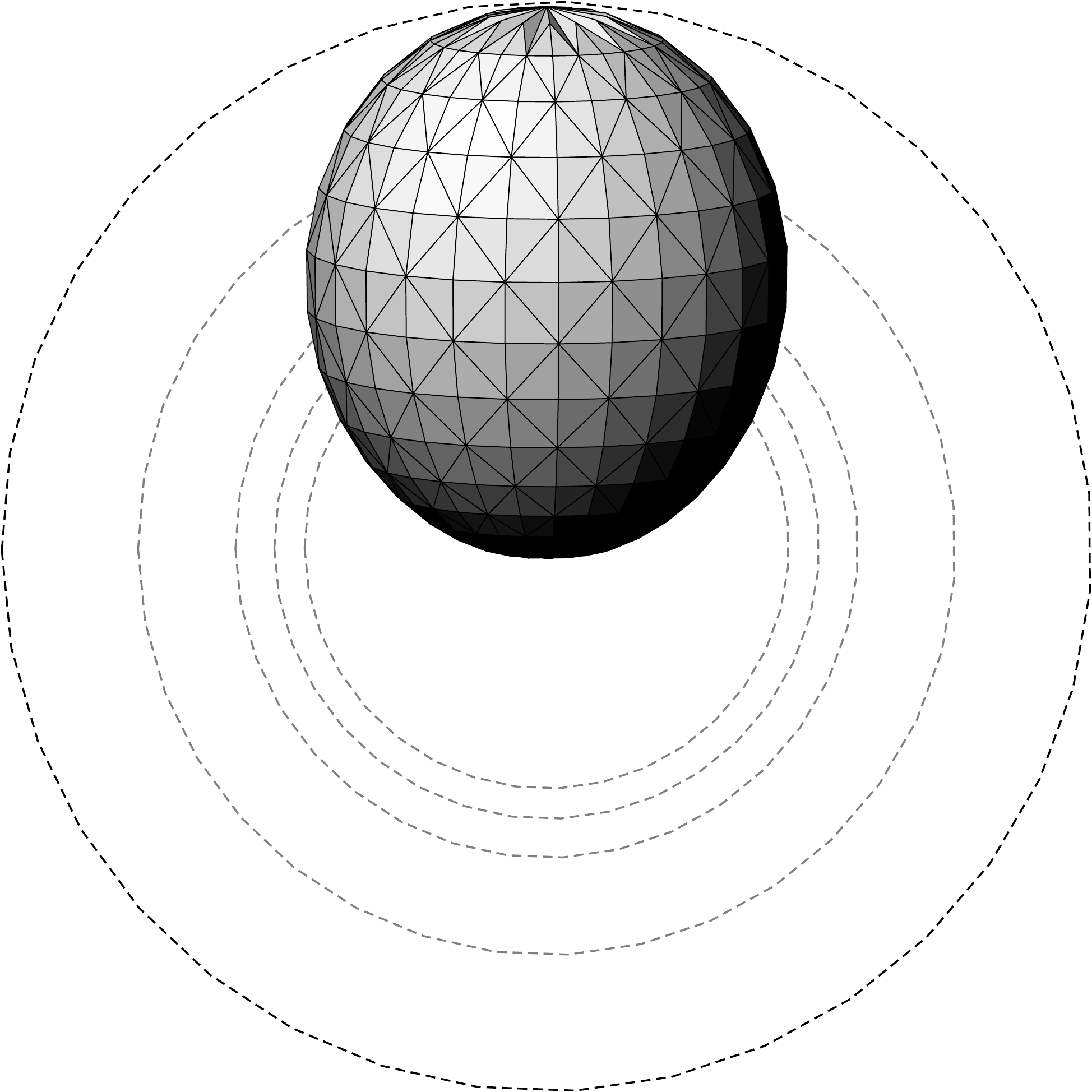}
    \end{overpic}
    \begin{overpic}[scale=.134]{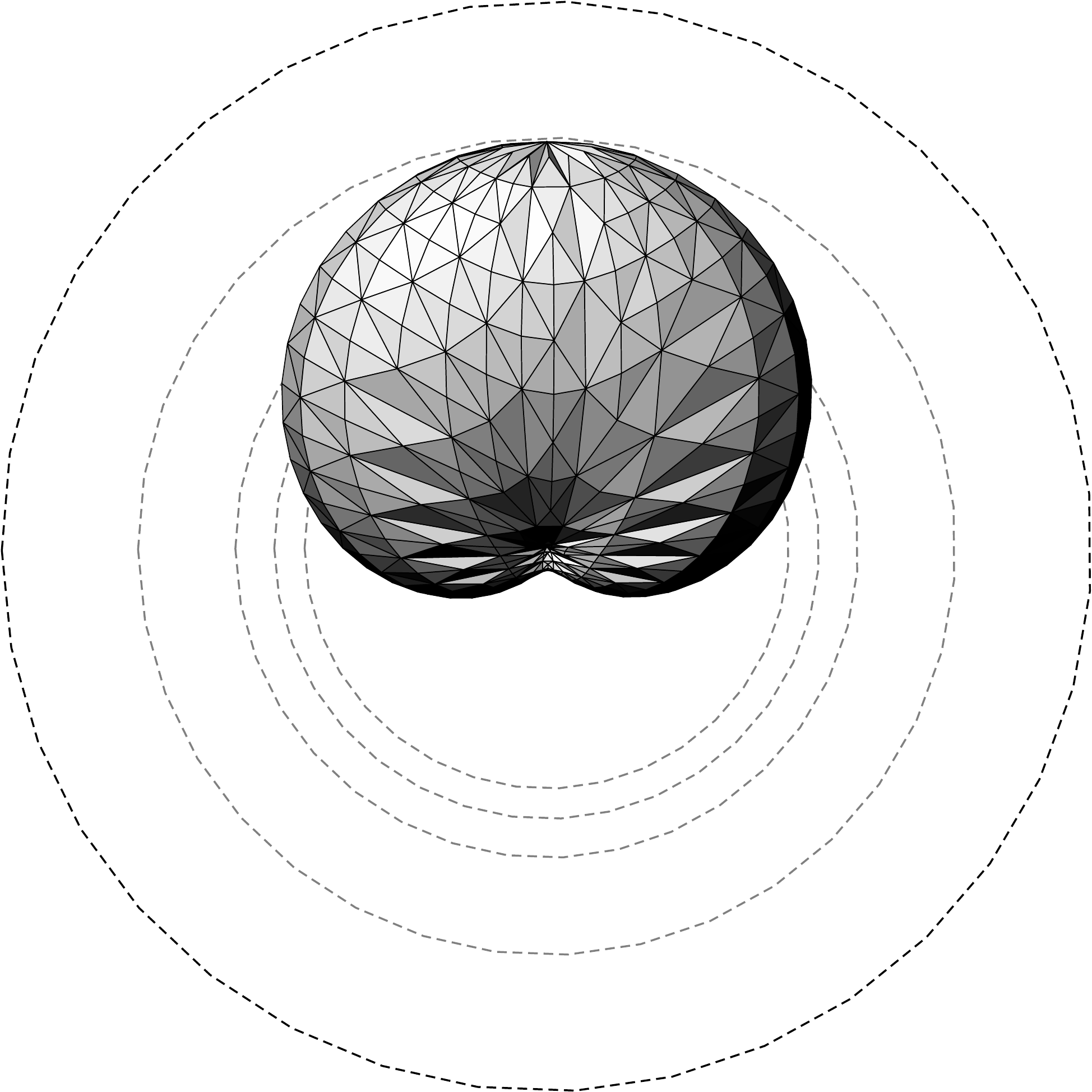}
    \end{overpic}
    \begin{overpic}[scale=.134]{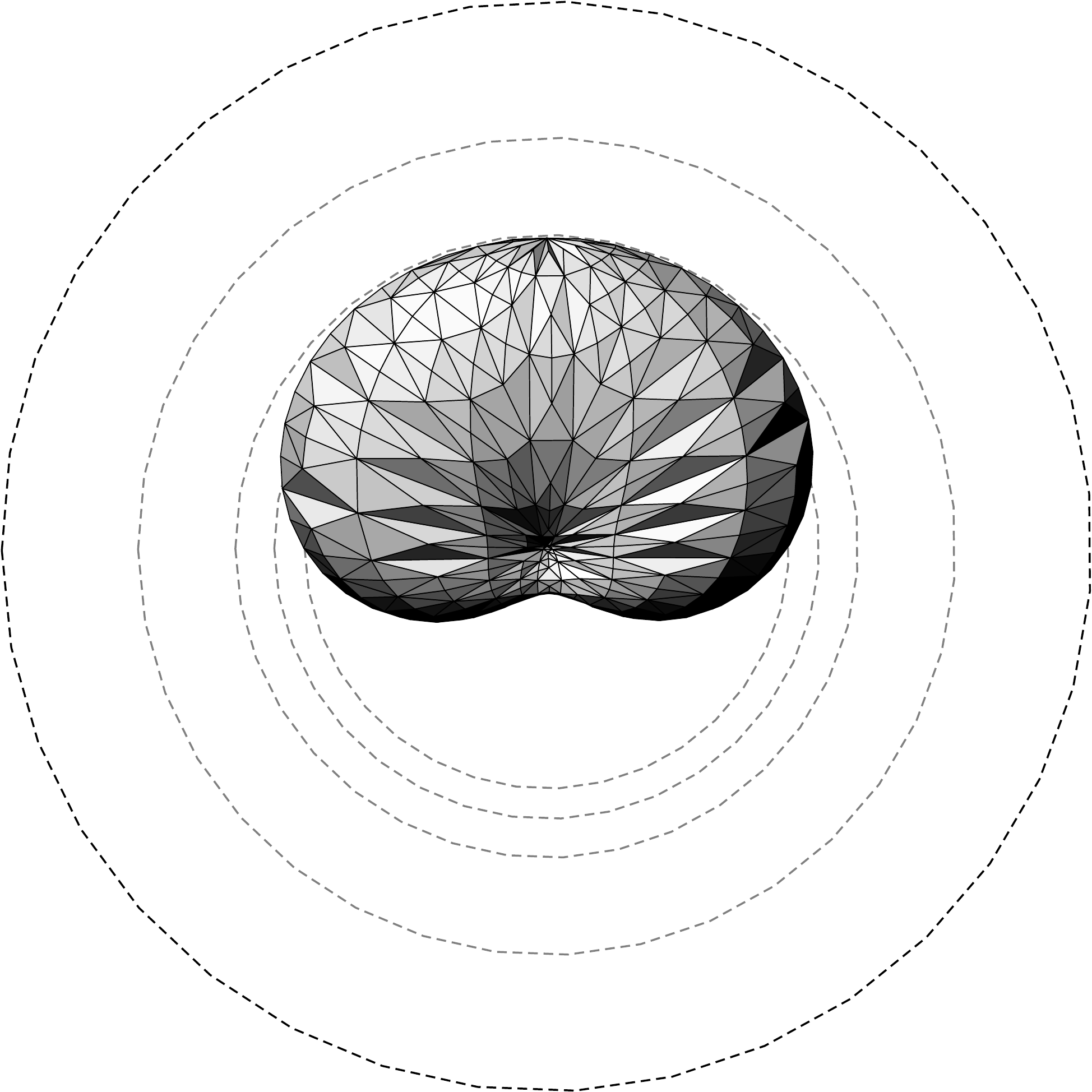}
    \end{overpic}
    \begin{overpic}[scale=.134]{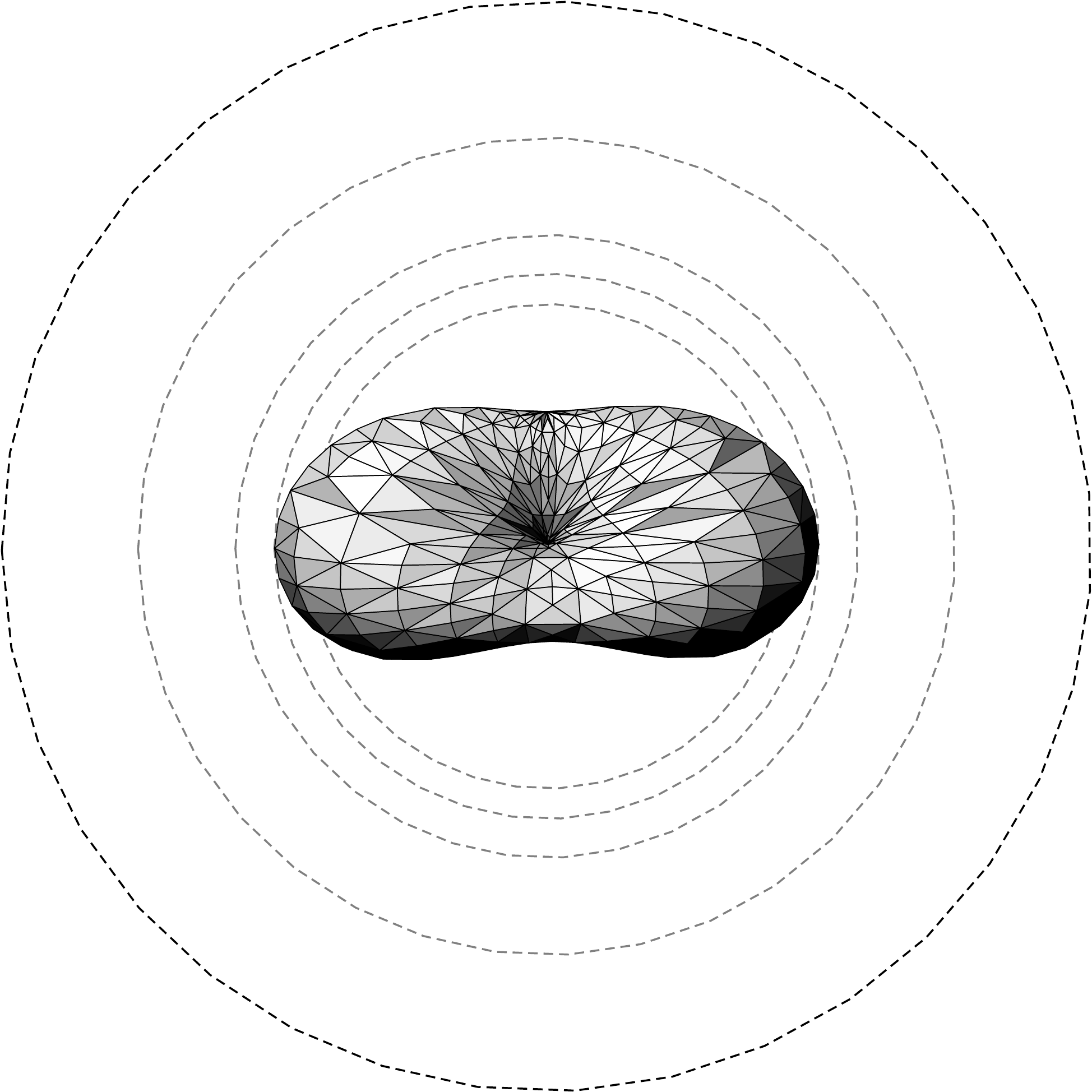}
    \end{overpic}
    \begin{overpic}[scale=.134]{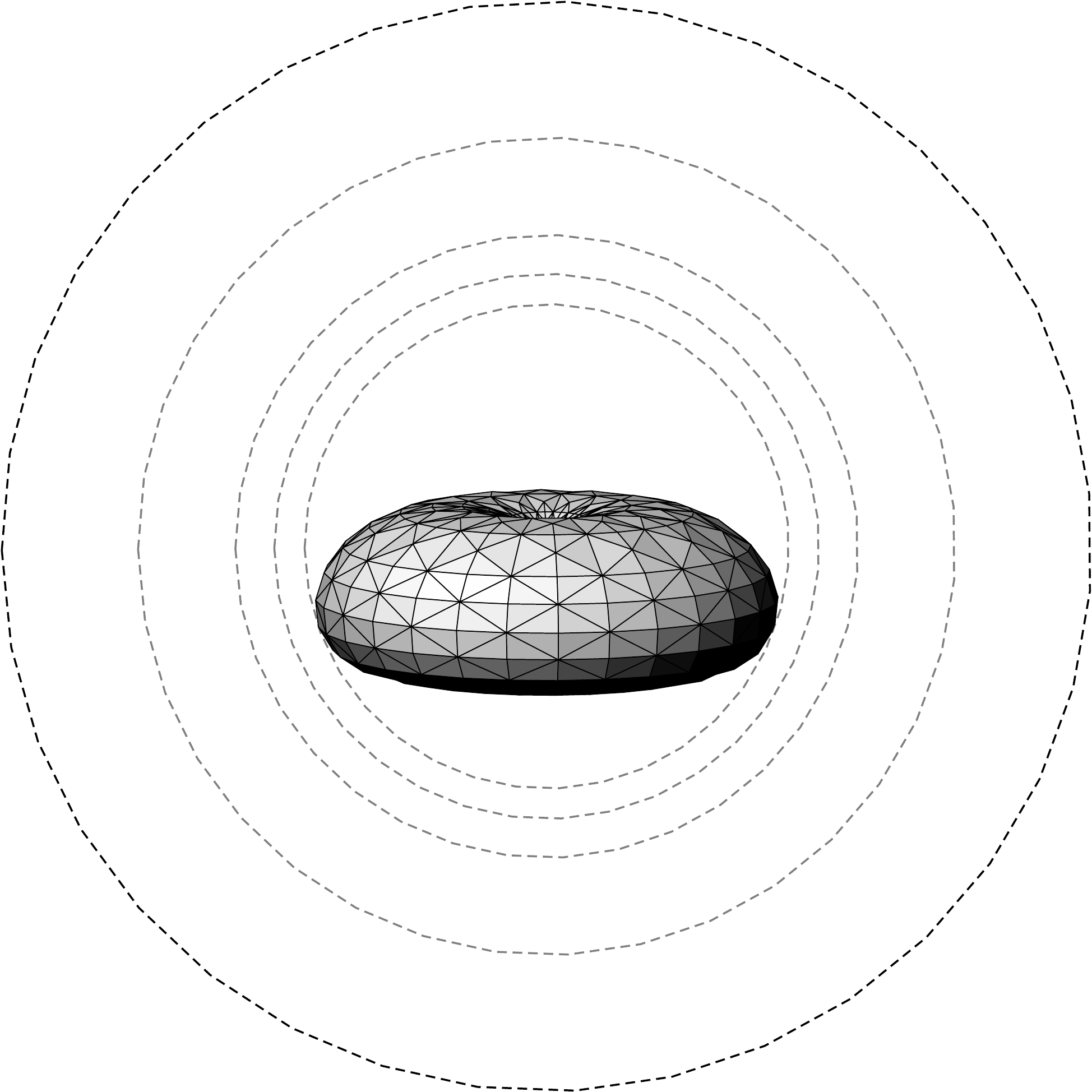}
    \end{overpic}
  \end{minipage}
  \begin{minipage}{138mm}
    \vspace{10mm}
    \hspace{2mm}
    \begin{overpic}[scale=.32]{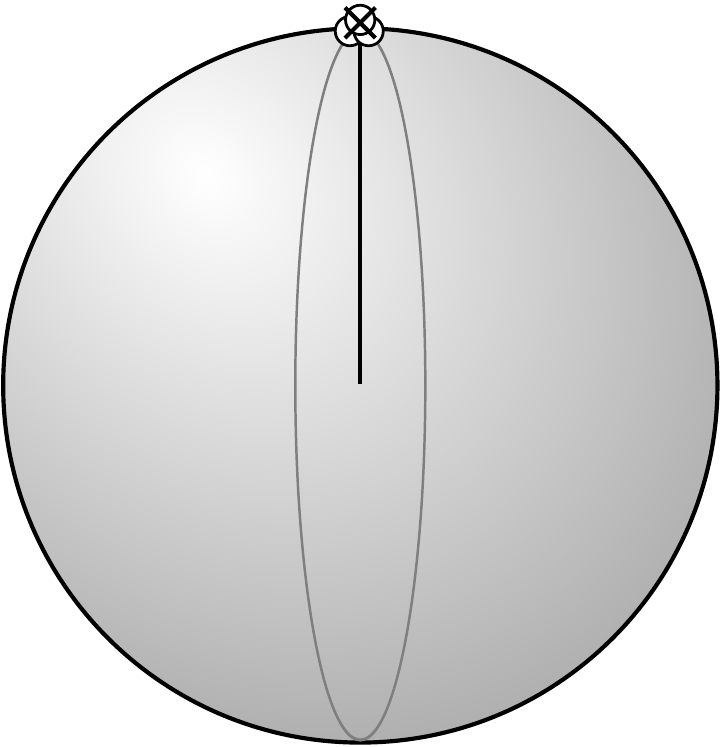}
      \put(-8,0){(a)}
      \put(35,113){$\sym{0}$}
    \end{overpic}
    \hspace{2mm}
    \begin{overpic}[scale=.32]{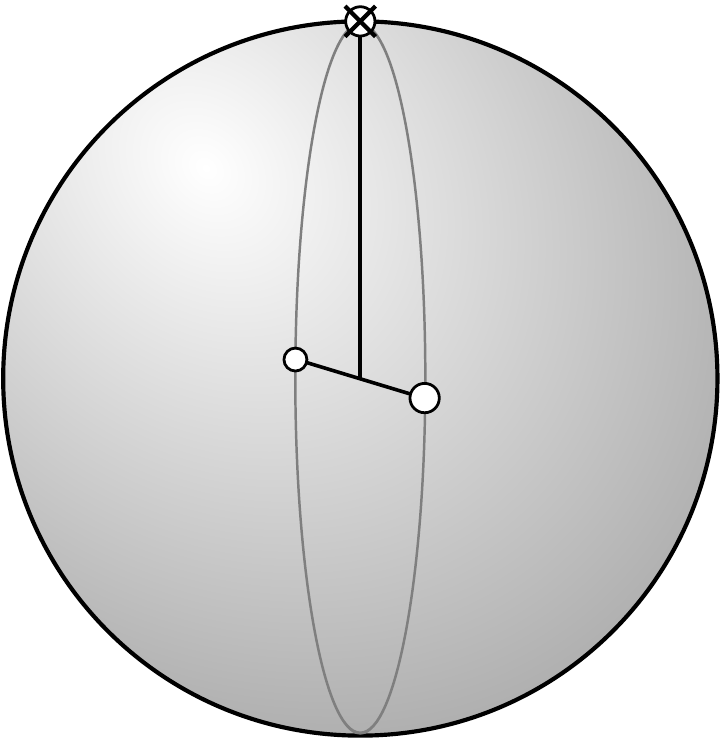}
      \put(-8,0){(b)}
      \put(2,113){$\sqrt{3} \sym{0} + \sym{2}$}
    \end{overpic}
    \hspace{2mm}
    \begin{overpic}[scale=.32]{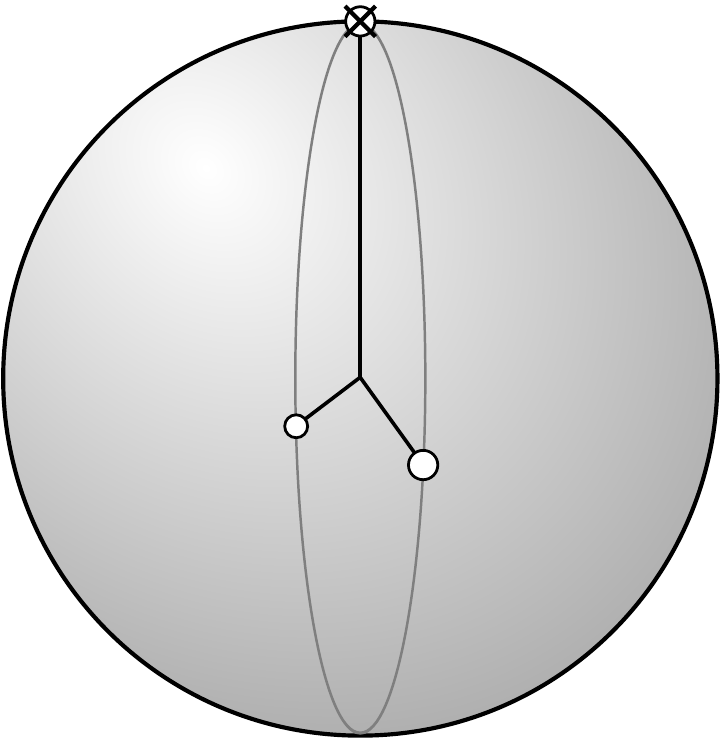}
      \put(-8,0){(c)}
      \put(1,113){$2 \sym{0} + \sqrt{3} \sym{2}$}
    \end{overpic}
    \hspace{2mm}
    \begin{overpic}[scale=.32]{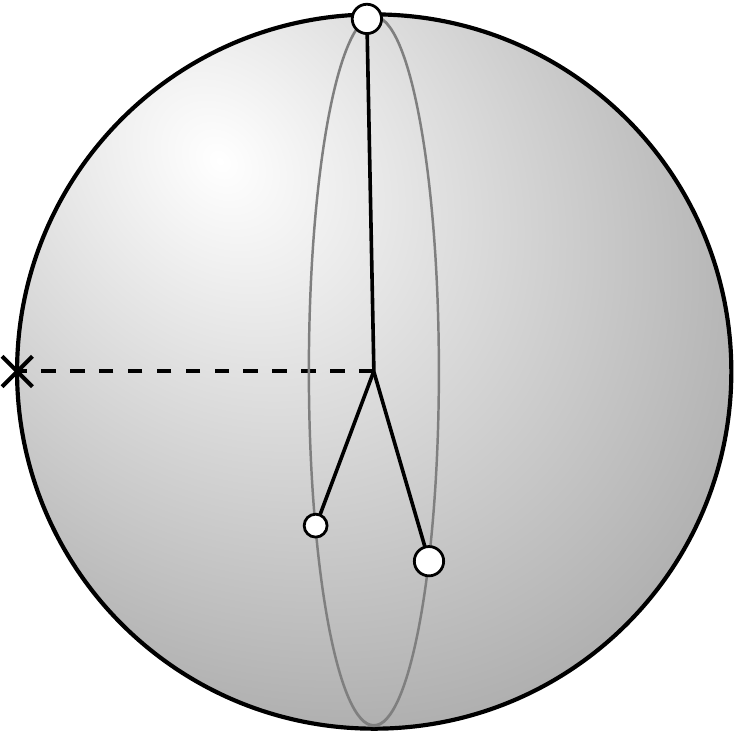}
      \put(-8,0){(d)}
      \put(7,113){$\sym{0} + \sqrt{3} \sym{2}$}
    \end{overpic}
    \hspace{2mm}
    \begin{overpic}[scale=.32]{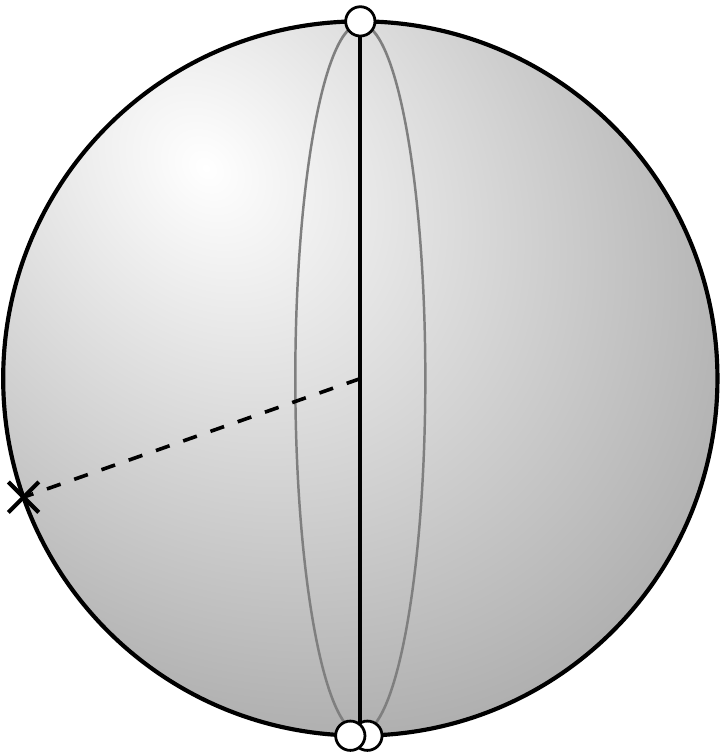}
      \put(-8,0){(e)}
      \put(36,113){$\sym{2}$}
    \end{overpic}
  \end{minipage}
  \begin{minipage}{138mm}
    \centering
    \vspace{10mm}
    \begin{overpic}[scale=1.05]{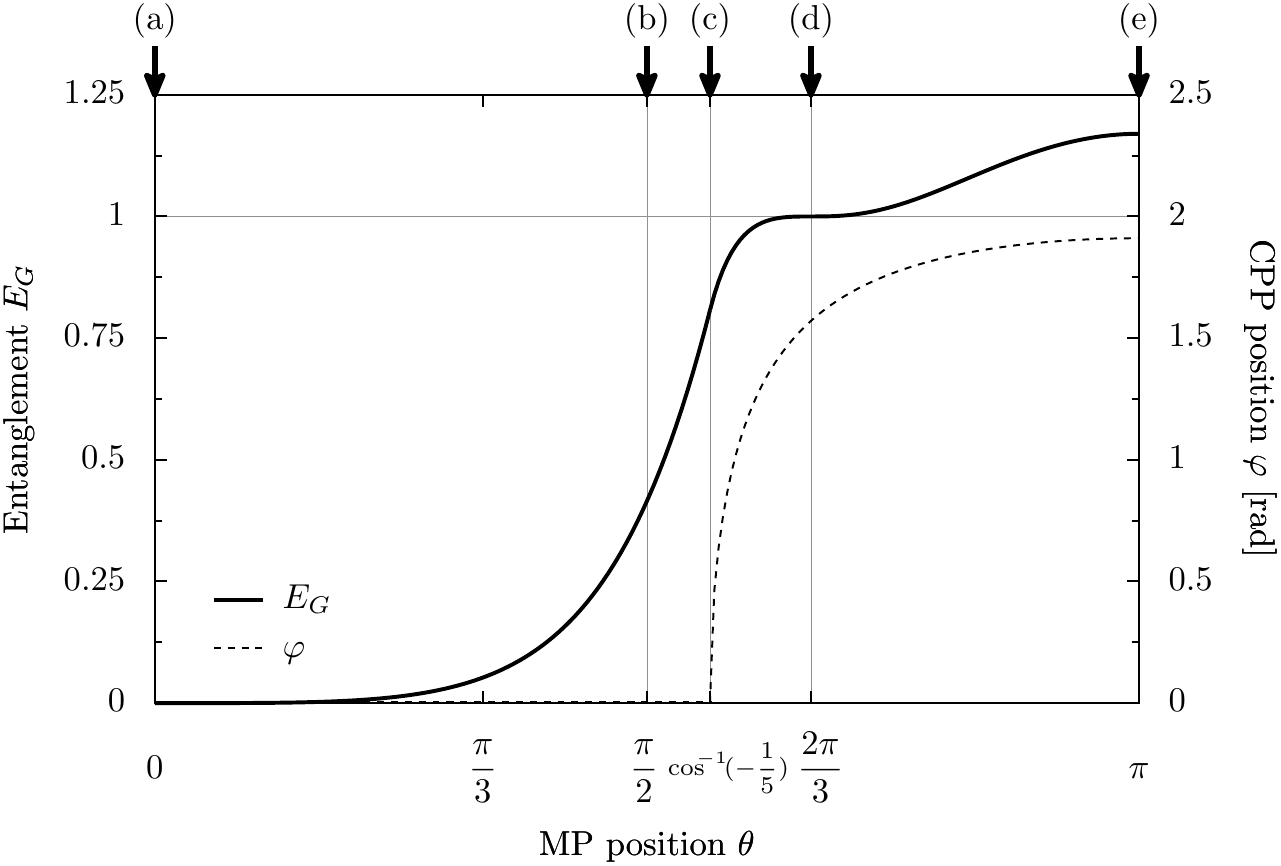}
    \end{overpic}
  \end{minipage}
  \caption[Geometric entanglement and CPPs of 3 qubit symmetric
  states]{\label{3_graph} The diagram shows how the values of $\Eg$
    and the location of the positive \ac{CPP} change as the \ac{MP}
    distribution of \protect\eqref{mp_3_form} is modified.  The
    \ac{CPP} remains on the north pole until the moving \acp{MP} have
    reached a latitude slightly below the equator, as seen in the
    Majorana representation (c). From that point onwards the \ac{CPP}
    rapidly moves southwards and reaches the equator at the \ac{GHZ}
    state (d). After this, the \ac{CPP} and $\Eg$ undergo only small
    changes until the W state (e) is reached.  Plots of the spherical
    amplitude function $g^2$ for the five marked states are shown on
    top of their Majorana representations. The values of $G^2$ are
    $G^2_{\text{a}} = 1$, $G^2_{\text{b}} = \tfra{3}{4}$,
    $G^2_{\text{c}} = \tfra{4}{7}$, $G^2_{\text{d}} = \frac{1}{2}$ and
    $G^2_{\text{e}} = \tfra{4}{9}$, and they give the radii of the
    dashed concentric circles.}
\end{figure}

In contrast to this, there is no unique way in geometry to measure the
\quo{distance} between three points on a sphere.  Similarly, no unique
entanglement measure and no total order exists for three qubit states.
Furthermore, unlike the two qubit case, a generic state of three
qubits cannot be cast positive or symmetric \cite{Dur00}.  We
therefore focus on a subset of three qubit symmetric states that
include highly entangled states.  For this consider the following
three \acp{MP}
\begin{equation}\label{mp_3_form}
  \ket{\phi_{1}} = \ket{0} \ens , \quad \ket{\phi_{2, 3}} =
  \co_{\theta} \ket{0} \pm \I \, \si_{\theta} \ket{1} \ens ,
\end{equation}
with the parametrisation $\theta \in [0,\pi]$.  Starting out with all
\acp{MP} on the north pole ($\theta = 0$), two of the \acp{MP} are
moved southwards as a complex conjugate pair until they reach the
south pole ($\theta = \pi$). The change of the \acp{CPP} and the
entanglement is studied as a function of $\theta$.  From
\eq{majorana_definition} it is found that the \acp{MP} give rise to
the state
\begin{equation}\label{3_mp_state1}
  \ket{\psi^{\text{s}} ( \theta )}  =
  \frac{\sqrt{3} \co_{\theta}^2 \sym{0} +
    \si_{\theta}^2 \sym{2} }
  {\sqrt{3 \co_{\theta}^4 + \si_{\theta}^4 }} \ens .
\end{equation}
This state is positive, so it suffices to find a positive \ac{CPP}.
Determining the absolute maximum of the spherical amplitude function
\eqref{spher_amp_funct} with the ansatz $\ket{\sigma} = \co_{\varphi}
\ket{0} + \si_{\varphi} \ket{1}$ is straightforward, yielding the
relationship between $\varphi$ and $\theta$:
\begin{equation}\label{3_mp_state2}
  \co_{\varphi}^2 =
  \frac{\si_{\theta}^2}{6 \si_{\theta}^2 - 3} \ens .
\end{equation}
The codomain of the left-hand side is $[0,1]$, but the right-hand side
lies outside this range for $\theta < \arccos (- \tfra{1}{5})$.  For
these values $\ket{\sigma} = \ket{0}$ is the only \ac{CPP}.
\Fig{3_graph} shows the change of $\varphi$ with $\theta$.  It is seen
that from $\theta = \arccos (- \tfra{1}{5})$ onwards the \ac{CPP}
abruptly leaves the north pole and moves towards the south pole along
the positive half-circle.  This behaviour can be explained with the
shape of the spherical amplitude function.  As seen in
\fig{3_graph}(c), the function $g^2$ is very flat around the north
pole which facilitates fast changes in the position of the global
maximum.  From \eq{3_mp_state1} and \eqref{3_mp_state2} the \ac{GM}
can be calculated and its graph is displayed in \fig{3_graph}. It is
found that $\Eg$ is monotonously increasing, which is in accordance
with the results of \cite{Tamaryan09}. Interestingly, the entanglement
reaches a saddle point at the \ac{GHZ} state ($\theta = \tfra{2
  \pi}{3}$) before it peaks at the W state ($\theta = \pi$).

\subsection{Totally invariant states and additivity}\label{invariant_and_additivity}

Quantum states that are the stationary points of an energy functional
regardless of the parameter values of the underlying system are called
\textbf{inert} states.  The inert states of spin-$j$ systems have been
fully characterised by their \acp{MP}: A state is inert \ac{iff} its
\ac{MP} distribution is invariant under a subgroup of the rotation
group $\text{SO}(3)$ acting on the Majorana sphere, but any small
variation of the \acp{MP} (excluding the joint rotations
\eqref{symmetric_lu}) results in a change of the symmetry group
\cite{Makela07}.  Because of the isomorphism between the states of a
spin-$j$ particle and the symmetric states of $2j$ qubits, this
definition can be extended to symmetric $n$ qubit states.  To avoid
confusion, the physically motivated term \quo{inert} is replaced with
\quo{\textbf{totally invariant}} \cite{Markham11}.  Regardless of the
underlying physical system, an \ac{MP} distribution is thus called
totally invariant if it is invariant under a subgroup of
$\text{SO}(3)$, but any small variation of the \acp{MP} changes the
subgroup.

Examples of totally invariant states are the \textbf{Platonic states},
which are defined as the quantum states whose \acp{MP} lie at the
vertices of the Platonic solids, the five highly symmetric convex
polyhedra whose edges, vertices and angles are all congruent.  The
tetrahedron state was already introduced as a four qubit symmetric
state in \sect{visualisation}. Treated as a state of a spin-$2$
system, the tetrahedron state represents an inert state.

Taking \ac{LU} equivalence into account, the subgroups of
$\text{SO}(3)$ and their symmetry implications can be listed as
follows:
\begin{itemize}
  \verycompactlist
\item special orthogonal group $\text{SO}(2)$: continuous $Z$-axis
  rotational symmetry
  
\item orthogonal group $\text{O}(2)$: continuous $Z$-axis rotational
  symmetry \& $X$-$Y\!$-plane symmetry
  
\item cyclic group $C_{m}$: discrete $Z$-axis rotational symmetry
  
\item dihedral group $D_{m}$: discrete $Z$-axis rotational symmetry \&
  $X$-$Y\!$-plane symmetry
  
\item tetrahedral group $T$: symmetry group of tetrahedron
  
\item octahedral group $O$: symmetry group of octahedron and cube
  
\item icosahedral group $Y$: symmetry group of icosahedron and
  dodecahedron
\end{itemize}
The continuous symmetries $\text{O}(2)$ and $\text{SO}(2)$ can be
fulfilled only by Dicke states.  The Dicke state $\sym{n,k}$ is a
totally invariant state of $\text{SO}(2)$ for all $n$ and $k$, with
the exception of $k = \frac{n}{2}$ for even $n$.  This is because the
equally balanced states $\sym{n,\frac{n}{2}}$ are the totally
invariant states of $\text{O}(2)$.  There are no totally invariant
states for the cyclic group $\text{C}_{m}$, but the remaining groups
$D_{m}$, $T$, $O$ and $Y$ all give rise to a multitude of totally
invariant states \cite{Makela07,Markham11}.

Markham \cite{Markham11} recently found that all totally invariant
symmetric $n$ qubit states satisfy \eq{meas_eq}, i.e. their amount of
entanglement is the same for the three distance-like entanglement
measures:
\begin{lemma}\label{invariant_equivalent}
  Let $\rho = \pure{\psi^{\text{s}}}$ be a totally invariant symmetric
  $n$ qubit pure state. Then
  \begin{equation}\label{invariant_equiv}
    \Eg ( \rho ) = E_{\text{R}} ( \rho ) =
    E_{\text{Rob}} ( \rho ) \ens .
  \end{equation}
\end{lemma}
Zhu \etal \cite{Zhu10} showed that positive states are strongly
additive under the \ac{GM}:
\begin{lemma}\label{positive_additive}
  Let $\rho$ be a positive state, pure or mixed. Then $\rho$ is
  strongly additive under $\Eg$, i.e. the following holds for all
  states $\sigma$:
  \begin{equation}\label{pos_additive}
    \Eg ( \rho \otimes \sigma ) =
    \Eg ( \rho ) + \Eg ( \sigma ) \ens .
  \end{equation}
\end{lemma}

In general the three measures $\Eg$, $E_{\text{R}}$ and
$E_{\text{Rob}}$ are not additive in the multipartite scenario
(cf. \cite{Werner02,Vollbrecht01}), and for the \ac{GM} it was shown
that beyond a certain amount of entanglement (which is present in
almost all states) states can never be strongly additive \cite{Zhu10}.
\Lemref{positive_additive} is not in conflict with this, since we
already argued at the end of \sect{pos_results} that positive states
are in general considerably less entangled than generic states.  Chen
\etal \cite{Chen11} found that symmetric states whose \acp{MP} are all
distributed within some half sphere are additive with respect to the
\ac{GM}.  In particular, this implies the additivity of all two and
three qubit symmetric states.  For larger $n$, however, it is clear
from the Majorana representation and the spherical amplitude function
that states with such an imbalance in their \ac{MP} distribution
cannot have much geometric entanglement.  We can therefore conclude
that \emph{additivity of states under the \ac{GM} is a signature of
  low entanglement.}

What can be said about the additivity of symmetric states under the
relative entropy of entanglement $E_{\text{R}}$ and the logarithmic
robustness of entanglement $E_{\text{Rob}}$?  We will combine
\lemref{invariant_equivalent} and \lemref{positive_additive} to show
that many symmetric states of interest are additive under
$E_{\text{R}}$ and $E_{\text{Rob}}$ in a sense of additivity that is
stronger than regular additivity, but weaker than strong additivity.
For this we will use the quantity $\Egt = \Eg - S$ which was
introduced in \sect{gm_def} and which coincides with $\Eg$ for pure
states.

\begin{theorem}\label{pos_inv}
  Let $\rho = \pure{\psi^{\text{s}}}$ be a pure symmetric $n$ qubit
  state that is positive and totally invariant.  Then $\rho$ is
  strongly additive under $\Eg$ and additive under $E_{\text{R}}$ and
  $E_{\text{Rob}}$. Furthermore, for arbitrary states $\sigma$ the
  following holds
  \begin{subequations}\label{sesquiadditivity}
    \begin{align}
      \Egt ( \sigma )& = E_{\text{R}} ( \sigma )&
      {}& \Longrightarrow&
      E_{\text{R}} ( \rho \otimes \sigma )& =
      E_{\text{R}} ( \rho ) + E_{\text{R}} ( \sigma )
      \ens , \label{sesquiadditivity1} \\
      \Egt ( \sigma )& = E_{\text{Rob}} ( \sigma )&
      {}& \Longrightarrow&
      E_{\text{Rob}} ( \rho \otimes \sigma )& =
      E_{\text{Rob}} ( \rho ) + E_{\text{Rob}} ( \sigma )
      \ens . \label{sesquiadditivity2}
    \end{align}
  \end{subequations}
\end{theorem}
\begin{proof}
  It is known \cite{Vedral98} or obvious that the three measures
  $\Eg$, $E_{\text{R}}$ and $E_{\text{Rob}}$ are subadditive, i.e. $E
  ( \rho \otimes \sigma ) \leq E ( \rho ) + E ( \sigma )$ for
  arbitrary $\rho$ and $\sigma$, and that the von Neumann entropy is
  strongly additive, i.e. $S ( \rho \otimes \sigma ) = S ( \rho ) + S
  ( \sigma )$ for arbitrary $\rho$ and $\sigma$.  Now let $\rho$ be a
  pure symmetric $n$ qubit state that is positive and totally
  invariant, and let $\sigma$ be an arbitrary state for which $\Egt (
  \sigma ) = E_{\text{x}} ( \sigma )$ holds, where $E_{\text{x}}$ can
  be either $E_{\text{R}}$ or $E_{\text{Rob}}$.  Then
  \begin{equation*}
    \Eg ( \rho ) + \Egt ( \sigma ) 
    \stackrel{\eqref{pos_additive}}{=}
    \Egt ( \rho \otimes \sigma )
    \stackrel{\eqref{meas_ineq}}{\leq}
    E_{\text{x}} ( \rho \otimes \sigma ) \leq
    E_{\text{x}} ( \rho ) + E_{\text{x}} ( \sigma )
    \stackrel{\eqref{invariant_equiv}}{=}
    \Eg ( \rho ) + \Egt ( \sigma ) \ens ,
  \end{equation*}
  which implies that $E_{\text{x}} ( \rho \otimes \sigma ) =
  E_{\text{x}} ( \rho ) + E_{\text{x}} ( \sigma )$.  The strong
  additivity of $\rho$ under $\Eg$ is clear from
  \lemref{positive_additive}, and the additivity of $\rho$ under
  $E_{\text{R}}$ or $E_{\text{Rob}}$ follows as a special case from
  the previous equation by setting $\sigma := \rho$.
\end{proof}
The main result of this theorem is that a considerable amount of
symmetric states is additive under $E_{\text{R}}$ and
$E_{\text{Rob}}$.  Trivial examples are the Dicke states $\sym{n,k}$
and the \ac{GHZ} states $\frac{1}{\sqrt{2}} ( \sym{n,0} + \sym{n,n}
)$, which are positive and totally invariant, thus satisfying the
conditions of \theoref{pos_inv}. In \chap{solutions} it will be seen
that for systems with a low number of qubits many highly or maximally
entangled symmetric states are positive as well as totally invariant.
This automatically results in the interesting property that these
states are additive and equivalent under the three distance-like
entanglement measures.

The strong additivity under $E_{\text{R}}$ and $E_{\text{Rob}}$ could
not be proven, but the statements \eqref{sesquiadditivity1} and
\eqref{sesquiadditivity2} represent a considerable extension of the
regular additivity.  The necessary condition for this is automatically
fulfilled by states that fulfil \eq{meas_eq}, in particular stabiliser
states, Dicke states, \permantisymm basis states
\cite{Hayashi06,Hayashi08,Markham07} and totally invariant symmetric
states \cite{Markham11}.  In this way we were able to extend the set
of symmetric states that is known to be additive under $E_{\text{R}}$
and $E_{\text{Rob}}$.  It still remains an open question, however,
whether arbitrary symmetric states are additive or even strongly
additive under the entanglement measures discussed here. This is one
of the open questions put forward in the conclusion of \cite{Zhu10},
and to date no counterexamples for the additivity of symmetric states
are known.

Finally, we remark that \theoref{pos_inv} could also have been
formulated by omitting the requirement of positivity, and instead
requiring that the \acp{MP} of $\rho$ are all confined to some
half-sphere on the Majorana sphere, including the bordering great
circle. This property would then guarantee the additivity of $\rho$
under \ac{GM} \cite{Chen11} required to prove all implications of the
theorem (except the strong additivity under \ac{GM}).  However, it is
clear that the only totally invariant symmetric states whose \acp{MP}
occupy at most half of the Majorana sphere are the Dicke states and
the \ac{GHZ} states.  Since these states are positive, they are
already accounted for in the given formulation of \theoref{pos_inv},
thus making the alternative formulation redundant.

\section{Extremal point distributions}\label{extremal_point}

For symmetric states the injective tensor norm appearing in the
definition \eqref{geo_def} of the \ac{GM} can be concisely expressed
in terms of the \acp{MP} and one \ac{CPS} $\ket{\Lambda} =
\ket{\sigma}^{\otimes n}$:
\begin{equation}\label{bloch_product}
  \abs{\bracket{\psi^{\text{s}}}{\Lambda}} =
  \frac{n!}{\sqrt{K}} \prod_{i=1}^{n} \,
  \abs{\bracket{\phi_{i}}{\sigma}} \ens .
\end{equation}
This is precisely the global maximum of the spherical amplitude
function \eqref{spher_amp_funct}.  Therefore, to determine the
\ac{CPP} of a given symmetric state, the absolute value of a product
of scalar products has to be maximised. From a geometrical point of
view, the factors $\bracket{\phi_{i}}{\sigma}$ are the angles between
the corresponding Bloch vectors on the Majorana sphere, and thus the
determination of the \ac{CPP} can be viewed as an optimisation problem
for a product of geometrical angles.

From a comparison with the min-max problem \eqref{minmax} of the
general case it is clear that the task of finding the maximally
entangled symmetric states can be concisely formulated as the
geometrical optimisation problem
\begin{equation}\label{maj_problem}
  \min_{ \{ \ket{\phi_{i}} \} } \frac{1}{\sqrt{K}}
  \left( \max_{ \ket{\sigma} } \prod_{i=1}^{n}  \,
    \abs{ \bracket{\phi_{i}}{\sigma}} \right) \ens .
\end{equation}
In other words, the maximum value of the spherical amplitude function
must be as small as possible.  This \textbf{Majorana problem} bears
all the properties of an optimisation problem on the surface of a
sphere in $\mbbrr$.  These kinds of problems deal with arrangements of
a finite number of points on a sphere so that an extremal property is
fulfilled \cite{Whyte52}.  There are infinite possibilities to define
such optimisation problems, but two particularly well-known problems
that have been extensively studied in the past are the following:

\textbf{\toths problem,} also known as Fejes' problem and Tammes'
problem, asks how $n$ points have to be distributed on the unit sphere
so that the minimum distance of all pairs of points becomes maximal
\cite{Whyte52}. This problem was first raised by the biologist Tammes
in 1930 while trying to explain the observed distribution of pores on
pollen grains \cite{Tammes30}. Recasting the $n$ points as unit
vectors $\bmr{r}_{i} \in \mbbrr$, the following cost function needs to
be maximised:
\begin{equation}
  f_{\text{\tothx}} ( \bmr{r}_{1} , \bmr{r}_{2} , \dots , \bmr{r}_{n} )
  = \min_{i < j} \, \abs{ \bmr{r}_{i} - \bmr{r}_{j} } \ens .
\end{equation}
The point configurations that solve this problem are called spherical
codes or sphere packings \cite{WeissteinSphere}.  The latter term
refers to the equivalent problem of placing $n$ identical spheres of
maximal possible radius around a central unit sphere, touching the
unit sphere at the points that solve \toths problem.

\textbf{Thomson's problem,} also known as Coulomb problem, asks how
$n$ point charges which are confined to the surface of a sphere can be
distributed so that the potential energy is minimised.  The charges
interact with each other only through Coulomb's inverse square
law. Devised by J. J. Thomson in 1904, this problem raises the
question about the stable patterns of up to 100 electrons on a
spherical surface \cite{Thomson04}.  Its cost function is given by the
Coulomb energy and needs to be minimised.
\begin{equation}\label{thomson_def}
  f_{\text{Thomson}} ( \bmr{r}_{1} , \bmr{r}_{2} , \dots ,
  \bmr{r}_{n} ) = \sum_{i < j} \, \frac{1}{\abs{\bmr{r}_{i} -
      \bmr{r}_{j}}} \ens .
\end{equation}
The original motivation for Thomson's problem was to determine the
stable electron distribution of atoms in the plum pudding model.
While this model has been superseded by modern quantum theory, there
is a wide array of novel applications for Thomson's problem or its
generalisation to similar interaction potentials.  Among these are
multi-electron bubbles in liquid $^4$He \cite{Albrecht87,Leiderer95},
surface ordering of liquid metal drops confined in Paul traps
\cite{Davis97}, the shell structure of spherical viruses
\cite{Marzec93}, \quo{colloidosomes} for encapsulating biochemically
active substances \cite{Dinsmore02}, fullerene patterns of carbon
atoms \cite{Kroto85} and the Abrikosov lattice of vortices in
superconducting metal shells \cite{Dodgson97}.

It should be noted that, to some extent, the definition of Thomson's
problem runs contrary to classical electrical theory, because
Earnshaw's theorem rules out the existence of stable equilibrium
configurations of a collection of discrete charges under the influence
of the electric force alone \cite{Aspden87}. For example, if one were
to place $n$ negative charges $-q$ around a central positive charge
$+nq$, then this configuration would quickly collapse instead of
assuming a solution of Thomson's problem.  This explains why the
definition of Thomson's problem requires the rather mathematical
assumption of the point charges being confined to the surface of a
sphere.  The existence of physical appearances of stable electron
patterns in liquid Helium \cite{Albrecht87,Leiderer95} can be readily
explained by the surface tension of the macroscopic drops which
exhibit a positive mirror charge on their surface, in conjunction with
the quantum-mechanical Pauli principle \cite{Albrecht87}.  The latter
prevents the electrons from falling back into the liquid Helium,
thereby turning them into a 2D electron gas described by a 1D
hydrogenic spectrum \cite{Platzman99}. In this sense, the macroscopic
system provides the electrons with a restriction to a spherical
surface, akin to the mathematical definition of Thomson's problem.

The definitions of \toths problem and Thomson's problem are clearly
different from each other, but they share the same solutions for $n =
2-6,12$. Leech \cite{Leech57} showed that for these numbers the
equilibrium distributions of Thomson's problem are invariant under
replacing Coulomb's $r^{-2}$ law by the limiting form $r^{-l}, l
\rightarrow \infty$, and this \quo{infinitely repulsive interaction}
gives rise to the solutions of T\'{o}th's problem.  Exact solutions to
T\'{o}th's problem are known for $n_{\text{To}} = 2-12,24$, and
therefore the exact solutions to Thomson's problem for $n_{\text{Th}}
= 2-6,12$ are automatically derived this way \cite{Erber91}. Exact
solutions to Thomson's problem are furthermore known for
$n_{\text{Th}} = 7,8$, but even for numbers as small as $9$ and $11$,
exact solutions remain elusive \cite{Whyte52}.  With the help of
numerics, however, putative and approximate solutions have been found
for a wide range of $n$ in both problems
\cite{Ashby86,Altschuler94,Sloane,Wales}.

The solutions to $n = 2, 3$ are trivial and given by the dipole and
equilateral triangle, respectively.  For $n = 4,6,8,12,20$ the
vertices of the highly symmetric Platonic solids -- the five regular
convex polyhedra whose edges, vertices and angles are all congruent --
are natural candidates, but, as seen in \fig{platonic}, they are the
actual solutions only for $n = 4,6,12$ \cite{Berezin85}.  For $n =
8,20$ the solutions are not Platonic solids and are different for the
two problems. The solutions for $n=4-12$ will be covered in more
detail alongside the Majorana problem in \chap{solutions}.

\begin{figure}
  \centering
  \begin{minipage}{135mm}
    \hfill
    \begin{overpic}[scale=1.27]{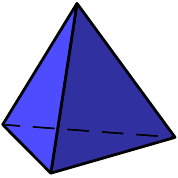}
      \put(-12,-3){(a)}
    \end{overpic}
    \hspace{0.5mm}
    \begin{overpic}[scale=1.26]{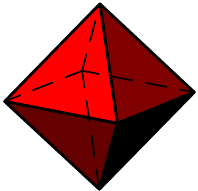}
      \put(-2,-3){(b)}
    \end{overpic}
    \hspace{1.5mm}
    \begin{overpic}[scale=1.15]{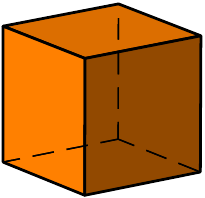}
      \put(-13,-3){(c)}
    \end{overpic}
    \hspace{1.5mm}
    \begin{overpic}[scale=1.15]{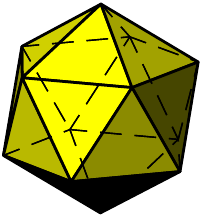}
      \put(-7,-3){(d)}
    \end{overpic}
    \hspace{1mm}
    \begin{overpic}[scale=1.21]{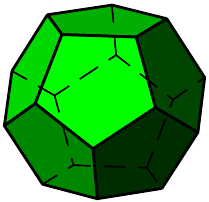}
      \put(-10,-3){(e)}
    \end{overpic}
  \end{minipage}
  \begin{minipage}{135mm}
    \centering
    \vspace{5mm}
    \begin{overpic}[scale=1.14]{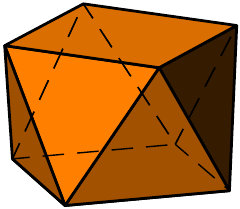}
      \put(-12,0){(f)}
    \end{overpic}
    \hspace{5mm}
    \begin{overpic}[scale=0.95]{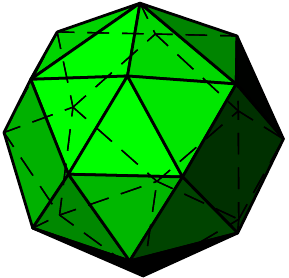}
      \put(-10,0){(g)}
    \end{overpic}
  \end{minipage}
  \caption[Platonic solids]{\label{platonic} Displayed from left to
    right in the top row are the five Platonic solids, the
    tetrahedron, octahedron, cube, icosahedron and dodecahedron. Their
    number of vertices is $n = 4,6,8,12$ and $20$.  The solutions to
    \toths and Thomson's problem are given by the Platonic solids only
    for $n=4,6$ and $12$.  For \toths problem the solutions for $n=8$
    and $n=20$ are shown in (f) and (g), respectively.  The polyhedron
    in (f) is the cubic antiprism, which is obtained from the regular
    cube by rotating one face by 45$^\circ$, followed by a slight
    compression along the direction perpendicular to the rotated face
    in order to return all edges to equal length.  The polyhedron
    shown in (g) consists of 30 triangles and 3 rhombuses.}
\end{figure}

The restriction of the points to the surface of the unit sphere, as
opposed to the interior of the sphere, is decisive for the solutions
of both problems.  In \toths problem it is clear that for larger $n$
the nearest-neighbour distances would be decreased by placing some
points inside the sphere.  For Thomson's problem this is not as
obvious, and several decades passed before it was realised that only
for $n < 12$ the electrons will all remain on the surface if given the
opportunity to occupy the interior of the sphere
\cite{Berezin85b,Berezin86}.

Both the classical problems and the Majorana problem are isotropic in
the sense that all directions in space are equal, and this makes it
reasonable to expect that the solutions exhibit certain symmetric
features.  For example, one could expect that the centre of mass of
the $n$ points always coincides with the sphere's middle point. This
is, however, not the case, as the solution to \toths problem for $n=7$
\cite{Erber91} or the solution to Thomson's problem for $n = 11$ shows
\cite{Erber91, Ashby86}. Furthermore, the solutions need not be
unique. For \toths problem, the first incident of this is $n = 5$
\cite{Melnyk77}, and for Thomson's problem at $n=15$ \cite{Erber91}
and $n = 16$ \cite{Ashby86}.  These aspects show that it is, in
general, hard to make statements about the form of the \quo{most
  spread out} point distributions on the sphere.  The Majorana problem
\eqref{maj_problem} is considered to be equally tricky, especially
because the normalisation factor $K$ depends on the \acp{MP}.

The Majorana problem shares a similarity with \toths problem in that
it is formulated as a min-max-problem, but a crucial difference is
that the positions of all $n$ points jointly influence the value of
the cost function.  In \toths problem the cost function only depends
on the smallest two-point distance, ignoring all other distances.  The
prefactor $K = K(\{ \ket{\phi_i} \}_{i = 1 \ldots n})$ depends on the
relative positions of the \acp{MP}, and from \eq{majorana_definition}
it is seen that $K$ increases with decreasing angles between the
individual Majorana points. Therefore, while the factor in brackets in
\eq{maj_problem} assumes small values for highly spread out \ac{MP}
distributions, the outer factor $\frac{1}{\sqrt{K}}$ will be
large. Conversely, when \acp{MP} move together, the factor in brackets
increases while the outer factor decreases.  This makes the solutions
of the Majorana problem highly nontrivial, and the solutions need not
be maximally spread out over the sphere in a conventional sense, as
the two coinciding \acp{MP} of the three qubit $\ket{\text{W}}$ state
demonstrate.

\section{Analytic results about MPs and CPPs}\label{analytic}

This section is mainly concerned with the interdependence between the
mathematical form of $n$ qubit symmetric states and their Majorana
representation.  For example, it is examined how the \acp{MP} and
\acp{CPP} are distributed for states whose coefficients are real,
positive or vanishing. In some of these cases the \acp{MP} and
\acp{CPP} form distinct patterns on the Majorana sphere that can be
described by symmetries.  In this context, care has to be taken as to
the meaning of the word \quo{symmetric}: Permutation-symmetric states
were introduced in \sect{symmetric_states}, and such states can be
visualised on the Majorana sphere.  Their \ac{MP} distributions may or
may not exhibit certain geometric symmetries in $\mbbrr$, such as
rotational and reflective symmetries.  For example, the \ac{GHZ}, W
and tetrahedron state of \fig{ghz_w_picture} and
\fig{tetrahedron_visual} all have a discrete or continuous rotational
symmetry around the $Z$-axis, as well as several reflective symmetries
along planes running through the origin of the sphere, e.g. the
$X$-$Z$-plane.

\subsection{Generalised Majorana representation}\label{general_maj_rep}

In the following a generalised version of the Majorana representation
\eqref{majorana_definition} will be derived which will prove helpful
e.g. for the analysis of real and positive states.  The property of a
symmetric state to be real or positive can often be inferred from its
\ac{MP} distribution.  As an example, the tetrahedron state
$\ket{\Psi_{4}} = \sqrt{\tfra{1}{3}} \sym{0} + \sqrt{\tfra{2}{3}}
\sym{3}$ is positive, even though its \acp{MP} are not all
positive. The first \ac{MP} $\ket{\alpha} := \ket{\phi_{1}} = \ket{0}$
is a positive qubit state, and a permutation of the remaining \acp{MP}
according to \eq{majorana_definition} yields a positive \ac{GHZ}-type
three qubit symmetric state $\ket{\beta} := \sum_{\text{perm}}
\ket{\phi_{2}} \ket{\phi_{3}} \ket{\phi_{4}} = \tfra{2}{\sqrt{3}}
\ket{000} + \tfra{4 \sqrt{2}}{\sqrt{3}} \ket{111}$. The tetrahedron
state can be reconstructed, up to normalisation, from all the
permutations of $\ket{\alpha}$ and $\ket{\beta}$ over the bipartitions
of the physical qubits into two subsets with one and three qubits,
respectively:
\begin{equation}
  \ket{\Psi_{4}} \, \propto \,
  \ket{\alpha}_{1} \ket{\beta}_{234} +
  \ket{\alpha}_{2} \ket{\beta}_{134} +
  \ket{\alpha}_{3} \ket{\beta}_{124} +
  \ket{\alpha}_{4} \ket{\beta}_{123} \ens .
\end{equation}
In the following this idea is formalised to arbitrary states and
arbitrary partitions.  It should be remembered that the \acp{MP}
representing a symmetric $n$ qubit state are abstract entities rather
than physical parts of the underlying system, and therefore partitions
of the set of \acp{MP} are fundamentally different from partitions of
the system's qubits.  Partitions of the \acp{MP} will be denoted by
$\mathcal{S} = \{ \mathcal{S}_{1}, \ldots , \mathcal{S}_{k} \}$ with
$\mathcal{S}_{i} = \{ \ket{\phi_{i_{1}}} , \ldots ,
\ket{\phi_{i_{m_{i}}}} \}$ for $i = 1 , \ldots , k$, and
$\sum_{i=1}^{k} m_{i} = n$.  Partitions of the physical qubits of the
system will be denoted by $\mathcal{P} = \{ \mathcal{P}_{1}, \ldots ,
\mathcal{P}_{l} \}$ with $\mathcal{P}_{i} = \{ i_{1} , \ldots ,
i_{r_{i}} \}$ for $i = 1 , \ldots , l$, and $\sum_{i=1}^{l} r_{i} =
n$.  The notation $\ket{\phi}_{i_{x}}$ is used to describe a single
qubit state of particle $i_{x}$, and $\ket{\psi}_{\mathcal{P}_{i}}$ is
used to describe an $r_{i}$-qubit state over the particles $i_{1} ,
i_{2} , \ldots , i_{r_{i}}$.
\begin{theorem}\label{theo_general_maj_rep}
  Let $\psis$ be a symmetric state of $n$ qubits with \acp{MP}
  $\ket{\phi_{1}}, \ldots ,\ket{\phi_{n}}$, and let $\mathcal{S} = \{
  \mathcal{S}_{1}, \ldots , \mathcal{S}_{k} \}$ be a partition of the
  \acp{MP}.  $\psis$ can then be written, up to a prefactor, as
  \begin{equation}\label{general_maj_rep_eq}
    \psis \propto \sum_{\text{partitions}}^{ \{
      \mathcal{P}_{1}^{j}, \ldots , \mathcal{P}_{k}^{j} \} }
    \ket{\psi_{\mathcal{S}_{1}}}_{\mathcal{P}_{1}^{j}} \otimes \cdots
    \otimes \ket{\psi_{\mathcal{S}_{k}}}_{\mathcal{P}_{k}^{j}} \ens ,
  \end{equation}
  where the $m_{i}$-qubit symmetric states
  $\ket{\psi_{\mathcal{S}_{i}}} := \sum_{\text{perm}}
  \ket{\phi_{i_{1}}} \cdots \ket{\phi_{i_{m_{i}}}}$ are composed from
  $\mathcal{S}_{i} = \{ \ket{\phi_{i_{1}}} , \ldots ,
  \ket{\phi_{i_{m_{i}}}} \}$ via the Majorana representation
  \eqref{majorana_definition}, and where the sum runs over all the
  partitions $\mathcal{P}^{j} = \{ \mathcal{P}_{1}^{j}, \ldots ,
  \mathcal{P}_{k}^{j} \}$ that satisfy $\abs{\mathcal{P}_{i}^{j}} =
  \abs{\mathcal{S}_{i}}$ for all $i$.
\end{theorem}
\begin{proof}
  For simplicity, we only consider a bipartition $\mathcal{S} = \{
  \mathcal{S}_{1}, \mathcal{S}_{2} \}$ of the \acp{MP}, with
  $\mathcal{S}_{1} = \{ \ket{\phi_{1}} , \ldots , \ket{\phi_{m}} \}$
  and $\mathcal{S}_{2} = \{ \ket{\phi_{m+1}} , \ldots , \ket{\phi_{n}}
  \}$.  The general case directly follows from this by mathematical
  induction.  The bipartitions of the system's qubits are denoted by
  $\mathcal{P}^{j} = \{ \mathcal{P}_{1}^{j} , \mathcal{P}_{2}^{j} \}$,
  $j = 1, \ldots , \tbinom{n}{m}$, with $\mathcal{P}_{1}^{j} = \{
  j_{1} , \ldots , j_{m} \}$ and $\mathcal{P}_{2}^{j} = \{ j_{m+1} ,
  \ldots , j_{n} \}$. Note that the subsystems in a product state can
  be shuffled, e.g. $\ket{\alpha}_{1} \ket{\beta}_{2} \ket{\gamma}_{3}
  \equiv \ket{\beta}_{2} \ket{\gamma}_{3} \ket{\alpha}_{1}$.
  \begin{align*}
    \psis& \propto \sum_{\text{perm}}^{ \{ 1, \ldots , n \} }
    \ket{\phi_{P(1)}}_{1} \cdots \ket{\phi_{P(n)}}_{n} =
    \sum_{\text{perm}}^{ \{ 1, \ldots , n \} }
    \ket{\phi_{1}}_{P(1)} \cdots \ket{\phi_{n}}_{P(n)} \\
    {}& = \sum_{\text{perm}}^{ \{ 1, \ldots , n \} } \Big(
    \ket{\phi_{1}}_{P(1)} \cdots \ket{\phi_{m}}_{P(m)} \Big) \Big(
    \ket{\phi_{m+1}}_{P(m+1)} \cdots
    \ket{\phi_{n}}_{P(n)} \Big) \\
    {}& = \sum_{\text{bipartitions}}^{ \{ \mathcal{P}_{1}^{j} , \,
      \mathcal{P}_{2}^{j} \} } \Bigg( \sum_{\text{perm}}^{ \{
      j_{1}, \ldots , j_{m} \} } \ket{\phi_{1}}_{P(j_{1})}
    \cdots \ket{\phi_{m}}_{P(j_{m})} \Bigg) \Bigg(
    \sum_{\text{perm}}^{ \{ j_{m+1}, \ldots , j_{n} \} }
    \ket{\phi_{m+1}}_{P(j_{m+1})} \cdots
    \ket{\phi_{n}}_{P(j_{n})} \Bigg) \\
    {}& = \sum_{\text{bipartitions}}^{ \{ \mathcal{P}_{1}^{j} , \,
      \mathcal{P}_{2}^{j} \} }
    \ket{\psi_{\mathcal{S}_{1}}}_{\mathcal{P}_{1}^{j}}
    \ket{\psi_{\mathcal{S}_{2}}}_{\mathcal{P}_{2}^{j}} \ens .
  \end{align*}
  From the identity $n! = \binom{n}{m} m!  (n-m)!$ it can be verified
  that the second and third line contain the same number of summands.
\end{proof}

\Eq{general_maj_rep_eq} can be understood as a \textbf{generalised
  Majorana representation} for arbitrary partitions, which contains
the regular Majorana representation as the special case $\mathcal{S} =
\{ \mathcal{S}_{1}, \ldots , \mathcal{S}_{n} \}$, with
$\mathcal{S}_{i} = \{ \ket{\phi_{i}} \}$ for all $i$.

The following corollary asserts that the number of \acp{MP} lying on
either pole of the Majorana sphere is immediately given by the
smallest and largest nonvanishing coefficient of a symmetric state.

\begin{corollary}\label{mp_poles}
  Let $\psis = \sum\limits_{m=0}^{n} a_{m} \sym{m}$ be a symmetric
  state of $n$ qubits.
  \begin{itemize}
  \item $l = \, \min \{ m \vert \, a_{m} \neq 0 \}$ \quad $\,\,
    \Longleftrightarrow$ \quad $l$ \acp{MP} lie on the south pole
    $\ket{1}$.
    
  \item $k = \max \{ m \vert \, a_{m} \neq 0 \}$ \quad
    $\Longleftrightarrow$ \quad $n-k$ \acp{MP} lie on the north pole
    $\ket{0}$.
  \end{itemize}
\end{corollary}
\begin{proof}
  Assume that $l$ \acp{MP} of $\psis$ lie on the south pole.  From
  \theoref{theo_general_maj_rep} it follows that one can write $\psis
  \propto \sum_{\text{partitions}} \ket{1}^{\otimes \, l}
  \ket{\varphi}$, where $\ket{\varphi} = \sum_{i=0}^{n-l} b_{i}
  \sym{i}$ is an ($n-l$)-qubit symmetric state.  Since the \acp{MP} of
  $\ket{\varphi}$ all have nonvanishing $\ket{0}$-components, it
  follows that $b_{0} \neq 0$, and therefore $\min \{ m \vert \, a_{m}
  \neq 0 \} = l$.  The converse statement\footnote{In terms of logic
    the statement \protect\quo{$l = \min \{ m \vert a_{m} \neq 0 \} \:
      \Longrightarrow \: l \text{ MPs are } \ket{1}$} is equivalent to
    the statement \quo{$r \neq l \text{ MPs are } \ket{1} \:
      \Longrightarrow \: l \neq \min \{ m \vert a_{m} \neq 0 \} $}.}
  follows by assuming that the number of \acp{MP} lying on the south
  pole is $r \neq l$, leading to $l \neq r = \min \{ m \vert \, a_{m}
  \neq 0 \}$.
  
  The statement about \acp{MP} on the north pole follows by the same
  arguments.
\end{proof}

This corollary is easily verified by examples such as
$\ket{\text{GHZ}_{3}} = \tfra{1}{\sqrt{2}} ( \sym{0} + \sym{3} )$
which has no \acp{MP} on the poles, or $\ket{\text{W}_{3}} = \sym{1}$
which has two \acp{MP} on the north pole and one \ac{MP} on the south
pole.

Rotational symmetries appear frequently in the Majorana
representations of symmetric states, and by means of the
\ac{LU}-equivalence mediated by the symmetric unitary operations
$U^{\text{s}} = U \otimes \cdots \otimes U$ in \eq{symmetric_lu} and
\eqref{lu_rotation}, it suffices to investigate only rotations around
the $Z$-axis of the Majorana sphere.  These are of a particularly
simple mathematical form, with the single-qubit rotation $\rotz$ of
\eqref{z_rotationmatrix} generalising to $Z$-axis rotations of a
symmetric $n$ qubit state as $\rotzs := \text{R}_{\text{z}}^{\otimes
  n}$.  The effect of $\rotzs$ on $\psis = \sum_{k = 0}^{n} a_k
\sym{k}$ is then
\begin{equation}\label{rot_z}
  \rotzs (\varphi ) \psis = \sum_{k = 0}^{n}
  a_k \E^{ \I k \varphi } \sym{k} \ens .
\end{equation}
$\psis$ is rotationally symmetric around the $Z$-axis \ac{iff} $\rotzs
(\varphi ) \psis \propto \psis$ for some $0 < \varphi < 2 \pi$.  In
the case of a discrete rotational symmetry the possible rotational
angles are (up to multiples) restricted to $\varphi = \tfra{2
  \pi}{m}$, with $m \in \mbbn$, $1<m \leq n$.  From \eq{rot_z} it is
then easy to determine the necessary and sufficient conditions for a
rotational $Z$-axis symmetry of the \acp{MP}.

\begin{lemma}\label{rot_symm}
  The \acp{MP} of a symmetric $n$ qubit state $\psis$ have a
  rotational $Z$-axis symmetry with rotational angle $\varphi =
  \tfra{2 \pi}{m}$ ($1 < m \leq n$) \ac{iff}
  \begin{equation}\label{rot_cond}
    \exists \, 0 \leq l < m : \quad
    \psis = \sum_{j=0}^{\lfloor \frac{n-l}{m} \rfloor}
    a_{l + j m} \sym{l + j m} \ens .
  \end{equation}
\end{lemma}
\begin{proof}
  Assume that $\psis$ can be written in the above form. Then
  \begin{multline*}
    \rotzs \big( \tfra{2 \pi}{m} \big) \psis =
    \sum_{j} a_{l + j m}
    \exp \left( \tfra{\I 2 \pi}{m} (l + j m) \right)
    \sym{l + j m} \\
    = \sum_{j} a_{l + j m}
    \exp \left( \tfra{\I 2 \pi l}{m} \right) \sym{l + j m} =
    \E^{\I \delta} \psis \ens ,
    \vspace{-1em}
  \end{multline*}
  with $\delta = \tfra{2 \pi l}{m}$, and therefore $\psis$ is
  rotationally symmetric around the $Z$-axis.
  
  Conversely, if $\psis = \sum_{k=0}^{n} a_{k} \sym{k}$ is
  rotationally symmetric, then $\rotzs ( \tfra{2 \pi}{m}) \psis =
  \sum_{k=0}^{n} a_{k} \exp \big( \tfra{\I 2 \pi k}{m} \big) \sym{k} =
  \E^{\I \delta} \psis$ for some $\delta \in \mbbr$. For this to hold,
  the value of $\exp \big( \frac{\I 2 \pi k}{m} \big)$ must be the
  same for all $k$ with $a_{k} \neq 0$, and because of this, the $k$
  can be cast as $k_{j} = l + j m$ with integers $0 \leq l < m$ and
  $0<j< \lfloor \frac{n-l}{m} \rfloor$.
\end{proof}

In other words, a sufficient number of coefficients need to vanish,
and the spacings between nonvanishing coefficients must be multiples
of $m$. For example, a symmetric state of the form $\psis = a_3
\sym{3} + a_7 \sym{7} + a_{15} \sym{15}$ is rotationally symmetric
with $\varphi = \tfra{\pi}{2}$, because the spacings between
nonvanishing coefficients are multiples of $4$.

We remark that \lemref{rot_symm} could also have been proved with the
generalised Majorana representation of
\theoref{theo_general_maj_rep}. The idea for this is that the discrete
rotational symmetry around the $Z$-axis necessitates that the \acp{MP}
that do not lie on the poles must be equidistantly spaced along
horizontal circles. Each such circle of \acp{MP} then represents a
\ac{GHZ}-type state $\ket{\psi_{\mathcal{S}_{i}}} = \alpha_{i} \ket{00
  \dots 0} + \beta_{i} \ket{11 \dots 1}$. Combining all these
$\ket{\psi_{\mathcal{S}_{i}}}$ via \eq{general_maj_rep_eq} then gives
rise to a state of the form \eqref{rot_cond}.

\subsection{Real symmetric states}\label{real_symm_states}

For symmetric states with real coefficients the following result is
immediately clear from \eq{general_maj_rep_eq}.
\begin{corollary}\label{real_gen_maj}
  If the $\ket{\psi_{\mathcal{S}_{i}}}$, $i = 1 , \ldots , k$ of the
  generalised Majorana representation \eqref{general_maj_rep_eq} are
  all real, then $\psis$ is also real.
\end{corollary}

Next it is shown that the \acp{MP} and \acp{CPP} of real states
exhibit a reflection symmetry with respect to the $X$-$Z$-plane which
cuts the Majorana sphere in half.  In mathematical terms, the
reflection of a Bloch vector $\ket{\phi} = \co_{\theta} \ket{0} +
\E^{\I \varphi} \si_{\theta} \ket{1}$ along the $X$-$Z$-plane is the
complex conjugate vector $\cc{\ket{\phi}} = \co_{\theta} \ket{0} +
\E^{- \I \varphi} \si_{\theta} \ket{1}$.

\begin{lemma}\label{maj_real}
  Let $\psis$ be a symmetric state of $n$ qubits.  $\psis$ is real
  \ac{iff} all its \acp{MP} are reflective symmetric with respect to
  the $X$-$Z$-plane of the Majorana sphere.
\end{lemma}
\begin{proof}
  ($\Rightarrow$) Let $\psis$ be a real state.  Then $\psis =
  \cc{\psis}$, and since Majorana representations are unique up to a
  global phase, $\psis$ has the same \acp{MP} as $\cc{\psis}$.
  Therefore the complex conjugate $\cc{\ket{\phi_{i}}}$ of each
  \ac{MP} $\ket{\phi_{i}}$ is also an \ac{MP}.

  ($\Leftarrow$) Let the \acp{MP} of $\psis$ be symmetric with respect
  to the $X$-$Z$-plane. Then for every non-real \ac{MP}
  $\ket{\phi_{i}} = \co_{\theta_{i}} \ket{0} + \E^{\I \varphi_{i}}
  \si_{\theta_{i}} \ket{1}$ its complex conjugate
  $\cc{\ket{\phi_{i}}}$ is also an \ac{MP}.  Define a partition
  $\mathcal{S} = \{ \mathcal{S}_{1}, \ldots , \mathcal{S}_{k} \}$ of
  the \acp{MP} where $\mathcal{S}_{1}$ contains all the real \acp{MP}
  and the remaining $\mathcal{S}_{i}$ each contain a complex conjugate
  pair of \acp{MP}: $\mathcal{S}_{i} = \{ \ket{\phi_{i}} ,
  \cc{\ket{\phi_{i}}} \}$.  The two qubit states
  $\ket{\psi_{\mathcal{S}_{i}}} = \ket{\phi_{i}} \cc{\ket{\phi_{i}}} +
  \cc{\ket{\phi_{i}}} \ket{\phi_{i}} \propto \co_{\theta_{i}}^{2}
  \sym{0} + \sqrt{2} \co_{\theta_{i}} \si_{\theta_{i}} \cos
  \varphi_{i} \sym{1} + \si_{\theta_{i}}^{2} \sym{2}$ are all real,
  and from \corref{real_gen_maj} it follows that $\psis$ is real.
\end{proof}

\begin{corollary}\label{cpp_real}
  Let $\psis$ be a symmetric state of $n$ qubits. If $\psis$ is real,
  then all its \acp{CPP} are reflective symmetric with respect to the
  $X$-$Z$-plane of the Majorana sphere.
\end{corollary}
\begin{proof}
  Consider the complex conjugate of the optimisation problem
  \eqref{maj_problem}.  It follows from \lemref{maj_real} that for any
  \ac{CPP} $\ket{\sigma}$ the complex conjugate $\cc{\ket{\sigma}}$ is
  also a \ac{CPP}.
\end{proof}

\subsection{Positive symmetric states}\label{pos_symm_states}

Particularly strong results can be obtained for the Majorana
representations of symmetric states with positive coefficients.
First, we restate \corref{real_gen_maj} for the positive case:

\begin{corollary}\label{pos_gen_maj}
  If the $\ket{\psi_{\mathcal{S}_{i}}}$, $i = 1 , \ldots , k$ of the
  generalised Majorana representation \eqref{general_maj_rep_eq} are
  all positive, then $\psis$ is also positive.
\end{corollary}

In the remainder of this section it will be shown that the Majorana
representations of positive symmetric states are of two basic types.
The first type exhibits a rotational $Z$-axis symmetry which forces
the \acp{MP} and \acp{CPP} into predictable and easily analysable
patterns on the Majorana sphere.  In particular, an upper bound for
the number of \acp{CPP} of such states can be derived.  The other type
of Majorana representation does not exhibit a $Z$-axis symmetry, and
its \acp{CPP} are all restricted to the half-circle of positive Bloch
vectors.

The only states with a continuous rotational $Z$-axis symmetry are the
Dicke states, i.e. the states whose \acp{MP} all lie on the poles.
This trivial case will not be considered in the following, and instead
it is assumed that at least one \ac{MP} does not lie on a
pole. Rotational $Z$-axis symmetries must then be discrete, with a
minimal rotational angle $\varphi = \tfra{2 \pi}{m}$, $m \in \mbbn$
and $1<m \leq n$.  From this symmetry and from \lemref{maj_real} the
allowed distribution patterns of the \acp{MP} can be fully
characterised: All \acp{MP} that do not lie on the poles must be
equidistantly spread along horizontal circles with neighbouring
spherical distances of $\varphi = \tfra{2 \pi}{m}$. The $m$ \acp{MP}
of such a circle represent an $m$ qubit \ac{GHZ}-type state
$\ket{\psi_{\mathcal{S}} } = \alpha \sym{0} + \beta \sym{m}$ by means
of the Majorana representation \eqref{majorana_definition}.

\begin{figure}
  \centering
  \begin{overpic}[scale=.65]{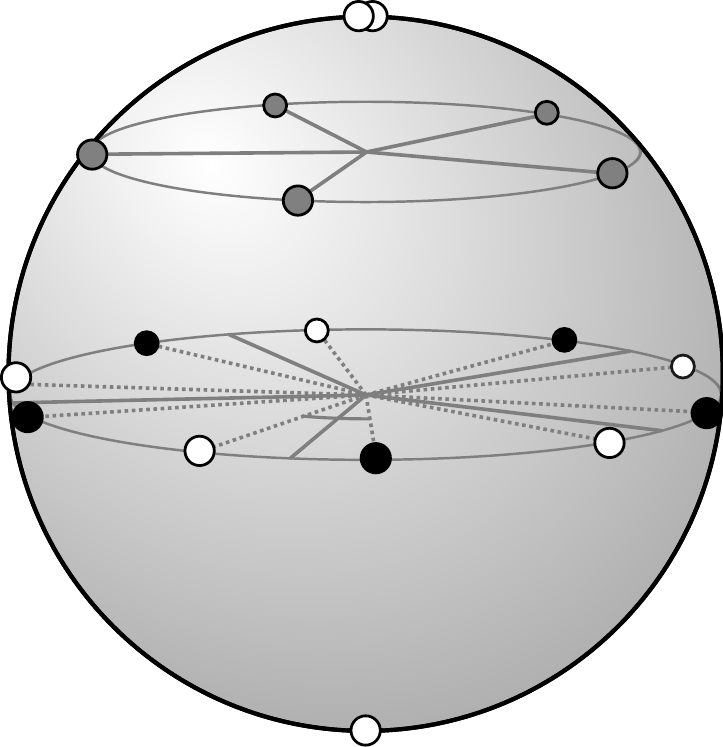}
    \put(-4,86){{$\ket{\psi_{\mathcal{S}}^{+}}$}}
    \put(-20,36){{$\ket{\psi_{\mathcal{S}}^{+ \vartheta}}$}}
    \put(-21,52){{$\ket{\psi_{\mathcal{S}}^{- \vartheta}}$}}
    \put(33,75.3){{\small $\varphi$}}
    \put(36.3,39.7){{\small $\vartheta$}}
    \put(43.8,39.2){{\small $\vartheta$}}
  \end{overpic}
  \caption[MPs of a real symmetric state with $Z$-axis rotational
  symmetry]{\label{symm_example} An exemplary \ac{MP} distribution of
    a positive symmetric 18 qubit state with a $Z$-axis rotational
    symmetry is shown. The minimal rotational angle is $\varphi =
    \tfra{2 \pi}{5}$. Two \acp{MP} lie on the north pole, one on the
    south pole, five on a single basic circle and 10 on two
    intertwined circles.  The circles of \acp{MP} which give rise to
    the five qubit \ac{GHZ}-type states
    $\ket{\psi_{\mathcal{S}}^{+}}$, $\ket{\psi_{\mathcal{S}}^{+
        \vartheta}}$, and $\ket{\psi_{\mathcal{S}}^{- \vartheta}}$ are
    coloured gray, black and white, respectively.}
\end{figure}
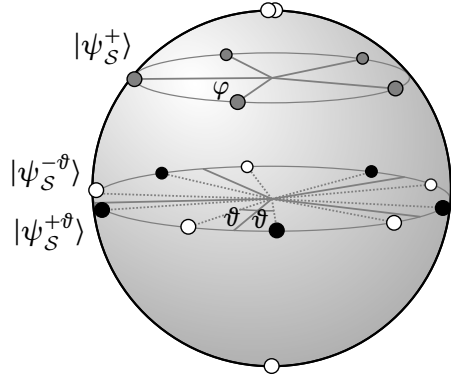

If $\ket{\psi_{\mathcal{S}} }$ is real (denoted as
$\ket{\psi_{\mathcal{S}}^{\pm} }$), then \lemref{maj_real} implies
that the complex conjugate of each \ac{MP} is also an \ac{MP} of that
circle.  A horizontal circle of \acp{MP} with this property is shown
in \fig{symm_example}.

If $\ket{\psi_{\mathcal{S}} }$ is not real, then \lemref{maj_real}
implies that for some \acp{MP} the complex conjugate is not part of
this circle. For the composite state $\psis$ to be real, this
necessitates that the \acp{MP} of the \quo{complex conjugate circle}
$\cc{\ket{\psi_{\mathcal{S}} }}$ are also part of $\psis$.  As shown
in \fig{symm_example}, this gives rise to two intertwined horizontal
circles $\ket{\psi_{\mathcal{S}}^{+ \vartheta} } = \alpha \sym{0} +
\E^{\I m \vartheta} \beta \sym{m}$ and $\ket{\psi_{\mathcal{S}}^{-
    \vartheta} } = \alpha \sym{0} + \E^{- \I m \vartheta} \beta
\sym{m}$, where the azimuthal angle of each \ac{MP} is shifted by an
angle $\pm \vartheta$ from the position it would occupy for the
corresponding non-phased state $\ket{\psi_{\mathcal{S}}^{+} } = \alpha
\sym{0} + \beta \sym{m}$ with $\alpha , \beta \geq 0$.

All horizontal circles of \acp{MP} present in a real symmetric state
with $Z$-axis rotational symmetry can be decomposed into these two
principal types, and from \eq{general_maj_rep_eq} it is clear that the
resulting state $\psis$ is real, rotationally symmetric and that the
degrees of freedom present in the horizontal circles of \acp{MP}
manifest themselves in the freedoms of the nonvanishing coefficients
of $\psis$.  The additional requirement of positivity for $\psis$
merely restricts the basic type of \ac{MP} circle to positive states
$\ket{\psi_{\mathcal{S}}^{+} }$, and the intertwined type to those
with an angle $\vartheta \leq \frac{\pi}{2 m}$.

The following lemma asserts strong restrictions on the possible
locations of the \acp{CPP} of positive symmetric states with or
without $Z$-axis rotational symmetries.

\begin{lemma}\label{pos_symm_cpp}
  Let $\psis$ be a positive symmetric $n$ qubit state, excluding the
  Dicke states.
  
  \begin{enumerate}
  \item[(a)] If $\psis$ has a $Z$-axis rotational symmetry with
    minimal rotational angle $\varphi = \tfra{2 \pi}{m}$, then all its
    \acp{CPP} $\ket{\sigma (\theta, \varphi)} = \co_{\theta} \ket{0} +
    \E^{\I \varphi} \si_{\theta} \ket{1}$ are restricted to the $m$
    azimuthal angles $\varphi_{r} = \tfra{2 \pi r}{m}$ with $r \in
    \mbbz$.  Furthermore, if $\ket{\sigma (\theta, \varphi_{r} )}$ is
    a \ac{CPP} for some $r$, then it is also a \ac{CPP} for all other
    values of $r$.
    
  \item[(b)] If $\psis$ has no $Z$-axis rotational symmetry, then all
    its \acp{CPP} are positive.
  \end{enumerate}
\end{lemma}
\begin{proof}
  The proof runs similar to the one of \lemref{lem_pos_cps}, where the
  existence of at least one positive \ac{CPP} was established.  We use
  the notation $\psis = \sum_{k} a_{k} \sym{k}$ with $a_{k} \geq 0$,
  and $\ket{\lambda} = \ket{\sigma}^{\otimes n}$.
  
  Case (a): Consider a non-positive \ac{CPP} $\ket{\sigma} =
  \co_{\theta} \ket{0} + \E^{\I \kappa} \si_{\theta} \ket{1}$ with
  $\kappa = \tfra{2 \pi s}{m}$, $s \in \mbbr$, and define
  $\ket{\lambda^{+}} := \ket{\sigma^{+}}^{\otimes n} = \left(
    \co_{\theta} \ket{0} + \si_{\theta} \ket{1} \right)^{\otimes
    n}$. Then
  \begin{equation*}
    \abs{\bracket{\psi^{\text{s}}}{\lambda}} =  \bigg| \sum_k \E^{\I k
      \kappa} a_{k} \co_{\theta}^{n-k} \si_{\theta}^{k}
    \sqrt{\tbinom{n}{k}} \bigg| \leq \sum_k a_{k} \co_{\theta}^{n-k}
    \si_{\theta}^{k} \sqrt{\tbinom{n}{k}} =
    \abs{\bracket{\psi^{\text{s}}}{\lambda^{+}}} \ens .
  \end{equation*}
  If this inequality is strict, then $\ket{\sigma}$ is not a \ac{CPP}.
  Since this would contradict the initial assumption, we must have an
  equality. Then for any two indices $k_i$ and $k_j$ of nonvanishing
  coefficients $a_{k_i}$ and $a_{k_j}$ the following must hold:
  $\E^{\I k_{i} \kappa} = \E^{\I k_{j} \kappa}$.  This can be
  reformulated as $k_{i} \, s \bmod m = k_{j} \, s$, or equivalently
  \begin{equation}\label{pos_symm_cpp_cond1}
    ( k_{i} - k_{j} ) \, s \bmod m = 0 \ens .
  \end{equation}
  Because $\varphi = \tfra{2 \pi}{m}$ is the minimal rotational angle,
  $m$ is the largest integer that satisfies \eq{rot_cond}, and thus
  there exist $k_i$ and $k_j$ with $a_{k_i}, a_{k_j} \neq 0$ s.t. $k_i
  - k_j = m$.  From this and from \eq{pos_symm_cpp_cond1} it follows
  that $s \in \mbbz$.  Therefore all \acp{CPP} are of the form
  $\ket{\sigma (\theta, \varphi_{r} )}$ with $r \in \mbbz$, and if
  $\ket{\sigma (\theta, \varphi_{r} )}$ is a \ac{CPP} for some $r$,
  then it is also a \ac{CPP} for all other $r \in \mbbz$.
  
  Case (b): Considering a \ac{CPP} $\ket{\sigma} = \co_{\theta}
  \ket{0} + \E^{\I \rho} \si_{\theta} \ket{1}$ with $\rho = 2 \pi r$
  and $r \in \mbbr$, we need to show that $r \in \mbbz$.  Defining
  $\ket{\sigma^{+}} = \co_{\theta} \ket{0} + \si_{\theta} \ket{1}$,
  and using the same line of argumentation as above, the equation
  $\E^{\I k_{i} \rho} = \E^{\I k_{j} \rho}$ must hold for any pair of
  nonvanishing $a_{k_i}$ and $a_{k_j}$.  This is equivalent to
  \begin{equation}\label{pos_symm_cpp_cond2}
    ( k_{i} - k_{j} ) \, r \bmod 1 = 0 \ens ,
  \end{equation}
  or $( k_{i} - k_{j} ) \, r \in \mbbz$, in particular $r \in
  \mathbb{Q}$. If there exist indices $k_{i}$ and $k_{j}$ of
  nonvanishing coefficients s.t. $k_{i} - k_{j} = 1$, then $r \in
  \mbbz$, as desired.  Otherwise consider $r = \tfra{x}{y}$ with
  coprime $x,y \in \mbbn$, $x < y$.  Because $\psis$ is not
  rotationally symmetric, the negation of \lemref{rot_symm} yields
  that, for any two $k_{i}$ and $k_{j}$ ($a_{k_{i}}, a_{k_{j}} \neq
  0$) with $k_{i} - k_{j} = \alpha > 1$, there must exist a different
  pair $k_{p}$ and $k_{q}$ ($a_{k_p}, a_{k_q} \neq 0$) with $k_{p} -
  k_{q} = \beta > 1$ s.t.  $\alpha$ is not a multiple of $\beta$ and
  vice versa.  From $r = \tfra{x}{y}$ and \eq{pos_symm_cpp_cond2}, it
  follows that $y = \alpha$ as well as $y = \beta$. This is a
  contradiction, so $r \in \mbbz$.
\end{proof}

A trivial consequence of the previous theorem is the following
statement that considerably simplifies the determination of the
geometric entanglement.
\begin{corollary}\label{positive_cpp}
  Every positive symmetric state has at least one positive symmetric
  \ac{CPP}.
\end{corollary}

With the confinement of the \acp{CPP} to certain azimuthal angles, as
described by \lemref{pos_symm_cpp}, it is possible to make the
following statements about the number and locations of the \acp{CPP}.
\begin{theorem}\label{cpp_number_locations}
  The Majorana representation of every positive symmetric state
  $\psis$ of $n$ qubits, excluding the Dicke states, belongs to one of
  the following three mutually exclusive classes.
  
  \begin{enumerate}
  \item[(a)] $\psis$ has a $Z$-axis rotational symmetry, with only the
    two poles as possible \acp{CPP}.
    
  \item[(b)] $\psis$ has a $Z$-axis rotational symmetry, and at least
    one \ac{CPP} is non-positive.
    
  \item[(c)] $\psis$ has no $Z$-axis rotational symmetry, and all
    \acp{CPP} are positive.
  \end{enumerate}
  Regarding the states from class (b) and (c), the following
  assertions can be made about the number of \acp{CPP} for $n \geq 3$:
  \begin{enumerate}
  \item[(b)] If both poles are occupied by at least one \ac{MP} each,
    then there are at most $2n-4$ \acp{CPP}, else there are at most
    $n$ \acp{CPP}.

  \item[(c)] There are at most $\lceil \tfra{n+2}{2} \rceil$
    \acp{CPP}.
  \end{enumerate}
\end{theorem}
\begin{proof}
  We start with the first part of the theorem. Case (c) has already
  been shown in \lemref{pos_symm_cpp}, so we only need to consider
  states $\psis$ with a $Z$-axis rotational symmetry.  If all
  \acp{CPP} are either $\ket{0}$ or $\ket{1}$, then we have case (a),
  otherwise there is at least one \ac{CPP} $\ket{\sigma}$ which does
  not lie on a pole.  If this $\ket{\sigma}$ is non-positive, we have
  case (b), and if $\ket{\sigma}$ is positive, \lemref{pos_symm_cpp}
  states the existence of another, non-positive \ac{CPP}, thus again
  resulting in case (b). This concludes the first part of the proof.

  Now consider the second part of the theorem.  We start with case
  (b), i.e. positive states that have a $Z$-axis rotational symmetry
  with minimal rotational angle $\varphi = \tfra{2 \pi}{m}$, and $1<m
  \leq n$.  According to \eq{bloch_product}, any \ac{CPP}
  $\ket{\sigma}$ maximises the function $\prod_{i=1}^{n}
  \abs{\bracket{\sigma}{\phi_{i}}}$, and from \corref{positive_cpp} it
  follows that there must be at least one positive \ac{CPP}
  $\ket{\sigma} = \co_{\theta} \ket{0} + \si_{\theta} \ket{1}$. First
  we derive the maximum possible number of positive \acp{CPP}, from
  which an upper bound for the total number of \acp{CPP} follows from
  \lemref{pos_symm_cpp}.

  For an \ac{MP} distribution with $k$ \acp{MP} on the north pole, $l$
  \acp{MP} on the south pole and the remaining $n-k-l$ \acp{MP} on
  horizontal circles, the function whose absolute value has to be
  maximised is
  \begin{equation*}
    f(\theta) = \bracket{\sigma}{0}^{k} \bracket{\sigma}{1}^{l}
    \prod_{r} h_{1} (\theta_{r}) \prod_{s} h_2
    (\vartheta_{s} , \theta_{s}) \ens ,
  \end{equation*}
  where $h_1 (\theta_{r}) = \prod_{i=1}^{m} \bracket{\sigma}{\phi_{i}
    ( \theta_{r} )}$ contains the factors contributed by a single
  circle with $m$ \acp{MP} at inclination $\theta_{r}$, and $h_2
  (\vartheta_{s}, \theta_{s}) = \prod_{i=1}^{m}
  \bracket{\sigma}{\phi_{i} ( + \vartheta_{s} , \theta_{s} )}
  \bracket{\sigma}{\phi_{i} ( - \vartheta_{s} , \theta_{s} )}$
  represents the factors contributed by two circles intertwined at
  azimuthal angles $\pm \vartheta_{s}$ with $2m$ \acp{MP}, and
  inclination $\theta_{s}$.  A simple calculation yields
  \begin{equation*}
    h_1 (\theta_{r}) = \co_{\theta}^{m} \co_{\theta_r}^{m} +
    \si_{\theta}^{m} \si_{\theta_r}^{m} \ens , \quad
    h_2 (\vartheta_{s} , \theta_{s}) =
    \co_{\theta}^{2m} \co_{\theta_s}^{2m} + 2 \cos
    (m \vartheta_{s}) \co_{\theta}^{m} \si_{\theta}^{m}
    \co_{\theta_s}^{m} \si_{\theta_s}^{m} + \si_{\theta}^{2m}
    \si_{\theta_s}^{2m} \ens .
  \end{equation*}
  Thus $f$ can be written in the form
  \begin{equation*}
    f(\theta) = \co_{\theta}^{k} \si_{\theta}^{l} \sum_{i=0}^{p}
    a_i  \co_{\theta}^{(p-i)m} \si_{\theta}^{im} = \sum_{i=0}^{p}
    a_i  \co_{\theta}^{k+(p-i)m} \si_{\theta}^{l+im} \ens ,
  \end{equation*}
  where the $a_i$ are positive-valued coefficients, and $p$ is the
  number of basic circles ($k+l+pm=n$).  The number of zeros of
  $f'(\theta)$ in $\theta \in (0, \pi)$ gives a bound on the number of
  positive \acp{CPP}. The form of $f'(\theta)$ is qualitatively
  different for $m=2$ and $m>2$. With the substitution $x = \tan
  \tfra{\theta}{2}$ the equation $f'(\theta) = 0$ for $m=2$ becomes
  \begin{gather*}
    a_{0} l + \left( \sum_{i=1}^{p} b_i x^{2i} \right) -
    a_{p} k x^{2p+2} = 0 \ens , \\
    \text{with} \quad b_{i} = a_{i} (l+2i) - a_{i-1} (k+2(p-i)+2) \in
    \mbbr \ens .
  \end{gather*}
  This is a real polynomial in $x$, with the first and last
  coefficient vanishing if no \acp{MP} lie on the south pole ($l=0$)
  and north pole ($k=0$), respectively.  Descartes' rule of signs
  states that the number of positive roots of a real polynomial is at
  most the number of sign differences between consecutive nonzero
  coefficients, ordered by descending variable exponent. From this and
  the fact that the codomain of $x$ is $\mbbr^{+}$, a calculation
  yields that for $m=2$ there are at most $p-1$, $p$ or $p+1$ extrema
  of $f(\theta)$ lying in $\theta \in (0,\pi)$, depending on whether
  $k$ and $l$ are zero or not.
  
  For $m>2$, we obtain the analogous result
  \begin{gather*}
    a_{0} l + \Bigg( \sum_{i=1}^{p} - c_i x^{im-(m-2)} +
    d_i x^{im} \Bigg) - a_{p} k x^{pm+2} = 0 \ens , \\
    \text{with} \quad c_{i} = a_{i-1} (k+(p-i)m+m) \in \mbbr^{+} \ens
    , \quad \text{and} \quad d_{i} = a_{i} (l+im) \in \mbbr^{+} \ens .
  \end{gather*}
  From Descartes' rule of signs it is found that there exist $2p-1$,
  $2p$ or $2p+1$ extrema of $f(\theta)$ in $\theta \in (0, \pi)$,
  depending on whether $k$ and $l$ are zero or not.
  
  With these results the maximum number of global maxima of
  $f(\theta)$ can be determined.  Case differentiations have to be
  performed with regard to $m=2$ or $m>2$, whether $k$ and $l$ are
  zero or not and whether $p$ is even or odd.  The non-positive
  \acp{CPP} are obtained by considering the rotational $Z$-axis
  symmetry. For any positive \ac{CPP} not lying on a pole, there are
  $m-1$ other, non-positive \acp{CPP} lying at the same inclination
  (cf.  \lemref{pos_symm_cpp}).  For $m=2$, the maximum possible
  number of \acp{CPP} is $\tfrac{n}{2} + 1$ ($n$ even) or
  $\tfrac{n+1}{2}$ ($n$ odd). This is significantly less than in the
  general case $m>2$ where a lengthy calculation yields $2n-4$ as the
  maximum number of \acp{CPP}.  Interestingly, this bound decreases to
  $n$ if at least one of the two poles is free of \acp{MP}.
  
  Now consider case (c), i.e. states with no $Z$-axis rotational
  symmetry. All \acp{MP} of a positive state must either lie on the
  positive half circle or form complex conjugate pairs
  (cf. \lemref{maj_real} and \theoref{theo_general_maj_rep}). From
  this the optimisation function follows as
  \begin{equation*}
    f(\theta) = \sum_{i=0}^{n} a_i  \co_{\theta}^{n-i}
    \si_{\theta}^{i} \ens ,
  \end{equation*}
  with real $a_i$. Calculating $f'(\theta)$ yields the condition for
  the extrema:
  \begin{equation*}
    a_{1} + \left( \sum_{i=1}^{n-1} b_i x^{i} \right)
    - a_{n-1} x^{n} = 0 \ens , \quad \text{with} \quad
    b_{i} = a_{i+1} (i+1) - a_{i-1} (n-i+1) \ens .
  \end{equation*}
  The maximum number of \acp{CPP} again follows from Descartes'
  rule. All \acp{CPP} are now restricted to the positive half circle
  (which includes the poles), yielding at most $\tfrac{n+3}{2}$
  \acp{CPP} for odd $n$ and $\tfrac{n+2}{2}$ for even $n$.
\end{proof}

\cleardoublepage

\chapter{Maximally Entangled Symmetric States}
\label{solutions}

\begin{quotation}
  In this chapter we present the candidates for maximal symmetric
  entanglement of up to 12 qubits with respect to the geometric
  measure of entanglement. These solutions of the \quo{Majorana
    problem} were found by a combination of analytical and numerical
  methods, which are explained in the first part of this chapter.
  With the help of the Majorana representation the point distributions
  of the solutions can be compared to those of \toths and Thomson's
  classical optimisation problems on the sphere. The chapter concludes
  with a summary and discussion of the obtained results.
\end{quotation}

\section{Methodology}\label{methodology}

Exact solutions for \toths and Thomson's problem of distributing $n$
points over the surface of a sphere are known only for a few values of
$n$, with the highest one being $n=24$. Still, this compares
favourably to the Majorana problem \eqref{maj_problem} for which no
analytical solution beyond $n=3$ is known \cite{Chen10}. Due to the
complexity of an analytical treatment of this optimisation problem, it
makes sense to employ the help of numerics.  The combination of
analytical and numerical methods used for our search for the maximally
entangled symmetric states will be outlined in the following.

\subsection{Positive states}\label{pos_states}

For several reasons a particular emphasis has been placed on the
search for the maximally entangled symmetric state among positive
states.  Firstly, positive states are considerably easier to
investigate, because the number of parameters in the general form of
the state is reduced by half, and because the existence of at least
one positive \ac{CPP} (cf. \corref{positive_cpp}) simplifies the
determination of the geometric entanglement.  In particular, it is
sometimes possible to analytically determine the exact form of the
estimated maximally entangled state by requiring that the values of
the spherical amplitude function \eqref{spher_amp_funct} coincide at
two local maxima on the positive half-circle of the Majorana sphere,
thus fulfilling the necessary condition of at least two \acp{CPP} from
\corref{numberofcpp}.  This strategy will be employed for determining
the exact form of some states discussed in this chapter.
Additionally, the positive case is easier to investigate because
\lemref{pos_symm_cpp} and \theoref{cpp_number_locations} restrict the
number and locations of \acp{CPP}. To determine the locations of all
\acp{CPP}, it suffices to find the positive \acp{CPP}, because all
other \acp{CPP} follow from \lemref{pos_symm_cpp}.  In the case of
\quo{Platonic states} with positive coefficients it will be seen that
the \acp{CPP} follow without any calculations from
\lemref{pos_symm_cpp} and the rotation groups alone.

It is reasonable to expect that the Majorana representations of highly
entangled positive symmetric states are rotationally symmetric around
the $Z$-axis, because otherwise the \acp{CPP} can only lie on the
positive half-circle of the Majorana sphere
(cf. \lemref{pos_symm_cpp}), which results in an imbalance of the
spherical amplitude function $g( \theta , \varphi )$.  Because of
\lemref{rot_symm}, rotationally symmetric states have a large number
of vanishing coefficients, which considerably reduces the complexity
of numerical searches for high and maximal entanglement.

Another argument is that -- recast as quantum states by means of the
Majorana representation -- many of the solutions to \toths and
Thomson's problem for lower $n$ are given by positive
states\footnote{One reason for this is that these point distributions
  often exhibit a rotational symmetry, thus leading to a low number of
  nonvanishing basis states (cf. \protect\lemref{rot_symm}), which in
  turns makes it more likely that the state can be cast without
  phases.}.  One could thus expect that in many cases the solutions to
the Majorana problem can be cast as positive states too.  On the other
hand, arguments were presented in \sect{pos_results} that for systems
with sufficiently many parties the entanglement of positive states is
considerably lower than that of general states, and a similar
behaviour is expected for the subset of symmetric states.

\begin{table}
  \centering
  \caption[Geometric entanglement of selected symmetric states]
  {\label{table4.1}
    Geometric entanglement of some $n$ qubit Dicke states and
    superpositions of two Dicke states. For the first four states
    the \acp{CPP} and $\Eg$ can be determined exactly,
    while for the latter four states precise values are known
    only in the asymptotic limit $n \to \infty$.
    The weight of the basis states in superpositions was
    chosen to yield maximal entanglement.
    With the exception of
    $\frac{1}{\sqrt{1+e}} \left( \sqrt{e} \sym{1} + \sym{n}
    \right)$ all states are totally invariant, and therefore
    they are additive under $\Eg$, $E_{\text{R}}$ and
    $E_{\text{Rob}}$, with the amount of entanglement
    coinciding for these three measures (cf.
    \protect\theoref{pos_inv}).}
  \begin{tabular}{c|ccc}
    \toprule
    $\psisn$ & positive \acp{CPP} $\ket{\sigma}$ &
    $\Eg \big( \psisn \big)$ &
    $\lim\limits_{n \to \infty} \Eg \big( \psisn \big)$ \\
    \midrule
    $\sym{0}$  & $\ket{0}$ & $0$ & $0$ \\
    $\sym{1}$ &
    $\sqrt{\frac{n-1}{n}} \ket{0} + \sqrt{\frac{1}{n}} \ket{1}$ &
    $\log_2 \left( \frac{n}{n-1} \right)^{n-1}$ &
    $\log_2 (e)$ \\[0.3em]
    $\sym{\frac{n}{2}}$ &
    $\frac{1}{\sqrt{2}} \left( \ket{0} + \ket{1} \right)$ &
    $\log_2 \Big( \frac{2^{n}}{\binom{n}{n/2}} \Big)$ &
    $\log_2 \sqrt{\frac{n \pi}{2}}$ \\
    $\frac{1}{\sqrt{2}} \left( \sym{0} + \sym{n} \right)$ &
    $\ket{0}$ and $\ket{1}$ & $1$ & $1$ \\
    \toprule
    $\psisn$ &
    \multicolumn{2}{c}{positive \acp{CPP}
      $\ket{\sigma}$ for $n \to \infty$} &
    $\lim\limits_{n \to \infty} \Eg \big( \psisn \big)$ \\
    \midrule
    $\frac{1}{\sqrt{1+e}} \left( \sqrt{e}
      \sym{1} + \sym{n} \right)$ &
    \multicolumn{2}{c}{$\sqrt{\frac{n-1}{n}} \ket{0} +
      \sqrt{\frac{1}{n}} \ket{1}$} &
    $\log_2 (1 + \E )$ \\[0.3em]
    $\frac{1}{\sqrt{2}} \left( \sym{1} + \sym{n-1} \right)$ &
    \multicolumn{2}{c}{$\sqrt{\frac{n-1}{n}} \ket{0} +
      \sqrt{\frac{1}{n}} \ket{1}$
      and $ \sqrt{\frac{1}{n}} \ket{0} +
      \sqrt{\frac{n-1}{n}} \ket{1}$} &
    $\log_2 (2 \E )$ \\[0.3em]
    $\frac{1}{\sqrt{2}} \big( \sym{\frac{n}{3}} +
    \sym{\frac{2n}{3}} \big)$ &
    \multicolumn{2}{c}{$\sqrt{\frac{2}{3}} \ket{0} + \sqrt{\frac{1}{3}} \ket{1}$
      and $\sqrt{\frac{1}{3}} \ket{0} + \sqrt{\frac{2}{3}} \ket{1}$} &
    $\log_2 \big( \frac{4}{3} \sqrt{n \pi} \big)$ \\[0.3em]
    $\frac{1}{\sqrt{2}} \big( \sym{\frac{n}{4}} +
    \sym{\frac{3n}{4}} \big)$ &
    \multicolumn{2}{c}{$\frac{\sqrt{3}}{2} \ket{0} + \frac{1}{2} \ket{1}$
      and $ \frac{1}{2} \ket{0} + \frac{\sqrt{3}}{2} \ket{1}$} &
    $\log_2 \sqrt{\frac{3 n \pi}{2}}$ \\
    \bottomrule
  \end{tabular}
\end{table}

Symmetric states with no more than two basis states can always be cast
positive, regardless of the number of parties.  Some exemplary
calculations were performed with such states, and the results are
shown in \tabref{table4.1}.  Listed are the entanglement of some $n$
qubit Dicke states and superpositions of two Dicke states, both for
fixed $n$ and in the asymptotic limit $n \to \infty$.  It can be seen
that states tend to have more geometric entanglement if they contain
Dicke states with a relatively balanced number of excitations.  For
example, the entanglement of the most balanced Dicke state
$\sym{{\frac{n}{2}}}$ scales as $\Eg = \Order ( \log_2 \sqrt{n} )$,
while $\ket{\text{GHZ}_{n}} = \tfra{1}{\sqrt{2}} \left( \sym{0} +
  \sym{n} \right)$ has an entanglement of only $\Eg = 1$, regardless
of $n$.

Finally, we outline how the most entangled positive symmetric states
were found.  Case differentiations were performed for all combinations
of vanishing and nonvanishing basis states, with the approximate value
of the maximal entanglement determined numerically in each case.  In
this fashion all the cases with significantly lower entanglement could
be ruled out. For the cases that exhibited high entanglement the
precise amount was calculated analytically (where possible) or
approximated numerically to a high precision. In all cases it was
found that the maximally entangled states can be expressed with a low
number of nonvanishing basis states.  This finding justifies the
omission of some cases with large numbers of nonvanishing basis states
in the numerical search of $n > 5$ qubits due to their complexity.
However, all possible states with a rotational symmetry around the
$Z$-axis were taken into account.  For these reasons we can be
confident that the most entangled states found this way are indeed the
maximally entangled positive symmetric ones.

\subsection{General states}\label{gen_states}

In the case of general symmetric states we do not have as many
analytic tools as for positive states, so the search is over a far
bigger set of possible states, and we can be less confident in our
results.  In particular, the candidates for maximal entanglement in
the general case should be treated with a certain amount of caution,
because states with yet more entanglement may exist.

A helpful result from a theoretical viewpoint is \corref{numberofcpp}
which provides the necessary conditions of at least two \acp{CPP} and
that the maximally entangled symmetric state must lie in the span of
its \acp{CPS}.  For all of our candidates for maximal symmetric
entanglement it was verified that these conditions are met.  From a
numerical viewpoint the known solutions to \toths and Thomson's
problem -- converted to symmetric $n$ qubit states via
\eq{majorana_definition} -- readily provide nontrivial lower bounds on
the maximal symmetric entanglement.  Martin \etal \cite{Martin10}
computed the geometric entanglement of these states for up to $n =
110$ and found that the solutions of Thomson's problem generally yield
higher entanglement than those of \toths problem.  Furthermore, they
found that the entanglement of Thomson's solutions scales as
$\Eg^{\text{Th}} \approx \log_2 \frac{(n+1)}{1.71}$, which is close to
the upper bound $\Eg \leq \log_2 (n+1)$ derived in \sect{upper_bound},
and thus leaves only a narrow corridor for the maximal values of
$\Eg$.  The classical solutions, particularly those for Thomson's
problem, are therefore a good starting point for an explicit search
for the maximally entangled symmetric state.  In some cases (for
$n=10$ in \sect{majorana_ten} and $n=11$ in \sect{majorana_eleven}) we
found the conjectured maximally entangled symmetric state by making
small modifications to the point distributions of the classical
solutions.  Another strategy to find highly entangled states is to
consider states with certain symmetry features in their \ac{MP}
distribution, such as rotational and reflective symmetries.

\section{Results}\label{maxent_results}

Before discussing the cases of 4 to 12 qubits as well as the 20 qubit
case, we remark that another study of highly and maximally entangled
symmetric states was independently performed by Martin \etal
\cite{Martin10}, and that their results are similar to ours. In their
paper they focused on using databases \cite{Sloane,Wales} with the
known numerical solutions of \toths and Thomson's problem to derive
the geometric entanglement of the corresponding symmetric states for
up to $n=110$ in a straightforward manner, and they found the
maximally entangled symmetric states for up to $n=6$ qubits.  In our
publication \cite{Aulbach10} we studied the cases of up to $n=12$
qubits in much more detail. In particular, we presented candidates for
maximal symmetric entanglement for each $n$, discussed the Majorana
representations of highly entangled states, and discovered that the
spherical volume function is a very useful tool for understanding the
distribution patterns of \acp{MP} and \acp{CPP}.  A summary of the
properties of the symmetric states investigated in this chapter can be
found in \tabref{table4.5}.

\subsection{Four qubits}\label{majorana_four}

The \acp{MP}, \acp{CPP} and the geometric entanglement of the
tetrahedron state were already discussed in \sect{visualisation}, with
different visualisations shown in \fig{tetrahedron_visual}.  For four
points both \toths and Thomson's problem are solved by the vertices of
the regular tetrahedron \cite{Whyte52}, and our numerical search for
the maximally entangled symmetric state returned this Platonic solid
too. Furthermore, the tetrahedron state $\ket{\Psi_{4}} =
\sqrt{\tfra{1}{3}} \sym{0} + \sqrt{\tfra{2}{3}} \sym{3}$ satisfies the
necessary conditions of \corref{numberofcpp}, because it can be
written as a linear combination of its \acp{CPS}: $\ket{\Psi_{4}} =
\tfrac{\sqrt{3}}{4} \big( \ket{\phi_1}^{\otimes 4} +
\ket{\phi_2}^{\otimes 4} + \ket{\phi_3}^{\otimes 4} +
\ket{\phi_4}^{\otimes 4} \big)$.  In the following we focus on the
various interesting properties of the tetrahedron state.

Firstly, we note that the tetrahedron state is totally invariant under
the tetrahedral symmetry group $T \subset \text{SO}(3)$, and because
$\ket{\Psi_{4}}$ is a positive state, \theoref{pos_inv} gives us the
equivalence and additivity of this state under $\Eg$, $E_{\text{R}}$
and $E_{\text{Rob}}$.  Further results can be derived from the
tetrahedral symmetry group: Even though $\text{R}_{T}^{\otimes 4}
\ket{\Psi_{4}} = \ket{\Psi_{4}}$ holds for all $\text{R}_{T} \in T$,
the individual \acp{MP} are not necessarily left invariant:
$\text{R}_{T} \ket{\phi_{i}}= \ket{\phi_{j}}$, with $i,j \in \{
1,2,3,4 \}$.  As seen in \fig{4_cpps}, this can be viewed as a
permutation of the \acp{MP}. Naturally, the $\text{R}_T$ can also be
viewed as rotations of the Majorana sphere along an axis running
though the sphere.  Because $\ket{\Psi_{4}}$ is positive,
\lemref{pos_symm_cpp} restricts the allowed locations for \acp{CPP} to
the three half-circles shown as blue lines in \fig{4_cpps}(a). These
lines rotate with the Majorana sphere under the action of
$\text{R}_T$, which allows us to apply \lemref{pos_symm_cpp} again in
the new orientation.  As shown in \fig{4_cpps}(b) and (c), two
successive rotations give rise to further restricting areas, coloured
green and red.  Any \ac{CPP} can only lie at the intersections of the
blue, green and red lines.  From \fig{4_cpps}(c) it can be seen that
these intersections are precisely the locations of the four \acp{MP},
thus providing us with all the \acp{CPP} of $\ket{\Psi_{4}}$ without
the need for any calculations.  This is remarkable, because the
determination of \acp{CPP} is usually a highly nontrivial task which
requires at least some analytical or numerical effort.

\begin{figure}
  \centering
  \begin{overpic}[scale=.45]{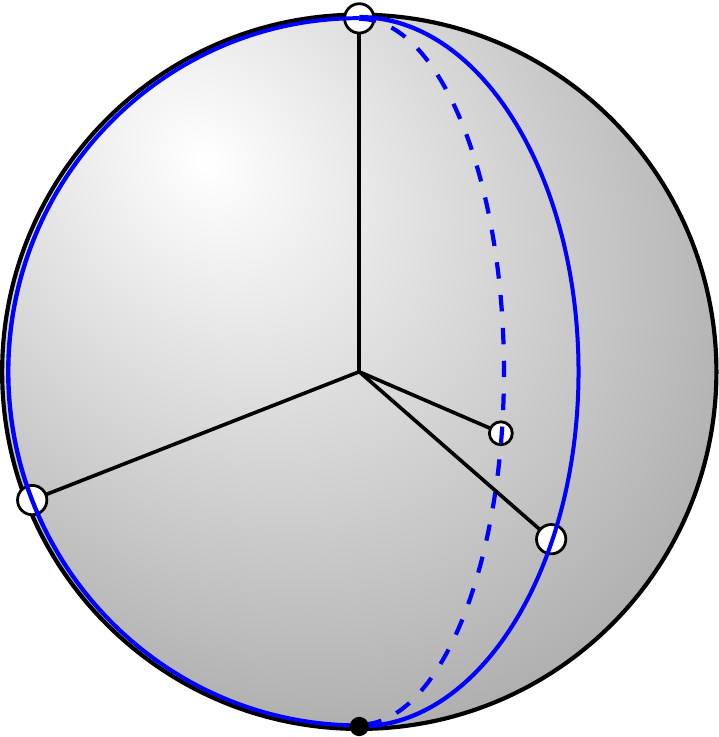}
    \put(-5,0){(a)}
    \put(43,103){$1$}
    \put(-7,23){$2$}
    \put(73,13){$3$}
    \put(59,49){$4$}
  \end{overpic}
  \hspace{7mm}
  \begin{overpic}[scale=.45]{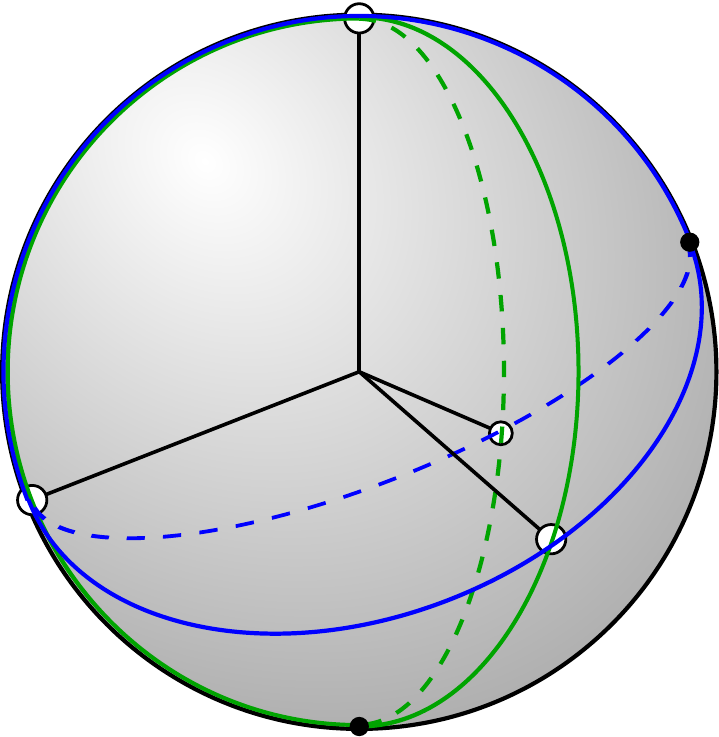}
    \put(-5,0){(b)}
    \put(43,103){$3$}
    \put(-7,23){$1$}
    \put(73,13){$2$}
    \put(59,49){$4$}
  \end{overpic}
  \hspace{7mm}
  \begin{overpic}[scale=.45]{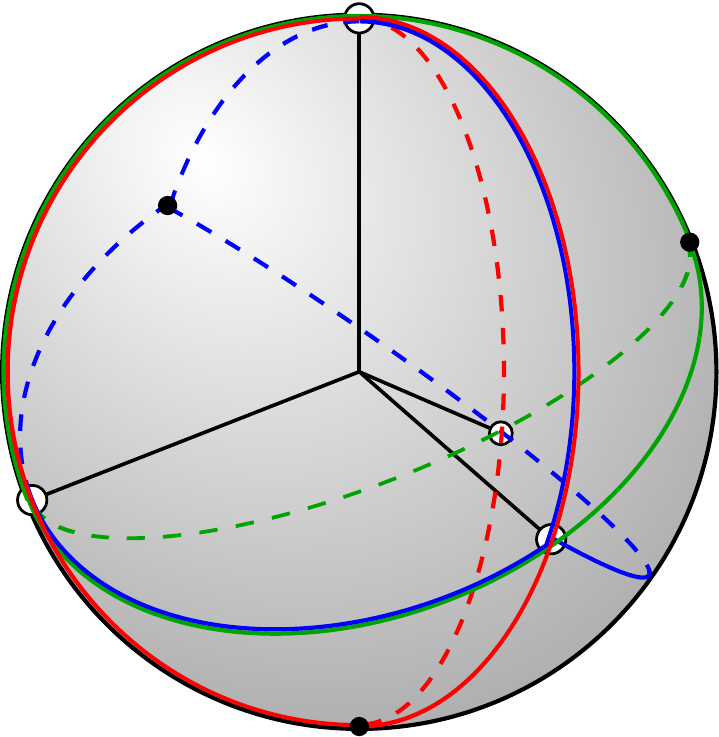}
    \put(-5,0){(c)}
    \put(43,103){$2$}
    \put(-7,23){$3$}
    \put(73,13){$1$}
    \put(59,49){$4$}
  \end{overpic}
  \caption[Determining the CPPs from the rotation
  group]{\label{4_cpps} The four \acp{CPP} of the tetrahedron state
    can be directly obtained from the tetrahedral symmetry group and
    from \protect\lemref{pos_symm_cpp}.  Rotations from $T \subset
    \text{SO}(3)$ amount to permutations of the \acp{MP} and thus
    provide additional restrictions to the allowed locations of the
    \acp{CPP}.  An $\text{R} (\frac{2 \pi}{3})$-rotation along the
    axis given by the Bloch vector of \ac{MP} $\ket{\phi_4}$ is
    performed twice between (a) and (c).  Any \ac{CPP} must lie at the
    intersections of the blue, green and red lines shown in (c),
    yielding the locations of the four \acp{MP}.}
\end{figure}
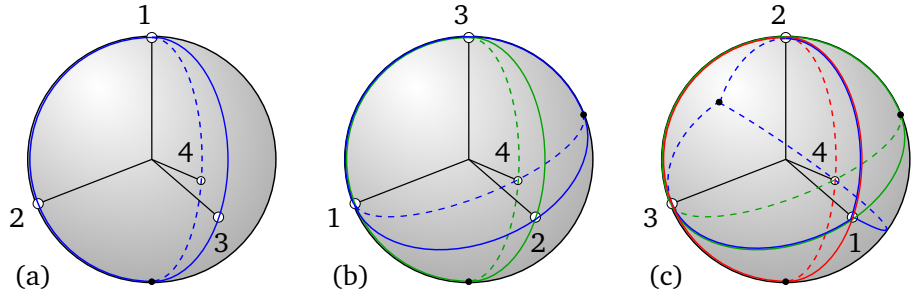

Apart from being the unique state with maximal geometric entanglement
among symmetric states, the special position of the tetrahedron state
in $\mathh = (\mbbc^2)^{\otimes 4}$ can be noticed in other ways.  As
a Platonic state, it is an optimal state for reference frame alignment
\cite{Kolenderski08}, and in terms of symmetric informationally
complete positive-operator-valued measures (SIC POVM)
\cite{Zauner11,Renes04} it was found that the tetrahedron state is the
unique state that can be generated in the setting of a two-dimensional
Hilbert space \cite{Renes04,Chen11}.  In \sect{global_ent} it will be
outlined that -- along with the four qubit cluster state and \ac{GHZ}
state -- the tetrahedron state is one of the three maximally entangled
four qubit states under a monotone that requires all $k$-tangles with
$k < 4$ to vanish \cite{Osterloh05,Dokovic09}. Under more stringent
requirements it is even the only state to be maximally entangled
\cite{Gisin98}.  Furthermore, witnesses for all types of 4 qubit
symmetric entanglement can be constructed from the tetrahedron state,
something that is not possible with other prominent states
\cite{Bastin11}.  Finally, through a private communication with Mio
Murao I became aware of the following bipartite decomposition of the
tetrahedron state:
\begin{equation}\label{tetra_bi}
  \begin{split}
    \ket{\Psi_{4}}& = \tfrac{1}{\sqrt{3}} \ket{00}\ket{00} +
    \tfrac{1}{\sqrt{6}} \left( \ket{01}\ket{11} + \ket{10}\ket{11} +
      \ket{11}\ket{01} + \ket{11}\ket{10} \right)
    \\
    {}& = \tfrac{1}{\sqrt{3}} \left( \ket{a}\ket{a} + \ket{b}\ket{c} +
      \ket{c}\ket{b} \right) \ens ,
  \end{split}
\end{equation}
with $\ket{a} = \ket{00}$, $\ket{b} = \frac{1}{\sqrt{2}} \left(
  \ket{01} + \ket{10} \right)$ and $\ket{c} = \ket{11}$ as
orthonormalised 2-qubit states. In this sense the tetrahedron state
contains maximally entangled bipartite qutrit (three-level) systems
along any split into two 2-qubit subsystems.  Viewed as such a
2-qutrit system, the entanglement of the state \eqref{tetra_bi} is
$\Eg = E_{\text{R}} = E_{\text{Rob}} = \log_2 d = \log_2 3$, and
because it is a pure bipartite state, it is additive under all three
measures.  Therefore the tetrahedron state retains its entanglement
and additivity properties when viewed as a 2-qutrit state instead of a
4-qubit state.  This bears resemblance to the 4 qubit cluster state
\cite{Briegel01}, which can be written in the form
\begin{equation}\label{4_qubit_cluster}
  \ket{\Psi^{c}_{4}} = \tfrac{1}{2} \big( \ket{0000} + \ket{0101} +
  \ket{1010} - \ket{1111} \big) \ens .
\end{equation}
This state has the entanglement $\Eg = E_{\text{R}} = E_{\text{Rob}} =
2$ and is additive under the three measures \cite{Markham07}. Taking
the bipartite cut along neighbouring qubits ($12-34$ or $14-23$)
clearly results in a maximally entangled bipartite state of two
four-level subsystems yielding the entanglement $\Eg = E_{\text{R}} =
E_{\text{Rob}} = \log_2 d = 2$. Unlike the tetrahedron state, however,
the cluster state is not symmetric, and indeed the bipartite cut along
diametrically opposite qubits ($13-24$) yields less entanglement:
\begin{equation}\label{4_qubit_cluster_diag_cut}
  \begin{split}
    \ket{\Psi^{c}_{4}}_{13-24} &= \tfrac{1}{2} \big( \ket{00}\ket{00}
    + \ket{00}\ket{11} + \ket{11}\ket{00} - \ket{11}\ket{11} \big) \\
    {}& = \tfrac{1}{\sqrt{2}} \big( \ket{a} \ket{\widetilde{a}} +
    \ket{b} \ket{\widetilde{b}} \big) \ens ,
  \end{split}
\end{equation}
with $\ket{a} = \ket{00}$, $\ket{b} = \ket{11}$, $\ket{\widetilde{a}}
= \tfrac{1}{\sqrt{2}} \left( \ket{00} + \ket{11} \right)$ and
$\ket{\widetilde{b}} = \tfrac{1}{\sqrt{2}} \left( \ket{00} - \ket{11}
\right)$. Under this bipartition only 1 ebit of entanglement is
obtained.

\subsection{Five qubits}\label{majorana_five}

For five points, the solution to Thomson's problem is given by three
of the charges lying on the vertices of an equatorial triangle and the
other two lying at the poles \cite{Ashby86,Marx70}.  Such a trigonal
bipyramid is also a solution to \toths problem, but it is not
unique\footnote{An example of another solution is the square pyramid
  obtained from the regular octahedron by removing one vertex, and a
  continuum of solutions is given by two fixed vertices on the poles
  and the other three vertices lying on the equatorial circle with
  spherical distances $\tfra{\pi}{2} \leq s_{\text{min}} \leq \tfra{2
    \pi}{3}$ between each pair.  In all of these configurations the
  minimum pairwise distance is $s_{\text{min}} = \tfra{\pi}{2}$.}
\cite{Toth43,Melnyk77,Schutte51}.  The corresponding quantum state,
the \quo{trigonal bipyramid state} $\ket{\psi_{5}} =
\tfra{1}{\sqrt{2}} ( \sym{1} + \sym{4} )$, has the \acp{MP}
\begin{equation}\label{5_tribi_maj}
  \ket{\phi_{1}} = \ket{0} \ens , \quad
  \ket{\phi_{2, 3, 4}} = \tfrac{1}{\sqrt{2}}
  ( \ket{0} + \E^{\I \kappa} \ket{1} ) \ens , \quad
  \ket{\phi_{5}} = \ket{1} \ens ,
\end{equation}
with $\kappa = 0, \tfra{2 \pi}{3}, \tfra{4 \pi}{3}$.  As seen in
\fig{bloch_5}(a) and (b), the state has three \acp{CPP} which coincide
with the equatorial \acp{MP}, yielding $\Eg (\ket{\psi_{5}}) = \log_2
\left( \tfra{16}{5} \right) \approx 1.678 \: 072$.  The trigonal
bipyramid state is positive and totally invariant under the dihedral
symmetry group $D_{m}$ for $ m=3$, which implies that $\ket{\psi_{5}}$
is equivalent and additive under $\Eg$, $E_{\text{R}}$ and
$E_{\text{Rob}}$.

\begin{figure}
  \centering
  \begin{overpic}[scale=.44]{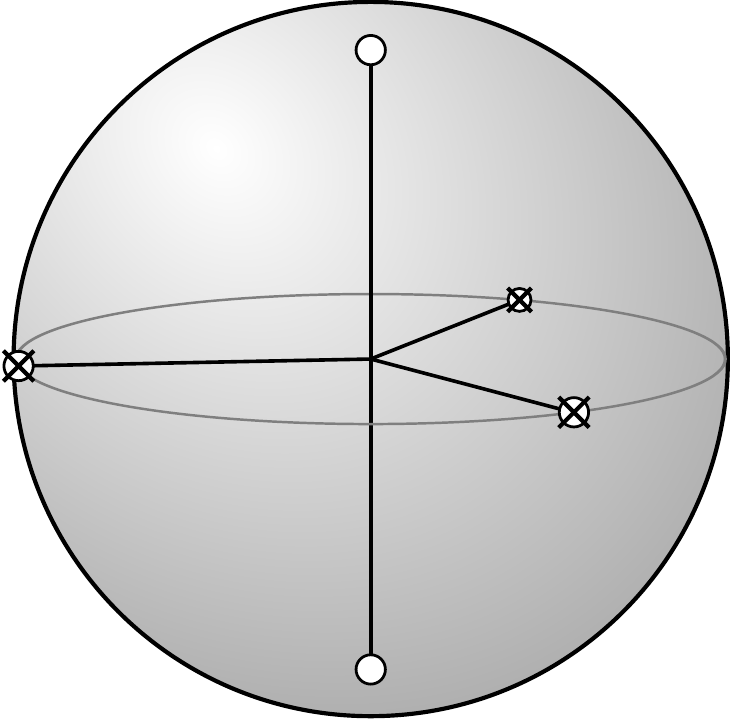}
    \put(-7,-2){(a)}
    \put(29,83){$\ket{\phi_{1}}$}
    \put(-19,50){$\ket{\phi_{2}}$}
    \put(74,29){$\ket{\phi_{3}}$}
    \put(67,64){$\ket{\phi_{4}}$}
    \put(29,10){$\ket{\phi_{5}}$}
  \end{overpic}
  \begin{overpic}[scale=.36]{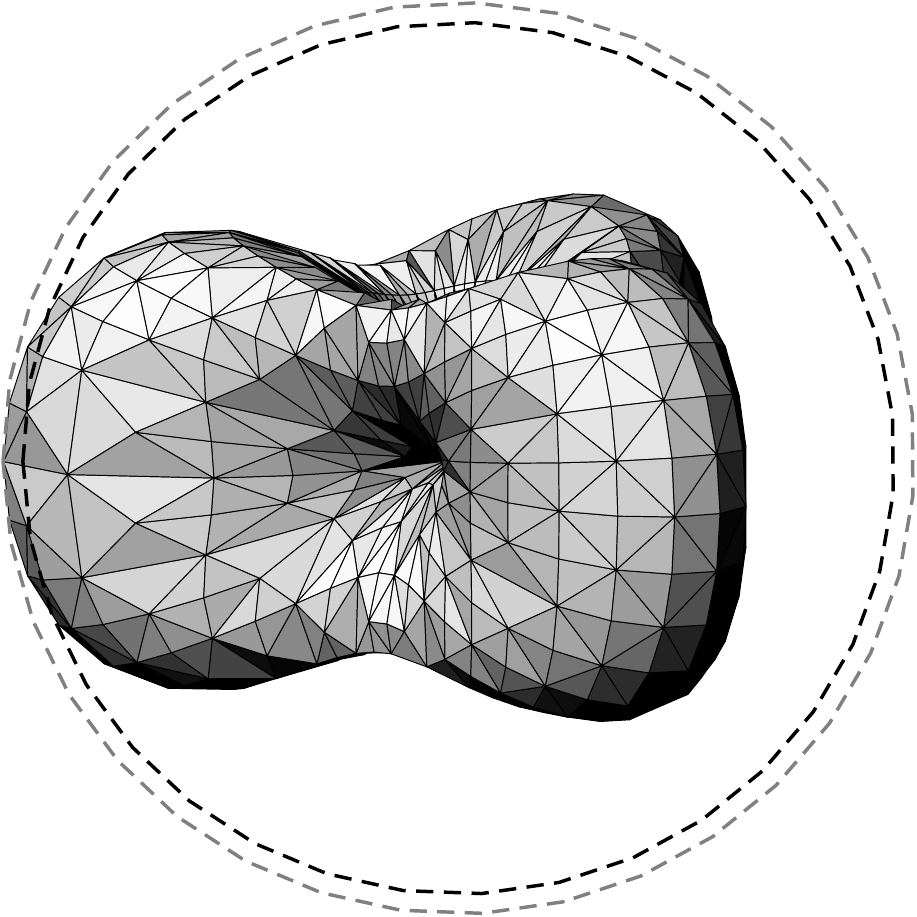}
    \put(-7,-2){(b)}
  \end{overpic}
  \hfill
  \begin{overpic}[scale=.44]{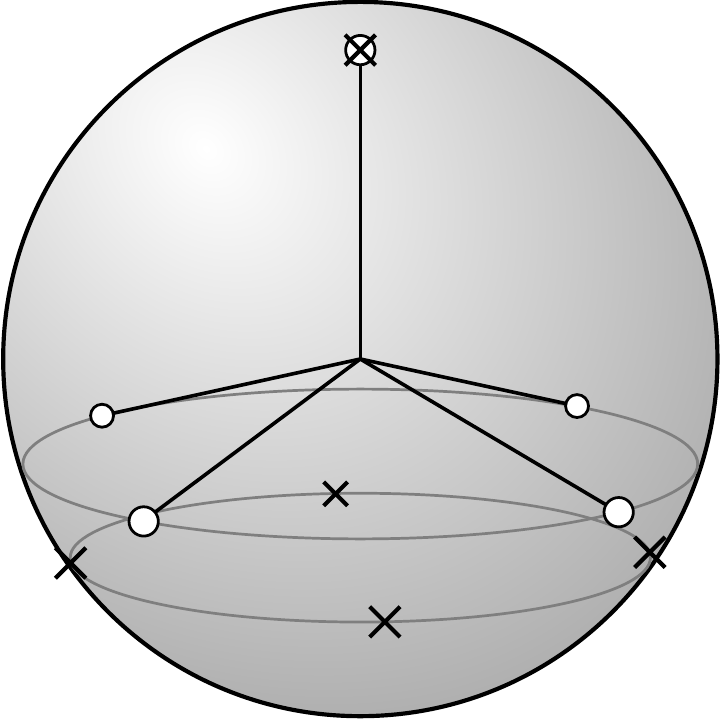}
    \put(-7,-2){(c)}
    \put(29,84){$\ket{\phi_{1}}$}
    \put(20,16){$\ket{\phi_{2}}$}
    \put(64,20){$\ket{\phi_{3}}$}
    \put(74,50){$\ket{\phi_{4}}$}
    \put(4,49){$\ket{\phi_{5}}$}
  \end{overpic}
  \begin{overpic}[scale=.36]{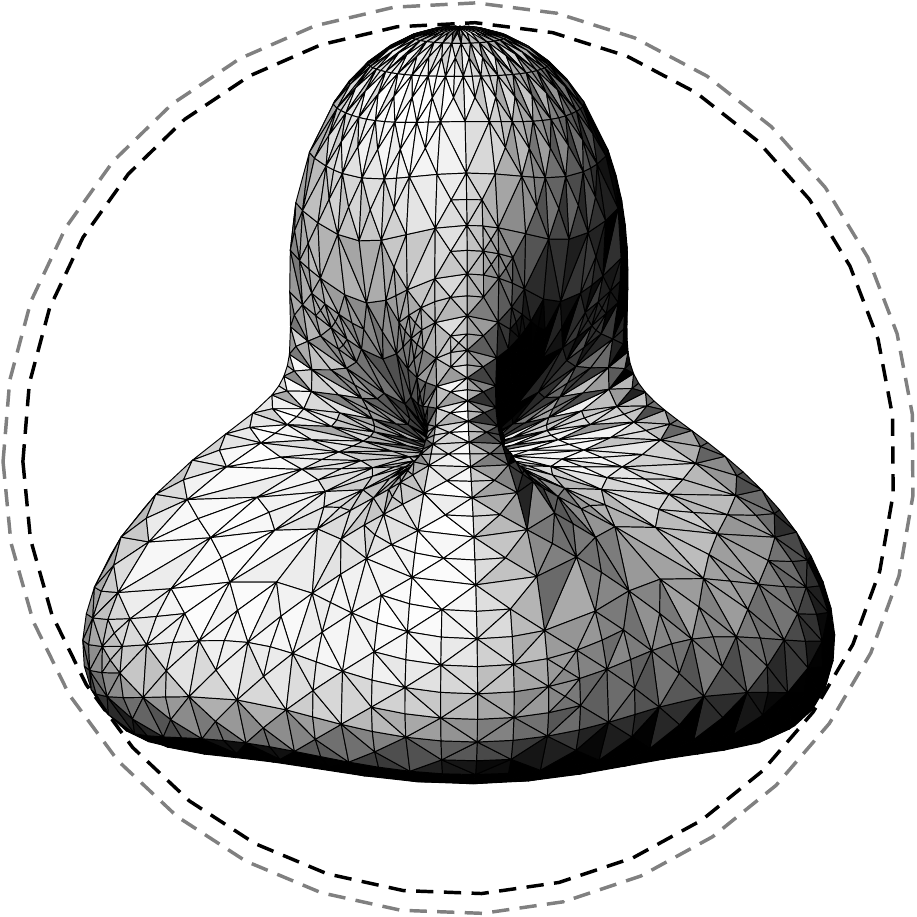}
    \put(-7,-2){(d)}
  \end{overpic}
  \caption[Majorana representation of 5 qubit states]{\label{bloch_5}
    For five qubits the Majorana representation and the spherical
    amplitude function $g^2 ( \theta , \varphi )$ of the
    \protect\quo{trigonal bipyramid state} $\ket{\psi_{5}}$ is shown
    in (a) and (b), respectively. The same visualisations are shown
    for the \protect\quo{square pyramid state} $\ket{\Psi_{5}}$ in (c)
    and (d), respectively. The dashed circles in (b) and (d) mark the
    maximum values of $g^2$, with the outer gray circle corresponding
    to the less entangled state $\ket{\psi_{5}}$ and the inner black
    circle to the more entangled state $\ket{\Psi_{5}}$.}
\end{figure}
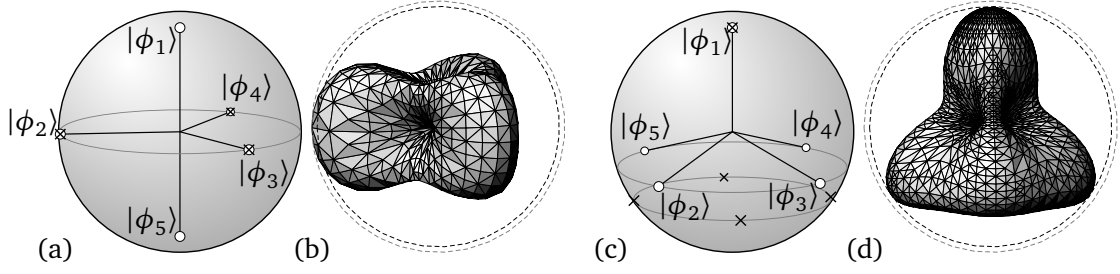

However, a numerical search among symmetric five qubit states yields
states with higher entanglement.  The conjectured maximally entangled
state has the \ac{MP} distribution of the square pyramid\footnote{This
  square pyramid cannot be a solution of \protect\toths problem,
  because the spherical distance between some \acp{MP}
  (e.g. $\ket{\phi_{2}}$ and $\ket{\phi_{3}}$) is $s_{\text{min}} <
  \tfra{\pi}{2}$.}  shown in \fig{bloch_5}(c), which corresponds to
the analytic form
\begin{equation}\label{5_opt_form}
  \ket{\Psi_{5}} = \frac{\sym{0} + A \sym{4}}{\sqrt{1 + A^2}}
  \ens ,
\end{equation}
where the \acp{MP} are
\begin{equation}\label{5_opt_maj}
  \ket{\phi_{1}} = \ket{0} \ens , \quad
  \ket{\phi_{2, 3, 4, 5}} = \sqrt{\alpha} \ket{0} + \E^{\I \kappa}
  \sqrt{1 - \alpha} \ket{1} \ens ,
\end{equation}
with $\kappa = \tfra{\pi}{4}, \tfra{3 \pi}{4}, \tfra{5 \pi}{4},
\tfra{7 \pi}{4}$. The relationship between $A$ and $\alpha$ is
\begin{equation}\label{A_alpha_rel}
  A = \frac{(1 - \alpha)^2}{\sqrt{5} \alpha^2} \ens .
\end{equation}
The value of $A$ which maximises the entanglement of $\ket{\Psi_{5}}$
gives rise to a state with five \acp{CPP}, one on the north pole and
the other four lying on a horizontal circle below the plane with the
\acp{MP}, i.e.
\begin{equation}\label{5_opt_cpp}
  \ket{\sigma_{1}} = \ket{0} \ens , \quad
  \ket{\sigma_{2, 3, 4, 5}} = x_{\text{m}} \ket{0} +
  k \, y_{\text{m}} \ket{1} \ens ,
\end{equation}
with $x_{\text{m}}^2 + y_{\text{m}}^2 = 1$ and $k = 1, \I, -1, -\I$.
Approximate values of these quantities are:
\begin{equation}\label{5_opt_form_approx}
  A \approx 1.531 \: 538 \: 191 \ens ,
  \quad \alpha \approx 0.350 \: 806 \: 560 \ens ,
  \quad x_{\text{m}} \approx 0.466 \: 570 \: 328 \ens .
\end{equation}
The exact values can be determined analytically.  Since
$\ket{\Psi_{5}}$ is positive, it suffices to investigate the maxima of
the spherical amplitude function $g ( \theta , \varphi )$ along the
positive half-circle: $g ( \theta ) \equiv g ( \theta , 0 )$.  Using
the parameterisation $x: = \co_{\theta} \in [0,1]$, an analysis shows
that the global maximum of $g ( x )$ becomes minimal when the value
$g(1)$ at the local maximum $x=1$ equals the value $g(x_{\text{m}})$
at the non-trivial maximum $x_{\text{m}} \in (0,1)$.  With the ansatz
$g(1) = g(x_{\text{m}})$ it follows that $A = \frac{1 -
  x_{\text{m}}^5} {\sqrt{ 5 } x_{\text{m}} y_{\text{m}}^4 }$, and the
requirement $ \frac{\D g}{\D x} (x_{\text{m}}) = 0$ yields $4
x_{\text{m}}^5 - 5 x_{\text{m}}^2 + 1 = 0$.  A polynomial division by
the trivial root $x_{\text{m}} = 1$ reduces this quintic equation to a
quartic one:
\begin{equation}\label{5_quartic}
  4 x_{\text{m}}^4 + 4 x_{\text{m}}^3 + 4 x_{\text{m}}^2 -
  x_{\text{m}} - 1 = 0 \ens .
\end{equation}
The real root in the interval $[0,1]$ can be determined analytically
by a reduction to cubic equations and Cardano's Formula \cite{Tignol}:
\begin{equation}
  \begin{split}
    x_{\text{m}}& = \frac{1}{4} \left( \sqrt{8 z - 3} - 1 + \sqrt{
        \tfra{10}{\sqrt{8 z - 3}} - 2 - 8 z } \, \right) \ens , \\
    \text{with} \quad z& = \frac{1}{24} \left( \sqrt[3]{ 550 + 30
        \sqrt{345} } + \sqrt[3]{ 500 - 30 \sqrt{345} } \, \right) +
    \frac{1}{6} \ens .
  \end{split}
\end{equation}
This $x_{\text{m}}$ establishes the nontrivial positive \ac{CPP}
$\ket{\sigma_2} = x_{\text{m}} \ket{0} + \sqrt{1 - x_{\text{m}}^2}
\ket{1}$, and by inserting it into $A = \frac{1 - x_{\text{m}}^5}
{\sqrt{ 5 } x_{\text{m}} y_{\text{m}}^4 }
\stackrel{\eqref{5_quartic}}{=} \frac{\sqrt{5}}{4 x_{\text{m}} (1 -
  x_{\text{m}}^2)}$, it yields the explicit form of $\ket{\Psi_{5}}$.
The parameter $\alpha$ of the \acp{MP} follows by solving
\eq{A_alpha_rel}.  From the \ac{MP} distribution in \fig{bloch_5}(c)
it is clear that $\ket{\Psi_{5}}$ remains invariant under the cyclic
symmetry group $C_{m}$ with $m=4$. However, it is not totally
invariant, because the latitude of the horizontal circle of \acp{MP}
can be varied without changing the rotation group.

The amount of entanglement $\Eg (\ket{\Psi_{5}}) = \log_2 ( 1 + A^2 )
\approx 1.742 \: 269$ of the square pyramid state is considerably
higher than that of the trigonal bipyramid state.  Martin \etal
\cite{Martin10} independently found a square pyramid state as the
maximally entangled symmetric five qubit state, and we verified that
their state is the same as ours.

\subsection{Six qubits}\label{majorana_six}

The regular octahedron, a Platonic solid, is the unique solution to
\toths and Thomson's problem.  The corresponding \quo{octahedron
  state} $\ket{\Psi_{6}} = \tfra{1}{\sqrt{2}} ( \sym{1} + \sym{5} )$
was numerically found to solve the Majorana problem for six qubits
too.  In the orientation shown in \fig{bloch_6}(a) the \acp{MP} are
\begin{equation}\label{6_mp}
  \ket{\phi_{1}} = \ket{0} \ens , \quad
  \ket{\phi_{2}} = \ket{1} \ens , \quad
  \ket{\phi_{3, 4, 5, 6}} = \tfrac{1}{\sqrt{2}}
  \big( \ket{0} + \E^{\I \kappa}
  \ket{1} \big) \ens ,
\end{equation}
with $\kappa = \tfra{\pi}{4}, \tfra{3 \pi}{4}, \tfra{5 \pi}{4},
\tfra{7 \pi}{4}$.  The octahedron state has a positive computational
form and is totally invariant under the octahedral symmetry group $O
\subset \text{SO}(3)$, implying that it is equivalent and additive
under $\Eg$, $E_{\text{R}}$ and $E_{\text{Rob}}$.  Furthermore, the
rotational invariance enables us to determine the \acp{CPP} with
\lemref{pos_symm_cpp}. As seen in \fig{bloch_6}(d), only one rotation
suffices to determine the eight \acp{CPP} at the centre of each face
of the octahedron, forming a cube inside the Majorana sphere:
\begin{equation}\label{6_cpp}
  \ket{\sigma_{1, 2, 3, 4}} = \sqrt{\tfrac{\sqrt{3}+1}{2 \sqrt{3}}}
  \ket{0} + k \sqrt{\tfrac{\sqrt{3}-1}{2 \sqrt{3}}} \ket{1}
  \ens , \hspace{2mm} \ket{\sigma_{5, 6, 7, 8}} =
  \sqrt{\tfrac{\sqrt{3}-1}{2 \sqrt{3}}} \ket{0} +
  k \sqrt{\tfrac{\sqrt{3}+1}{2 \sqrt{3}}} \ket{1} \, ,
\end{equation}
with $k = 1,\I,-1,-\I$.  The geometric entanglement follows as $\Eg (
\ket{\Psi_{6}} ) = \log_2 \left( \tfra{9}{2} \right) \approx 2.169 \:
925$.
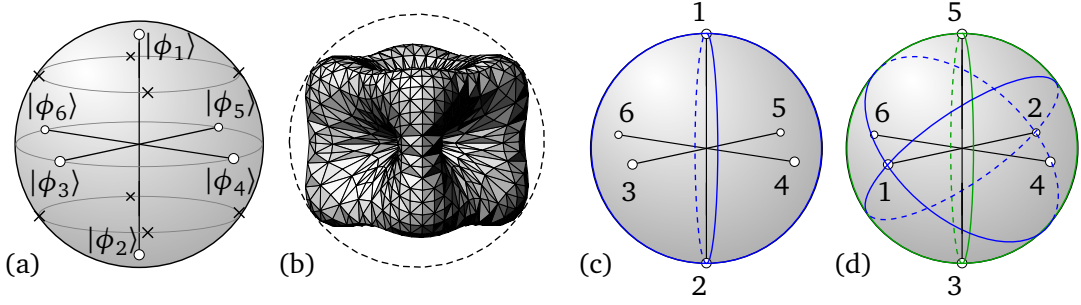
\begin{figure}
  \centering
  \begin{overpic}[scale=.45]{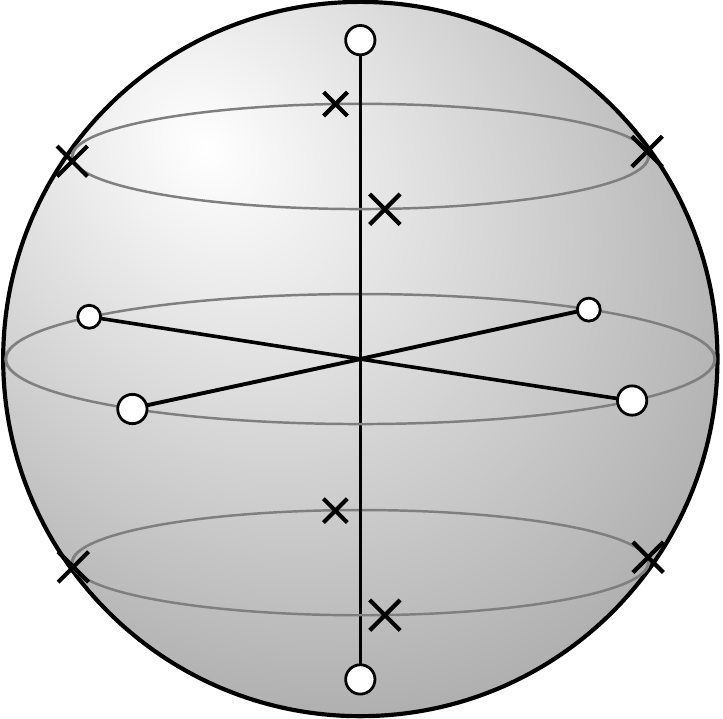}
    \put(-4,-2){(a)}
    \put(52,86){$\ket{\phi_{1}}$}
    \put(29,8){$\ket{\phi_{2}}$}
    \put(6,31){$\ket{\phi_{3}}$}
    \put(76,32){$\ket{\phi_{4}}$}
    \put(76,63){$\ket{\phi_{5}}$}
    \put(5,62){$\ket{\phi_{6}}$}
  \end{overpic}
  \hspace{1mm}
  \begin{overpic}[scale=.51]{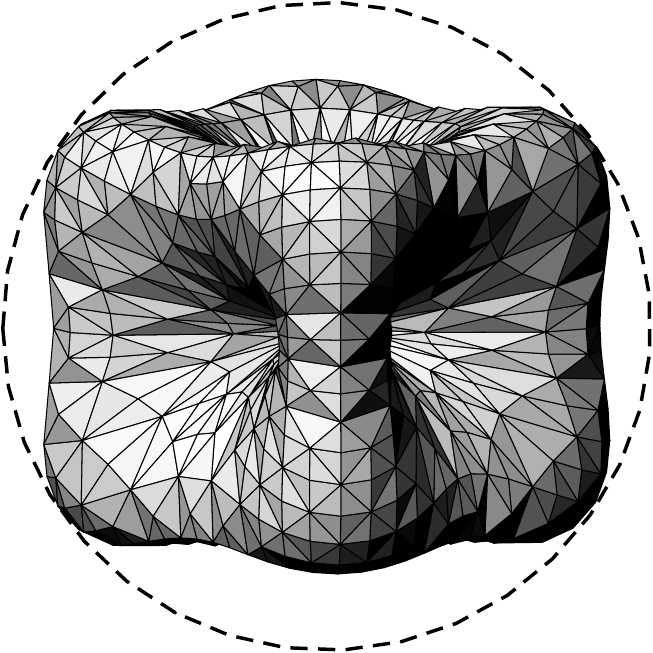}
    \put(-4,-2){(b)}
  \end{overpic}
  \hfill
  \begin{overpic}[scale=.42]{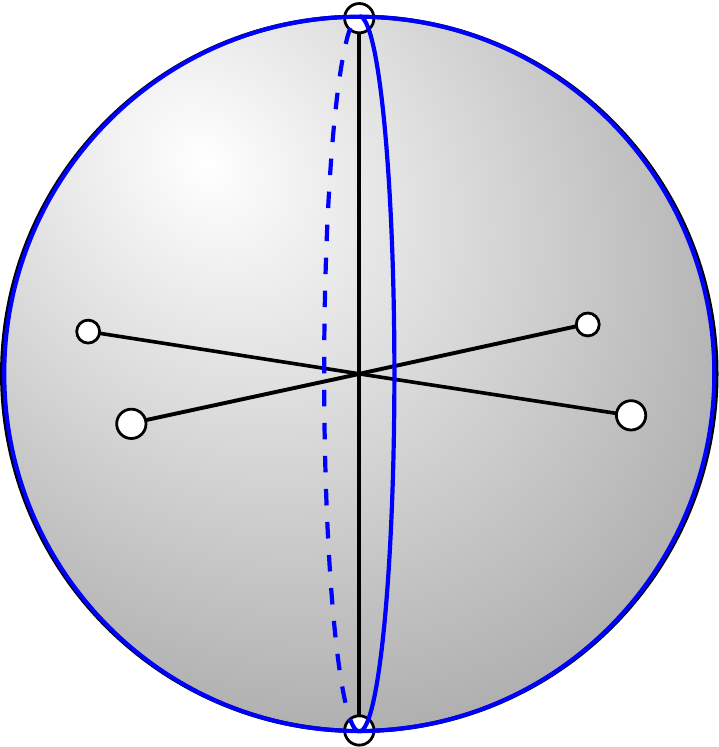}
    \put(-5,-2){(c)}
    \put(42,103){$1$}
    \put(42,-11){$2$}
    \put(13,28){$3$}
    \put(76,31){$4$}
    \put(75,63){$5$}
    \put(12,61){$6$}
  \end{overpic}
  \hspace{1mm}
  \begin{overpic}[scale=.42]{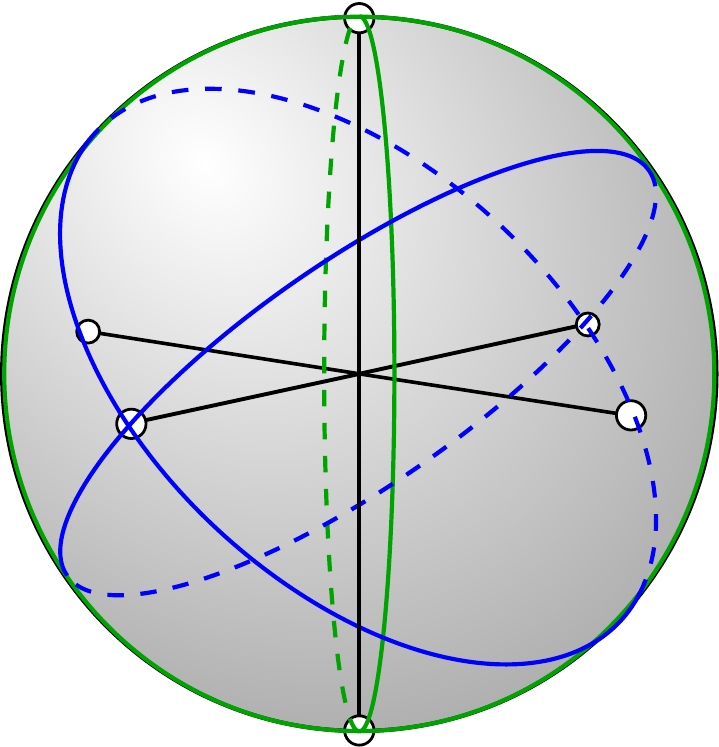}
    \put(-5,-2){(d)}
    \put(42,103){$5$}
    \put(42,-11){$3$}
    \put(13,28){$1$}
    \put(76,31){$4$}
    \put(75,63){$2$}
    \put(12,61){$6$}
  \end{overpic}
  \caption[Majorana representation of 6 qubit states]{\label{bloch_6}
    For six qubits the Majorana representation and the spherical
    amplitude function of the \protect\quo{octahedron state}
    $\ket{\Psi_{6}}$ are shown in (a) and (b), respectively.
    Analogous to the tetrahedron state, the \acp{CPP} follow directly
    from the octahedral symmetry group $O \subset \text{SO}(3)$.  Only
    one $\text{R} (\frac{\pi}{2})$-rotation along the axis spanned by
    the \acp{MP} $\ket{\phi_{4}}$ and $\ket{\phi_{6}}$ is necessary to
    unambiguously determine the eight \acp{CPP} at the intersections
    of the blue and green lines in (d).}
\end{figure}
In contrast to the tetrahedron state, where the \acp{MP} and \acp{CPP}
overlap, the \acp{CPP} of the octahedron state lie as far away from
the \acp{MP} as possible.  This is because the \acp{MP} of
$\ket{\Psi_{6}}$ form antipodal pairs, which leads to the spherical
amplitude function being zero at the location of each \ac{MP}.

\subsection{Seven qubits}\label{majorana_seven}

For seven points, the solutions to the two classical problems become
fundamentally different for the first time.  \toths problem is solved
by two triangles asymmetrically positioned about the equator and the
remaining point at the north pole \cite{Schutte51,Erber91}, or (1-3-3)
in the F{\"o}ppl notation \cite{Whyte52}.  Converting \toths solution
to Bloch vectors yields the \acp{MP}
\begin{equation}\label{7_maj_toth}
  \ket{\phi_{1}} = \ket{0} \ens , \quad
  \ket{\phi_{2, 3, 4}} = \co_{\theta} \ket{0} +
  \E^{\I \kappa} \si_{\theta} \ket{1} \ens , \quad
  \ket{\phi_{5, 6, 7}} = \co_{\vartheta} \ket{0} -
  \E^{\I \kappa} \si_{\vartheta} \ket{1} \ens ,
\end{equation}
with $\kappa = 0, \tfra{2 \pi}{3}, \tfra{4 \pi}{3}$, and their
inclinations are given by
\begin{equation}\label{7_maj_toth2}
  \co_{\theta} = \tfrac{1}{2} \csc ( \tfrac{2 \pi}{9} ) \ens , \quad
  \co_{\vartheta} = \sqrt{ \tfrac{1}{2} -
    \tfrac{\sqrt{3}}{6} \cot ( \tfrac{2 \pi}{9} ) } \ens .
\end{equation}
This non-positive state is of the form $\ket{\psi_{7}^{\text{\totx}}}
= \alpha \sym{0} - \beta \sym{3} - \gamma \sym{6}$, with the
approximate values for the coefficients being
\begin{equation}\label{7_coeff}
  \alpha \approx 0.295 \: 510 \ens , \quad
  \beta \approx 0.602 \: 458 \ens , \quad
  \gamma \approx 0.741 \: 430 \ens .
\end{equation}
The state is rotationally symmetric around the $Z$-axis and has three
\acp{CPP} $\ket{\sigma_{1, 2, 3}} = \co_{\phi} \ket{0} + \E^{\I
  \kappa} \si_{\phi} \ket{1}$, with $\kappa = 0, \tfra{2 \pi}{3},
\tfra{4 \pi}{3}$ and $\phi \approx 2.089 \: 603$, yielding $G^2
\approx 0.309 \: 326$ and $\Eg ( \ket{\psi_{7}^{\text{\totx}}} )
\approx 1.692 \: 798$.  \Fig{bloch_7} shows the Majorana
representation and the highly imbalanced spherical amplitude function
of $\ket{\psi_{7}^{\text{\totx}}}$.  The entanglement can be
considerably increased by varying the parameters \eqref{7_coeff},
which corresponds to changing the latitude of the two \ac{MP} circles
shown in \fig{bloch_7}(a). In this way a state with seven \acp{CPP}
and $\Eg \approx 2.146 \: 81$ can be obtained.

\begin{figure}
  \centering
  \begin{overpic}[scale=.44]{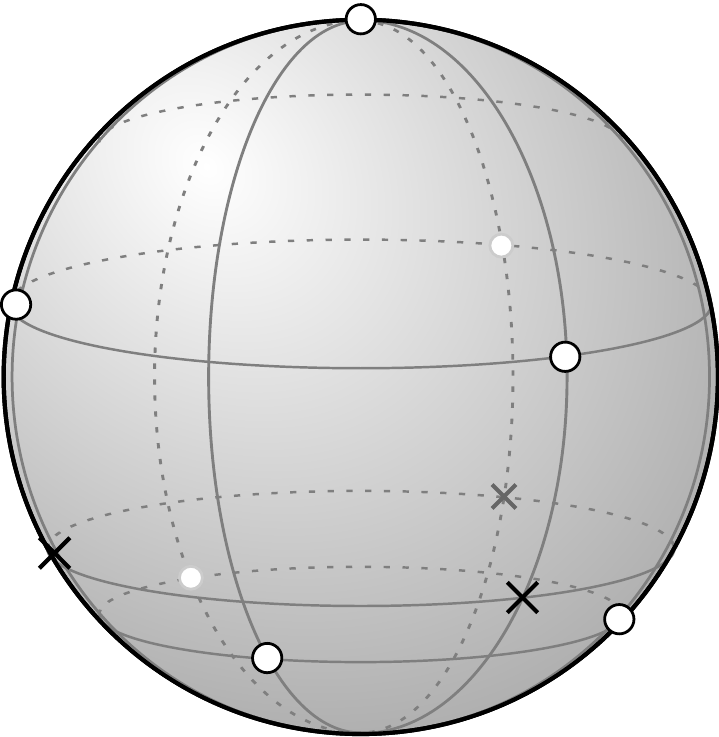}
    \put(-7,0){(a)}
  \end{overpic}
  \begin{overpic}[scale=.37]{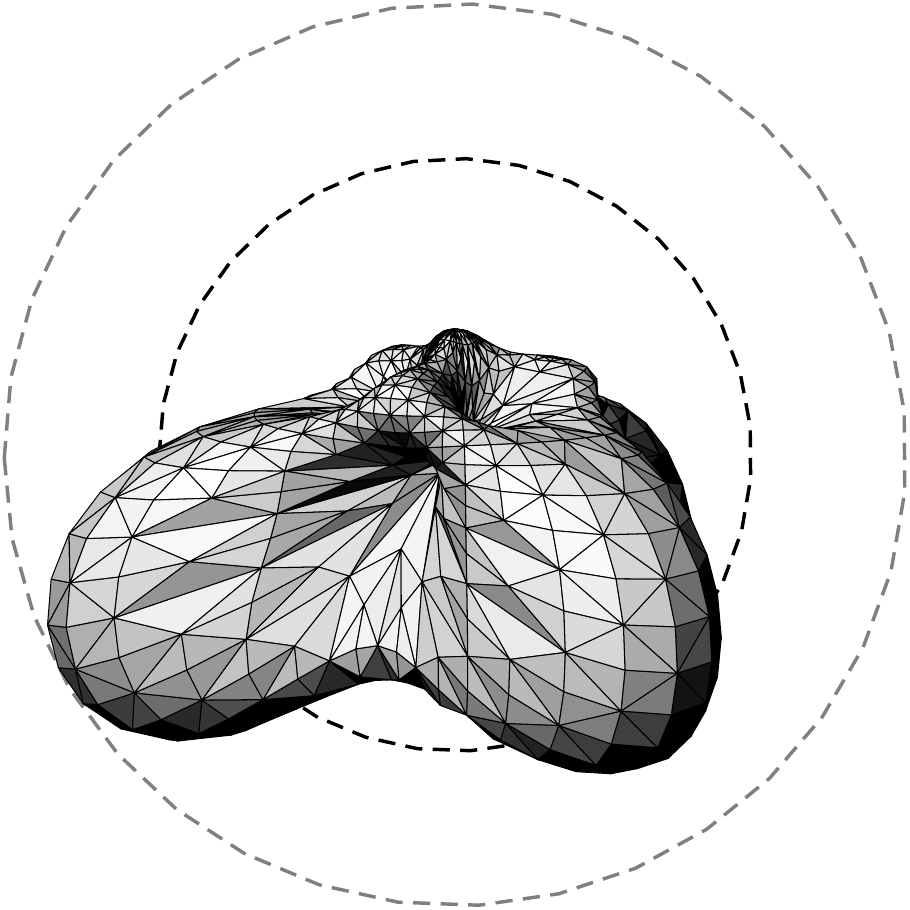}
    \put(-7,0){(b)}
  \end{overpic}
  \hfill
  \begin{overpic}[scale=.44]{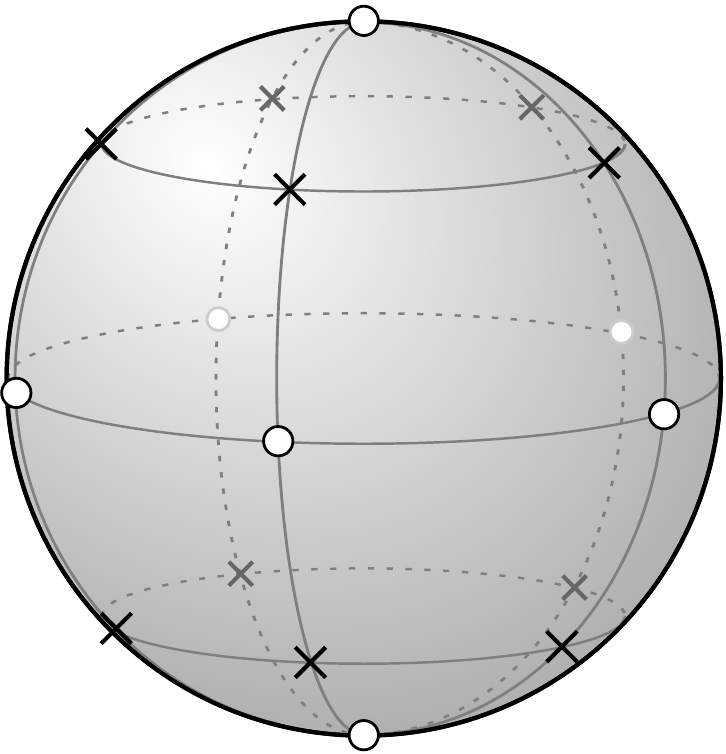}
    \put(-7,0){(c)}
  \end{overpic}
  \begin{overpic}[scale=.37]{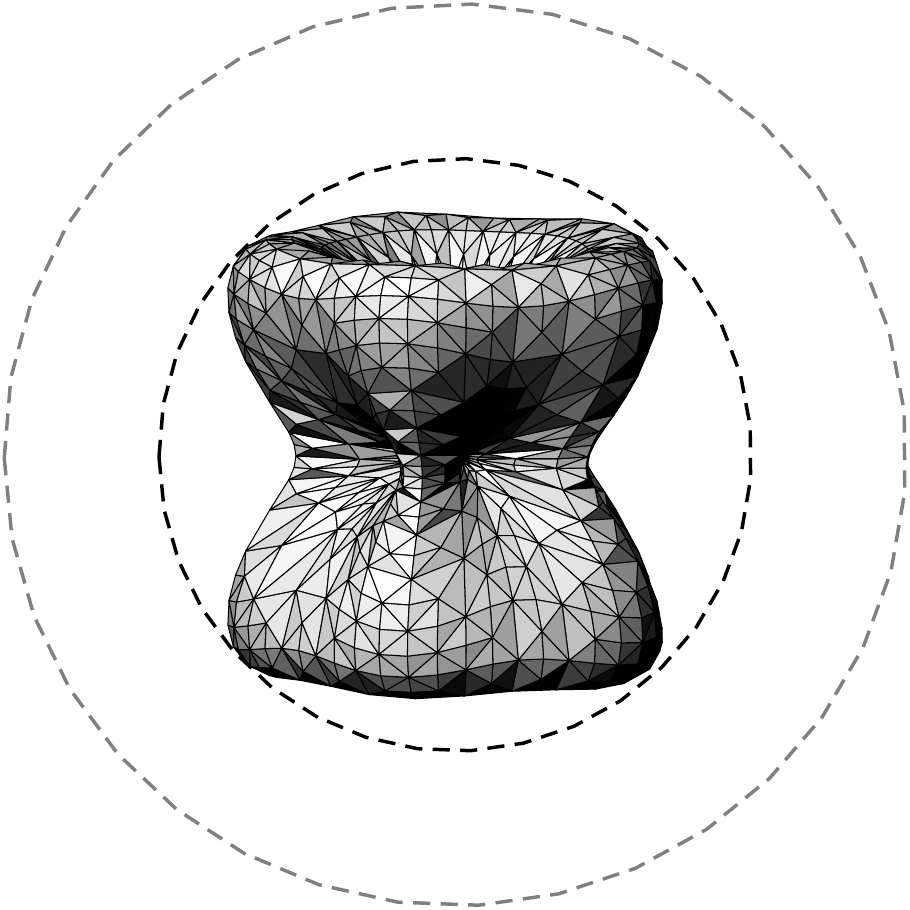}
    \put(-7,0){(d)}
  \end{overpic}
  \caption[Majorana representation of 7 qubit states]{\label{bloch_7}
    For seven qubits the Majorana representation and the spherical
    amplitude function $g^2 ( \theta , \varphi )$ of the solution of
    \protect\toths problem $\ket{\psi_{7}^{\text{\totx}}}$ is shown in
    (a) and (b), and for the \protect\quo{pentagonal dipyramid state}
    $\ket{\Psi_{7}}$ in (c) and (d), respectively. The outer and inner
    circle correspond to the maximum values of $g^2$ for
    $\ket{\psi_{7}^{\text{\totx}}}$ and $\ket{\Psi_{7}}$,
    respectively.}
\end{figure}

Thomson's problem is solved by the vertices of a pentagonal dipyramid
\cite{Ashby86,Marx70,Erber91}, where five points lie on an equatorial
pentagon and the other two on the poles.  The corresponding
\quo{pentagonal dipyramid state}, shown in \fig{bloch_7}, is
numerically found to be the solution to the Majorana problem, too.
The state is $\ket{\Psi_{7}} = \tfra{1}{\sqrt{2}} ( \sym{1} + \sym{6}
)$, and its \acp{MP} are
\begin{equation}\label{7_maj}
  \ket{\phi_{1}} = \ket{0} \ens , \quad
  \ket{\phi_{2, 3, 4, 5, 6}} = \tfrac{1}{\sqrt{2}} \big( \ket{0} +
  \E^{\I \kappa} \ket{1} \big) \ens , \quad
  \ket{\phi_{7}} = \ket{1} \ens ,
\end{equation}
with $\kappa = 0, \tfra{2 \pi}{5}, \tfra{4 \pi}{5}, \tfra{6 \pi}{5},
\tfra{8 \pi}{5}$. Despite this simple analytical form, the
determination of the \acp{CPP} is not trivial.  With the
parameterisation $x := \cos^2 \theta = (2 \co_{\theta}^2 - 1)^2$ the
positions of the ten \acp{CPP}
\begin{equation}\label{7_max_cpps}
  \ket{\sigma_{1,2,3,4,5}} = \co_{\theta} \ket{0} +
  \E^{\I \kappa}\si_{\theta} \ket{1} \ens , \quad
  \ket{\sigma_{6,7,8,9,10}} = \si_{\theta} \ket{0} +
  \E^{\I \kappa} \co_{\theta} \ket{1} \ens ,
\end{equation}
$\kappa = 0, \tfra{2 \pi}{5}, \tfra{4 \pi}{5}, \tfra{6 \pi}{5},
\tfra{8 \pi}{5}$, are given by the real root of the cubic equation
\begin{equation}\label{7_polynomial}
  49 x^3 + 165 x^2 - 205 x + 55 = 0
\end{equation}
in the interval $[0, \tfra{1}{2}]$. Approximate values are
$\co_{\theta} \approx 0.920 \: 574$ and $\si_{\theta} \approx 0.390 \:
567$, yielding $G^2 \approx 0.203 \: 247$ and $\Eg ( \ket{\Psi_{7}} )
\approx 2.298 \: 691 \: 396$.  Since $\ket{\Psi_{7}}$ is positive and
totally invariant under the dihedral symmetry group $D_{5}$, it
satisfies the requirements of \theoref{pos_inv}.

\subsection{Eight qubits}\label{majorana_eight}

The regular cube is a Platonic solid with eight vertices, and
therefore a natural candidate to study.  Its \ac{MP} locations can be
directly obtained from the \acp{CPP} \eqref{6_cpp} of the octahedron
state, which were found to form a cube, as seen in \fig{bloch_6}(a).
In the configuration shown in \fig{bloch_8_cube}(a) the \acp{MP} are

\begin{equation}\label{8_cube_mp}
  \ket{\phi_{1, 2, 3, 4}} = \co_{\theta} \ket{0} +
  \E^{\I \kappa} \si_{\theta} \ket{1} \ens , \quad
  \ket{\phi_{5, 6, 7, 8}} = \si_{\theta} \ket{0} +
  \E^{\I \kappa} \co_{\theta} \ket{1} \, ,
\end{equation}
with $\kappa = \tfra{\pi}{4}, \tfra{3 \pi}{4}, \tfra{5 \pi}{4},
\tfra{7 \pi}{4}$, and $\co_{\theta}^{2} = \tfrac{\sqrt{3}+1}{2
  \sqrt{3}}$, $\si_{\theta}^{2} = \tfrac{\sqrt{3}-1}{2 \sqrt{3}}$.
This gives rise to the cube state
\begin{equation}\label{8_cube}
  \ket{\psi_{8}^{\text{c}}} =
  \frac{1}{2 \sqrt{6}} \left( \sqrt{5} \sym{0} +
    \sqrt{14} \sym{4} + \sqrt{5} \sym{8} \right) \ens .
\end{equation}
This state is positive and totally invariant under the octahedral
symmetry group $O \in \text{SO}(3)$, thus meeting the prerequisites of
\theoref{pos_inv}.  Its \acp{CPP} can be obtained in the same manner
as for the tetrahedron and octahedron state by applying
\lemref{pos_symm_cpp}.  From \fig{bloch_8_cube}(c) it can be seen that
two rotations (e.g. $\rotx (\frac{\pi}{2})$ and $\roty
(\frac{\pi}{2})$) give rise to three areas, coloured blue, green and
red. The intersection of these three areas are six points that form
the vertices of a regular octahedron.  Thus the \acp{CPP} of the cube
state are identical to the \acp{MP} \eqref{6_mp} of the octahedron
state, up to an $\rotz ( \frac{\pi}{4} )$-rotation.  The entanglement
of the cube state follows as $\Eg ( \ket{\psi_{8}^{\text{c}}} ) =
\log_2 \left( \tfra{24}{5} \right) \approx 2.263 \: 034$.

\begin{figure}
  \centering
  \begin{overpic}[scale=0.5]{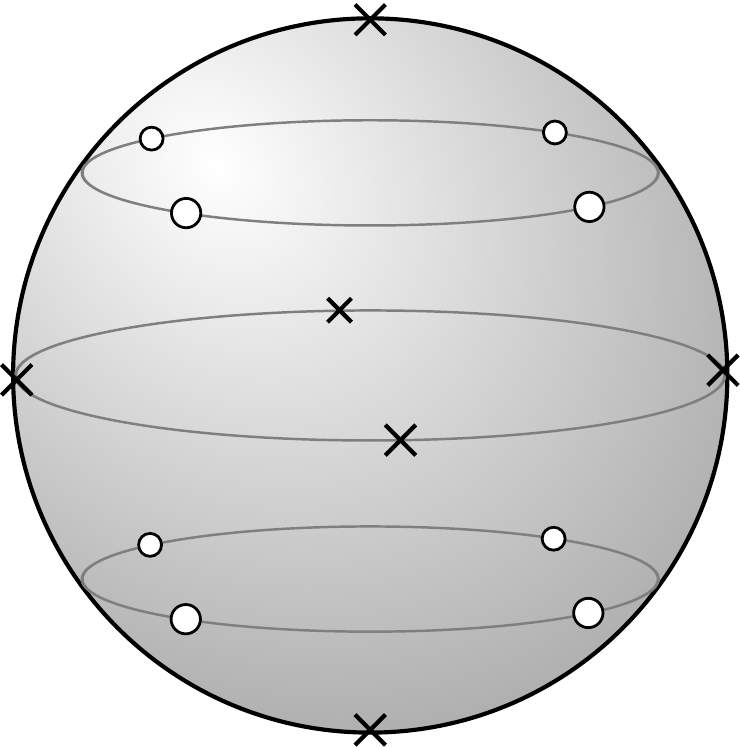}
    \put(-7.5,0){(a)}
  \end{overpic}
  \hspace{5mm}
  \begin{overpic}[scale=0.62]{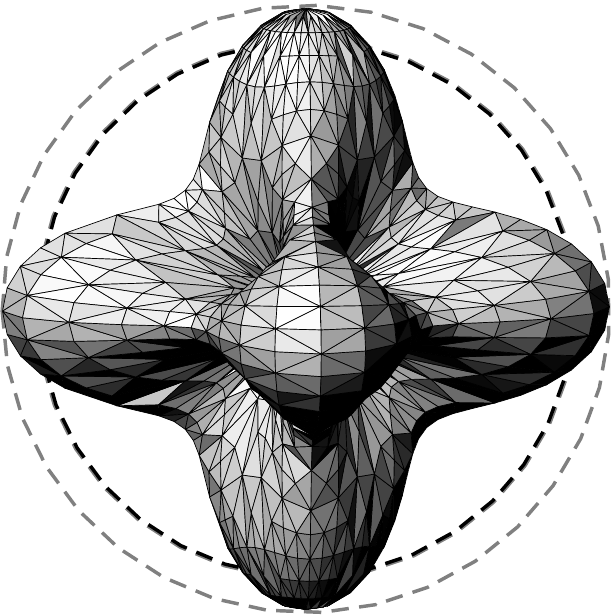}
    \put(-7.5,0){(b)}
  \end{overpic}
  \hspace{5mm}
  \begin{overpic}[scale=.5]{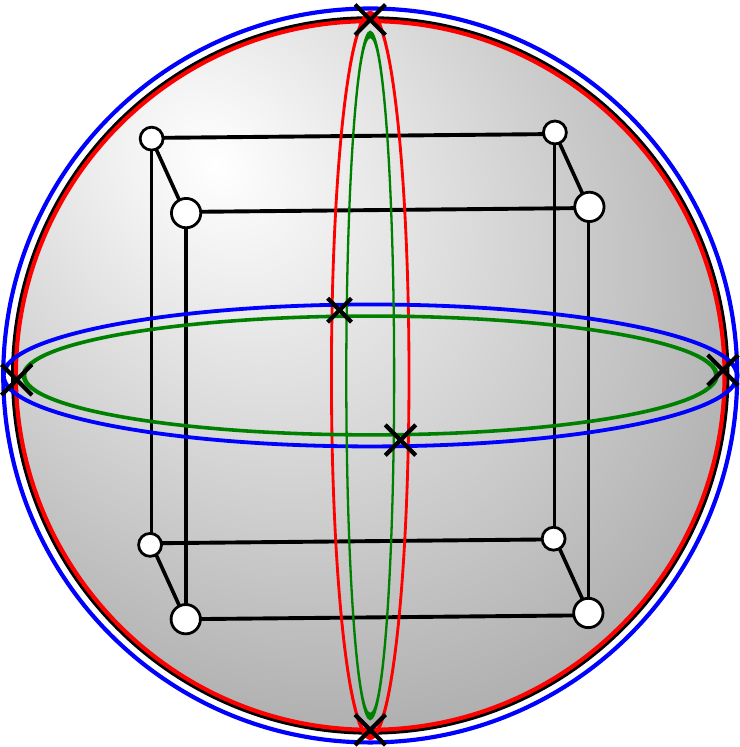}
    \put(-7.5,0){(c)}
  \end{overpic}
  \caption[Majorana representation of 8 qubit
  cube]{\label{bloch_8_cube} For eight qubits the Majorana
    representation and the spherical amplitude function $g^2 ( \theta
    , \varphi )$ of the cube state $\ket{\psi_{8}^{\text{c}}}$ are
    shown in (a) and (b), respectively. The \acp{CPP} can be directly
    determined from \protect\lemref{pos_symm_cpp} and the octahedral
    symmetry group $O \in \text{SO}(3)$ by performing an $\rotx
    (\frac{\pi}{2})$ and $\roty (\frac{\pi}{2})$ rotation, as shown in
    (c). The outer and inner circle in (b) correspond to the cube
    state $\ket{\psi_{8}^{\text{c}}}$ and the maximally entangled
    symmetric eight qubit state $\ket{\Psi_{8}}$, respectively.}
\end{figure}

A numerical search yields states that are considerably higher
entangled than the cube state.  The \quo{asymmetric pentagonal
  dipyramid state} shown in \fig{bloch_8}(a) is numerically found to
have the highest amount of entanglement. The exact analytic form of
this positive state is not known, but it can be numerically
approximated to high precision.  The form of the state is
$\ket{\Psi_{8}} = \alpha \sym{1} + \beta \sym{6}$, with approximate
values $\alpha \approx 0.671 \: 588 \: 032$ and $\beta \approx 0.740
\: 924 \: 770$, and the \acp{MP} are
\begin{equation}\label{8_maj_max}
  \ket{\phi_{1,2}} = \ket{0} \ens , \quad
  \ket{\phi_{3,4,5,6,7}} = \co_{\theta} \ket{0} +
  \E^{\I \kappa} \si_{\theta} \ket{1} \ens , \quad
  \ket{\phi_{8}} = \ket{1} \ens ,
\end{equation}
with $\kappa = 0, \tfra{2 \pi}{5}, \tfra{4 \pi}{5}, \tfra{6 \pi}{5},
\tfra{8 \pi}{5}$ and $\theta \approx 1.715 \: 218 \: 732$.  In
particular, there is a two-fold \ac{MP} degeneracy at the north pole,
similar to the W state of three qubits.  As seen in \fig{bloch_8}(a),
there are two rings with five \acp{CPP} each,
\begin{equation}\label{8_max_cpps}
  \ket{\sigma_{1,2,3,4,5}} = \co_{\vartheta} \ket{0} +
  \E^{\I \kappa}\si_{\vartheta} \ket{1} \ens , \quad
  \ket{\sigma_{6,7,8,9,10}} = \co_{\phi} \ket{0} +
  \E^{\I \kappa} \si_{\phi} \ket{1} \ens ,
\end{equation}
with $\kappa = 0, \tfra{2 \pi}{5}, \tfra{4 \pi}{5}, \tfra{6 \pi}{5},
\tfra{8 \pi}{5}$, $\co_{\vartheta} \approx 0.928 \: 479$ and
$\co_{\phi} \approx 0.525 \: 434$.  From this it follows $G^2 \approx
0.183 \: 619$ and $\Eg ( \ket{\Psi_{8}} ) \approx 2.445 \: 210$.

\begin{figure}
  \centering
  \begin{overpic}[scale=.44]{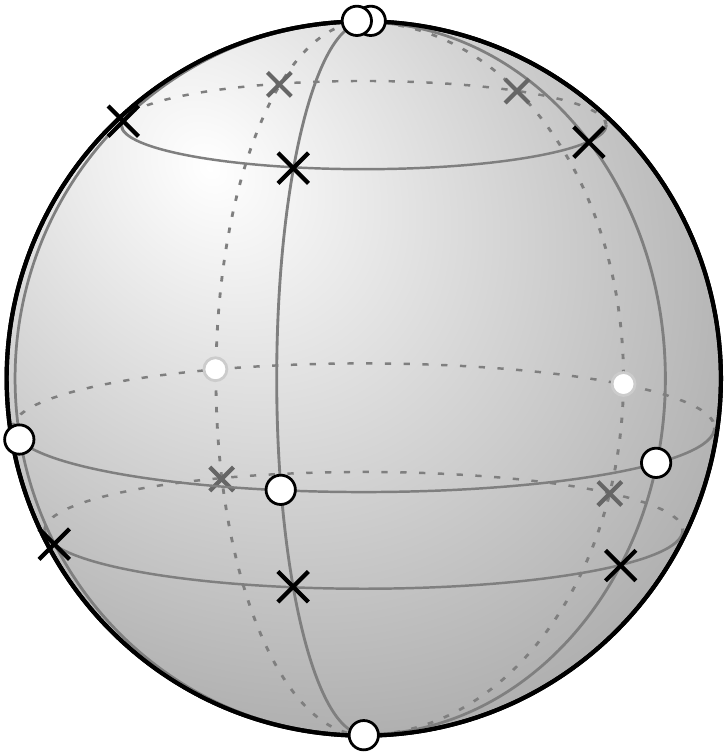}
    \put(-7,0){(a)}
  \end{overpic}
  \begin{overpic}[scale=0.56]{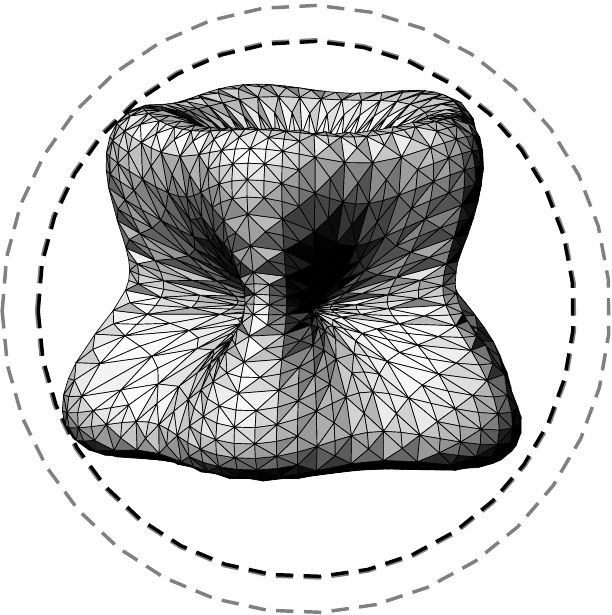}
    \put(-7,0){(b)}
  \end{overpic}
  \hfill
  \begin{overpic}[scale=.44]{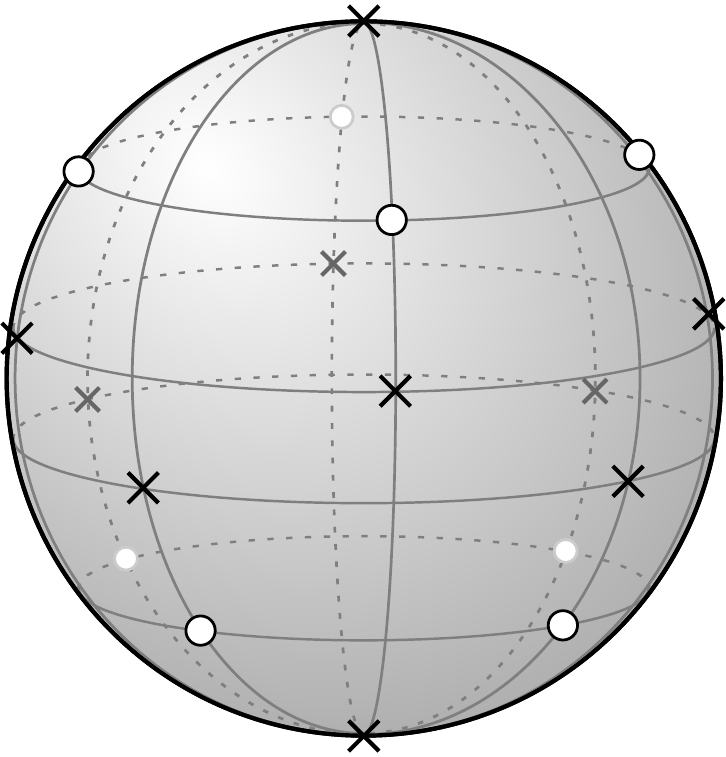}
    \put(-7,0){(c)}
  \end{overpic}
  \begin{overpic}[scale=0.56]{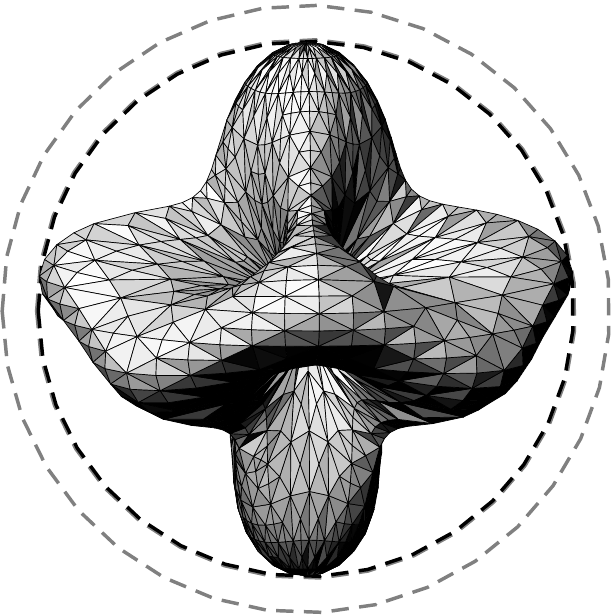}
    \put(-7,0){(d)}
  \end{overpic}
  \caption[Majorana representation of 8 qubit states]{\label{bloch_8}
    For eight qubits the \protect\quo{asymmetric pentagonal dipyramid
      state} $\ket{\Psi_{8}}$ shown in (a) and (b) is conjectured to
    be the maximally entangled state. A similarly highly entangled
    state is the optimal antiprism state $\ket{\Psi_{8}^{\text{a}}}$,
    shown in (c) and (d).  The outer and inner circles correspond to
    the cube state $\ket{\psi_{8}^{\text{c}}}$ and to
    $\ket{\Psi_{8}}$, respectively.}
\end{figure}

As mentioned in \sect{extremal_point}, the classical solutions are not
solved by the regular cube.  \toths problem for eight points is solved
by the cubic antiprism introduced and discussed in \fig{platonic}.
This antiprism is regular in the sense that all its sides have the
same length.  The solution to Thomson's problem is a slightly
different antiprism that is not regular and which can be obtained from
\toths antiprism by a slight expansion along the direction
perpendicular to the rotated face \cite{Marx70,Ashby86,Erber91}.  Cast
as symmetric states, antiprisms have the form
\begin{equation}\label{8_antiprism}
  \ket{\psi_{8}^{\text{a}}} = \frac{\sym{0} + A \sym{4} -
    \sym{8}}{\sqrt{2 + A^2}} \ens ,
\end{equation}
where the real parameter $A$ depends on the latitude of the two
\ac{MP} rings.  The \acp{MP} can be parameterised as
\begin{equation}\label{8_anti_maj}
  \begin{split}
    \ket{\phi_{1, 2, 3, 4}}& = \sqrt{a} \ket{0} + \E^{\I ( \kappa +
      \frac{\pi}{4})} \sqrt{1 - a} \ket{1} \ens , \\
    \ket{\phi_{5, 6, 7, 8}}& = \sqrt{1 - a} \ket{0} +
    \E^{\I \kappa } \sqrt{a} \ket{1} \ens ,
  \end{split}
\end{equation}
with $a \in [0,1]$, $\kappa = 0, \tfra{\pi}{2}, \pi, \tfra{3 \pi}{2}$,
and the maxima of the spherical amplitude function as
\begin{equation}\label{8_anti_cpp}
  \begin{split}
    \ket{\sigma_{1}}& = \ket{0} \ens , \quad
    \ket{\sigma_{2}} = \ket{1} \ens , \\
    \ket{\sigma_{3, 4, 5, 6}}& =
    x \ket{0} + \E^{\I \kappa} \sqrt{1 - x^2} \ket{1} \ens , \\
    \ket{\sigma_{7, 8, 9, 10}}& =
    \sqrt{1 - x^2} \ket{0} + \E^{\I (\kappa + \frac{\pi}{4})}
    x \ket{1} \ens ,
  \end{split}
\end{equation}
with $x \in [0,1]$, $\kappa = 0, \tfra{\pi}{2}, \pi, \tfra{3 \pi}{2}$.
The maximally entangled antiprism state $\ket{\Psi_{8}^{\text{a}}}$
can then be found by a calculation similar to the one performed for
the maximally entangled five qubit state.  From numerics it is clear
that $\ket{\Psi_{8}^{\text{a}}}$ has ten \acp{CPP}, one on each pole
and the others lying on two horizontal planes, see \fig{bloch_8}(c)
and (d).  It suffices to determine the latitude of the nontrivial
positive \ac{CPP}, so we use $g ( \theta ) \equiv g ( \theta , 0 )$
and the parameterisation $x: = \co_{\theta}$.  To turn all local
maxima into \acp{CPP} the value of $g(x)$ at $x = 1$ needs to equal
the value at the non-trivial maximum $x_{\text{m}} \in (0,1)$.  From
$g(1) = g(x_{\text{m}})$ it follows that $A = \frac{1 - x_{\text{m}}^8
  + y_{\text{m}}^8} {\sqrt{70} \, x_{\text{m}}^4 y_{\text{m}}^4 }$,
and from $ \frac{\D g}{\D x} (x_{\text{m}}) = 0$ it follows that
\begin{equation}\label{solve_8}
  x_{\text{m}}^6 - x_{\text{m}}^4 + 2 x_{\text{m}}^2 - 1 = 0 \ens .
\end{equation}
This amounts to solving a cubic equation, yielding the single real
root
\begin{equation}\label{solution_8}
  x_{\text{m}} = \sqrt{\frac{1}{3} \left( 1 + z -
      \frac{5}{z} \right) } \ens , \text{ with} \quad
  z = \sqrt[3]{\frac{11 + 3 \sqrt{69}}{2}} \ens .
\end{equation}
This $x_{\text{m}}$ establishes the locations of all nontrivial
\acp{CPP} and by inserting it into $A$, it yields the explicit form of
\eqref{8_antiprism}, as well as the entanglement $\Eg (
\ket{\Psi_{8}^{\text{a}}} ) = \log_2 (2 + A^2)$.  The latitude of the
\acp{MP} is found by solving a quartic equation that arises when
determining the \acp{MP} from the given form of the state: The value
of $a$ is given by the real root of $\sqrt{70} a^2 (1- a )^2 A - a^4 +
(1- a )^4 = 0$.  Approximate values of the quantities are:
\begin{equation}\label{8_anti_form_approx}
  x \approx 0.754 \: 878 \ens , \quad  A \approx 1.847 \: 592 \ens ,
  \quad a \approx 0.797 \: 565 \ens .
\end{equation}
The latitude of the upper \ac{MP} circle follows as $\theta \approx
0.933 \: 368 \: 783$, and the amount of entanglement is $\Eg (
\ket{\Psi_{8}^{\text{a}}} ) \approx 2.436 \: 587 \: 205$.  The optimal
antiprism state $\ket{\Psi_{8}^{\text{a}}}$ is thus only slightly less
entangled than the numerically determined maximally entangled
symmetric state of eight qubits.  Intriguingly, the maximally
entangled state $\ket{\Psi_{8}}$ is a positive state, whereas the
antiprism states cannot be cast with positive coefficients.

\begin{table}
  \begin{threeparttable}
    \caption[Comparison of selected eight qubit symmetric states]
    {\label{table4.2} Comparison of all the eight qubit symmetric
      states studied in \protect\sect{majorana_eight}.  For each state
      the latitude $\theta$ of the topmost circle of \acp{MP} as well
      as the geometric entanglement $\Eg$ is listed.  The entanglement
      of antiprism states decreases with increasing deviance of the
      \ac{MP} angle $\theta$ from that of the optimal antiprism state
      $\ket{\Psi_{8}^{\text{a}}}$.}
    \newcolumntype{R}{>{\centering\arraybackslash}X}
    \begin{tabularx}{\textwidth}{R|RR}
      \toprule
      State & \ac{MP} angle $\theta$ [rad] &
      Entanglement $\Eg$ \\
      \midrule
      Majorana solution  $\ket{\Psi_{8}}$ &
      $\approx 1.715 \: 218 \: 732$ & $\approx 2.445 \: 210 \: 159$ \\
      regular cube $\ket{\psi_{8}^{\text{c}}}$ &
      $\arccos \frac{1}{\sqrt{3}} \approx 0.955$ &
      $\log_2 \big( \tfra{24}{5} \big) \approx 2.263$ \\
      optimal antiprism $\ket{\Psi_{8}^{\text{a}}}$ &
      $\approx 0.933 \: 368 \: 783$\tnote{$\dagger$} &
      $\approx 2.436 \: 587 \: 205$\tnote{$\dagger$} \\[0.15em]
      Thomson antiprism $\ket{\psi_{8}^{\text{Th}}}$ &
      $\approx 0.975 \: 883 \: 252$ & $\approx 2.084 \: 181 \: 528$ \\
      \toth antiprism $\ket{\psi_{8}^{\text{\totx}}}$ &
      $\arctan \sqrt{2 \sqrt{2}} \approx 1.034$ &
      $\approx 1.711 \: 525 \: 327$\tnote{$\dagger$} \\
      \bottomrule
    \end{tabularx}
    \begin{tablenotes}
    \item [$\dagger$] {\footnotesize Closed-form analytic expressions
        are known, but not displayed due to their complicated form.}
    \end{tablenotes}
  \end{threeparttable}
\end{table}

The antiprism states that solve \toths and Thomson's problem each have
only the two \acp{CPP} $\ket{\sigma_{1}} = \ket{0}$ and
$\ket{\sigma_{2}} = \ket{1}$.  This imbalance of their spherical
amplitude functions is due to the two horizontal \ac{MP} circles being
closer to the equator than in the configuration seen in
\fig{bloch_8}(c).  As listed in \tabref{table4.2}, this leads to a
reduction of the geometric entanglement.  No analytic form is known
for the antiprism state $\ket{\psi_{8}^{\text{Th}}}$ which solves
Thomson's problem, but the latitude of its \acp{MP} can be numerically
determined by minimising a nonlinear function \cite{Marx70}, yielding
$\theta \approx 0.975 \: 883 \: 252$ and $\Eg (
\ket{\psi_{8}^{\text{Th}}} ) \approx 2.084 \: 181 \: 498$.  On the
other hand, the solution $\ket{\psi_{8}^{\text{\totx}}}$ of \toths
problem can be determined analytically from the known spherical
distance $s_{\text{min}} = \arccos \big( \frac{\sqrt{8}-1}{7} \big)$
between neighbouring pairs of points.  The latitude of the \ac{MP}
circle then follows as $\theta = \arctan \sqrt{2 \sqrt{2}}$, and the
analytical form of the state \eqref{8_antiprism} is given by $A =
\frac{1 - \tau^2}{\sqrt{70} \tau }$, with $\tau := \tan^4 (
\frac{\theta}{2} ) = \frac{1}{8} \big( \sqrt{1 + 2 \sqrt{2}} -1
\big)^{4}$.  The entanglement follows as $\Eg (
\ket{\psi_{8}^{\text{\totx}}} ) = \log_2 (2 + A^2) \approx 1.712$.

\subsection{Nine qubits}\label{majorana_nine}

For nine points, the solutions to \toths and Thomson's problem are
slightly different manifestations of a \quo{triaugmented triangular
  prism}. As shown in \fig{bloch_9}(a), three equilateral triangles
are positioned parallel but asymmetric to each other, with a
reflective symmetry along the $X$-$Y$-plane.  The \acp{MP} of this
configuration are
\begin{align}\label{9_maj_triaug}
  \ket{\phi_{1,2,3}}& = \co_{\theta} \ket{0} -
  \E^{\I \kappa} \si_{\theta} \ket{1} \ens , \nonumber \\
  \ket{\phi_{4,5,6}}& = \tfrac{1}{\sqrt{2}} \big( \ket{0} +
  \E^{\I \kappa} \ket{1} \big) \ens , \\
  \ket{\phi_{7,8,9}}& = \si_{\theta} \ket{0} -
  \E^{\I \kappa} \co_{\theta} \ket{1} \ens , \nonumber
\end{align}
with $\kappa = 0, \tfra{2 \pi}{3}, \tfra{4 \pi}{3}$.  This gives rise
to a real state
\begin{equation}\label{9_triaug}
  \ket{\psi_{9}} = \frac{\sym{0}  - A \big( \sym{3} + \sym{6} \big)
    + \sym{9}}{\sqrt{2 + 2 A^2}} \ens ,
\end{equation}
where the relationship between $A$ and the \acp{MP} is $A \tau
\sqrt{84} = - \tau^2 + \tau - 1$ with $\tau := \tan^3
(\frac{\theta}{2})$.  The single freedom of this configuration is the
inclination $\theta$ (or $\pi - \theta$) of the \acp{MP} that lie
outside the equator.

For all values of $A$ the spherical amplitude function of
$\ket{\psi_{9}}$ has local maxima at the three equatorial \acp{MP} and
at the poles.  From this it can be inferred that the most entangled
state of the form \eqref{9_triaug} is the one where these maxima yield
the same value.  The optimal state thus has the five \acp{CPP} shown
in \fig{bloch_9}(a), and a simple calculation yields $A = \frac{1 + 8
  \sqrt{2}}{2 \sqrt{21}}$ and $\Eg ( \ket{\psi_{9}} ) = \log_2
\frac{213 + 16 \sqrt{2}}{42} \approx 2.488$.  In contrast to this, the
configurations that solve the classical problems are not optimal. In
the solution to Thomson's problem the latitudes of the outer \acp{MP}
are closer to the equator than in \fig{bloch_9}(a), and even more so
in the solution to \toths problem. This induces an imbalance in the
spherical amplitude function, resulting in the two poles being the
only \acp{CPP}.  The geometric entanglement is $\Eg (
\ket{\psi_{9}^{\text{Th}}} ) \approx 2.434 \: 192 \: 780$ and $\Eg (
\ket{\psi_{9}^{\text{\totx}}} ) \approx 2.150 \: 714 \: 397$,
respectively.

\begin{figure}
  \centering
  \begin{overpic}[scale=.445]{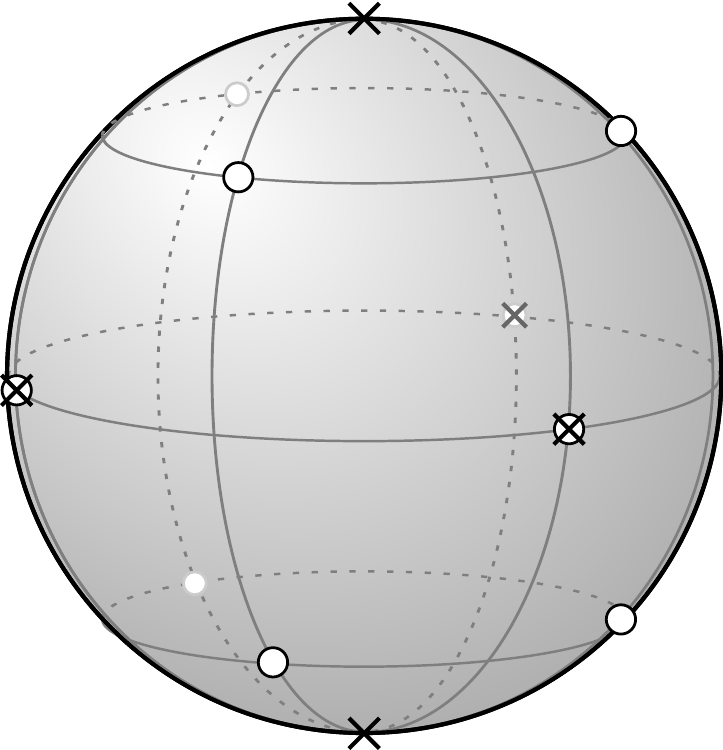}
    \put(-7,-2){(a)}
  \end{overpic}
  \begin{overpic}[scale=.645]{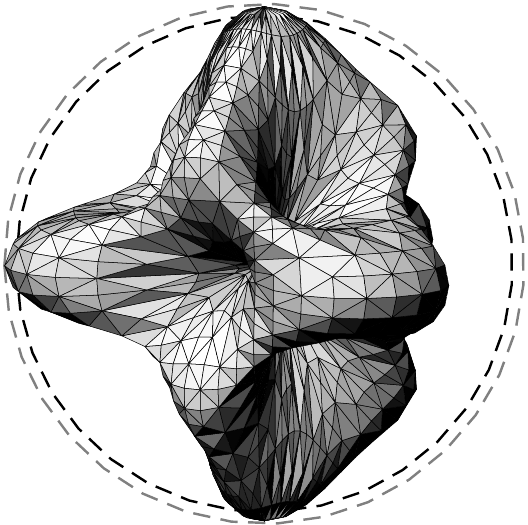}
    \put(-7,-2){(b)}
  \end{overpic}
  \hfill
  \begin{overpic}[scale=.445]{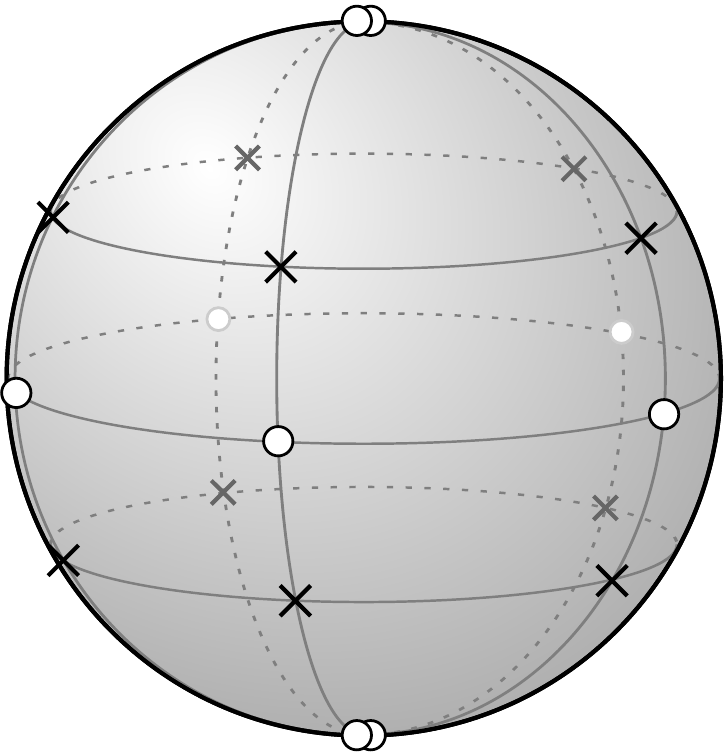}
    \put(-7,-2){(c)}
  \end{overpic}
  \begin{overpic}[scale=.645]{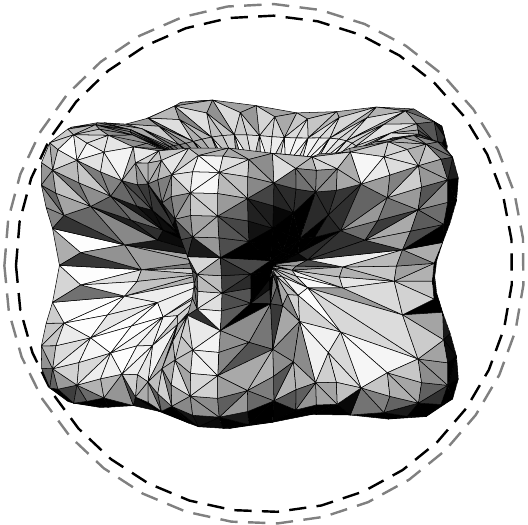}
    \put(-7,-2){(d)}
  \end{overpic}
  \caption[Majorana representation of 9 qubit states]{\label{bloch_9}
    For nine qubits the optimal \protect\quo{triaugmented triangular
      prism state} is shown in (a) and (b), and the
    \protect\quo{pentagonal dipyramid state} $\ket{\Psi_{9}}$ which is
    conjectured to be the maximally entangled symmetric nine qubit
    state is shown in (c) and (d).  The latter state has a two-fold
    \ac{MP} degeneracy on each pole.}
\end{figure}

The maximally entangled symmetric nine qubit state, however, does not
assume the form of a triaugmented triangular prism.  A numerical
search determines the state $\ket{\Psi_{9}} = \tfra{1}{\sqrt{2}} (
\sym{2} + \sym{7} )$, shown in \fig{bloch_9}(c), with the \acp{MP}
\begin{equation}\label{9_maj_maxent}
  \ket{\phi_{1,2}} = \ket{0} \ens , \quad
  \ket{\phi_{3,4,5,6,7}} = \tfrac{1}{\sqrt{2}} \big( \ket{0} +
  \E^{\I \kappa} \ket{1} \big) \ens , \quad
  \ket{\phi_{8,9}} = \ket{1} \ens ,
\end{equation}
and $\kappa = 0, \tfra{2 \pi}{5}, \tfra{4 \pi}{5}, \tfra{6 \pi}{5},
\tfra{8 \pi}{5}$.  This is a positive state with \ac{MP} degeneracies,
and the state is totally invariant under the dihedral symmetry group
$D_{5}$.  The \acp{CPP} lie on two circles
\begin{equation}\label{9_cpp_maxent}
  \ket{\sigma_{1,2,3,4,5}} = \co_{\theta} \ket{0} +
  \E^{\I \kappa} \si_{\theta} \ket{1} \ens , \quad
  \ket{\sigma_{6,7,8,9,10}} = \si_{\theta} \ket{0} +
  \E^{\I \kappa} \co_{\theta} \ket{1} \ens ,
\end{equation}
with $\kappa = 0, \tfra{2 \pi}{5}, \tfra{4 \pi}{5}, \tfra{6 \pi}{5},
\tfra{8 \pi}{5}$.  Unlike the \acp{MP}, the \acp{CPP} do not have a
simple analytical form.  They can however be determined in the same
way as done for the seven qubit case.  With the substitution $x :=
\cos^2 \theta$, the inclination follows from the real root of $81 x^3
+ 385 x^2 - 245 x + 35 = 0$ in the interval $[0,0.3]$.  Approximate
values are $\co_{\theta} \approx 0.860 \: 122$ and $\si_{\theta}
\approx 0.510 \: 087$, from which one obtains $\Eg ( \ket{\Psi_{9}} )
\approx 2.553 \: 960 \: 277$, which is a significantly higher amount
of entanglement than for the most entangled triaugmented triangular
prism state.

\subsection{Ten qubits}\label{majorana_ten}

The solution to \toths problem is an arrangement of the form (2-2-4-2)
in the F{\"o}ppl notation \cite{Melnyk77,Whyte52}, with only two
\acp{CPP}, and the numerically determined entanglement $\Eg (
\ket{\psi_{10}^{\text{\totx}}} ) \approx 1.958 \: 874 \: 344$ is
relatively low.

Thomson's problem is solved by a \quo{gyroelongated square bipyramid},
a polyhedron that arises from a cubic antiprism by placing square
pyramids on each of the two square surfaces\footnote{In a narrower
  sense, the gyroelongated square bipyramid is the unique polyhedron
  that arises from the regular antiprism (sides of equal length) by
  the requirement that all faces are equilateral triangles, which
  makes it one of the eight convex \emph{deltahedra}. This deltahedron
  does however not have a circumsphere that touches all its vertices,
  and therefore it does not directly translate to a spherical point
  distribution.}.  The \acp{MP}, shown in \fig{bloch_10}(a), have the
form
\begin{equation}\label{10_maj_gyro}
  \begin{split}
    \ket{\phi_{1}}& = \ket{0} \ens , \quad
    \ket{\phi_{2,3,4,5}} = \co_{\theta} \ket{0} +
    k \si_{\theta} \ket{1}  \ens , \\
    \ket{\phi_{10}}& = \ket{1} \ens , \quad
    \ket{\phi_{6,7,8,9}} = \si_{\theta} \ket{0} +
    k \E^{\I \frac{\pi}{4}} \co_{\theta} \ket{1}  \ens ,
  \end{split}
\end{equation}
with $k = 0, \I, -1, - \I$.  This gives rise to a real state
\begin{equation}\label{10_gyro}
  \ket{\psi_{10}} = \frac{\sym{1} + 
    A \sym{5} - \sym{9}}{\sqrt{2 + A^2}} \ens .
\end{equation}
The relationship between $A$ and the \acp{MP} is described by $A \tau
\sqrt{252} = 1 - \tau^2$ with $\tau := \tan^4 (\frac{\theta}{2})$.
The state $\ket{\psi_{10}}$ has eight \acp{CPP}
\begin{equation}\label{10_cpp_maxent}
  \ket{\sigma_{1,2,3,4}} = \co_{\vartheta} \ket{0} +
  k \si_{\vartheta} \ket{1} \ens , \quad
  \ket{\sigma_{5,6,7,8}} = \si_{\vartheta} \ket{0} +
  k \E^{\I \frac{\pi}{4}} \co_{\vartheta} \ket{1} \ens ,
\end{equation}
with $k = 0, \I, -1, - \I$, and where the latitude $\vartheta$ depends
on the precise form of \eq{10_gyro}.  An analytical treatment of
\eq{10_gyro} and \eqref{10_cpp_maxent} is very difficult, so we limit
ourselves to numerics.

The entanglement obtained from the point distribution of Thomson's
solution is $\Eg ( \ket{\psi_{10}^{\text{Th}}} ) \approx 2.731 \: 632
\: 770$, and a numerical analysis reveals that this state is very
close to the maximally entangled state of the form \eqref{10_gyro}.  A
small modification of the latitude of the \acp{MP} yields the extremal
entanglement $\Eg ( \ket{\Psi_{10}} ) \approx 2.737 \: 432 \: 003$ at
$\theta \approx 1.142 \: 46$, with the latitude of the \acp{CPP} given
by $\vartheta \approx 1.048$.  The state is shown in
\fig{bloch_10}(a), and it is proposed to be the maximally entangled 10
qubit symmetric state.

\begin{figure}
  \centering
  \begin{overpic}[scale=.445]{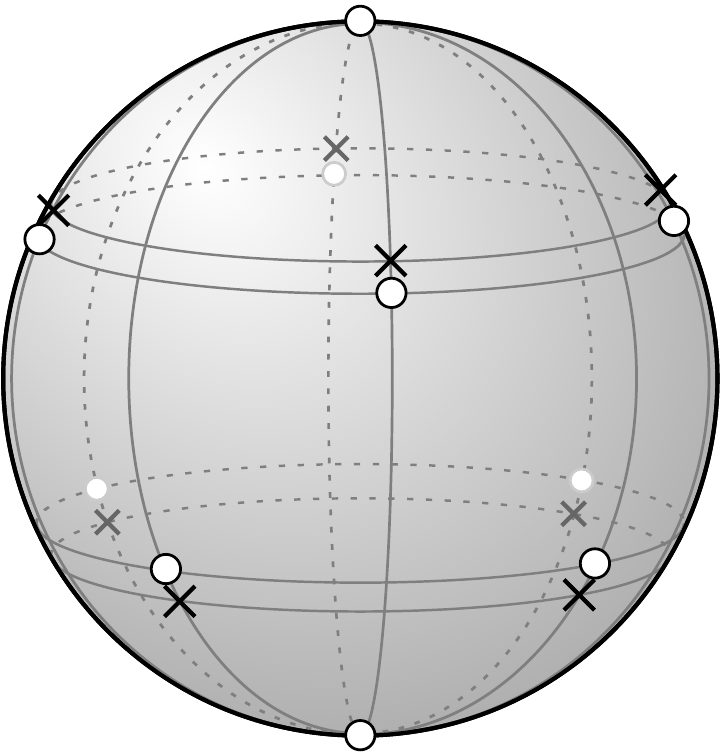}
    \put(-7.5,0){(a)}
  \end{overpic}
  \begin{overpic}[scale=.74]{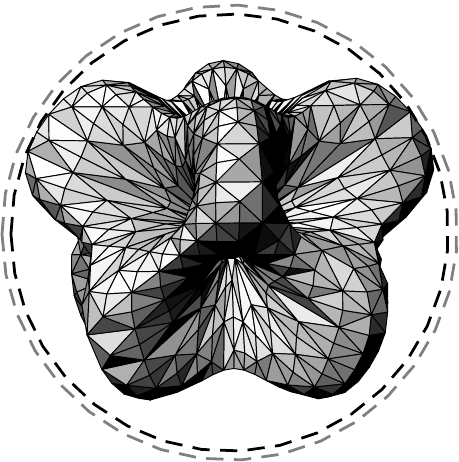}
    \put(-7.5,0){(b)}
  \end{overpic}
  \hfill
  \begin{overpic}[scale=.445]{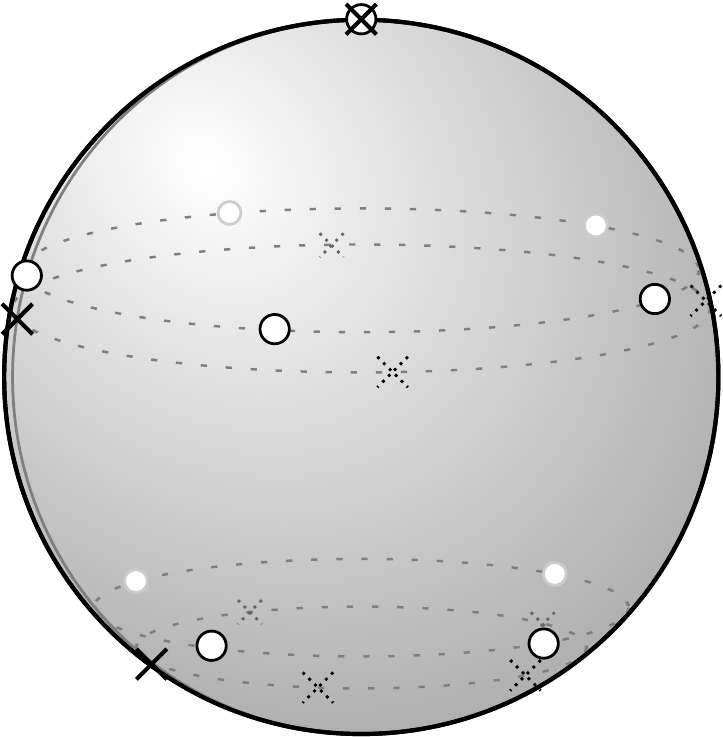}
    \put(-7.5,0){(c)}
  \end{overpic}
  \hspace{0.5mm}
  \begin{overpic}[scale=.445]{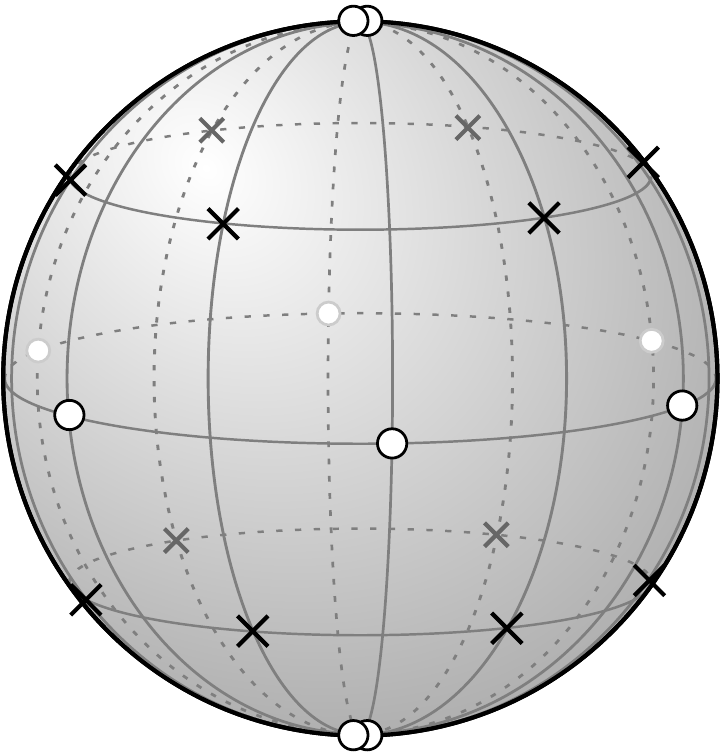}
    \put(-7.5,0){(d)}
  \end{overpic}
  \caption[Majorana representation of 10 qubit
  states]{\label{bloch_10} For 10 qubits the conjectured maximally
    entangled state $\ket{\Psi_{10}}$, shown in (a) and (b), takes the
    form of a gyroelongated square bipyramid. The Majorana
    representation of the numerically determined maximally entangled
    positive state $\ket{\Psi_{10}^{\text{pos}}}$ is shown in (c), and
    the positive state $\ket{\psi_{10}^{\text{pos}}}$ shown in (d) has
    almost the same amount of entanglement as
    $\ket{\Psi_{10}^{\text{pos}}}$.  The state
    $\ket{\Psi_{10}^{\text{pos}}}$ has three \acp{CPP}, with the
    locations of further local maxima of $g^2 ( \theta , \varphi )$
    indicated by dashed crosses.  The outer and inner circle in (b)
    corresponds to the value of $G^2$ for
    $\ket{\Psi_{10}^{\text{pos}}}$ and $\ket{\Psi_{10}}$,
    respectively.}
\end{figure}

The 10 qubit case is the first one where the conjectured maximally
entangled symmetric state cannot be cast with positive coefficients.
Since the search for maximal entanglement is more reliable within the
subset of positive states, we will separately consider the positive
case.

A numerical search returns a state of the form
$\ket{\Psi_{10}^{\text{pos}}} = \alpha \sym{0} + \beta \sym{4} +
\gamma \sym{9}$ as the positive state with the highest amount of
geometric entanglement, namely $\Eg ( \ket{\Psi_{10}^{\text{pos}}} )
\approx 2.679 \: 763 \: 092$.  The approximate values of the
coefficients are
\begin{equation}\label{10_positive_coeff}
  \alpha \approx 0.395 \: 053 \: 091 \ens , \quad
  \beta \approx  0.678 \: 420 \: 822 \ens , \quad
  \gamma \approx 0.619 \: 417 \: 665 \ens .
\end{equation}
The \ac{MP} distribution is shown in \fig{bloch_10}(c).  From
\lemref{rot_symm} it is clear that this state is not rotationally
symmetric around the $Z$-axis.  The state has only three \acp{CPP},
which are all positive (cf.  \theoref{cpp_number_locations}), but the
spherical amplitude function $g( \theta , \varphi )$ has six other
local maxima with values close to those at the \acp{CPP}.  The
positions of these local maxima are shown as dashed crosses in
\fig{bloch_10}(c).  One would expect that the \acp{MP} on the two
\quo{circles}, one with five \acp{MP} and another with four \acp{MP},
have the same latitude and are equidistantly spaced apart. However,
this is not the case, as the locations of the \acp{MP} deviate by very
small amounts from such a regular distribution.  Indeed, since
equidistant circles of \acp{MP} correspond to \ac{GHZ}-type states, it
can be seen from \theoref{theo_general_maj_rep} that for perfect
\ac{MP} rings the state $\ket{\Psi_{10}^{\text{pos}}}$ would need to
have four nonvanishing basis states.

We mention that there exists a fully rotationally symmetric and
totally invariant (under the dihedral group $D_{6}$) positive state
that comes very close to $\ket{\Psi_{10}^{\text{pos}}}$ in terms of
entanglement. Its form is $\ket{\psi_{10}^{\text{pos}}} =
\tfra{1}{\sqrt{2}} ( \sym{2} + \sym{8} )$, and its Majorana
representation is shown in \fig{bloch_10}(d).  The 12 \acp{CPP} can be
determined as the solutions of a quadratic equation, yielding
\begin{equation}\label{10_sigma}
  \begin{split}
    \ket{\sigma_{1,2, \ldots ,6}}& =
    \tfrac{1}{\sqrt{3 - \sqrt{3}}} \ket{0} + \E^{\I \kappa}
    \tfrac{1}{\sqrt{3 + \sqrt{3}}} \ket{1} \ens , \\
    \ket{\sigma_{7,8, \ldots ,12}}& =
    \tfrac{1}{\sqrt{3 + \sqrt{3}}} \ket{0} + \E^{\I \kappa}
    \tfrac{1}{\sqrt{3 - \sqrt{3}}} \ket{1} \ens ,
  \end{split}
\end{equation}
with $\kappa = 0, \tfra{\pi}{3}, \tfra{2 \pi}{3}, \pi, \tfra{4
  \pi}{3}, \tfra{5 \pi}{3}$.  The entanglement is $\Eg (
\ket{\psi_{10}^{\text{pos}}} ) = \log_2 \left( \tfra{32}{5} \right)
\approx 2.678 \: 072$, which is less than $0.1 \%$ difference from
$\Eg ( \ket{\Psi_{10}^{\text{pos}}} )$.

\subsection{Eleven qubits}\label{majorana_eleven}

The known numerical solution to Thomson's problem has the form
(1-2-4-2-2) in F{\"o}ppl notation \cite{Whyte52}, yielding the
approximate entanglement $\Eg ( \ket{\psi_{11}^{\text{Th}}} ) \approx
2.482 \: 570$.  On the other hand, the solution to \toths problem is
obtained by removing one vertex of the regular icosahedron, yielding a
pentagonal antiprism with a pentagonal pyramid on one of the two
pentagonal surfaces, or (1-5-5) \cite{Melnyk77}.  From the known
geometric properties of the icosahedron the solution is found
analytically to be $\ket{\psi_{11}^{\text{\totx}}} =
\tfra{\sqrt{462}}{25} \sym{0} + \tfra{11}{25} \sym{5} -
\tfra{\sqrt{42}}{25} \sym{10}$.  Unsurprisingly, the corresponding
spherical amplitude function is very imbalanced, with the single
\ac{CPP} lying antipodal to the removed icosahedron vertex, yielding
$\Eg ( \ket{\psi_{11}^{\text{\totx}}} ) = \log_2 \left(
  \tfra{625}{462} \right) \approx 0.435 \: 963$.  By varying the
latitude of the \ac{MP} circles, however, it is possible to obtain a
state with much higher entanglement: The state shown in
\fig{bloch_11}(a) is rotationally symmetric around the $Z$-axis and
has 11 \acp{CPP}.  The form of the state is $\ket{\Psi_{11}} = \alpha
\sym{0} + \beta \sym{5} - \gamma \sym{10}$, with approximate values
\begin{equation}\label{11_maxent_coeff}
  \alpha \approx 0.376 \: 611 \: 967 \ens , \quad
  \beta  \approx 0.715 \: 661 \: 256 \ens , \quad
  \gamma \approx 0.588 \: 211 \: 181 \ens .
\end{equation}
Its \acp{MP} are
\begin{equation}\label{11_mp_maxent}
  \ket{\phi_{1}} = \ket{0} \ens , \ens
  \ket{\phi_{2,3,4,5,6}} = \co_{\theta} \ket{0} -
  \E^{\I \kappa} \si_{\theta} \ket{1} \ens , \ens
  \ket{\phi_{7,8,9,10,11}} = \si_{\vartheta} \ket{0} +
  \E^{\I \kappa} \co_{\vartheta} \ket{1} \ens ,
\end{equation}
with $\kappa = 0, \tfra{2 \pi}{5}, \tfra{4 \pi}{5}, \tfra{6 \pi}{5},
\tfra{8 \pi}{5}$, and approximate latitudinal angles $\theta \approx
1.168 \: 499 \: 343$ and $\vartheta \approx 2.253 \: 247 \: 569$.  The
entanglement is $\Eg ( \ket{\Psi_{11}} ) \approx 2.817 \: 698 \: 505$,
making this state the potentially maximally entangled symmetric state
of 11 qubits.

\begin{figure}
  \centering
  \begin{overpic}[scale=.45]{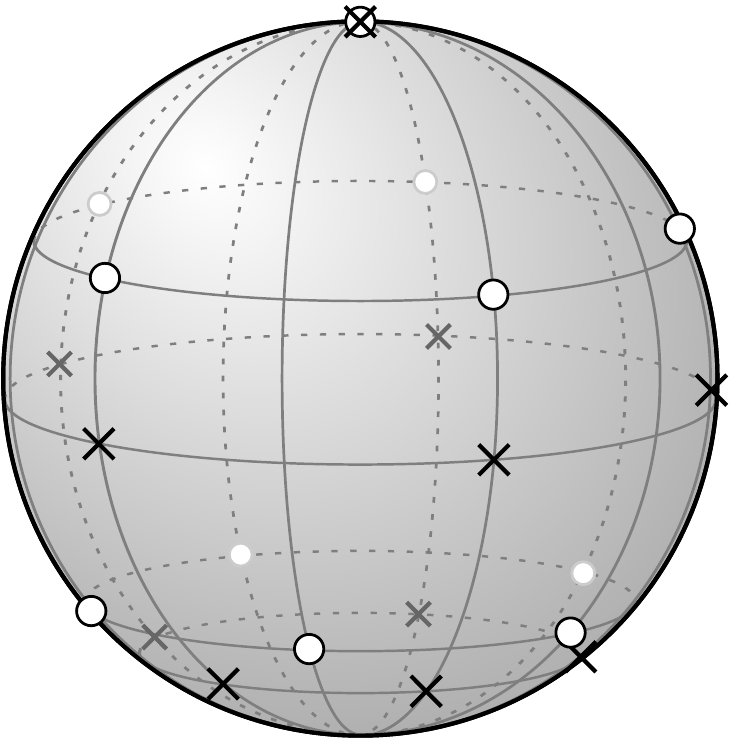}
    \put(-7,-2){(a)}
  \end{overpic}
  \begin{overpic}[scale=.77]{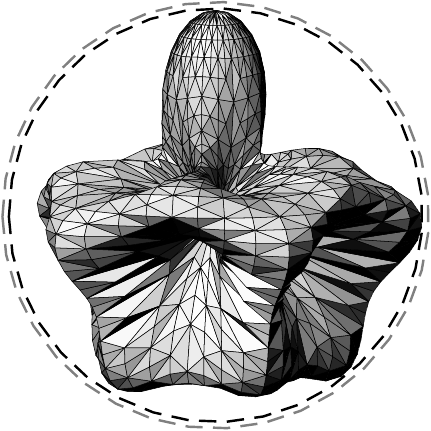}
    \put(-7,-2){(b)}
  \end{overpic}
  \hfill
  \begin{overpic}[scale=.45]{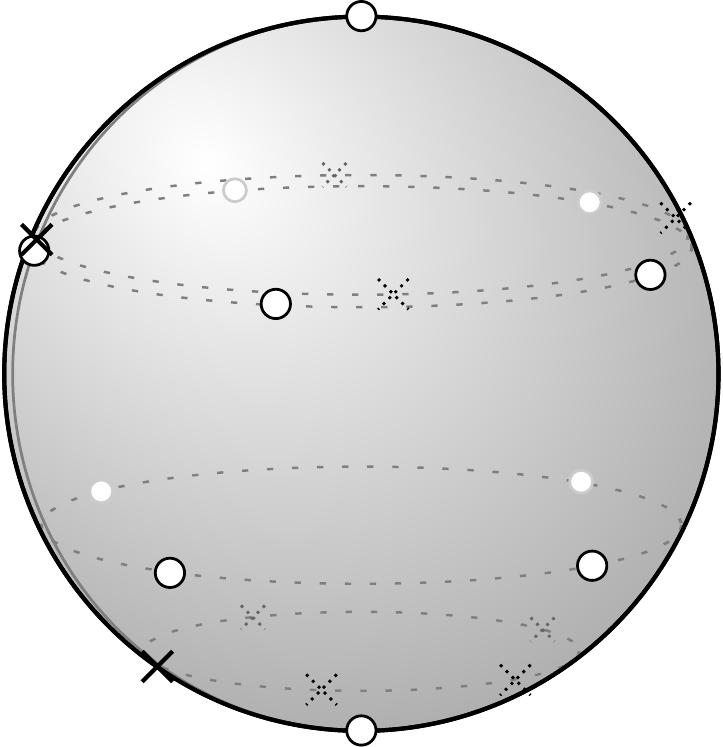}
    \put(-7,-2){(c)}
  \end{overpic}
  \begin{overpic}[scale=.77]{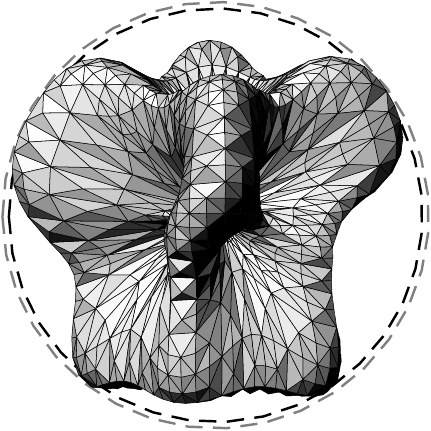}
    \put(-7,-2){(d)}
  \end{overpic}
  \caption[Majorana representation of 11 qubit
  states]{\label{bloch_11} For 11 qubits the candidate for maximal
    entanglement $\ket{\Psi_{11}}$ is shown in (a) and (b).  The
    numerically determined maximally entangled positive state
    $\ket{\Psi_{11}^{\text{pos}}}$, shown in (c) and (d), has only two
    \acp{CPP}, but its spherical amplitude function has seven more
    local maxima with values very close to those at the \acp{CPP}.}
\end{figure}

Analogous to the 10 qubit case, the numerically determined maximally
entangled positive symmetric state does not have a rotational
symmetry. The state, shown in \fig{bloch_11}(c) and (d), is of the
form $\ket{\Psi_{11}^{\text{pos}}} = \alpha \sym{1} + \beta \sym{5} +
\gamma \sym{10}$, with the approximate values
\begin{equation}\label{11_positive_coeff}
  \alpha \approx 0.550 \: 982 \: 113 \ens , \quad
  \beta \approx  0.578 \: 058 \: 577 \ens , \quad
  \gamma \approx 0.601 \: 886 \: 195 \ens .
\end{equation}
This state has only two \acp{CPP}, but the spherical amplitude
function has seven more local maxima with values close to the
\acp{CPP}.  The geometric entanglement of this state is $\Eg (
\ket{\Psi_{11}^{\text{pos}}} ) \approx 2.773 \: 622 \: 669$.

\subsection{Twelve qubits}\label{majorana_twelve}

For 12 points both \toths and Thomson's problem are solved by the
icosahedron.  Due to the high symmetry present in Platonic solids, the
icosahedron state is also a strong candidate for maximal symmetric
entanglement in the 12 qubit case.  The state can be cast with real
coefficients $\ket{\Psi_{12}} = \tfra{\sqrt{7}}{5} \sym{1} -
\tfra{\sqrt{11}}{5} \sym{6} - \tfra{\sqrt{7}}{5} \sym{11}$, and its
\acp{MP} can be derived from the known angles and distances in the
icosahedron:
\begin{equation}\label{12_maj}
  \begin{split}
    \ket{\phi_{1}}& = \ket{0} \ens , \quad
    \ket{\phi_{2, 3, 4, 5, 6}} = \sqrt{ \tfrac{3+\sqrt{5}}{5
        +\sqrt{5}}} \ket{0} + \E^{\I \kappa}
    \sqrt{\tfrac{2}{5+\sqrt{5}}} \ket{1} \ens , \\
    \ket{\phi_{12}}& = \ket{1} \ens , \quad
    \ket{\phi_{7, 8, 9, 10, 11}} = \sqrt{ \tfrac{2}{5+\sqrt{5}}}
    \ket{0} - \E^{\I \kappa} \sqrt{
      \tfrac{3+\sqrt{5}}{5 +\sqrt{5}}} \ket{1} \ens ,
  \end{split}
\end{equation}
with $\kappa = 0, \tfra{2 \pi}{5}, \tfra{4 \pi}{5}, \tfra{6 \pi}{5},
\tfra{8 \pi}{5}$. The \ac{MP} distribution is shown in
\fig{bloch_12}(a).  From the icosahedral symmetry and the antipodal
pairs of \acp{MP}, it can be easily inferred that there exist 20
\acp{CPP}, one at the centre of each face of the icosahedron,
describing a dodecahedron on the Majorana sphere.  Although
\lemref{lem_pos_cps} and \lemref{pos_symm_cpp} cannot be applied to
the icosahedron state, its \acp{CPP} can be verified analytically by
considering the values of the spherical amplitude function $g ( \theta
, \varphi )$ within the area of one spherical triangle spanned by
three neighbouring \acp{MP}.  The \acp{CPP} thus obtained are
\begin{equation}\label{12_sigma}
  \begin{split}
    \ket{\sigma_{1, 2, 3, 4, 5}}& = \text{a}_{+} \ket{0} -
    \E^{\I \kappa} \text{a}_{-} \ket{1}\ens , \quad
    \ket{\sigma_{11, 12, 13, 14, 15}} = \text{b}_{-} \ket{0} +
    \E^{\I \kappa} \text{b}_{+} \ket{1} \ens , \\
    \ket{\sigma_{6, 7, 8, 9, 10}}& = \text{b}_{+} \ket{0} -
    \E^{\I \kappa} \text{b}_{-} \ket{1}\ens , \quad
    \ket{\sigma_{16, 17, 18, 19, 20}} = \text{a}_{-} \ket{0} +
    \E^{\I \kappa} \text{a}_{+} \ket{1} \ens ,
  \end{split}
\end{equation}
with $\kappa = 0, \tfra{2 \pi}{5}, \tfra{4 \pi}{5}, \tfra{6 \pi}{5},
\tfra{8 \pi}{5}$, and
\begin{equation}\label{12_sigma_coeff}
  \text{a}_{\pm} = \sqrt{\tfrac{1}{2} \pm \tfrac{1}{2}
    \sqrt{\tfrac{5+2 \sqrt{5}}{15}}} \ens , \quad
  \text{b}_{\pm} = \sqrt{\tfrac{1}{2} \pm \tfrac{1}{2}
    \sqrt{\tfrac{5-2 \sqrt{5}}{15}}} \ens .
\end{equation}
With the knowledge of the exact positions of the \acp{MP} and
\acp{CPP}, the entanglement follows as $\Eg ( \ket{\Psi_{12}} ) =
\log_2 \left( \tfra{243}{28} \right) \approx 3.117 \: 458$. Naturally,
the icosahedron state is totally invariant under the icosahedral
rotation group $Y$, so it follows from \lemref{invariant_equivalent}
that its entanglement is the same for the three distance-like
entanglement measures.  However, since $\ket{\Psi_{12}}$ is not
positive, the conditions of \theoref{pos_inv} are not fulfilled, and
it is not known whether $\ket{\Psi_{12}}$ is additive under the
various entanglement measures.

\begin{figure}
  \begin{overpic}[scale=1.08]{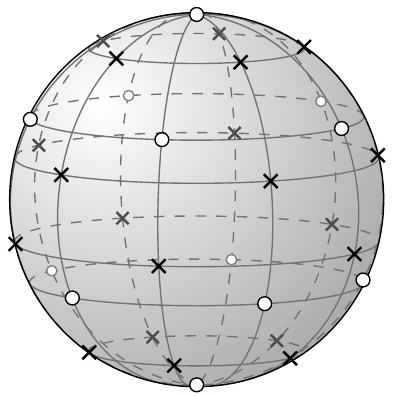}
    \put(-5,0){(a)}
  \end{overpic}
  \begin{overpic}[scale=1.42]{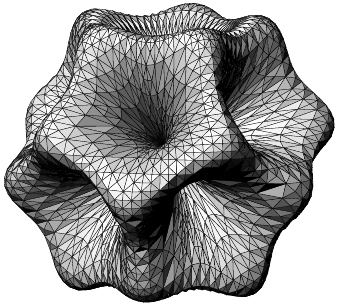}
    \put(-5,0){(b)}
  \end{overpic}
  \hfill
  \begin{overpic}[scale=.56]{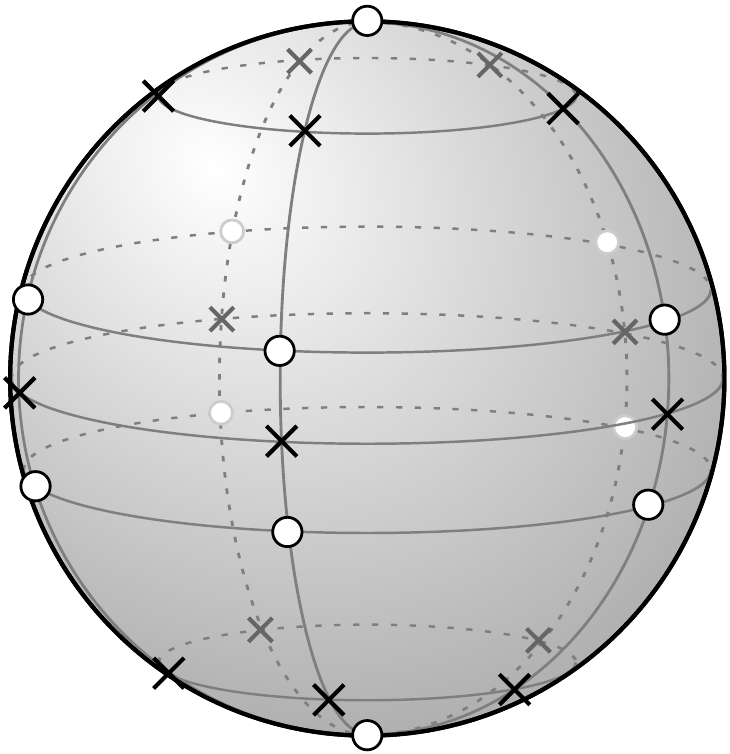}
    \put(-5,0){(c)}
  \end{overpic}
  \caption[Majorana representation of 12 qubit
  states]{\label{bloch_12} For 12 qubits the icosahedron state
    $\ket{\Psi_{12}}$, shown in (a) and (b), is conjectured to be the
    maximally entangled symmetric state.  In the subset of positive
    states the state $\ket{\Psi_{12}^{\text{pos}}}$ shown in (c) is
    detected as the maximally entangled one.}
\end{figure}

The numerical search for the maximally entangled positive state yields
a state of the form $\ket{\Psi_{12}^{\text{pos}}} = \alpha \sym{1} +
\beta \sym{6} + \alpha \sym{11}$ with
\begin{equation}\label{12_positive_coeff}
  \alpha \approx 0.555 \: 046 \: 977 \ens , \quad
  \beta  \approx 0.619 \: 552 \: 827 \ens . \quad
\end{equation}
From \fig{bloch_12}(c) it can be seen that this state is similar to
the icosahedron, with one of the horizontal circles of \acp{MP}
rotated by 36$^\circ$ so that it is aligned with the \acp{MP} of the
other circle.  There are 15 \acp{CPP} distributed on three circles,
with one circle coinciding with the equator.  The approximate amount
of entanglement is $\Eg ( \ket{\Psi_{12}^{\text{pos}}} ) \approx 2.993
\: 524 \: 700$.

\subsection{Twenty qubits}\label{majorana_twenty}

For the sake of completeness we mention the 20 qubit case, because it
contains the dodecahedron, the Platonic solid with the largest number
of vertices.  It was seen that the 20 \acp{CPP} of the icosahedron
state describe a regular dodecahedron, and therefore the \acp{MP} of
the dodecahedron state are immediately given by \eq{12_sigma}.  The
analytic form of the dodecahedron state is
\begin{equation}\label{dodecahedron_state}
  \ket{\psi_{20}} = \tfrac{1}{25 \sqrt{3}}
  \Big( \! \sqrt{187} \ket{S_{0}} + \sqrt{627} \ket{S_{5}} +
  \sqrt{247} \ket{S_{10}} - \sqrt{627} \ket{S_{15}} + \sqrt{187}
  \ket{S_{20}} \! \Big) \, .
\end{equation}
Its Majorana representation is shown in \fig{bloch_20}(a), and its
spherical volume function $g^{\frac{2}{3}} ( \theta , \varphi )$ is
shown in \fig{bloch_20}(b).  From the icosahedral symmetry and the
antipodal configuration of the \acp{MP} it can be easily inferred that
this state has 12 \acp{CPP}, one at the centre of each face of the
dodecahedron. Therefore the \acp{CPP} are given by \eq{12_maj}.  With
$\ket{\sigma_{1}} = \ket{0}$ being a \ac{CPP}, we immediately obtain
$G^2 = \frac{187}{1875}$ and $\Eg ( \ket{\psi_{20}} ) = \log_2
\frac{1875}{187} \approx 3.325 \: 780$.  Like the icosahedron state,
the dodecahedron state is totally invariant under the icosahedral
symmetry group $Y$, but it cannot be cast as a positive
state. Therefore its entanglement coincides for the three
distance-like entanglement measures, but additivity results are not
known.

\begin{figure}
  \centering
  \begin{overpic}[scale=0.84]{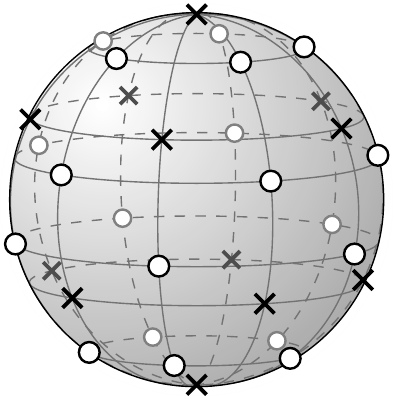}
    \put(-7,0){(a)}
  \end{overpic}
  \begin{overpic}[scale=.24]{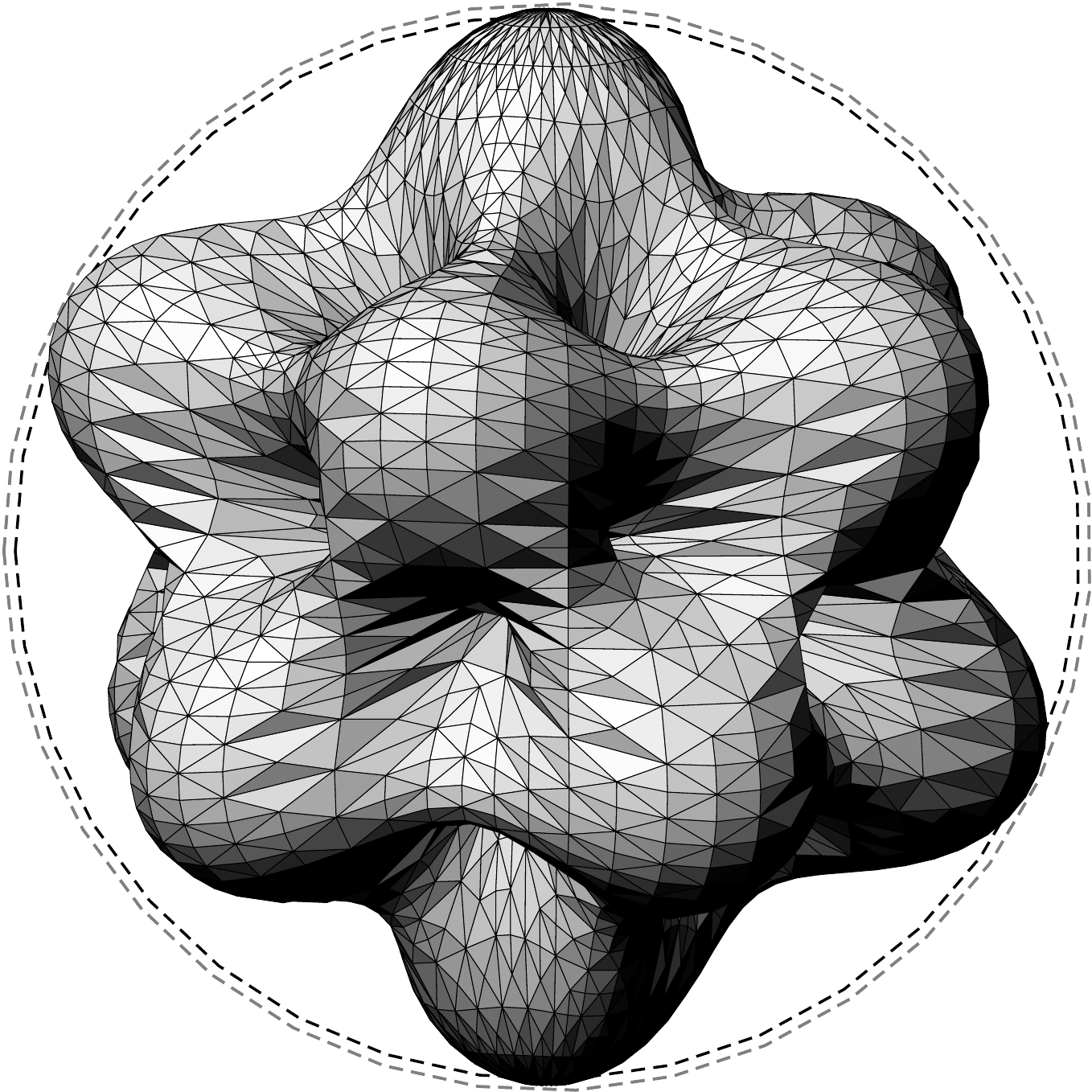}
    \put(-7,0){(b)}
  \end{overpic}
  \hfill
  \begin{overpic}[scale=.24]{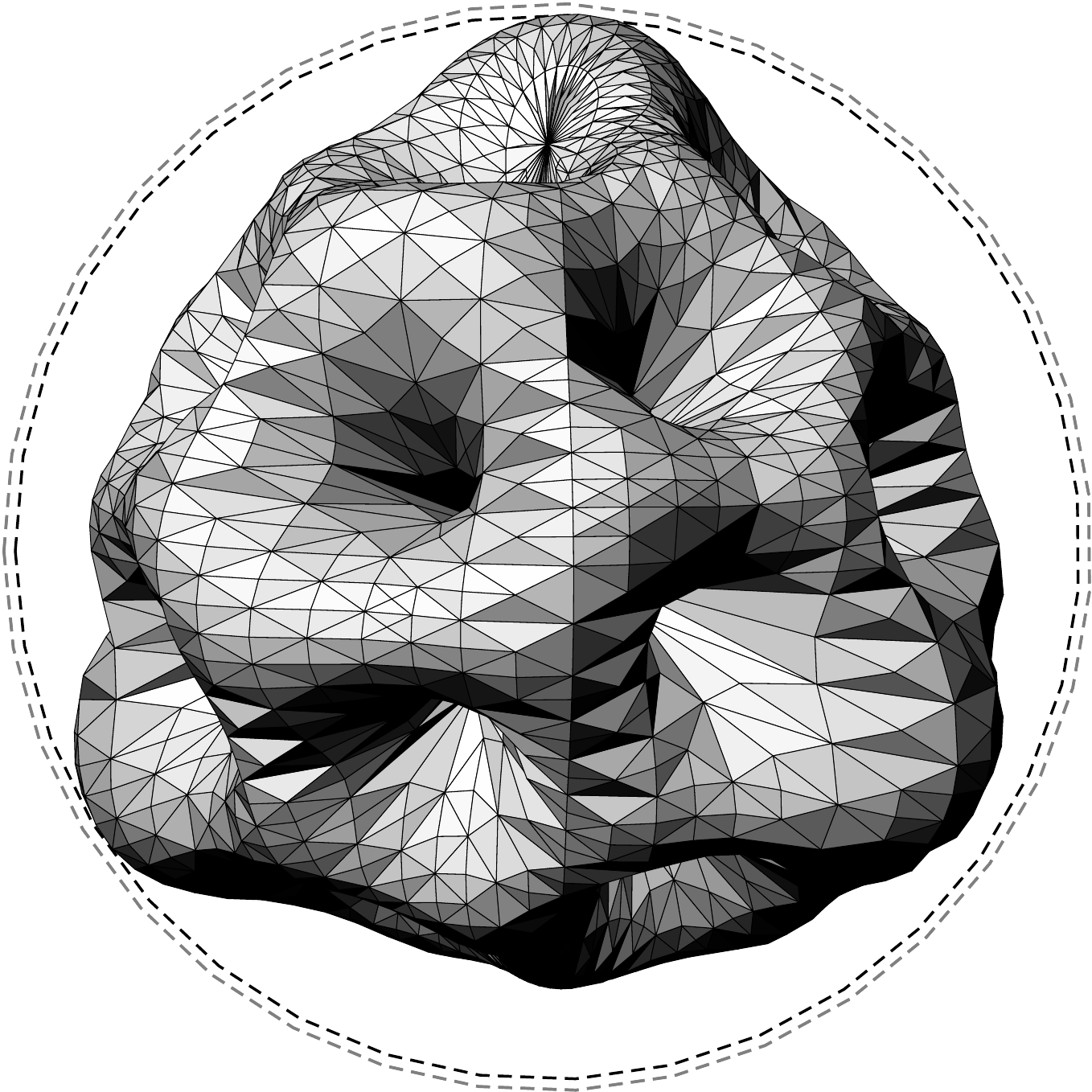}
    \put(-7,0){(c)}
  \end{overpic}
  \begin{overpic}[scale=.24]{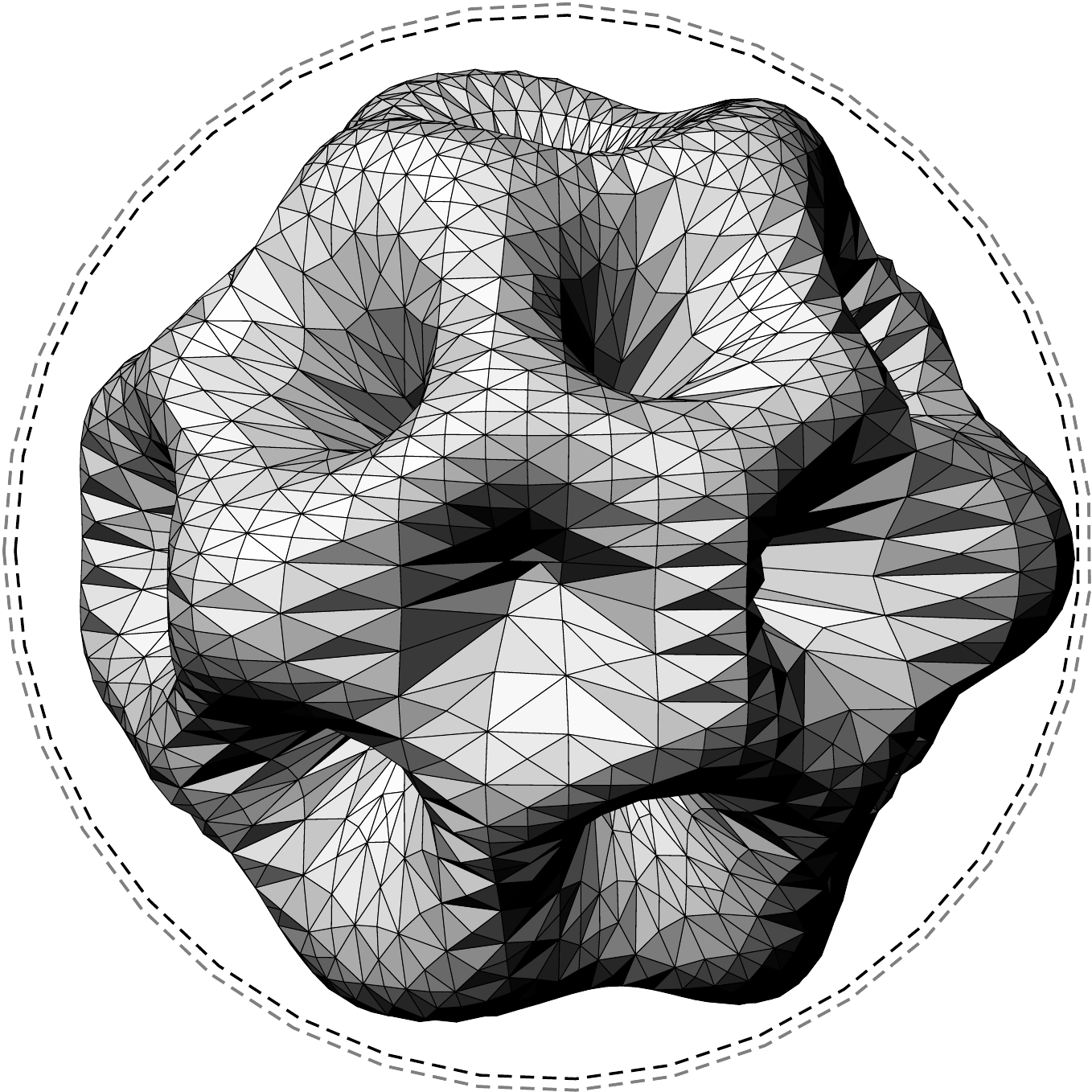}
    \put(-7,0){(d)}
  \end{overpic}
  \caption[Majorana representation of 20 qubit
  states]{\label{bloch_20} For 20 qubits the Majorana representation
    and the spherical volume function $g^{\frac{2}{3}} ( \theta ,
    \varphi )$ of the dodecahedron state $\ket{\psi_{20}}$ is shown in
    (a) and (b), respectively.  This state is not maximally entangled,
    and two counterexamples are the solutions of \protect\toths and
    Thomson's problem, numerically computed as spherical volume
    functions in (c) and (d), respectively. The radius of the outer
    and inner circle are the maximal values of $g^{\frac{2}{3}}$ for
    $\ket{\psi_{20}}$ and $\ket{\psi_{20}^{\text{Th}}}$,
    respectively.}
\end{figure}

As mentioned in \sect{extremal_point}, the dodecahedron does not solve
either of the classical problems.  Here we show that it does not solve
the Majorana problem either.  This can be easily seen by converting
the numerically known point distributions of \toths and Thomson's
problem into 20 qubit symmetric states and determining their
entanglement.  Their spherical volume functions $g^{\tfra{2}{3}} (
\theta , \varphi )$ are shown in \fig{bloch_20}(c) and (d),
respectively, and the numerically derived values of their geometric
entanglement are $\Eg ( \ket{\psi_{20}^{\text{\totx}}} ) \approx 3.327
\: 075$ and $\Eg ( \ket{\psi_{20}^{\text{Th}}} ) \approx 3.418 \:
012$.  Thus the solution of \toths problem is only marginally more
entangled than the dodecahedron state, but the solution of Thomson's
problem has a significantly higher amount of entanglement.  The latter
state has only three \acp{CPP}, which describe an equilateral triangle
on the equator, so it is reasonable to expect that yet higher
entangled 20 qubit symmetric states exist.

\section{Summary and Discussion}\label{discussion}

In the following we discuss the results gathered about highly and
maximally entangled symmetric states from several viewpoints, and
formulate some results and conjectures.

\subsection{Entanglement properties}\label{entanglement_scaling}

In \chap{geometric_measure} it was found that the maximal geometric
entanglement of $n$ qubit states scales linearly, whereas the maximal
symmetric entanglement scales logarithmically.  Combining the upper
and lower bounds for the symmetric case, it is seen that the maximal
symmetric $n$ qubit entanglement scales as
\begin{equation}\label{symmetricbounds}
  \log_2 \sqrt{\tfrac{n \pi}{2}} \leq \Eg^{\text{max}}
  \leq \log_2 (n+1) \ens ,
\end{equation}
i.e. polylogarithmically between $\Order (\log \sqrt{n})$ and $\Order
(\log n)$.  Stronger lower bounds can be found numerically from the
known solutions of \toths and Thomson's problem by translating their
point distributions into the corresponding symmetric states and
determining their entanglement.  Martin \etal \cite{Martin10} did this
for up to $n = 110$ and found $\Eg ( \ket{\Psi_{n}^{\text{Th}}} )
\approx \log_2 \frac{(n+1)}{1.71} = \log_2 (n+1) - 0.775$ for the
solutions of Thomson's problem.  While this comes close to the upper
bound, the fluctuations for the individual $n$ can be large.  In
contrast to this, the explicit form of equally weighted superpositions
of Dicke states with pseudorandom phases is known for all $n$ and
their entanglement exhibits very small fluctuations \cite{Martin10}.
The best entanglement scaling found for such states is $\Eg (
\ket{\Psi_{n}} ) \approx \log_2 \frac{(n+1)}{2.22}$, which is slightly
below that of Thomson's solutions \cite{Martin10}.

For $3$ qubits the maximally entangled state $\ket{\text{W}}$ is
symmetric. On the other hand, for $n > 5$ qubits the maximally
entangled state can be neither symmetric nor \ac{LU}-equivalent to a
symmetric state, because the lower bound $\Eg \geq \frac{n}{2}$ for
general states is higher than the upper bound $\Eg \leq \log_2 (n+1)$
for symmetric states. Regarding the cases of $n = 4,5$ qubits, we
consider the entanglement of the maximally entangled symmetric states
derived in the previous section, and find that $\Eg ( \ket{\Psi_{n}} )
< \frac{n}{2}$ in both cases, which implies that the maximally
entangled states of the general Hilbert space can be symmetric only
for $n \leq 3$ qubits.

\begin{table}
  \begin{threeparttable}
    \caption[Values of maximal symmetric $n$ qubit entanglement for up
    to $n = 12$]{\label{table4.3} Entanglement values for symmetric
      $n$ qubit states in terms of the geometric measure.  Listed from
      left to right are the entanglement of the most entangled Dicke
      state, the maximally entangled positive symmetric state, the
      conjectured maximally entangled symmetric state and the upper
      bound on symmetric entanglement.  The inequalities $\Eg \big(
      \sym{\lfloor n/2 \rfloor} \big) \leq \Eg \big(
      \ket{\Psi^{\text{pos}}_{n}} \big) \leq \Eg \big( \ket{\Psi_{n}}
      \big) < \log_2 (n + 1)$ hold for all $n$, and wherever the
      amount of entanglement does not increase from left to right, the
      respective right-hand cell has been left blank.}
    \newcolumntype{R}{>{\centering\arraybackslash}X}
    \begin{tabularx}{\textwidth}{l|RRRR}
      \toprule
      $n$ & $\Eg \big( \sym{\lfloor n/2 \rfloor} \big)$ &
      $\Eg \big( \ket{\Psi^{\text{pos}}_{n}} \big)$
      & $\Eg \big( \ket{\Psi_{n}} \big)$ & $\log_2 (n+1)$ \\
      \midrule
      2 & $1$ & & & $\log_2 3$ \\
      3 & $\log_2 (9/4)$ & & & $2$ \\
      4 & $\log_2 (8/3)$ & $\log_2 3$ & & $\log_2 5$ \\
      5 & $\approx 1.532 \: 824 \: 877$ &
      $\approx 1.742 \: 268 \: 948$\tnote{$\dagger$} & &
      $\approx 2.584 \: 962 \: 501$ \\
      6 & $\log_2 (16/5)$ & $\log_2 (9/2)$ & & $\log_2 7$ \\
      7 & $\approx 1.767 \: 313 \: 935$ &
      $\approx 2.298 \: 691 \: 396$\tnote{$\dagger$} & & $3$ \\
      8 & $\approx 1.870 \: 716 \: 983$ &
      $\approx 2.445 \: 210 \: 159$ & &
      $\approx 3.169 \: 925 \: 001$ \\
      9 & $\approx 1.942 \: 404 \: 615$ &
      $\approx 2.553 \: 960 \: 277$\tnote{$\dagger$} & &
      $\approx 3.321 \: 928 \: 095$ \\
      10 & $\approx 2.022 \: 720 \: 077$ &
      $\approx 2.679 \: 763 \: 092$ &
      $\approx 2.737 \: 432 \: 003$ & $\approx 3.459 \: 431 \: 619$ \\
      11 & $\approx 2.082 \: 583 \: 285$ &
      $\approx 2.773 \: 622 \: 669$ &
      $\approx 2.817 \: 698 \: 505$ &
      $\approx 3.584 \: 962 \: 501$ \\
      12 & $\approx 2.148 \: 250 \: 959$ &
      $\approx 2.993 \: 524 \: 700$ &
      $\log_2 (243/28)$ & $\approx 3.700 \: 439 \: 718$ \\
      \bottomrule
    \end{tabularx}
    \begin{tablenotes}
    \item [$\dagger$] {\footnotesize Closed-form analytic expressions
        are known, but not displayed due to their complicated form.}
    \end{tablenotes}
  \end{threeparttable}
\end{table}

\Tabref{table4.3} summarises the largest entanglement values that we
found for symmetric $n$ qubit states with positive and general
coefficients for up to 12 qubits.  For comparison purposes, the upper
and lower bound are also listed.  Where closed-form expressions could
not be found for the entanglement of the positive and general
solutions, numerical values were calculated with a precision of at
least ten digits.  The values for $\Eg \left(
  \ket{\Psi^{\text{pos}}_{n}} \right)$ can be considered reliable in
the sense that we detected the maximally entangled state with a high
likelihood.  In contrast to this, the values $\Eg \left(
  \ket{\Psi_{n}} \right)$ for general symmetric states are less
reliable: While the entanglement of the candidates was calculated with
high precision, there is no guarantee that these states are indeed the
maximally entangled ones.  However, even if more entangled states do
exist, they are likely to have only a slightly higher amount of
entanglement.

\begin{figure}
  \centering
  \begin{overpic}[scale=1.2]{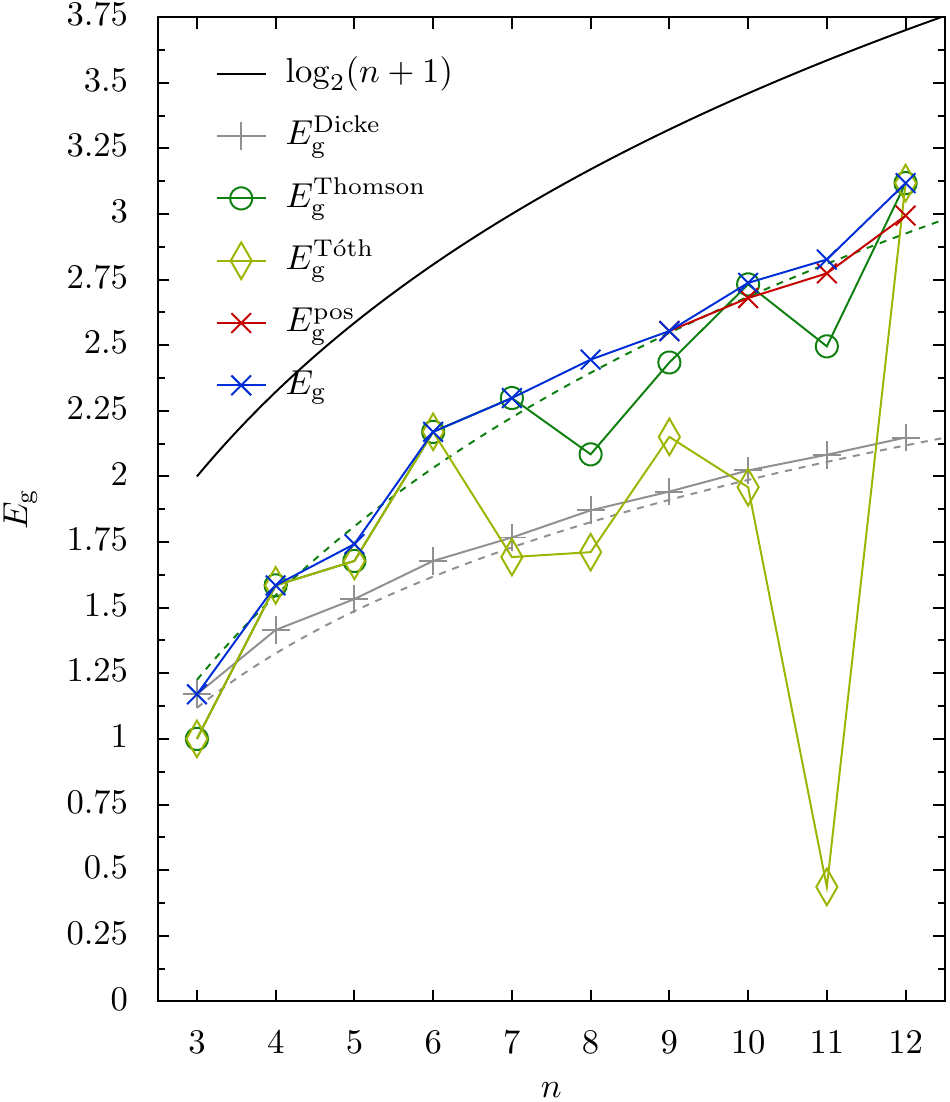}
  \end{overpic}
  \caption[Diagram with symmetric entanglement
  values]{\label{entanglement_graph} The geometric entanglement of
    different $n$ qubit symmetric states is shown.  The numerically
    determined maximally entangled symmetric states are represented by
    blue crosses. For $n = 10-12$ these states are not positive, and
    the corresponding maximally entangled positive states are denoted
    by red crosses.  The known upper and lower bound on maximal
    symmetric entanglement is shown as a black and gray line,
    respectively, with the Stirling approximation of the equally
    balanced Dicke states displayed as a dotted gray line.  The
    solutions of \protect\toths problem (olive diamonds) and Thomson's
    problem (green circles) yield nontrivial lower bounds for the
    maximal symmetric entanglement. The fitting $\Eg^{\text{Thomson}}
    \approx \log_2 \frac{(n+1)}{1.71}$ for the solutions of Thomson's
    problem of up to $n = 110$ was derived in \protect\cite{Martin10}
    and is displayed as a dashed green line.  Because of the
    relationship \protect\eqref{meas_pure_ineq} the values of $\Eg$
    are lower bounds for the maximal symmetric entanglement of the
    relative entropy of entanglement $E_{\text{R}}$ and the
    logarithmic robustness of entanglement $E_{\text{Rob}}$.}
\end{figure}

\begin{table}
  \begin{threeparttable}
    \caption[Summary of the properties of the investigated symmetric
    states]
    {\label{table4.5} Summary of the properties of the symmetric
      states discussed in \protect\chap{solutions}. The second column
      indicates whether the states are positive or not (the
      non-positive states listed here cannot be turned into positive
      states by symmetric \ac{LU} operations).  The third column
      indicates whether states are totally invariant, and if they are,
      the corresponding symmetry group is listed (cf.
      \protect\sect{invariant_and_additivity}).  All positive states
      are strongly additive under $\Eg$
      (cf. \protect\lemref{positive_additive}), and for all totally
      invariant states the three measures $\Eg$ $E_{\text{R}}$ and
      $E_{\text{Rob}}$ coincide
      (cf. \protect\lemref{invariant_equivalent}). States that are
      positive as well as totally invariant are furthermore additive
      under $E_{\text{R}}$ and $E_{\text{Rob}}$
      (cf. \protect\theoref{pos_inv}).  The last three columns
      indicate whether the listed states are the (conjectured)
      solutions of the point distribution problems discussed in
      \protect\sect{extremal_point}.}
    \newcolumntype{R}{>{\centering\arraybackslash}X}
    \begin{tabularx}{\textwidth}{R|cc|ccc}
      \toprule
      Families of states & \multirow{2}{*}{positive} & totally & \toth & Thomson & Majorana \\
      ($n \geq 3$ and $0 < k < \frac{n}{2}$) & & invar. & solution  & solution & solution \\
      \midrule
      balanced Dicke states $\sym{n,\frac{n}{2}}$ & \yestick & $\text{O}(2)$  & \notick & \notick & \notick \\
      imbalanced Dicke states $\sym{{n,k}}$ & \yestick & $\text{SO}(2)$ & \notick & \notick & \small{$n=3$} \\
      GHZ states $\ket{\text{GHZ}_{n}}$ & \yestick & $D_{n}$ & \small{$n=3$} & \small{$n=3$} & \notick \\
      \multirow{2}{*}{$\frac{1}{\sqrt{2}} \left( \sym{k} + \sym{n-k} \right)$} & \multirow{2}{*}{\yestick} &
      \multirow{2}{*}{$D_{n-2k}$} & \small{$n=5,6$} & \small{$n=5,6,7$} & \small{$n=6,7,9$} \\[-0.15em]
      &  &  & \small{$(k=1)$} & \small{$(k=1)$} & \small{$(k=1,2 \tnote{$\dagger$})$} \\
      \toprule
      States discussed in & \multirow{2}{*}{positive} & totally & \toth & Thomson & Majorana \\
      \protect\sect{maxent_results} & & invar. & solution & solution & solution \\
      \midrule
      tetrahedron $\ket{\Psi_{4}}$ & \yestick & $T$ & \yestick & \yestick & \yestick \\
      trigonal bipyramid $\ket{\psi_{5}}$ & \yestick & $D_{3}$ & \yestick & \yestick & \notick \\
      square pyramid $\ket{\Psi_{5}}$ & \yestick & \notick & \notick & \notick & \yestick \\
      octahedron $\ket{\Psi_{6}}$ & \yestick & $O$ & \yestick & \yestick & \yestick \\
      (1-3-3) $\ket{\psi_{7}^{\text{\totx}}}$ & \notick & \notick & \yestick & \notick & \notick \\
      pentagonal dipyramid $\ket{\Psi_{7}}$ & \yestick & $D_{5}$ & \notick & \yestick & \yestick \\
      regular cube $\ket{\psi_{8}^{\text{c}}}$ & \yestick & $O$ & \notick & \notick & \notick \\
      asymmetric pentagonal & \multirow{2}{*}{\yestick} & \multirow{2}{*}{\notick} &
      \multirow{2}{*}{\notick} & \multirow{2}{*}{\notick} & \multirow{2}{*}{\yestick} \\
      dipyramid $\ket{\Psi_{8}}$ & & & & & \\
      antiprism states $\ket{\psi_{8}^{\text{a}}}$ & \notick & \notick &
      \yestick\tnote{$*$} & \yestick\tnote{$*$} & \notick \\
      triaugmented triangular & \multirow{2}{*}{\notick} & \multirow{2}{*}{\notick} &
      \multirow{2}{*}{\yestick\tnote{$*$}} & \multirow{2}{*}{\yestick\tnote{$*$}} & \multirow{2}{*}{\notick} \\
      prism states $\ket{\psi_{9}}$ & & & & & \\
      pentagonal dipyramid $\ket{\Psi_{9}}$ & \yestick & $D_{5}$ & \notick & \notick & \yestick \\
      (2-2-4-2) $\ket{\psi_{10}^{\text{\totx}}}$ & \notick & \notick & \yestick & \notick & \notick \\
      gyroelongated square & \multirow{2}{*}{\notick} & \multirow{2}{*}{\notick} &
      \multirow{2}{*}{\notick} & \multirow{2}{*}{\yestick\tnote{$*$}} & \multirow{2}{*}{\yestick\tnote{$*$}} \\
      bipyramid states $\ket{\psi_{10}}$ & & & & & \\
      maximal. ent. positive $\ket{\Psi_{10}^{\text{pos}}}$ & \yestick & \notick & \notick & \notick & \notick \\[0.15em]
      rotat. symm. positive $\ket{\psi_{10}^{\text{pos}}}$ & \yestick & $D_{6}$ & \notick & \notick & \notick \\
      (1-2-4-2-2) $\ket{\psi_{11}^{\text{Th}}}$ & \notick & \notick & \notick & \yestick & \notick \\
      (1-5-5) states $\ket{\psi_{11}}$ & \notick & \notick & \yestick\tnote{$*$} & \notick & \yestick\tnote{$*$} \\
      maximal. ent. positive $\ket{\Psi_{11}^{\text{pos}}}$ & \yestick & \notick & \notick & \notick & \notick \\
      icosahedron $\ket{\Psi_{12}}$ & \notick & $Y$ & \yestick & \yestick & \yestick \\
      maximal. ent. positive $\ket{\Psi_{12}^{\text{pos}}}$ & \yestick & \notick & \notick & \notick & \notick \\
      dodecahedron $\ket{\psi_{20}}$ & \notick & $Y$ & \notick & \notick & \notick \\
      \bottomrule
    \end{tabularx}
    \begin{tablenotes}
    \item [$*$] {\footnotesize The solutions correspond to different
        values of the parameter(s) that describe the states.}
    \item [$\dagger$] {\footnotesize $k=1$ for $n=6,7$, and $k=2$ for $n=9$.}
    \end{tablenotes}
  \end{threeparttable}
\end{table}

The diagram in \fig{entanglement_graph} displays the entanglement of
our candidates and solutions along with the entanglement of the
classical problems and the upper and lower bounds.  It is seen that
the solutions of Thomson's problem generally exhibit a higher amount
of entanglement than those of \toths problem, thus demonstrating that
for large $n$ the solutions for Thomson's problem are generally a
better approximation for the Majorana problem than the solutions of
\toths problem.

In \tabref{table4.5} the qualitative properties (positive, totally
invariant, solution of an optimisation problem) of all the symmetric
states investigated in this chapter are listed.  Note that positivity
automatically implies strong additivity under $\Eg$
(\lemref{positive_additive}), that total invariance implies $\Eg =
E_{\text{R}} = E_{\text{Rob}}$ (\lemref{invariant_equivalent}), and
that the simultaneous existence of these two properties additionally
implies additivity under $E_{\text{R}}$ and $E_{\text{Rob}}$
(\theoref{pos_inv}).

\subsection{Number and locations of MPs}\label{number_mp}

The spherical amplitude function $g ( \theta , \varphi ) =
\abs{\bracket{\psi^{\text{s}}}{\sigma (\theta , \varphi ) }^{\otimes
    n}}$ proved to be a valuable tool for numerically determining the
\ac{MP} locations of a given symmetric state $\psis$, because the
zeroes of this function coincide with the antipodes of the \acp{MP}.
By considering the power $g^{\frac{2}{3}} ( \theta , \varphi )$ of
this function, we obtained the spherical volume function which
describes a three-dimensional volume that is constant for all $n$
qubit symmetric states (cf. \corref{cor_volume}).  This function can
be used to explain the Majorana representation of highly entangled
symmetric states: A \quo{bunching} of \acp{MP} in a small area,
e.g. in one half-sphere of the Majorana sphere leads to high values of
$g^{\frac{2}{3}}$ in that area, and this imbalance of the spherical
volume function leads to low entanglement.  This explains the tendency
of \acp{MP} to spread out widely over the sphere, in a similar fashion
to the classical problems.  Rather surprisingly, however, there also
exist highly entangled states where two or more \acp{MP} coincide (as
seen for $n =3,8,9$).  This is intriguing because such configurations
are the least optimal ones for classical point distributions. Again,
this can be explained with the constant integration volume: Because
the zeroes of $g^{\frac{2}{3}}$ are the antipodes of the \acp{MP}, a
lower number of \emph{different} \acp{MP} means that the spherical
volume function has fewer zeroes, and due to the constant volume, this
can lead to smaller values at the global maxima.

With regard to the Platonic solids, it was found that they solve the
Majorana problem only for $n=4,6,12$, which is in full analogy to the
classical problems.  How can this be understood?  For \toths problem
an intuitive description was already given for $n=8$ in
\fig{platonic}: By turning the cube into a regular antiprism, the
nearest neighbour distances can be increased, at the expense of
breaking the Platonic symmetry. In general, Thomson's problem and the
Majorana problem also favour such increased nearest-neighbour
distances. For $n=4,6,12$ the Platonic solids are composed of regular
triangles, whereas the cube ($n=8$) is composed of regular squares and
the dodecahedron ($n=20$) of regular pentagons. From this it can be
inferred that the vertices of optimal point distributions tend to form
triangles with their nearest neighbours.

Summing up the behaviour of \acp{MP} of maximally entangled symmetric
states, we can say that they prefer to be either well spaced apart
from each other, or to coincide into degeneracies.  Like the classical
point distributions, the \acp{MP} tend to describe polyhedra that are
made up mostly or entirely of triangles.  Because phased states are in
general considerably higher entangled than positive states
(cf. \theoref{theo_positive_entanglement}, \corref{max_pos_ent} and
\cite{Zhu10}), and because positive coefficients impose strong
restrictions on the locations of \acp{MP} and \acp{CPP} (cf.
\sect{pos_symm_states}), it is expected that for larger $n$ the
maximally entangled symmetric states no longer exhibit any rotational
and reflective symmetries in their Majorana representation.  For
Thomson's problem, the first distribution without any symmetry
features (and therefore no representation as a real state,
cf. \corref{cpp_real}) arises at $n=13$, and for \toths problem at
$n=15$. It is therefore reasonable to expect that the situation is
similar for the solutions of the Majorana problem.

From a mathematical point of view, an interesting question is in which
cases the \acp{MP} and \acp{CPP} of certain states (such as the
maximally entangled ones) can be derived analytically as an algebraic
or closed-form number.  The positions of the \acp{MP} and \acp{CPP}
are often given by the roots of polynomial equations.  Abel's
impossibility theorem states that the general quintic and higher
equation is impossible to solve algebraically \cite{Tignol}. In the
cases $n=7$ and $n=9$ we encountered such polynomials, but we could
reduce them to cubic equations by suitable substitutions.  This may
not be possible in general, and Galois theory may then be useful in
answering the question of algebraic solvability \cite{Tignol}.

\subsection{Number and locations of CPPs}\label{number_cpp}

Excluding the Dicke states with their continuous ring of \acp{CPP},
one observes that candidates for maximal entanglement tend to have a
large number of \acp{CPP}.  The prime example is the case of five
qubits, where the classical solution with only three \acp{CPP} is less
entangled than the \quo{square pyramid} state which has five
\acp{CPP}.  In \theoref{cpp_number_locations} it was shown that $2n-4$
is an upper bound on the number of \acp{CPP} of positive symmetric $n$
qubit states.  In \tabref{table4.4} this bound is compared to the
number of \acp{CPP} of all candidates and solutions. It can be seen
that the bound is obeyed by all states, including those that cannot be
cast with positive coefficients.  In many cases the number of
\acp{CPP} comes close to the bound ($n=5,8$) or coincides with it ($n
= 4,6,7,12$).  This raises the question whether the upper bound of
$2n-4$ on the number of \acp{CPP} also holds for general symmetric
states. One indication in favour of this conjecture is given by
Euler's formula for convex polyhedra, which states that a convex
polyhedron with $n$ vertices can have at most $2n-4$ faces, with the
bound being strict \ac{iff} all faces are triangles.  This is
intriguing because our proof of \theoref{cpp_number_locations} is of a
very technical nature, where the number $2n-4$ arises in a seemingly
arbitrary fashion.  This hints at a deeper lying connection between
the faces spanned by the \acp{MP} and the number of local maxima
present in the spherical amplitude function $g ( \theta , \varphi
)$. We therefore formulate the following conjecture:
\begin{conjecture}\label{conjecture_cpp}
  With the exception of the Dicke states, every $n$ qubit symmetric
  state has at most $2n-4$ \acp{CPP}.
\end{conjecture}

\begin{table}
  \centering
  \caption[Number of CPPs and faces of highly entangled symmetric
  states]{\label{table4.4}
    The number of \acp{CPP} and polyhedral faces in the Majorana
    representation of the solutions or conjectured solutions are
    listed. The upper bound $2n-4$ must hold for the number of
    faces (due to Euler's formula) and for the number of \acp{CPP}
    (due to \protect\theoref{cpp_number_locations}) of
    $\ket{\Psi^{\text{pos}}_{n}}$. Entries are omitted where the
    underlying state is the same as the conjectured solution
    $\ket{\Psi_{n}}$ of the Majorana problem.}
  \begin{tabular}{l|cccccc}
    \toprule
    $n$ & \acp{CPP} $\ket{\psi_{n}^{\text{\totx}}}$ &
    \acp{CPP} $\ket{\psi_{n}^{\text{Th}}}$ &
    \acp{CPP} $\ket{\Psi^{\text{pos}}_{n}}$ &
    \acp{CPP} $\ket{\Psi_{n}}$ &
    faces $\ket{\Psi_{n}}$ & $2n-4$ \\
    \midrule
    4  &   &   &    & 4  & 4  & 4  \\
    5  & 3 & 3 &    & 5  & 5  & 6  \\
    6  &   &   &    & 8  & 8  & 8  \\
    7  & 3 &   &    & 10 & 10 & 10 \\
    8  & 2 & 2 &    & 10 & 10 & 12 \\
    9  & 3 & 3 &    & 10 & 10 & 14 \\
    10 & 2 & 8 & 3  & 8  & 16 & 16 \\
    11 & 1 & 2 & 2  & 11 & 16 & 18 \\
    12 &   &   & 15 & 20 & 20 & 20 \\
    \bottomrule
  \end{tabular}
\end{table}

What can we say about lower bounds on the number of \acp{CPP}?  For
maximally entangled symmetric $n$ qubit states \corref{numberofcpp}
predicts the existence of only two distinct \acp{CPP}, but our results
show that in general there is a considerably larger number of
\acp{CPP}, and that the \acp{CPP} tend to be well spread out over the
sphere.  This makes sense from the viewpoint of the necessary
condition outlined in \corref{numberofcpp}, namely that it must be
possible to write maximally entangled symmetric states as linear
combinations of their \acp{CPS}.

For $n=10,11$ the numerically determined maximally entangled positive
symmetric states do not exhibit a rotational symmetry. This is
somewhat surprising, because \lemref{pos_symm_cpp} implies that the
\acp{CPP} can then only lie on the positive half-circle of the
Majorana sphere, thus likely resulting in an imbalance of the
spherical amplitude function $g( \theta , \varphi )$. However, this
imbalance is only very weakly pronounced for the positive solutions of
$n=10,11$, with the non-global maxima of $g( \theta , \varphi )$
coming very close to the value at the \acp{CPP}. It was seen that both
$\ket{\Psi_{10}^{\text{pos}}}$ and $\ket{\Psi_{11}^{\text{pos}}}$ are
cast with only three nonvanishing basis states, and that
\theoref{theo_general_maj_rep} implies that shifting the \acp{MP} in a
way that each horizontal \ac{MP} ring assumes a rotational $Z$-axis
symmetry would result in four nonvanishing basis states.  It thus
seems that, at least for positive symmetric states, a lower number of
nonvanishing basis states is more favourable than Majorana
representations with certain symmetry features.

It was found that for $3 < n \leq 12$ the maximally entangled
symmetric $n$ qubit states are not Dicke states.  This result can be
easily extended to $n > 12$ by comparing the entanglement scaling of
the equally balanced Dicke state \eqref{dicke_ent_scaling} to e.g. the
superpositions of Dicke states shown in \tabref{table4.1}, or to the
entanglement scaling of the symmetric states defined for all $n$ in
\cite{Martin10}.  Since Dicke states are the only states whose
Majorana representation exhibits a continuous rotational symmetry, we
obtain the following result:
\begin{corollary}\label{discrete_cpp}
  For $n > 3$ the maximally entangled symmetric $n$ qubit states with
  respect to the geometric measure have only a finite number of
  \acp{CPP}.
\end{corollary}
This finding is interesting in light of the question raised in
\sect{gen_results}, namely whether maximally entangled states of
arbitrary multipartite systems have a discrete or continuous amount of
\acp{CPS} (see also Tamaryan \etal \cite{Tamaryan08,Tamaryan10}).  The
answer for the general (non-symmetric) case is not known, but the
investigation of the symmetric $n$ qubit case gives reason to believe
that for most multipartite systems the maximally entangled states have
only a finite number of distinct \acp{CPS}.

\cleardoublepage

\chapter{Classification of Symmetric State
  Entanglement}\label{classification}
 
\begin{quotation}
  In the previous chapter the entanglement of symmetric states was
  investigated primarily from a quantitative point of view. Now the
  focus shifts towards the qualitative characterisation of symmetric
  states. The concepts of \ac{LOCC} and \ac{SLOCC} equivalence are
  adapted to the symmetric case, and the \acf{DC}, an entanglement
  classification scheme specifically for symmetric states, is
  reviewed. It is found that \ac{SLOCC} operations between symmetric
  states are described by the \mob transformations of complex
  analysis. This allows for an intuitive visualisation, as well as
  practical uses such as the determination of whether two symmetric
  states belong to the same \ac{SLOCC} class.  The symmetric
  \ac{SLOCC} and \ac{DC} classes for up to five qubits are studied in
  detail, and representative states are derived for each entanglement
  class.  Connections are made to known \ac{SLOCC} invariants as well
  as related works, such as the \acfp{EF} \cite{Verstraete02} or
  alternative definitions of maximal entanglement
  \cite{Osterloh05,Gisin98}.
\end{quotation}

\section{Entanglement classification schemes for
  symmetric states}\label{overview_entclass}

The entanglement classification schemes \ac{LOCC} and \ac{SLOCC} were
already discussed in \sect{entanglement_classes}. In particular, it
was seen that \ac{SLOCC} equivalence gives rise to a coarser partition
than \ac{LOCC} equivalence in the sense that every \ac{LOCC} operation
is also an \ac{SLOCC} operation, but not vice versa.  The concepts of
\ac{LOCC} and \ac{SLOCC} equivalence are now applied to the subset of
symmetric states, and a comparison is made to the \acf{DC}
\cite{Bastin09}, an entanglement classification scheme designed
specifically for symmetric states.

\subsection{Symmetric LOCC and SLOCC
  operations}\label{symm_locc_slocc}

The condition for \ac{LOCC} equivalence between two arbitrary $n$
qudit states formulated in \eq{gen_locc_cond} is a special case of the
\ac{SLOCC} equivalence \eqref{gen_slocc_cond}.  The special linear
group $\text{SL}(d, \mbbc )$ contains all invertible $d \times d$
complex matrices with unit determinant, which explains why \ac{SLOCC}
operations are also known as \acfp{ILO} \cite{Dur00}.  Note that
$\text{SL}(d, \mbbc )$ contains $\text{SU}(d)$ as a subgroup.  In the
following we focus on the qubit case ($d = 2$) and on
permutation-symmetric states.

Given two symmetric $n$ qubit states $\psis$ and $\phis$, is there a
way to simplify \eq{gen_locc_cond} and \eq{gen_slocc_cond} to take
permutation-symmetry into account?  Mathonet \etal \cite{Mathonet10}
recently discovered that there always exists a symmetric \ac{ILO}
between two \ac{SLOCC}-equivalent symmetric states:
\begin{equation}\label{slocccond}
  \psis \stackrel{\text{SLOCC}}{\longleftrightarrow} \phis
  \quad \Longleftrightarrow \quad
  \exists \, \mathb \in \slc :
  \psis = \mathb^{\otimes n} \phis \ens .
\end{equation}
This statement is far from obvious, in a fashion that bears
resemblance to the existence of symmetric \acp{CPS} for all symmetric
$n$ qubit states.  Just as the result of H\"{u}bener \etal
\cite{Hubener09} greatly simplifies the quantitative determination of
the geometric entanglement of symmetric states, \eq{slocccond} greatly
simplifies the qualitative decision problem of whether two given $n$
qubit symmetric states belong to the same \ac{SLOCC} class or
not. Instead of considering arbitrary \acp{ILO} $\mathb_{1} \otimes
\cdots \otimes \mathb_{n} \in \slc^{\otimes n}$ with $6n$ real
\ac{d.f.}, it suffices to consider only the six \ac{d.f.} present in
$\slc$, regardless of the number of qubits.

Another similarity between the results of H\"{u}bener \etal
\cite{Hubener09} and Mathonet \etal \cite{Mathonet10} is that there
are exceptions to the converse statements. Regarding the first result,
while symmetric $n$ qubit states always possess at least one symmetric
\ac{CPS}, all the \acp{CPS} are necessarily symmetric only for $n \geq
3$ qubits \cite{Hubener09}. Regarding the second result, if two
symmetric $n$ qubit states are \ac{SLOCC}-equivalent, then there must
exist a symmetric \ac{ILO} between them, but there may also exist
non-symmetric \acp{ILO} \cite{Mathonet10}.  However, non-symmetric
\acp{ILO} between symmetric states exist only for states that belong
to the separable class, the W class and the \ac{GHZ} class. For $n
\geq 4$ qubits these three \ac{SLOCC} classes constitute only an
infinitesimal fraction in the set of all \ac{SLOCC} classes
\cite{Mathonet10}.

From the arguments in \cite{Mathonet10} it can be easily inferred that
\eq{slocccond} holds in analogous form for \ac{LOCC}
operations\footnote{This was also explicitly derived alongside an
  extension to mixed symmetric states in \protect\cite{Cenci11}.}:
\begin{equation}\label{locccond}
  \psis \stackrel{\text{LOCC}}{\longleftrightarrow} \phis
  \quad \Longleftrightarrow \quad
  \exists \, \matha \in \suc :
  \psis = \matha^{\otimes n} \phis \ens .
\end{equation}
This reduces the complexity of determining the \ac{LOCC}-equivalence
of two symmetric states from the $3n$ \ac{d.f.} present in $\matha_{1}
\otimes \cdots \otimes \matha_{n} \in \suc^{\otimes n}$ to the three
\ac{d.f.} of $\suc$.  The three real \ac{d.f.} present in $\suc$ were
already identified as the rotations \eqref{symmetric_lu} of the
\ac{MP} distribution on the Majorana sphere. Therefore \eq{locccond}
implies that two symmetric $n$ qubit states are \ac{LOCC}-equivalent
\ac{iff} their \ac{MP} distributions can be converted into each other
by a rotation on the Majorana sphere.

\begin{figure}
  \centering
  \includegraphics[scale=1.2]{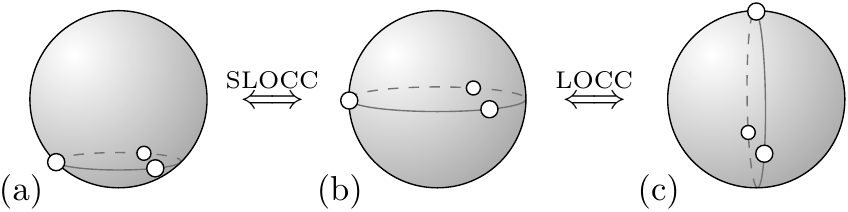}
  \caption[Majorana representations of LOCC and SLOCC equivalent
  states]{\label{ghztrafos} The \ac{MP} distributions of three
    \ac{GHZ}-type symmetric states of three qubits are shown. The
    \ac{GHZ}-state $\sym{0} + \sym{3}$ displayed in (b) is
    \ac{LOCC}-equivalent to the rotated \ac{GHZ}-state $\sym{0} +
    \sqrt{3} \sym{2}$ shown in (c).  The \ac{GHZ}-type state $\alpha
    \sym{0} + \beta \sym{3}$ in (a) is \ac{SLOCC}-equivalent, but not
    \ac{LOCC}-equivalent to the others.}
\end{figure}

Is it possible to make a similar operational statement with regard to
the \ac{MP} distribution of \ac{SLOCC}-equivalent symmetric states?
From \eq{majorana_definition} and \eq{slocccond} it is clear that
$\mathb \in \slc$ acts on each \ac{MP} individually.  Therefore, once
the action of $\slc$ on an individual Bloch vector is understood, one
automatically understands how \ac{MP} distributions transform under
the action of symmetric \ac{SLOCC} operations.  Because of $\suc
\subset \slc$, three of the six \ac{d.f.} of the special linear group
$\slc$ can be identified as the usual rotations on the Bloch sphere.
From mathematics it is known that the Lie group $\slc$ is a double
cover of the \mob group, the automorphism group on the Riemann sphere.
Because of this, the transformation of the \acp{MP} under a symmetric
\ac{SLOCC} operation is described by a \textbf{\mob transformation} of
complex analysis, with the Majorana sphere in lieu of the Riemann
sphere. The \mob transformations will be covered in detail in
\sect{mobius}, and here we only present the example in
\fig{ghztrafos}, showing three \ac{GHZ}-type states that are \ac{LOCC}
or \ac{SLOCC}-equivalent to each other.

\subsection{Degeneracy configuration}\label{degeneracyconfig}

The \textbf{\acf{DC}} is an entanglement classification scheme
introduced specifically for $n$ qubit symmetric states
\cite{Bastin09}.  Its definition incorporates the Majorana
representation by counting the number of identical \acp{MP} of a given
symmetric state.  Each $n$ qubit symmetric state belongs to exactly
one \textbf{\ac{DC} class} $\mathd_{n_1 , \ldots , n_d}$ with $n = n_1
+ \ldots + n_d$ ($n_1 \geq \ldots \geq n_d$), and where $n_1$ stands
for the number of \acp{MP} coinciding on one point of the Bloch
sphere, $n_2$ for those coinciding at a different point, and so on.
We call the $n_{i}$ the \textbf{degeneracy degrees}, and the number
$d$ the \textbf{diversity degree}. The number of different \ac{DC}
classes into which the Hilbert space of $n$ qubits is partitioned is
given by the partition function $p(n)$.  The usefulness of the concept
of \ac{DC} classes can be seen from the fact that, due to the
non-singular nature of \acp{ILO}, the \ac{MP} degeneracy of a given
symmetric state remains invariant under symmetric \ac{SLOCC}
operations: $\ket{\phi_{i}} = \ket{\phi_{j}} \Leftrightarrow \mathb
\ket{\phi_{i}} = \mathb \ket{\phi_{j}}$ for all $\ket{\phi_{i}} ,
\ket{\phi_{j}} \in \mbbc^{2}$ and all $\mathb \in \slc$.  On the other
hand, two states that belong to the same \ac{DC} class do not
necessarily belong to the same \ac{SLOCC} class \cite{Bastin09}.  Thus
we arrive at the following refinement hierarchy:
\begin{theorem}\label{hierarchy}
  The symmetric subspace of every \emph{n} qubit Hilbert space has the
  following refinement hierarchy of entanglement partitions:
  \begin{equation}
    \text{LOCC} \leq \text{SLOCC} \leq \text{DC} \ens .
  \end{equation}
\end{theorem}

An obvious advantage of \ac{DC} classes over \ac{SLOCC} classes is
that the number of entanglement classes remains finite for arbitrary
$n$. This is in stark contrast to the number of \ac{SLOCC} classes,
which becomes infinite for $n \geq 4$ qubits, even when considering
only the symmetric subset.  Furthermore, operational implications have
been found for the concept of \ac{DC} classes: Each \ac{DC} class can
be unambiguously associated with specific parameter configurations in
experiments \cite{Bastin09}.

\section{\mob transformations}\label{mobius}

As mentioned in the previous section, \ac{SLOCC} operations between
multiqubit symmetric states can be understood by means of the \mob
transformations from complex analysis.  This intriguing link was
independently discovered and described\footnote{Unbeknownst to me as
  well as to Ribeiro and Mosseri during the writing of our
  manuscripts, some partial properties were already discovered by
  Kolenderski \protect\cite{Kolenderski10}. In that paper the effect
  of $\text{GL}(2, \mbbc )$ operations on Bloch vectors is described,
  but the connection to the \protect\mob transformations of complex
  analysis is not made.}  by me \cite{Aulbach11} and by Ribeiro and
Mosseri \cite{Ribeiro11}.  First, the definition of \mob
transformations is recapitulated, and then the transformations are
employed to analyse and visualise the freedoms present in symmetric
\ac{SLOCC} operations.

\subsection{Introduction}\label{mobius_def}

The \mob transformations are defined in complex analysis as the
bijective holomorphic\footnote{A complex-valued function of a complex
  variable is holomorphic if it is complex differentiable everywhere
  on its domain.  Complex differentiability is a very strong
  requirement, resulting in many fascinating properties of holomorphic
  functions.}  functions that project the extended complex plane
$\cext = \mbbc \cup \{ \infty \}$ onto itself \cite{Forst}.  These
isomorphic functions $f : \cext \to \cext$ take the form of rational
functions
\begin{equation}\label{mobiusform}
  f(z) = \frac{az + b}{cz + d} \ens ,
\end{equation}
with $a, b, c, d \in \mbbc$, and $ad - bc \neq 0$. The latter
condition ensures that $f$ is invertible. In the case $c \neq 0$ the
domain of $f$ is $\mbbc \backslash \{ - \tfra{d}{c} \}$ and the
codomain is $\mbbc \backslash \{ \tfra{a}{c} \}$, while for $c = 0$
both the domain and codomain are $\mbbc$. The extension to a bijective
mapping $f : \cext \to \cext$ is mediated by $f(- \tfra{d}{c}) :=
\infty$, $f ( \infty ) := \tfra{a}{c}$ for $c \neq 0$, and $f( \infty
) := \infty$ for $c = 0$. The coefficients give rise to the matrix
representation $\mathb = \left(
  \begin{smallmatrix}
    a & b \\
    c & d
  \end{smallmatrix}
\right)$ of the \mob group, and from \eq{mobiusform} it is clear that
it suffices to consider those $\mathb$ with determinant one (i.e., $ad
- bc = 1$).  Since $+\mathb$ and $-\mathb$ describe the same
transformation $f (z)$, the \mob group is isomorphic to the projective
special linear group $\pslc = \slc / \{ \pm \one \}$.

As outlined in \sect{stereo_proj}, all points of $\cext$ can be
projected onto the Riemann sphere by means of an inverse stereographic
projection.  With this projection the roots $\{ z_1 , \ldots , z_n \}$
of the Majorana polynomial \eqref{majpoly} are projected to the
surface of the Majorana sphere, where they become the \acp{MP}.  The
action of a \mob transformation $f : \cext \to \cext$ on the roots $\{
z_1 , \ldots , z_n \}$ in the extended complex plane then translates
on the Majorana sphere to a generalised rotation $\mathb \in \slc$
acting on each \ac{MP}, which is precisely a symmetric \ac{SLOCC}
operation of the form \eqref{slocccond}.  We can therefore view the
\mob transformations (or equivalently symmetric \ac{SLOCC} operations)
either as automorphisms on $\cext$ or as automorphisms on
$\mathcal{S}^2$, with the isomorphism between these two manifolds
described by the stereographic projection.

\begin{figure}
  \centering
  \includegraphics[scale=1.1]{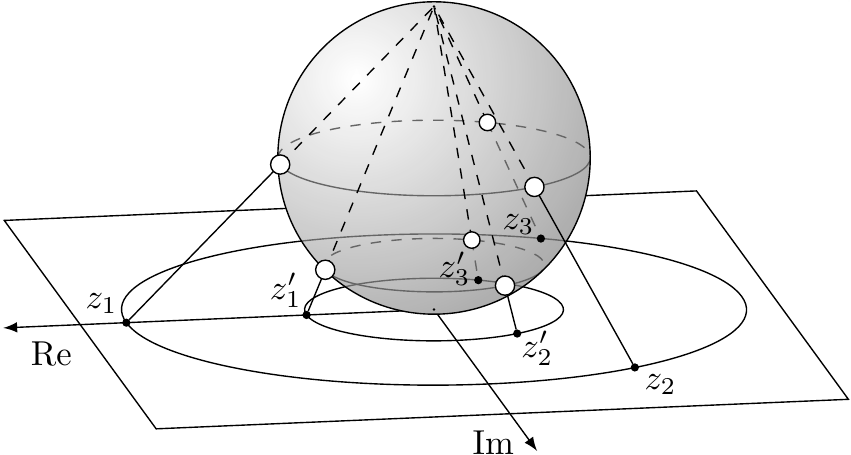}
  \caption[Stereographic projection of MPs]{\label{sterproj_ghz} A
    stereographic projection through the north pole of the Majorana
    sphere mediates between the Majorana roots in the complex plane
    and the \acp{MP} on the surface of the sphere. The \ac{SLOCC}
    operation of \protect\fig{ghztrafos} is facilitated by the
    transformation $f(z) = \tfra{z}{2}$ which maps the set of roots
    $\{ z_1 , z_2 , z_3 \}$ onto the set $\{ z'_1 , z'_2 , z'_3 \}$,
    thus lowering the ring of \acp{MP}.}
\end{figure}

As an example, \fig{sterproj_ghz} shows the action of the \mob
transformation $f(z) = \tfra{z}{2}$ which transforms the \acp{MP} of
the distribution shown in \fig{ghztrafos}(b) into that of
\fig{ghztrafos}(a).  It can be seen that circles remain circles under
the action of this transformation, both on the sphere and in the
complex plane.  Intriguingly, this property holds for all \mob
transformations: Both on the Riemann sphere and in the complex plane
circles are projected onto circles, where we consider straight lines
in the complex plane to be circles too \cite{Knopp}. Furthermore,
angles are preserved under \mob transformations, i.e. two lines or
circles that meet at an angle $\alpha$ will still meet at an angle
$\alpha$ after the transformation.  These properties\footnote{It is
  said that one picture is worth a thousand words, and this is
  probably even more true for a video.  To gain a good understanding
  of the \protect\mob transformations it is recommended to watch the
  beautiful video clip of Arnold \protect\etal \protect\cite{Arnold}
  which has featured in a visualisation competition of Science
  magazine.} become more understandable when taking into account that
every \mob transformation \eqref{mobiusform} can be composed from the
following elementary operations \cite{Forst}:
\begin{itemize}
\item Rotation \& Dilation: \quad $z \longmapsto a z \ens ,$
  \hspace{7.05mm} with $a \in \mbbc \backslash \{ 0 \}$.
  \begin{enumerate}
  \item[i)] Rotation: \hspace{11.45mm} $z \longmapsto \E^{\I \varphi}
    z \ens ,$ \hspace{4.4mm} with $\varphi \in \mbbr$.
    
  \item[ii)] Dilation: \hspace{12.35mm} $z \longmapsto r z \ens ,$
    \hspace{7.2mm} with $r > 0$.
  \end{enumerate}
  
\item Translation: \hspace{15.55mm} $z \longmapsto z + b \ens ,$ \quad
  with $b \in \mbbc$.
  
\item Inversion: \hspace{18.7mm} $z \longmapsto \frac{1}{z} \ens .$
\end{itemize}

\mob transformations can be categorised into different types,
depending on the values of the trace and eigenvalues of the
transformation matrix $\mathb$.  There exist \emph{parabolic,
  elliptic, hyperbolic} and \emph{loxodromic} \mob transformations
\cite{Knopp}, but a unifying feature is that two not necessarily
antipodal or distinct points on the Riemann sphere are left invariant.
This generalises the $\suc$ rotations, where the two invariant points
are the intersections of the rotation axis with the sphere.  As an
example, the \ac{SLOCC} operation shown in \fig{sterproj_ghz} is
mediated by a hyperbolic \mob transformation.  These transformations
are characterised by the two invariant points (here the north and
south pole) acting as attractive and repulsive centres, with the
\acp{MP} moving away from the repulsive centre towards the attractive
one.

A well-known property of \mob transformations is that for any two
ordered sets of three pairwise distinct points $\{ v_1 , v_2 , v_3 \}$
and $\{ w_1 , w_2 , w_3 \}$ there always exists exactly one \mob
transformation that maps one set to the other \cite{Knopp}.  This is
in general not possible for two sets of four pairwise distinct points,
but the \textbf{cross-ratio preservation} of \mob transformations
\cite{Knopp} can be employed to derive a necessary and sufficient
condition: An ordered quadruple of distinct complex numbers $\{ v_1,
v_2, v_3, v_4 \}$ can be projected onto another quadruple $\{ w_1,
w_2, w_3, w_4 \}$ by a \mob transformation \ac{iff}
\begin{equation}\label{cross-ratio}
  \frac{(v_1 - v_3)(v_2 - v_4)}{(v_2 - v_3)(v_1 - v_4)} = 
  \frac{(w_1 - w_3)(w_2 - w_4)}{(w_2 - w_3)(w_1 - w_4)} \ens .
\end{equation}

\subsection{Relationship to SLOCC operations}\label{slocc_mob_link}

From the preceding introduction of the \mob transformations and
\eq{slocccond} the following theorem is clear.
\begin{theorem}\label{slocc_mob}
  Two symmetric \emph{n} qubit states $\psis$ and $\phis$ are
  \ac{SLOCC}-equivalent \ac{iff} there exists a \mob transformation
  \eqref{mobiusform} between their Majorana roots.
\end{theorem}

How to determine whether such a \mob transformation exists?
Naturally, $\psis$ and $\phis$ must belong to the same \ac{DC} class,
as \ac{SLOCC} equivalence is a refinement of \ac{DC} equivalence.  One
crucial property of \mob transformations in this regard is that any
set of three pairwise distinct points can be projected onto any other.
This immediately leads to the following important result, first
described in \cite{Bastin09}.
\begin{corollary}\label{dc_three_points}
  If two symmetric \emph{n} qubit states $\psis$ and $\phis$ belong to
  the same \ac{DC} class $\mathd_{n_1 , \ldots , n_d}$ with diversity
  degree $d \leq 3$, then they are \ac{SLOCC}-equivalent.
\end{corollary}

This corollary implies that \ac{DC} classes with a diversity degree of
three or less consist of a single \ac{SLOCC} class. In particular,
this means that for two and three qubit systems the partition into
\ac{SLOCC} classes is the same as the partition into \ac{DC} classes.
The reverse of \corref{dc_three_points} clearly does not hold in
general, and for states with diversity degree $d=4$ the cross-ratio
preservation \eqref{cross-ratio} yields the following result:
\begin{corollary}\label{dc_four_points}
  Let $\psis$ and $\phis$ be two symmetric \emph{n} qubit states that
  belong to the same \ac{DC} class $\mathd_{n_1 , \ldots , n_4}$ with
  diversity degree $d = 4$. The Majorana roots of $\psis$ and $\phis$
  corresponding to the degeneracy $n_{i}$ are labelled $v_{i}$ and
  $w_{i}$, respectively. If the $v_{i}$ and $w_{i}$ fulfil
  \eq{cross-ratio}, then $\psis$ and $\phis$ are
  \ac{SLOCC}-equivalent.
\end{corollary}

Note that the ordering of the roots has to be taken into account,
because \eq{cross-ratio} does in general not remain true under
permutations, and because Majorana roots corresponding to degenerated
\acp{MP} have to be projected onto Majorana roots with the same
degeneracy.  In case of a \ac{DC} class that contains the same
degeneracy degree several times\footnote{This is the case for all
  \ac{DC} classes of up to and including 9 qubits, as the 10 qubit
  class $\mathd_{4 , 3 , 2 , 1}$ is the first \ac{DC} class with four
  different degeneracy degrees.}, i.e. $n_{i} = n_{j}$ for some $i
\neq j$, there is obviously more than one way to designate the indices
$n_{i}$ to the roots in decreasing order, and \eq{cross-ratio} needs
to hold only for one such ordering to obtain \ac{SLOCC}-equivalence.

Next we investigate the constituents of \ac{SLOCC} operations, and
identify the freedoms that do not correspond to \ac{LOCC} operations.
The \mob transformations \eqref{mobiusform} have six real \ac{d.f.},
and are described by $\text{SL}(2, \mbbc )$.  The polar decomposition
of linear algebra states that every invertible complex matrix can be
uniquely decomposed into a unitary matrix and a positive-semidefinite
Hermitian matrix \cite{Golub}.  We use this result to factorise the
\ac{d.f.} that genuinely belong to \ac{SLOCC} operations (i.e., which
cannot be realized by \ac{LOCC} operations), and show that this
factorisation corresponds to a clear and intuitive visualisation with
the Majorana sphere.

\begin{theorem}\label{decomp}
  Every \ac{SLOCC} operation between two symmetric $n$ qubit states
  can be factorised into an affine \mob transformation of the form
  \begin{equation}\label{mobiusform2}
    \widetilde{f}(z) = A z + B \ens , \quad
    \text{with} \quad A > 0 \ens , \ens B \in \mbbc \ens ,
  \end{equation}
  and an \ac{LOCC} operation. This decomposition is unique, and the
  set of transformations \eqref{mobiusform2} forms a group that is
  isomorphic to $\slc / \suc$.
\end{theorem}

\begin{proof}
  First, the existence of a factorisation of each \ac{SLOCC} operation
  into a transformation $\widetilde{f}$ of the form
  \eqref{mobiusform2} and an \ac{LOCC} operation is shown.  For each
  $\mathb = \left(
    \begin{smallmatrix}
      a & b \\
      c & d
    \end{smallmatrix}
  \right) \in \slc$ we define $\widetilde{\mathb} = \lambda \mathb$
  with $\lambda = \sqrt{a \cc{a} + c \cc{c}} > 0$. Since
  $\widetilde{\mathb}$ describes the same \ac{SLOCC} operation as
  $\mathb$, it suffices to show that $\widetilde{\mathb}$ can be
  decomposed into an \ac{LOCC} operation $\matha \in \suc$ and a \mob
  transformation of the form \eqref{mobiusform2}:
  \begin{equation*}
    \begin{pmatrix}
      \lambda a & \lambda b \\
      \lambda c & \lambda d
    \end{pmatrix} =
    \begin{pmatrix}
      \alpha & - \cc{\beta} \\
      \beta & \cc{\alpha}
    \end{pmatrix}
    \otimes
    \begin{pmatrix}
      A & B \\
      0 & 1
    \end{pmatrix} \ens ,
  \end{equation*}
  with $A > 0$ and $\alpha , \beta , B \in \mbbc$, $\alpha \cc{\alpha}
  + \beta \cc{\beta} = 1$.  For given parameters $a,b,c,d \in \mbbc$
  with $ad - bc = 1$, this is fulfilled for $\alpha =
  \frac{a}{\lambda}$, $\beta = \frac{c}{\lambda}$, $A = \lambda^{2}$
  and $B = \frac{\lambda^{2} b + \cc{c}}{a} = \frac{\lambda^{2} d -
    \cc{a}}{c}$ (for $a=0$ or $c=0$ only one of the two identities
  holds). This proves the existence of a factorisation.

  To show the uniqueness of factorisations, it is assumed that a given
  \ac{SLOCC} operation $\mathb \in \slc$ can be factorised, up to
  scalar prefactors $\lambda_{1} , \lambda_{2} \in \mbbc \backslash \{
  0 \}$, in the above way by two sets of parameters $\{ \alpha_{1} ,
  \beta_{1} , A_{1} , B_{1} \}$ and $\{ \alpha_{2} , \beta_{2} , A_{2}
  , B_{2} \}$. Elimination of $\mathb$ from the resulting matrix
  equations yields the condition
  \begin{equation*}
    \frac{\lambda_2}{\lambda_1}
    \begin{pmatrix}
      \alpha_1 & - \cc{\beta}_1 \\
      \beta_1 & \cc{\alpha}_1
    \end{pmatrix}
    \otimes
    \begin{pmatrix}
      A_1 & B_1 \\
      0 & 1
    \end{pmatrix} =
    \begin{pmatrix}
      \alpha_2 & - \cc{\beta}_2 \\
      \beta_2 & \cc{\alpha}_2
    \end{pmatrix}
    \otimes
    \begin{pmatrix}
      A_2 & B_2 \\
      0 & 1
    \end{pmatrix}
     .
  \end{equation*}
  A straightforward calculation yields $\big|
  \frac{\lambda_{2}}{\lambda_{1}} \big| = 1$, and from this it readily
  follows that the two sets of parameters must coincide.  This
  uniqueness implies that the set of transformations $\widetilde{f}$
  is isomorphic to $\slc / \suc$, and their group properties are
  easily verified explicitly.
\end{proof}

\Theoref{decomp} is closely related to the polar decomposition of
matrices, which has also been mentioned in connection with the Bloch
sphere in \cite{Kolenderski10}.  However, while the matrices
describing the affine transformations $\widetilde{f}$ are positive,
they are in general not Hermitian (unlike in the polar decomposition),
and the introduction of the prefactor $\lambda$ in the proof is
necessary because $\matha$ and $\mathb$ are defined to have unit
determinants.

\begin{figure}
  \centering
  \includegraphics[scale=1.3]{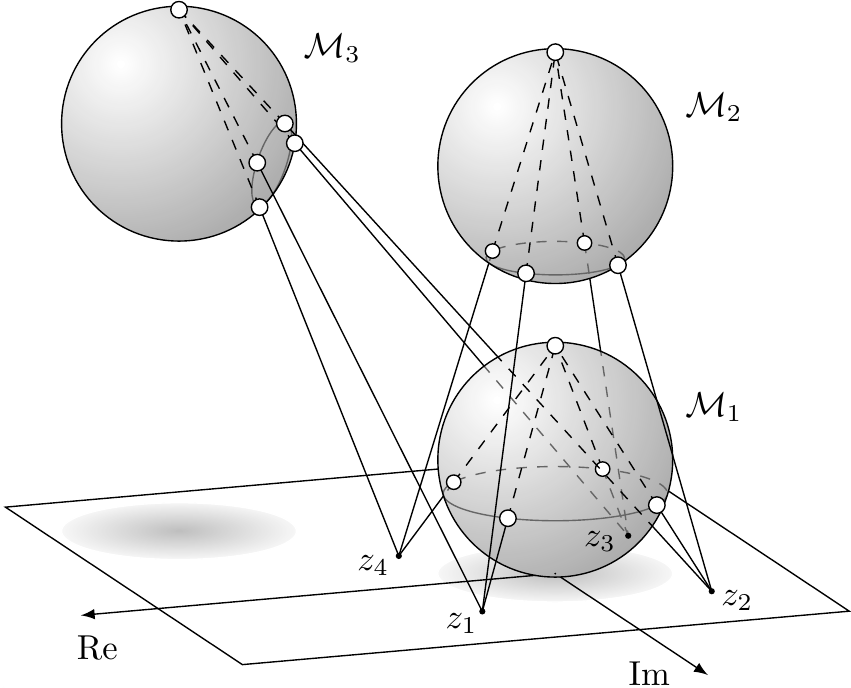}
  \caption[Alternative visualisation of \mob
  transformations]{\label{mob_translations} Alternative visualisation
    of \protect\mob transformations where a fixed set of complex
    points is projected onto the surface of a moving sphere.  The
    three innate freedoms of \ac{SLOCC} operations not present in
    \ac{LOCC} operations are then described by the translations of the
    Majorana sphere in $\mbbrr$.  The north pole of sphere
    $\mathm_{1}$ (with the \ac{MP} distribution of the five qubit
    \protect\quo{square pyramid state}) lies 2 units above the origin
    of the complex plane, while the one of $\mathm_{2}$ lies 5 units
    above, and $\mathm_{3}$ is additionally displaced horizontally by
    a vector $5 - 5 \I$.  The parameters $( A , B )$ of
    \protect\eq{mobiusform2} for the transformation of $\mathm_{1}$ to
    $\mathm_{2}$ and $\mathm_{3}$ are $( \tfra{5}{2} , 0 )$ and $(
    \tfra{5}{2} , 5 - 5 \I )$, respectively.}
\end{figure}

The orthodox way to visualise \mob transformations is to fix the
Riemann sphere in $\mbbrr$, and points $\{ z_1 , \ldots , z_n \}$ on
the complex plane are transformed to different points $\{ z'_1 ,
\ldots , z'_n \}$ under the action of \eqref{mobiusform}.
Alternatively, the points in the plane can be considered fixed, and
instead the Riemann sphere moves in $\mbbrr$, as shown in
\fig{mob_translations}.  The six \ac{d.f.} of the \mob transformations
are then split into three translational freedoms (movement of sphere
in $\mbbrr$) and three rotational freedoms ($\text{SU}(2)$-rotations
of sphere).  By considering these elementary operations it can be
verified by calculation that this is an equivalent way of viewing the
change of \acp{MP} on the sphere under the action of \mob
transformations.  A general \ac{SLOCC} operation between two symmetric
states is then described by a translation of the Majorana sphere in
$\mbbrr$, followed by a rotation.  In this approach the affine
transformations \eqref{mobiusform2} exactly describe the set of
translations in $\mbbrr$ that leave the sphere's north pole above the
complex plane.  The parameters of the affine function $\widetilde{f}
(z) = A z + B$ correspond to the translation as follows: The parameter
$A = \tfra{h_{2}}{h_{1}}$ is the ratio of the heights of the north
pole before ($h_{1}$) and after ($h_{2}$) the transformation, and $B$
is the horizontal displacement vector (cf. \fig{mob_translations}).
Regarding the subsequent rotation of the Majorana sphere, it is clear
that it leaves the relative \ac{MP} distribution invariant and can be
described by an \ac{LOCC} operation.

\section{Representative states for SLOCC classes}\label{representative}

In the following the \ac{SLOCC} and \ac{DC} classes of symmetric
states of up to $5$ qubits are characterised, and representative
states are given for each equivalence class. The aim is to provide
representations that are not only unique (i.e. the representative
states are all inequivalent to each other), but that also allow for a
simple analytical form as well as simple \ac{MP} distributions.
Before tackling this problem, we briefly investigate the relationship
between general and symmetric states under \ac{SLOCC} operations.  For
example, if it were possible to transform every nonsymmetric state
into a symmetric state by \ac{SLOCC}, then the restriction to
symmetric states would be merely an artificial one, because all states
in $\mathh$ could be represented by symmetric states.

\subsection{Relationship between symmetric and
  nonsymmetric states}\label{symm_nonsymm_relationship}

From a comparison of parameters it can be easily seen that the
aforementioned symmetrisation of generic states by \ac{SLOCC}
operations is a rare exception: Unnormalised pure states of $n$ qubits
are described by $2^n$ complex coefficients, and taking the global
phase into account, this leads to $2^{n+1} - 1$ independent real
\acf{d.f.}.  General \ac{SLOCC} operations (which include \ac{LU}
operations that can be associated with basis transformations and
standard forms) are described by $\text{SL}(2, \mbbc )^{n}$ and have
$6n$ real \ac{d.f.}, so the number of remaining independent \ac{d.f.}
is $2^{n+1} - 6n - 1$.  On the other hand, unnormalised symmetric $n$
qubit states are described by $n+1$ Dicke states, yielding $2(n+1) -
1$ independent real freedoms.  Since $2^{n+1} - 6n - 1 \leq 2(n+1) -
1$ holds only for $n \leq 4$ qubits, it is clear that generic states
of five and more qubits cannot be symmetrised by \ac{SLOCC}.  In
\sect{entanglement_classes} the \ac{SLOCC} equivalence classes of
systems with up to four qubits were already reviewed, and we will now
follow up on this by investigating whether these equivalence classes
contain symmetric states.

First, we note that the symmetrisation of the \ac{SLOCC} class of $n$
qubit product states is trivial, because every product state can be
turned into a symmetric state (e.g. $\ket{0}^{\otimes n}$) by an
\ac{LU} operation.  This \ac{SLOCC} class coincides with the \ac{DC}
class $\mathd_{n}$.

The other extreme with regard to symmetrisation are the \ac{SLOCC}
classes with states that are neither product states nor entangled over
all parties. As an example, the three qubit state
$\ket{\psi}_{\text{A-BC}} = \ket{000} + \ket{011} = \ket{0} \otimes
(\ket{00} + \ket{11})$ is \textbf{biseparable}, because the first
qubit is not entangled with the rest. This inherently asymmetric
property cannot be lifted by \ac{SLOCC} operations, since local
operations cannot create entanglement over disentangled parts. Thus
the \ac{SLOCC} class to which $\ket{\psi}_{\text{A-BC}}$ belongs does
not contain a single symmetric state.  As symmetric states are either
fully separable or entangled over all parties \cite{Ichikawa08}, this
implies that for three and more qubits there always exist
fundamentally nonsymmetric \ac{SLOCC} classes.

Every pure bipartite state (which includes the 2 qubit case) can be
cast as a symmetric state by means of the Schmidt decomposition
\eqref{schmidt_decomp}, which means that every state is
\ac{LU}-equivalent to a symmetric state. For 2 qubits there are two
\ac{SLOCC} equivalence classes, containing the separable and entangled
states, respectively.  Choosing the states $\sym{0}$ and $\sym{1}$ as
representatives of these \ac{SLOCC} classes, we display their Majorana
representations in \fig{slocc23}.

\begin{figure}
  \centering
  \includegraphics[scale=1.3]{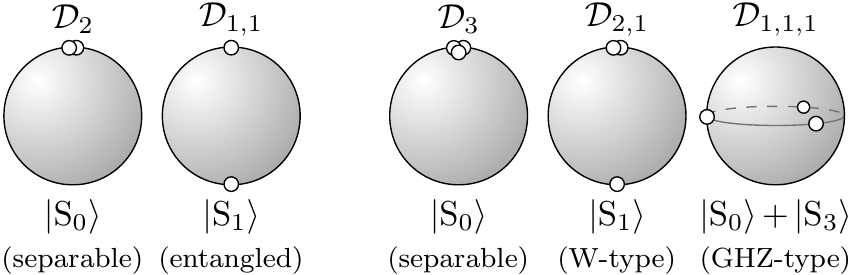}
  \caption[DC classes of 2 and 3 qubit states]{\label{slocc23} The
    \ac{SLOCC} and \ac{DC} classes of 2 and 3 qubit symmetric states
    are represented by the \ac{MP} distributions of characteristic
    states. Because of \protect\corref{dc_three_points} each \ac{DC}
    class consists of a single \ac{SLOCC} class.  Excluding the
    biseparable 3 qubit states, every 2 or 3 qubit state can be
    transformed into one of the symmetric states shown here by a
    (generally nonsymmetric) \ac{SLOCC} operation.}
\end{figure}

For 3 qubits there exist six different \ac{SLOCC} classes, the
separable class, the three biseparable classes, and the two
inequivalent classes with \ac{GHZ}-type and W-type entanglement
\cite{Dur00,Acin00}. All states of the latter two classes are
\ac{SLOCC}-equivalent to the $\ket{\text{GHZ}}$ and $\ket{\text{W}}$
state, respectively \cite{Dur00}.  Therefore, with the exception of
the biseparable states, every three qubit state can be turned into a
symmetric state by \ac{SLOCC}.  In \fig{slocc23} the three symmetric
\ac{SLOCC} classes are represented by the states $\sym{0}$, $\sym{1}$
and $\sym{0} + \sym{3}$.  From \corref{dc_three_points} it is clear
that the \ac{DC} classes coincide with the \ac{SLOCC} classes, with
$\mathd_{3}$ containing the separable states, $\mathd_{2,1}$ the
W-type entangled states and $\mathd_{1,1,1}$ the \ac{GHZ}-type
entangled states.

For 4 qubits the number of \ac{SLOCC} classes becomes infinite
\cite{Dur00}, even when considering only the subset of symmetric
states.  The symmetric entanglement classes of 4 qubits will be
investigated in detail in the next section, and in \sect{families}
these classes will be linked to the \acfp{EF} of Verstraete \etal
\cite{Verstraete02}.

\subsection{Four qubit symmetric classes}\label{rep_4_qubit}

For symmetric states of 4 qubits there exist five \ac{DC} classes and
a continuum of \ac{SLOCC} classes \cite{Bastin09}.  As shown in
\fig{slocc4}, four of the \ac{DC} classes coincide with \ac{SLOCC}
classes (which is clear from \corref{dc_three_points}), while the
generic class $\mathd_{1,1,1,1}$ with no \ac{MP} degeneracy is
comprised of a continuum of \ac{SLOCC} classes (cf. Figure 2 in
\cite{Markham11}).  We will now parameterise this continuum in a way
that exactly one state, acting as the \textbf{representative state},
is chosen from every \ac{SLOCC} class\footnote{The uniqueness implied
  by \protect\quo{exactly one} is an improvement over alternative
  representations such as the \acp{EF} \protect\cite{Verstraete02}
  where some of the representative states are \ac{SLOCC}-equivalent to
  each other.}.  The high symmetry present in an equidistant
distribution of three \acp{MP} along the equator facilitates the
restriction of the remaining \ac{MP} to a well-defined, connected area
on the sphere's surface:

\begin{theorem}\label{theorem4qubit}
  Every pure symmetric state of 4 qubits is \ac{SLOCC}-equivalent to
  exactly one state of the set
  \begin{gather*}
    \{ \sym{0} , \, \sym{1} , \, \sym{2} , \, 2 \sym{0} + t \sym{1} +
    \sym{3} + 2t \sym{4} \} \ens ,
    \\[0.75em]
    \text{with} \quad t = \E^{\I \varphi} \tan \tfra{\theta}{2} \ens ,
    \quad \text{and} \quad ( \theta , \varphi ) \in \bmr{\{} [0,
    \tfra{\pi}{2}) \times [0, \tfra{2 \pi}{3}) \bmr{\}} \cup \bmr{\{}
    \{ \tfra{\pi}{2} \} \times [0 , \tfra{\pi}{3}] \bmr{\}} \ens .
  \end{gather*}
\end{theorem}

\begin{proof}
  First it will be shown that every symmetric 4 qubit state $\psis$
  can be transformed by \ac{SLOCC} into one of the above states.  From
  the previous discussion and \fig{slocc4}, this is clear for all
  \ac{DC} classes except $\mathd_{1,1,1,1}$.  Given an arbitrary state
  $\psis \in \mathd_{1,1,1,1}$, there always exists a \mob
  transformation $f: \psis \mapsto \ket{\psi '^{\text{s}}}$ s.t.
  three \acp{MP} are projected onto the three corners of an
  equilateral triangle in the equatorial plane. If the fourth \ac{MP}
  $\ket{\phi_4}$ is not projected into the area parameterised by $(
  \theta , \varphi ) \in \bmr{\{} [0, \tfra{\pi}{2}) \times [0,
  \tfra{2 \pi}{3}) \bmr{\}} \cup \bmr{\{} \{ \tfra{\pi}{2} \} \times
  (0 , \tfra{\pi}{3}] \bmr{\}}$ (cf. \fig{slocc4}), then it can be
  projected into that area by a combination of $\{ \rotxs ( \pi ) ,
  \rotzs ( \tfra{2 \pi}{3} ) \}$-rotations which preserve the
  equatorial \ac{MP} distribution.

  It remains to show that this set of states is unique, i.e., two
  different \acp{MP} $\ket{\phi_4}$ and $\ket{\phi_4 '}$ within the
  aforementioned parameter range give rise to two different states
  $\psis \neq \ket{\psi '^{\text{s}}}$ which are
  \ac{SLOCC}-inequivalent.  By considering all 4! possible projections
  between the \acp{MP} of $\psis$ and $\ket{\psi '^{\text{s}}}$ it can
  be easily verified explicitly with the cross-ratio preservation
  \eqref{cross-ratio} that a transformation is possible only if
  $\ket{\phi_4} = \ket{\phi_4 '}$.
\end{proof}

\begin{figure}
  \centering
  \includegraphics[scale=1.3]{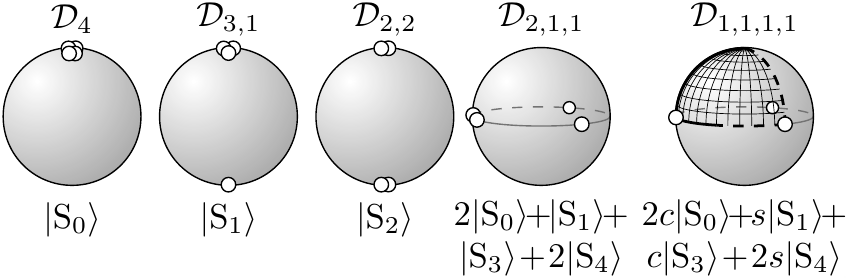}
  \caption[DC classes of 4 qubit states]{\label{slocc4} Only four of
    the five \ac{DC} classes of 4 qubit symmetric states coincide with
    a single \ac{SLOCC} class. Due to the continuum of \ac{SLOCC}
    classes present in $\mathd_{1,1,1,1}$, only three \acp{MP} can be
    fixed in its representative state, with the unique locations for
    the fourth \ac{MP} $c \ket{0} + s \ket{1}$ parameterising the set
    of representative states.  Here $c = \cos \tfra{\theta}{2}$ and $s
    = \E^{\I \varphi} \sin \tfra{\theta}{2}$, and the range of
    parameters is $( \theta , \varphi ) \in \bmr{ \{ } [0,
    \tfra{\pi}{2}) \times [0, \tfra{2 \pi}{3}) \bmr{ \} } \cup \bmr{
      \{ } \{ \tfra{\pi}{2} \} \times (0 , \tfra{\pi}{3}] \bmr{ \} }$,
    shown as a black grid.  The fixed equatorial \acp{MP} of the
    representative states are equidistantly spaced.}
\end{figure}

\subsection{Five qubit symmetric classes}\label{rep_5_qubit}

The \ac{DC} classes of 5 qubits and representative states for the
\ac{SLOCC} classes can be seen in \fig{slocc5}. The \ac{SLOCC} classes
of the generic class $\mathd_{1,1,1,1,1}$ can be parameterised by two
complex variables, corresponding to two \acp{MP} in the black and
white area, respectively.  Unlike the 4 qubit case, however, this
parameterisation is neither unique, nor confined to the generic
\ac{DC} class.  Different sets of parameters $( \theta_1 , \varphi_1 ,
\theta_2 , \varphi_2 ) \neq ( \theta'_1 , \varphi'_1 , \theta'_2 ,
\varphi'_2 )$ can give rise to \ac{SLOCC}-equivalent states, and for
$( \theta_1 , \varphi_1 ) = ( \theta_2 , \varphi_2 )$ the
corresponding state does not even belong to $\mathd_{1,1,1,1,1}$
because of an \acp{MP} degeneracy.  A unique set of representative
states can therefore be provided only for the subset of symmetric
states with at least one \ac{MP} degeneracy:

\begin{theorem}\label{theorem5qubitdeg}
  Every pure symmetric state of 5 qubits with an \ac{MP} degeneracy
  (i.e., diversity degree $d < 5$) is \ac{SLOCC}-equivalent to exactly
  one state of the set
  \begin{gather*}
    \{ \sym{0} , \, \sym{1} , \, \sym{2} , \, \sqrt{10} \left( \sym{0}
      + t \sym{5} \right) + t \sym{2} + \sym{3} + \sqrt{2} \left( 1 +
      t \right) \left( \sym{1} + \sym{4} \right) \} \ens ,
    \\[0.75em]
    \text{with} \quad t = \E^{\I \varphi} \tan \tfra{\theta}{2} \ens ,
    \quad \text{and} \quad ( \theta , \varphi ) \in \bmr{\{} [0,
    \tfra{\pi}{2} ) \times [0, 2 \pi) \bmr{\}} \cup \bmr{\{} \{
    \tfra{\pi}{2} \} \times [0 , \pi ] \bmr{\}} \ens .
  \end{gather*}
\end{theorem}

\begin{proof}
  The proof runs analogous to the one of \theoref{theorem4qubit}, with
  the observation that the representative states of the
  $\mathd_{3,1,1}$ and $\mathd_{2,2,1}$ class are readily subsumed in
  the parameter range of $\mathd_{2,1,1,1}$.  The fixed \acp{MP} of
  $\mathd_{2,1,1,1}$ are left invariant under a $\rotxs ( \pi
  )$-rotation, thus ensuring that the remaining \ac{MP} can be
  projected into the desired parameter range. The uniqueness is again
  verified by considering all possible cross-ratios.
\end{proof}

\begin{figure}
  \centering
  \includegraphics[scale=1.2]{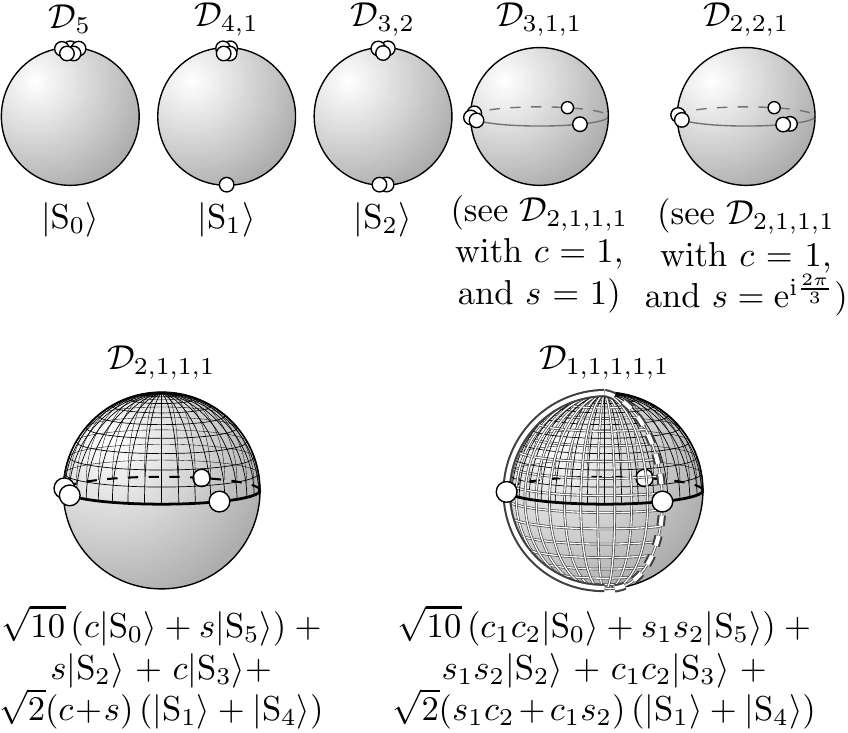}
  \caption[DC classes of 5 qubit states]{\label{slocc5} The first five
    of the seven \ac{DC} classes of 5 qubit symmetric states coincide
    with \ac{SLOCC} classes, while the representative states of
    $\mathd_{2,1,1,1}$ are parameterised by one \ac{MP} $c \ket{0} + s
    \ket{1}$ (black grid), and those of $\mathd_{1,1,1,1,1}$ by two
    \acp{MP} (black and white grid).  The parameter range for
    $\mathd_{2,1,1,1}$ is $( \theta , \varphi ) \in \bmr{\{} [0,
    \tfra{\pi}{2} ) \times [0, 2 \pi) \bmr{\}} \cup \bmr{\{} \{
    \tfra{\pi}{2} \} \times (0 , \pi ] \bmr{\}} \backslash \bmr{\{} \{
    \tfra{\pi}{2} \} \times \{ \tfra{2 \pi}{3} \} \bmr{\}}$.  For
    $\mathd_{1,1,1,1,1}$ the range of $( \theta_1 , \varphi_1 )$ is
    the same as $( \theta , \varphi )$, and $( \theta_2 , \varphi_2 )
    \in \bmr{\{} [0, \pi ] \times [0, \tfra{2 \pi}{3}) \bmr{\}}
    \backslash \bmr{\{} \{ \tfra{\pi}{2} \} \times \{ 0 \} \bmr{\}}$.
    The fixed equatorial \acp{MP} of the representative states are all
    equidistantly spaced.}
\end{figure}

An over-complete set of representative states for the general case can
be given as follows:

\begin{corollary}\label{theorem5qubit}
  Every pure symmetric state of 5 qubits is \ac{SLOCC}-equivalent to
  one or more state of the set
  \begin{equation*}
    \{ \sym{0}, \, \sym{1}, \, \sym{2} , \, \sqrt{10} \left( \sym{0} +
      t_1 t_2 \sym{5} \right) + t_1 t_2 \sym{2} + \sym{3} + \sqrt{2}
    \left( t_1 + t_2 \right) \left( \sym{1} + \sym{4} \right) \} \ens
    ,
  \end{equation*}
  \vspace{-3em}
  \begin{align*}
    \text{with} \quad t_i = \E^{\I \varphi_i} \tan \tfra{\theta_i}{2}
    \ens , \quad \text{and} \quad ( \theta_1 , \varphi_1 )& \in
    \bmr{\{} [0, \tfra{\pi}{2}] \times [0, 2 \pi) \bmr{\}} \cup
    \bmr{\{} \{
    \tfra{\pi}{2} \} \times (0 , \pi ] \bmr{\}} \ens , \\
    ( \theta_2 , \varphi_2 )& \in \bmr{\{} [0, \pi] \times [0, \tfra{2
      \pi}{3}) \bmr{\}} \ens .
  \end{align*}
\end{corollary}

\begin{proof}
  Only the generic class $\mathd_{1,1,1,1,1}$ needs to be considered.
  Given an arbitrary state of this class, three of its \acp{MP} can be
  projected onto the vertices of an equilateral triangle by means of a
  \mob transformation. These \acp{MP} are left invariant under $\{
  \rotxs ( \pi ) , \rotzs ( \tfra{2 \pi}{3} ) \}$-rotations. If the
  fourth \ac{MP} does not lie in the $( \theta_1 , \varphi_1 )$-area,
  it can be projected there by a $\rotxs ( \pi )$-rotation. Subsequent
  $\rotzs ( \tfra{2 \pi}{3} )$-rotations can project the fifth \ac{MP}
  into the $( \theta_2 , \varphi_2 )$-area, while leaving the fourth
  \ac{MP} in the $( \theta_1 , \varphi_1 )$-area.
\end{proof}

As the number of qubits increases, the picture gradually becomes more
complicated, because \ac{DC} classes with diversity degree $n$ contain
a continuous range of \ac{SLOCC} classes that is parameterised by
$n-3$ variables \cite{Bastin09}.

\section{Entanglement families of four qubits}\label{families}

The concept of \textbf{\acfp{EF}} was already briefly touched upon in
\sect{entanglement_classes}.  Derived by Verstraete \etal
\cite{Verstraete02} with some advanced methods of linear algebra, this
classification scheme reduces the complexity of the four qubit case by
replacing the infinite amount of \ac{SLOCC} classes with nine
different \acp{EF}.  The \ac{EF} a state belongs to does not change
under \ac{SLOCC} operations, thus making the partition into \ac{SLOCC}
classes a refinement of the partition into \ac{EF} classes (\ac{SLOCC}
$\leq$ \ac{EF}).  The nine \acp{EF} are represented by the following
ranges of states, with $a,b,c,d \in \mbbc$ being arbitrary complex
parameters.
\begin{itemize}
  \compactlist
\item $G_{abcd} =
  \frac{a+d}{2} \big( \ket{0000} + \ket{1111} \big) +
  \frac{a-d}{2} \big( \ket{0011} + \ket{1100} \big) +
  \frac{b+c}{2} \big( \ket{0101} + \ket{1010} \big) \\
  \hspace*{1.7cm}
  + \frac{b-c}{2} \big( \ket{0110} + \ket{1001} \big)$
  
\item $L_{abc_{2}} =
  \frac{a+b}{2} \big( \ket{0000} + \ket{1111} \big) +
  \frac{a-b}{2} \big( \ket{0011} + \ket{1100} \big) +
  c \big( \ket{0101} + \ket{1010} \big) + \ket{0110}$

\item $L_{a_{2} b_{2}} =
  a \big( \ket{0000} + \ket{1111} \big) +
  b \big( \ket{0101} + \ket{1010} \big) +
  \ket{0110} + \ket{0011}$

\item $L_{ab_{3}} =
  a \big( \ket{0000} + \ket{1111} \big) +
  \frac{a+b}{2} \big( \ket{0101} + \ket{1010} \big) +
  \frac{a-b}{2} \big( \ket{0110} + \ket{1001} \big) \\
  \hspace*{1.7cm}
  + \frac{\I}{\sqrt{2}} \big( \ket{0001} + \ket{0010} + \ket{0111}
  + \ket{1011} \big)$

\item $L_{a_{4}} =
  a \big( \ket{0000} + \ket{0101} + \ket{1010} + \ket{1111} \big)
  + \big( \I \ket{0001} + \ket{0110} - \I \ket{1011} \big)$

\item $L_{a_{2} 0_{3\oplus\bar{1}}} =
  a \big( \ket{0000} + \ket{1111} \big) +
  \big( \ket{0011} + \ket{0101} + \ket{0110} \big)$

\item $L_{0_{5\oplus\bar{3}}} =
  \ket{0000} + \ket{0101} + \ket{1000} + \ket{1110}$

\item $L_{0_{7\oplus\bar{1}}} =
  \ket{0000} + \ket{1011} + \ket{1101} + \ket{1110}$
  
\item $L_{0_{3\oplus\bar{1}}0_{3\oplus\bar{1}}} =
  \ket{0000} + \ket{0111}$
\end{itemize}
Up to permutations, every pure 4 qubit state is \ac{SLOCC}-equivalent
to a state from exactly one of these families. Unlike in our
\theoref{theorem4qubit}, however, the parameterisation of the \acp{EF}
is not unique, i.e. two different sets of parameters $( a,b,c,d ) \neq
( a' , b' , c' , d' )$ can give rise to two \ac{SLOCC}-equivalent
states.  This non-uniqueness can be already seen from the
non-normalised nature of the generic family $G_{abcd}$ which is due to
the choice of the parameters $a,b,c,d$ as the eigenvalues of a matrix
employed for the proof in \cite{Verstraete02}. A less trivial example
are the two symmetric states $\ket{\psi^{a}} = \sym{2}$ and
$\ket{\psi^{b}} = ( \sym{0} + \sym{4} ) + \sqrt{\tfra{2}{3}} \sym{2}$
which are both present in the family $G_{abcd}$. Their \ac{LU}
equivalence $\ket{\psi^{a}} \stackrel{\text{LU}}{\longleftrightarrow}
\ket{\psi^{b}}$ can be immediately seen from their \ac{MP}
distributions $\ket{\phi^{a}_{1,2}} = \ket{0}$, $\ket{\phi^{a}_{3,4}}
= \ket{1}$, and $\ket{\phi^{b}_{1,2}} = \tfra{1}{\sqrt{2}} (\ket{0} +
\I \ket{1})$, $\ket{\phi^{b}_{3,4}} = \tfra{1}{\sqrt{2}} (\ket{0} - \I
\ket{1})$.

Here we are interested in the subset of symmetric 4 qubit states.  In
the following we will determine in which \acp{EF} the symmetric
\ac{SLOCC} classes are located, and we will elucidate the relationship
between the \ac{DC} and \ac{EF} classes, both of which are coarser
partitions than the \ac{SLOCC} classes.

First, we identify the \acp{EF} of the \ac{SLOCC} and \ac{DC} classes
shown in \fig{slocc4}.  The separable state $\sym{0}$, and therefore
the entire $\mathd_{4}$ class, is \ac{LU}-equivalent to the state
$\ket{0110}$ embedded in the family $L_{{abc}_{2}}$ for parameters
$a=b=c=0$.  The W state $\sym{1}$ representing $\mathd_{3,1}$ is
recovered from the family $L_{{ab}_{3}}$ by setting $a=b=0$ and
spin-flipping the last two qubits.  The state $\sym{2}$ representing
$\mathd_{2,2}$ can be found in the general family $G_{abcd}$ by
setting $a=1, b=2, c=0, d=-1$.  A state of the degeneracy class
$\mathd_{2,1,1}$ is found in $L_{{abc}_{2}}$ by setting $a=1, b=0,
c=\tfra{1}{2}$ and spin-flipping the second and third qubit, yielding
the state $\ket{\psi}= \sqrt{\tfra{2}{5}} \sym{0} + \sqrt{\tfra{3}{5}}
\sym{2}$ which is made up of the \acp{MP} $\ket{\phi_{1,2}} = \ket{0}$
and $\ket{\phi_{3,4}} = \tfra{1}{2} \ket{0} \pm \I \tfra{\sqrt{3}}{2}
\ket{1}$.  The continuum of \ac{SLOCC} classes present in the generic
class $\mathd_{1,1,1,1}$ has previously been parameterised in
\cite{Bastin09} as $\left( \sym{0} + \sym{4} \right) + \mu \sym{2}$,
with $\mu \in \mbbc \backslash \{ \pm \sqrt{\tfra{2}{3}} \}$.  These
states are recovered from the general family $G_{abcd}$ for $a = 1 +
\tfra{\mu}{\sqrt{6}}, b = \sqrt{\tfra{2}{3}} \mu, c = 0, d=1 -
\tfra{\mu}{\sqrt{6}}$.  The reason for the exclusion of $\mu = \pm
\sqrt{\tfra{2}{3}}$ is that the \ac{MP} distribution then becomes
degenerate, and it was already seen above that $\ket{\psi^{b}} =
\left( \sym{0} + \sym{4} \right) + \sqrt{\tfra{2}{3}} \sym{2}$ is
\ac{LU}-equivalent to $\sym{2} \in \mathd_{2,2}$.  Summing up, we
found
\begin{itemize}
  \compactlist
\item $\mathd_{1,1,1,1}$, $\mathd_{2,2} \subset G_{abcd}$
  
\item $\mathd_{2,1,1}$, $\mathd_{4} \subset L_{{abc}_{2}}$
  
\item $\mathd_{3,1} \subset L_{{ab}_{3}}$
\end{itemize}
In particular, only three of the nine \acp{EF} contain all the states
of the symmetric subspace of the 4 qubit Hilbert space, including
those non-symmetric states that are \ac{SLOCC}-equivalent to symmetric
states. The other six \acp{EF} only contain genuinely non-symmetric
states that cannot be symmetrised by \ac{SLOCC} operations.
Furthermore, it is noteworthy that the families $G_{abcd}$ and
$L_{{abc}_{2}}$ each contain two different \ac{DC} classes.  This is
somewhat unexpected, because there exist nine different \acp{EF}, but
only five different \ac{DC} classes.  We can therefore conclude that
the \acp{EF} are not a particularly useful classification scheme for
symmetric states, due to their coarseness.  Since all states of a
given \ac{DC} class are contained in only one \ac{EF},
\theoref{hierarchy} can be specified for the 4 qubit case:
\begin{theorem}\label{hierarchy4qubit}
  The symmetric subspace of the 4 qubit Hilbert space has the
  following refinement hierarchy of entanglement partitions:
  \begin{equation}
    \text{LOCC} < \text{SLOCC} < \text{DC} < EF \ens .
  \end{equation}
\end{theorem}

\section[Determining SLOCC inequivalence from the MP
distribution]{Determining SLOCC inequivalence from the
  \protect\newline MP distribution}\label{determiningslocc}

The known properties of \mob transformations can be employed to
immediately determine whether two symmetric $n$ qubit states with the
same degeneracy of their \acp{MP} could be \ac{SLOCC}-equivalent.  As
outlined in \sect{mobius_def}, circles on the surface of the Majorana
sphere are always projected onto circles, and the angles at which two
circles meet are preserved.  These properties can be exploited by
identifying and comparing circles with \acp{MP} on them.

As an example, \fig{maxentstates_slocc} shows the \ac{MP}
distributions of some states investigated in \chap{solutions}.  The
two 5 qubit states $\ket{\psi_{5}}$ and $\ket{\Psi_{5}}$ are not
\ac{SLOCC}-equivalent, because $\ket{\Psi_5}$ exhibits a ring with
four \acp{MP}, whereas no ring with four \acp{MP} can be found for
$\ket{\psi_5}$.  Similarly, it can be shown that most of the states
investigated in \chap{solutions} are \ac{SLOCC}-inequivalent to each
other.  For 12 qubits it is not immediately clear that the positive
solution $\ket{\psi_{12}}$ and the icosahedron state
$\ket{\Psi_{12}}$, shown in \fig{maxentstates_slocc}, are
\ac{SLOCC}-inequivalent, since both states have several rings with
four or five \acp{MP} each.  In the icosahedron state
$\ket{\Psi_{12}}$ it is possible to identify 20 circles, each through
three adjacent \acp{MP} (the corners of all faces of the icosahedron),
so that the interior of each circle is free of \acp{MP}. This property
must be preserved under \mob transformations, but for
$\ket{\psi_{12}}$ it is not possible to find such twenty distinct
circles that are all free of \acp{MP} in their interior.

\begin{figure}
  \centering
  \begin{overpic}[scale=1.3]{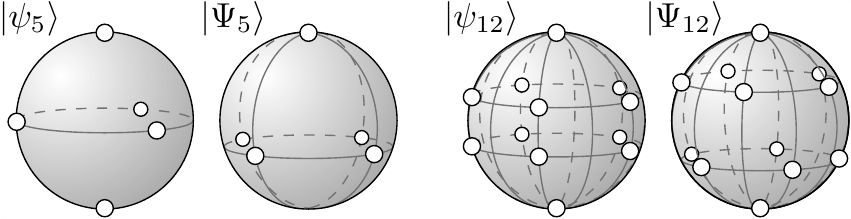}
    \put(-1,0){(a)}
    \put(24,0){(b)}
    \put(53,0){(c)}
    \put(77,0){(d)}
  \end{overpic}
  \caption[Determination of SLOCC inequivalence from MP
  distribution]{\label{maxentstates_slocc} The \ac{MP} distributions
    of four highly or maximally entangled symmetric states introduced
    in \protect\chap{solutions} are shown.  The 5 qubit
    \protect\quo{trigonal bipyramid state} $\ket{\psi_{5}}$ is
    \ac{SLOCC}-inequivalent to the \protect\quo{square pyramid state}
    $\ket{\Psi_{5}}$. Likewise, the 12 qubit icosahedron state
    $\ket{\Psi_{12}}$ cannot be reached from $\ket{\psi_{12}}$ by
    \ac{SLOCC} operations.}
\end{figure}
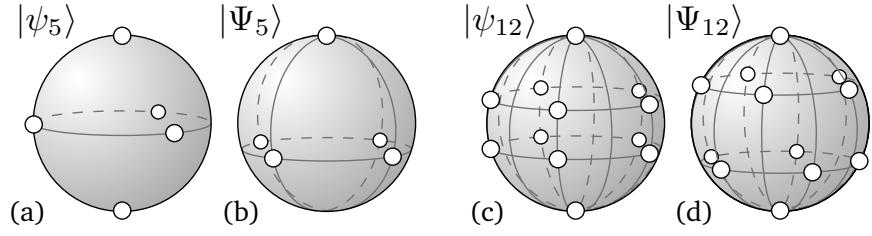

Markham \cite{Markham11} determined the \ac{SLOCC}-inequivalence of
the four qubit \ac{GHZ} state $\ket{\text{GHZ}_4}$ and tetrahedron
state $\ket{\text{T}}$ analytically from the values of their Schmidt
rank and geometric entanglement.  Interestingly, with the
geometrically motivated approach employed here there is no need for
such calculations: The \acp{MP} of $\ket{\text{GHZ}_4}$ all lie on the
single ring, but those of $\ket{\text{T}}$ don't.

The \ac{SLOCC}-inequivalence of all totally invariant states of up to
7 qubits was determined in \cite{Markham11} by considering the \ac{MP}
degeneracies as well as the Schmidt rank, and a conjecture was made
that all totally invariant symmetric states are
\ac{SLOCC}-inequivalent to each other.  Using the preservation of
circles and angles under \mob transformations, it is expected that
this conjecture becomes much easier to verify.

The existence of $n$ qubit states that are not \ac{LOCC} or
\ac{SLOCC}-equivalent to their complex conjugates has been affirmed
\cite{Acin00b,Kraus10b}, and the operational consequences of this for
distinguishing such states have been discussed
\cite{Kraus10b,Vicente11}.  Taking the complex conjugation into
account, the number of parameters to describe pure three qubit states
up to \ac{LU} were reduced from six to five \cite{Vicente11}.  In the
case of symmetric states we immediately see that complex conjugation
corresponds to a reflection of the \acp{MP} along the $X$-$Z$-plane.
We can therefore explain the \ac{LOCC}-inequivalence of complex
conjugate symmetric states with the geometric concept of
\textbf{chirality}, or handedness: Although having the same distances
and angles (and therefore the same geometric entanglement), the mirror
image of an arbitrary \ac{MP} distribution is in general not
\ac{LU}-equivalent to the original.  This idea can be easily extended
to \ac{SLOCC}-equivalence, with the result that general symmetric
states are not even \ac{SLOCC}-equivalent to their complex conjugate.

\section{Symmetric SLOCC invariants on the Majorana
  sphere}\label{symm_inv}

We are already familiar with the property of \mob transformations to
preserve circles and angles, as well as the cross-ratios
\eqref{cross-ratio} of ordered quadruples of points.  These quantities
can therefore be considered to be \textbf{symmetric \ac{SLOCC}
  invariants}.

A detailed study of symmetric \ac{LU} and \ac{SLOCC} invariants with
the Majorana representation was recently undertaken by Ribeiro and
Mosseri \cite{Ribeiro11}.  With regard to symmetric \ac{LU} operations
(which are equivalent to \ac{LOCC} operations between symmetric
states, cf. \eq{locccond}), they found that the well-known six \ac{LU}
invariants for 3 qubits \cite{Kempe99,Coffman00} can be expressed in
terms of the $3!$ angles between pairs of \acp{MP}.  With regard to
symmetric \ac{SLOCC} operations, the cross-ratios were identified as a
natural basis for invariants.  Given a cross-ratio,
\begin{equation}\label{cross-ratio_single}
  \lambda :=
  \frac{(v_1 - v_3)(v_2 - v_4)}{(v_2 - v_3)(v_1 - v_4)} \in \cext \ens ,
\end{equation}
the $24$ different possible ways to permute the entries $\{ v_1, v_2,
v_3, v_4 \}$ generally leads to six different values $\{ \lambda_{i}
\}$ for the cross ratio: $\{ \lambda , \tfra{1}{\lambda} , 1 - \lambda
, \tfra{1}{1 - \lambda}, \tfra{\lambda}{\lambda - 1} , \tfra{\lambda -
  1}{\lambda} \}$ \cite{Ribeiro11}.  These different values partition
the complex plane into six distinct regions, as seen in
\fig{slocc_invariants}(b).

\begin{figure}
  \centering
  \begin{minipage}{135mm}
    \centering
    \begin{overpic}[scale=1.3]{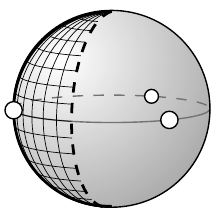}
      \put(-12,0){(a)}
      \put(31,155){$\mathd_{1,1,1,1}$}
      \put(15,127){$2 c \sym{0} + s \sym{1} +$}
      \put(18,107){$c \sym{3} + 2 s \sym{4}$}
    \end{overpic}
    \hspace{10mm}
    \begin{overpic}[scale=0.9]{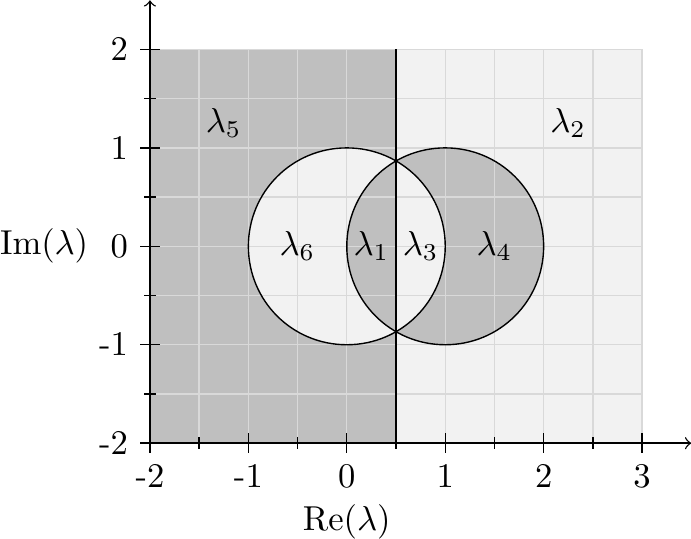}
      \put(0,0){(b)}
    \end{overpic}
  \end{minipage}
  \begin{minipage}{135mm}
    \vspace{8mm}
    \begin{overpic}[scale=1.02]{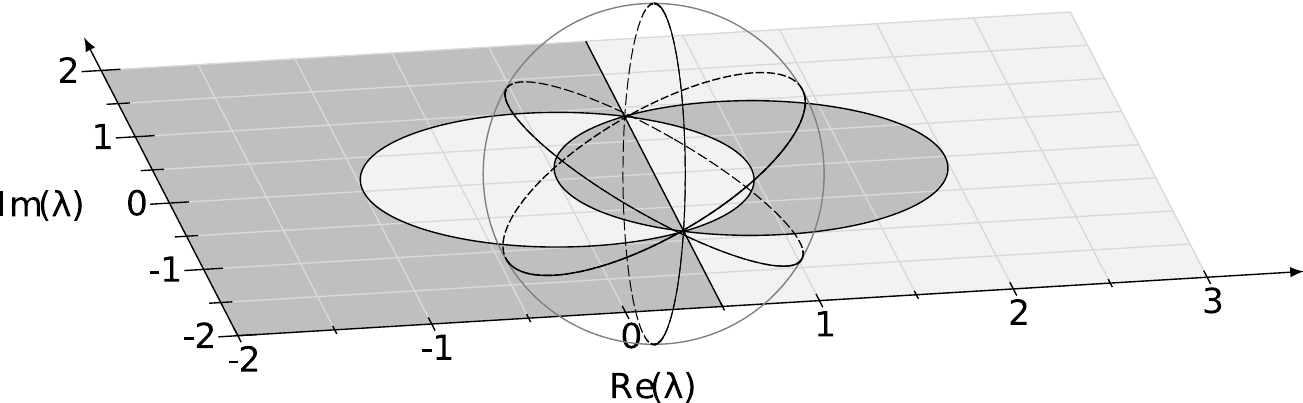}
      \put(2,3){(c)}
    \end{overpic}
  \end{minipage}
  \caption[Relationship between cross-ratio and SLOCC
  invariants]{\label{slocc_invariants} An alternative parameterisation
    of the generic 4 qubit \ac{DC} class $\mathd_{1,1,1,1}$ is shown
    in (a), with the parameterisation of the fourth \ac{MP} $c \ket{0}
    + s \ket{1}$ being $( \theta , \varphi ) \in \bmr{ \{ } [0, \pi)
    \times [0, \tfra{ \pi}{3}) \bmr{ \} }$. The six different areas
    $\{ \lambda_{i} \}$ in the complex plane that correspond to
    permutations of the entries in \protect\eq{cross-ratio_single} are
    shown in (b).  The plane is cut by the line $\RE (z) =
    \tfra{1}{2}$, as well as two unit circles with centres $z = 0$ and
    $z = 1$.  Aligning the poles of the sphere (a) with the points $z
    = \frac{1}{2} + \I$ and $z = \frac{1}{2} - \I$ in (b) yields the
    arrangement (c), which can be understood as a stereographic
    projection of the parameter range $( \theta , \varphi )$ onto one
    of the six areas $\{ \lambda_{i} \}$.}
\end{figure}
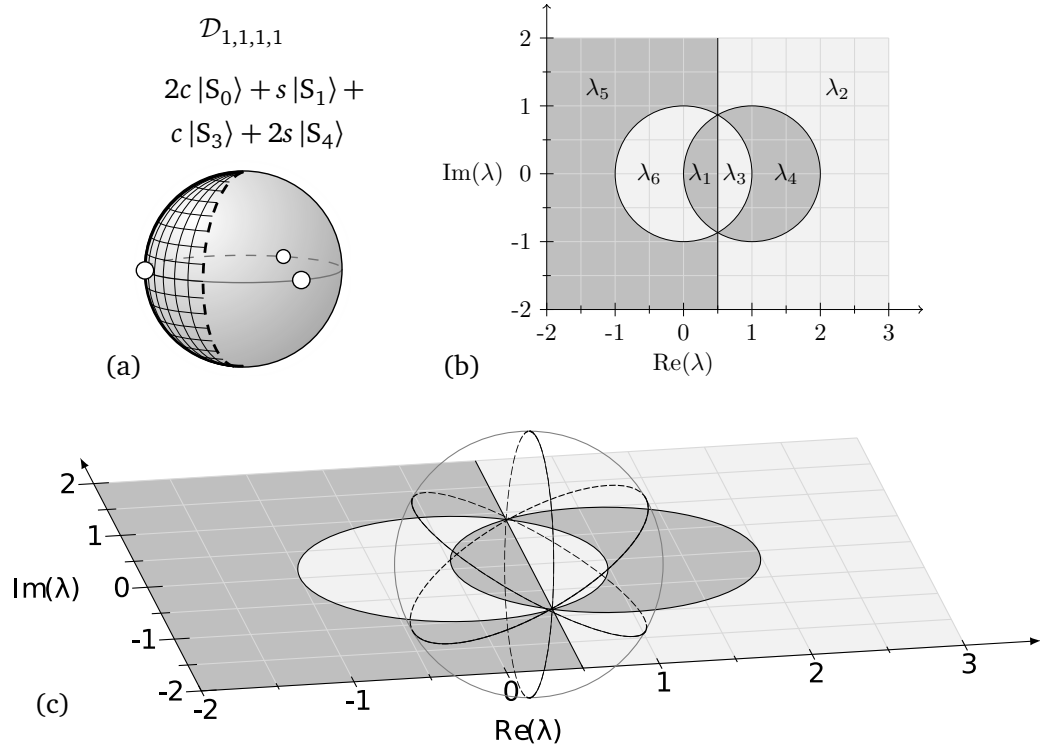

Here we point out a relationship of this partition to the generic
\ac{SLOCC} equivalence classes of 4 qubits, namely the degeneracy
class $\mathd_{1,1,1,1}$ studied in \sect{rep_4_qubit}.  The parameter
range chosen for the unique representative states of
$\mathd_{1,1,1,1}$ in \fig{slocc4} is not the only possibly choice,
and an equivalent parameterisation with a high degree of symmetry is
shown in \fig{slocc_invariants}(a).  There the fourth \ac{MP}
$\ket{\phi_{4}} = c \ket{0} + s \ket{1}$ has the parameterisation $(
\theta , \varphi ) \in \bmr{ \{ } [0, \pi) \times [0, \tfra{ \pi}{3})
\bmr{ \} }$.  As shown in \fig{slocc_invariants}, this parameter range
can be projected onto one of the six areas $\{ \lambda_{i} \}$ in the
complex plane by means of a stereographic projection.  This is by no
means a coincidence, but rather a consequence of \ac{SLOCC}
invariants: Every value within the parameter range of $( \theta ,
\varphi )$ corresponds to a unique state that is
\ac{SLOCC}-inequivalent to any other value within that range.
Likewise, two 4 qubit symmetric states with no \ac{MP} degeneracies
are \ac{SLOCC}-inequivalent \ac{iff} the cross-ratios
\eqref{cross-ratio_single} of their Majorana roots are different for
all possible permutations, which implies that the values of their
cross-ratios need to be considered only for one of the six areas $\{
\lambda_{i} \}$.

\section{Global entanglement measures}\label{global_ent}

Polynomial \ac{SLOCC} invariants that provide information about the
type of entanglement present in a system were already mentioned at the
end of \sect{entanglement_classes}.  Here we review two well-known
approaches in light of our results about symmetric states.

\subsection{Maximal
  \protect\texorpdfstring{$n$-tangles}{n-tangles}}\label{max_tangles}

Osterloh and Siewert \cite{Osterloh05,Osterloh06} constructed
entanglement monotones from antilinear operators that are invariant
under \ac{SLOCC} operations, and that can be understood as
generalisations of the concurrence (2-tangle) \cite{Wootters98} and
the 3-tangle \cite{Coffman00}.  This allows for the construction of a
\quo{global entanglement measure} with the aim to detect only genuine
$n$ qubit entanglement in the sense that it is blind for $k$-partite
entanglement with $k<n$.  More formally, an $n$ qubit state
$\ket{\Psi_{n}}$ is a \textbf{\quo{maximal $n$-tangle}} if the
following two conditions hold \cite{Osterloh05,Osterloh06,Osterloh10}:
\begin{enumerate}
\item[i)] All reduced density matrices of $\ket{\Psi_{n}}$ with rank
  $\leq 2$ are maximally mixed.
  
\item[ii)] All $k$-site reduced density matrices of $\ket{\Psi_{n}}$
  have zero $k$-tangle\footnote{No unique definition of the $k$-tangle
    for $k > 3$ exists in the literature, apart from the requirement
    that it should be an entanglement monotone generalising the
    2-tangle and 3-tangle.  Often the tangles are defined as
    polynomial $\slc$-invariants which are entanglement monotones by
    construction, and which take the form of homogeneous functions of
    the coefficients of the given state \cite{Verstraete03}.  For
    example, in \protect\cite{Osterloh05,Osterloh06} these
    entanglement monotones are constructed from the expectation values
    of antilinear operators.}  ($1 < k < n$).
\end{enumerate}
This definition of maximal entanglement is very similar to those
proposed in \cite{Gisin98,Verstraete03}.  The first condition implies
that a maximal amount of information is gained when reading out a
qubit, something that is closely related to the \emph{stochastic
  states} of \cite{Verstraete03}.  The second condition excludes
hybrids of various types of entanglement, thus following the concept
of \emph{monogamy} \cite{Coffman00}, i.e. that the total entanglement
is a resource distributed among different types of entanglement.
Occasionally, a third condition is imposed, namely that the maximally
$n$-tangled states have a phase-independent canonical form, i.e. that
the first two properties shall be unaffected by relative phases in the
coefficients \cite{Osterloh06}.  This implies that maximally
$n$-tangled states can always be written with positive coefficients,
which is quite interesting for us because of our focus on positive
states.  On the other hand, this is a first indication that
\quo{global entanglement measures} are qualitatively different from
e.g. the geometric measure, because the latter is not expected to be
maximised for positive states in general.  Another interesting
property of global entanglement measures is that they are closely
related to the concept of symmetry, in the sense that permutational
invariance was identified as a characteristic property of such
measures \cite{Coffman00}.

The $n$ qubit \ac{GHZ} state is a maximally $n$-tangled state for all
$n$, and in the case of 3 qubits it is the only such state.  Therefore
the \ac{SLOCC} class of 3 qubit states with genuine tripartite
entanglement is represented by the \ac{GHZ} state.

Another maximally $n$-tangled state for every $n \geq 3$ is the
following state,
\begin{equation}\label{x-states}
  \ket{X_{n}} = \sqrt{n} \sym{1} + \sqrt{n-2} \sym{n} \ens ,
\end{equation}
coined the \textbf{X-state} \cite{Osterloh10}.  Note that in the 3
qubit case the state $\ket{X_{3}}$ is \ac{LU}-equivalent to
$\ket{\text{GHZ}_{3}}$.  The maximally $n$-tangled states
$\ket{\text{GHZ}_{n}}$ and $\ket{X_{n}}$ are the two extremes in the
sense that $\ket{\text{GHZ}_{n}}$ is always the maximally $n$-tangled
state of minimal length whereas $\ket{X_{n}}$ is the one of maximal
length. Here, \emph{length} means the number of components in the
canonical form of a state \cite{Osterloh06}.

For 4 qubits there exist three inequivalent \ac{SLOCC}
invariants\footnote{ Ren \protect\etal \protect\cite{Ren08} discovered
  that one of the three entanglement monotones defined in
  \protect\cite{Osterloh05} is not permutation-invariant, and they
  proposed a new permutation-invariant monotone in place of the old
  one. The states detected by the corresponding invariant (the 4 qubit
  cluster states) are not affected by this redefinition.}, and a
corresponding \quo{basis} of three inequivalent maximally 4-tangled
states with neither 3-tangle nor concurrence has been determined
\cite{Osterloh05,Dokovic09}.  These states are the \ac{GHZ} state
\begin{equation}\label{4qubit_max1}
  \ket{\Psi^{a}_{4}} =
  \tfrac{1}{\sqrt{2}} \left( \sym{0} + \sym{4} \right) \ens ,
\end{equation}
the 4 qubit X-state
\begin{equation}\label{4qubit_max2}
  \ket{\Psi^{b}_{4}} =  \sqrt{\tfrac{2}{3}} \sym{1} +
  \sqrt{\tfrac{1}{3}} \sym{4} \ens ,
\end{equation}
and the 4 qubit cluster state
\begin{equation}\label{4qubit_max3}
  \ket{\Psi^{c}_{4}} = \tfrac{1}{2} \left( \ket{1111} + \ket{1100} +
    \ket{0010} + \ket{0001} \right) \ens .
\end{equation}
We immediately notice that $\ket{\Psi^{b}_{4}}$ is identical to the
tetrahedron state defined in \eqref{tetrahedron_state}, up to a
$\rotxs (\pi)$-rotation (spin-flip).  Quite surprisingly, this link
does not seem to have been discovered before, despite the state
\eqref{4qubit_max2} having periodically appeared in various analytical
forms in the literature since at least 1998.  In \cite{Gisin98} the
state was determined as the symmetric state that maximises a more
stringent version of the global entanglement measure (outlined further
below), and it is also the unique state that maximises a variant of
the global measure defined in \cite{Love07}.  The distinguished
position of the tetrahedron state for polynomial invariants and global
measures is further seen from the fact that it is the only one of the
three maximally 4-tangled states that can be detected by the
hyperdeterminant introduced by Miyake \cite{Miyake03}.  This is yet
further proof that the Platonic symmetry of the Majorana
representation is a signature of distinguished properties of the
underlying symmetric states.

The 4 qubit cluster state and its permutations are representative
states of one of the two classes of graph states that exist for 4
qubits, with the other class of graph states being represented by the
\ac{GHZ} state \cite{Hein04,Ren08,Bai08,Dokovic09}.  The positive
state $\ket{\Psi^{c}_{4}}$ displayed in \eq{4qubit_max3} is
\ac{LU}-equivalent to the canonical form for cluster states introduced
by Briegel \etal \cite{Briegel01}, which can be seen by flipping the
first two qubits and applying the Hadamard gate $H =
\tfrac{1}{\sqrt{2}} \left(
  \begin{smallmatrix}
    1 & 1 \\
    -1 & 1
  \end{smallmatrix}
\right)$ on the last two qubits:
\begin{equation}\label{max_ent_cluster_state}
  \begin{split}
    \sigma_{x} \otimes \sigma_{x} \otimes H
    \otimes H \, \ket{\Psi^{c}_{4}}& =
    \tfrac{1}{2} \big( \ket{00\!-\!-} + \ket{00\!+\!+} +
    \ket{11\!-\!+} + \ket{11\!+\!-} \big) \\
    {}& = \tfrac{1}{2} \big( \ket{0000} + \ket{0011} +
    \ket{1100} - \ket{1111} \big) \ens .
  \end{split}
\end{equation}
Somewhat surprisingly, $\ket{\Psi^{c}_{4}}$ is neither symmetric nor
\ac{LU}-equivalent\footnote{The \ac{LU}-inequivalence can be verified
  e.g. by the different eigenvalues of the reduced density matrices
  $\rho_{12} = \Trace_{34} \left( \pure{\Psi^{c}_{4}} \right)$ and
  $\rho_{23} = \Trace_{14} \left( \pure{\Psi^{c}_{4}} \right)$.  The
  question of whether $\ket{\Psi^{c}_{4}}$ is \ac{SLOCC}-equivalent to
  a symmetric state is irrelevant in this context, because the
  resulting state would no longer be in the normal form required for
  maximally $n$-tangled states.} to a symmetric state, which implies
that maximally $n$-tangled states are not necessarily symmetric,
despite the monotones for global entanglement being invariant under
qubit permutations.

The three different types of genuine entanglement detected by the
4-tangle are not distinguished by the \acp{EF} of \cite{Verstraete02},
because the states \eqref{4qubit_max1}, \eqref{4qubit_max2} and
\eqref{4qubit_max3} all belong to the generic family
$G_{abcd}$. Therefore, the global entanglement measure can be more
useful than the \acp{EF} to distinguish different types of 4 qubit
entanglement.

States of 5 and 6 qubits with maximal global entanglement were found
in \cite{Osterloh06}, and in the case of 5 qubits four different types
of entanglement are detected \cite{Osterloh06,Luque06}.  The normal
forms of the corresponding states are the \ac{GHZ} state, two states
that can be easily verified to be \ac{LU}-nonsymmetric, as well as the
5 qubit X-state
\begin{equation}\label{5qubit-x}
  \ket{\Psi^{d}_{5}} =
  \sqrt{\tfrac{5}{8}} \sym{1} + \sqrt{\tfrac{3}{8}} \sym{5} \ens .
\end{equation}
Unlike the other three states, however, the state $\ket{\Psi^{d}_{5}}$
does not satisfy all the conditions imposed on maximal $5$-tangles,
because it has a nonvanishing 4-tangle \cite{Osterloh06}.
Nevertheless, $\ket{\Psi^{d}_{5}}$ can be considered to have an
extremal amount of global entanglement. A comparison of \eq{5qubit-x}
with \eq{5_opt_form} reveals that the Majorana representation of
$\ket{\Psi^{d}_{5}}$ is a spin-flipped square pyramid state.  Defining
$\ket{\Phi_{5}} = \rotxs ( \pi ) \ket{\Psi^{d}_{5}} =
\sqrt{\tfrac{3}{8}} \sym{0} + \sqrt{\tfrac{5}{8}} \sym{4}$, we compare
the 5 qubit X-state to the maximally entangled symmetric 5 qubit state
$\ket{\Psi_{5}}$ in terms of the geometric measure, derived in
\sect{majorana_five}.  The latitudinal angle of the \ac{MP} circle of
$\ket{\Psi_{5}}$ is $\theta \approx 1.874$, whereas for $\ket{\Phi}$
one obtains $\tan^4 (\frac{\theta}{2}) = \frac{5}{\sqrt{3}}$, yielding
$\theta \approx 1.833$.  The imbalance present in the spherical
amplitude function of $\ket{\Phi_{5}}$ results in $\ket{\sigma} =
\ket{0}$ being the only \ac{CPP}, and the geometric entanglement is
$\Eg ( \ket{\Phi_{5}} ) = \log_2 (\frac{8}{3}) \approx 1.415$, which
is well below that of $\ket{\Psi_{5}}$ and which coincides with the
entanglement of the maximally entangled 5 qubit Dicke state $\Eg (
\sym{2} ) = \log_2 (\frac{8}{3})$.  Even though the small difference
in the latitudinal \ac{MP} angle leads to a large decrease of the
entanglement, the fidelity remains very close to the original state:
\begin{equation}\label{fidelity}
  F = \abs{\bracket{\Phi_{5}}{\Psi_{5}}}
  \stackrel{\eqref{5_opt_form}}{=} \frac{\sqrt{3} +
    \sqrt{5} A}{2 \sqrt{2} \sqrt{1 + A^2}} \approx 0.996 \: 752 \ens .
\end{equation}

It is remarkable that the 5 qubit X-state -- apart from the \ac{GHZ}
state the only symmetric state to be detected by the global
entanglement measure -- is so close to the square pyramid state which
we found to solve the Majorana problem of 5 qubits.  In contrast to
this, the classical optimisation problems of \toth and Thomson are
solved by the trigonal bipyramid state, a state with a qualitatively
different \ac{MP} distribution, as evidenced by the
\ac{SLOCC}-inequivalence shown in \sect{determiningslocc}.

\subsection{Maximal mixture in all reduced density
  matrices}\label{max_reduced}

Gisin \etal \cite{Gisin98} introduced and studied five different
criteria for maximal global entanglement.  Their investigation is
limited to symmetric states, which they justified with the argument
that all of the $n$ qubits of maximally entangled states should be
equivalent, with no privileged part.  The conclusion is that four of
the five criteria are compatible (with \ac{GHZ} states having maximal
global entanglement), while the fifth criterion is qualitatively
different from the others.  This criterion is that \emph{all} reduced
density matrices shall be maximally mixed, which is a more stringent
variant of condition~1 outlined at the beginning of
\sect{max_tangles}.  States that satisfy this strong criterion exist
only for $n=2,3,4$ and $6$ qubits.  Recast in our notation, the states
found in \cite{Gisin98} are:
\begin{subequations}\label{gisin_states}
  \begin{align}
    \ket{\Psi_{3}}_{\pm 1}& =
    \tfrac{1}{\sqrt{2}} \big( \sym{0} \pm \sym{3} \big) \ens , \\
    \ket{\Psi_{3}}_{\pm 2}& =
    \tfrac{1}{2 \sqrt{2}} \big( \sym{0} \pm \sqrt{3} \sym{1} -
    \sqrt{3} \sym{2} \mp \sym{3} \big) \ens , \\
    \ket{\Psi_{4}}_{\pm 1}& = \tfrac{1}{4}
    \big( - \sqrt{3} \sym{0} \pm 2 \sym{1} + \sqrt{2} \sym{2}
    \pm 2 \sym{3} - \sqrt{3} \sym{4} \big) \ens , \\
    \ket{\Psi_{4}}_{2 \phantom{\pm}}& =
    \tfrac{1}{2} \big( \sym{0} +
    \I \sqrt{2} \sym{2} + \sym{4} \big) \ens, \\
    \ket{\Psi_{6}}_{\pm 1}& = \tfrac{1}{\sqrt{2}}
    \big( \sym{1} \pm \sym{5} \big) \ens , \\
    \ket{\Psi_{6}}_{2 \phantom{\pm}}& =
    \tfrac{1}{4} \big( - \sqrt{3} \sym{0} + \sqrt{5} \sym{2} +
    \sqrt{5} \sym{4} - \sqrt{3} \sym{6} \big) \ens , \\
    \ket{\Psi_{6}}_{\pm 3}& =
    \tfrac{1}{3} \big( \sqrt{2} \sym{0} \pm
    \I \sqrt{5} \sym{3} + \sqrt{2} \sym{6} \big) \ens .
  \end{align}
\end{subequations}
The first index denotes the number $n$ of qubits and the second index
counts the different states.  Interestingly, the states with same $n$
are all \ac{LU}-equivalent to each other. This can be verified by
calculating and comparing their Majorana representations.  For
example, $\ket{\Psi_{3}}_{\pm 1}$ and $\ket{\Psi_{3}}_{\pm 2}$ are all
equivalent to the 3 qubit \ac{GHZ} state via symmetric \acp{LU}:
\begin{equation}\label{ghz_gisin}
  \ket{\Psi_{3}}_{+1} =
  \rotzs ( \pi ) \ket{\Psi_{3}}_{-1} =
  \rotzs ( \pi ) \rotxs ( \tfrac{\pi}{2} ) \ket{\Psi_{3}}_{+2} =
  \rotzs ( \pi ) \rotxs ( \tfrac{\pi}{2} )  \rotzs ( \pi )
  \ket{\Psi_{3}}_{-2} \ens .
\end{equation}
Analogous identities hold for the 4 and 6 qubit states listed above,
and these states represent the tetrahedron state and the octahedron
state, respectively.  We thus arrive at the conclusion that symmetric
states whose reduced density matrices are all maximally mixed possess
an exceptionally high amount of geometric symmetry in their \ac{MP}
distributions. Only four such states exist: Two antipodal points
($n=2$), the equilateral triangle ($n=3$), the regular tetrahedron
($n=4$), and the regular octahedron ($n=6$).

A common property of these point distributions is that they look
exactly the same from the viewpoint of each vertex\footnote{This is
  one of the possible ways to define the Platonic solids. Since
  polyhedra (three-dimensional polytopes) need to have at least four
  vertices, the configurations of two antipodal points ($n=2$) and the
  equilateral triangle ($n=3$) are not considered to be Platonic
  solids.}. It is therefore legitimate to ask why only two of the five
Platonic solids give rise to this kind of maximal global
entanglement. Considering that the cube ($n=8$) and the dodecahedron
($n=20$) neither solve the classical point distribution problems nor
maximise the \ac{GM}, it is perhaps not surprising that they are
missing. But how to explain the absence of the icosahedron ($n=12$)?
We put forward the conjecture that this is because the icosahedron
state cannot be cast with positive coefficients.  It is well-known
that the maximally $n$-tangled states generally allow for a positive
representation, which is why the criterion of a phase-independent
canonical form is sometimes added to the list of conditions for
maximally $n$-tangled states \cite{Osterloh06}.  It is therefore
conceivable that a positive computational representation (up to
\ac{LU}-equivalence) is a necessary property for any state whose
reduced density matrices are all maximally mixed.

\cleardoublepage

\chapter{Links and Connections}\label{connections}

\begin{quotation}
  The aim of this penultimate chapter is to outline some novel links
  between the study of quantum states in terms of the Majorana
  representation and other topics in mathematics and physics.  On
  their own these results are not strong enough to warrant their own
  chapters, and so they are subsumed as independent sections under
  this chapter.

  Firstly, the \quo{Majorana problem} is compared and contrasted to
  spherical point distribution problems that seek to find the most
  non-classical states, such as \quo{anticoherent} spin states and the
  \quo{queens of quantum}.  Secondly, the dual polyhedra of the five
  Platonic solids are linked to the Majorana representations of the
  corresponding symmetric states, thus discovering a quantum analogue
  to the Platonic duals of classical geometry.  Thirdly, the Majorana
  representation is employed to investigate the permutation-symmetric
  ground states of the \ac{LMG} model, with a new proof given for the
  derivation of its \acp{CPP}.
\end{quotation}

\section{\quo{Anticoherent} spin states and \quo{queens of
    quantum}}\label{anticoherent_queens}

As outlined in \sect{majorana_definition_sect}, there exists an
isomorphism between the states of a single spin-$j$ particle and the
symmetric states of $2j$ qubits, and this isomorphism is mediated by
the Majorana representation.  The coherent states of a quantum
mechanical particle are considered to be the most classical states,
and in terms of the Majorana representation these states are those
whose \acp{MP} all coincide at a single point, thus describing as
precisely as possible a \quo{classical} spin vector. This classical
nature of coherent states can furthermore be seen from their
resistance to entanglement formation \cite{Markham03}.  It is
therefore natural to ask whether there also exist \quo{least classical
  states} which are the opposite of spin-coherent states in some
sense, with the expectation that such states exhibit a large amount of
non-locality.  Zimba \cite{Zimba06} defined the \quo{anticoherent}
spin states as those whose polarisation vector vanishes, $\bmr{p} =
\braket{\psi | \bmr{n S} | \psi} = 0$, and whose corresponding
variance $\braket{\psi | ( \bmr{n S} )^{2} | \psi}$ is uniform over
the sphere. Here $\bmr{n}$ is a unit vector in $\mbbrr$, and $\bmr{S}
= ( S_{x} , S_{y} , S_{z} )$ is the spin operator.  This concept is
sometimes generalised to higher-order anticoherence, with a state
$\ket{\psi}$ being anticoherent of order $k$ if $\braket{\psi | (
  \bmr{n S} )^{i} | \psi}$ is independent of $\bmr{n}$ for all $i \leq
k$ \cite{Crann10}.  It has been shown that the quantum states
represented by the five Platonic solids are all anticoherent
\cite{Zimba06}, and this can be understood by means of the
mathematical concept of spherical designs \cite{Crann10}.

Bearing in mind the isomorphism between general spin-$j$ states and
symmetric states of $2j$ qubits, we put forward the question whether
the Majorana representation of the maximally entangled symmetric
states coincides with that of anticoherent spin states, and vice
versa. This question is motivated by the observation that the coherent
states of a spin-$j$ particle are described by $2j$ coinciding
\acp{MP}, which precisely corresponds to the symmetric $2j$ qubit
states with zero entanglement. Unfortunately, it turns out that no
such direct link exists for the \quo{opposite} case. From the fact
that the maximally entangled symmetric state of eight qubits is not a
Platonic solid, it follows that anticoherent \ac{MP} distributions do
not guarantee extremal entanglement.  On the other hand, the maximally
entangled symmetric state of five qubits, the square pyramid state
$\ket{\Psi_{5}}$ discussed in \sect{majorana_five}, is neither
anticoherent nor a spherical design, because the \quo{centre of mass}
of the five \acp{MP} does not coincide with the origin of the Majorana
sphere.  In terms of anticoherent spin states this leads to
$\braket{\Psi_{5} | S_z | \Psi_{5}} \neq 0$.  For spherical designs we
observe that by setting $p(x) = x$ in Definition 2 of \cite{Crann10},
it follows that for all spherical designs the \quo{centre of mass}
must coincide with the sphere's origin.

Anticoherent spin states are only one possibility to define the most
non-classical states.  Giraud \etal \cite{Giraud10} put forward the
concept of the \quo{queens of quantum}. These states have the property
to be furthest away (in terms of the Hilbert-Schmidt metric $\norm{A}
= \Trace (A^{\dagger} A)^{1/2}$) from the set of \quo{classical
  states}, with the latter defined as those states that can be written
as a convex sum of projectors onto coherent states.  This definition
bears resemblance to that of the relative entropy of entanglement
\cite{Vedral98,Vedral97}.  It turns out that the \quo{queens of
  quantum} can always be found in pure states, and that their Majorana
representations have a high degree of geometric symmetry with no
\ac{MP} degeneracies.  In general, however, their Majorana
representations differ from those of our maximally entangled symmetric
states, and they are not identical to the solutions of \toths and
Thomson's problem either.

Intriguingly, Martin \etal \cite{Martin10} found that the \quo{queens
  of quantum} coincide with the solutions of the Majorana problem if
the Hilbert-Schmidt distance is replaced with the Bures distance
$D_{B}$.  For product states the \quo{Bures quantumness} reads
\begin{equation}\label{bures1}
  Q_{B} \left( \pure{\psi} \right) = \min_{\rho_{c} \in \mathcal{C}}
  D_{B} \left( \pure{\psi} , \rho_{c} \right) =
  \min_{\rho_{c} \in \mathcal{C}}
  \sqrt{2 - 2 \sqrt{\braket{\psi | \rho_{c} | \psi}}} \ens ,
\end{equation}
where $\mathcal{C}$ is the convex hull of spin coherent states.  Thus
every $\rho_{c}$ can be written as $\rho_{c} = \sum_{i} \lambda_{i}
\pure{\Phi_{i}}$, where the $\ket{\Phi_{i}}$ are coherent states. This
problem can be reduced to finding the largest overlap of $\ket{\psi}$
with a coherent state $\ket{\Phi}$,
\begin{equation}\label{bures2}
  Q_{B} \left( \pure{\psi} \right) = 
  \sqrt{2 - 2 \max_{\ket{\Phi}} \abs{\bracket{\Phi}{\psi}}} \ens ,
\end{equation}
which is equivalent to the Majorana problem \eqref{maj_problem} of
symmetric states.  In other words, the Majorana representation of the
spin-$j$ state with the largest \quo{Bures quantumness} is identical
to that of the maximally entangled symmetric state of $2j$ qubits in
terms of the geometric measure \cite{Martin10}.

Concluding this section, we have seen that the Majorana
representations of our maximally entangled symmetric states are in
general different from \quo{anticoherent} spin states as well as the
original definition of the \quo{queens of quantum}, but they are
identical to the \quo{queens of quantum} in terms of the Bures
metric. This link means that any solution found for the Majorana
problem of $2j$ qubits will also be a spin-$j$ \quo{Bures-queen of
  quantum}, and vice versa.

\section{Dual polyhedra of the Platonic solids}\label{dualpoly}

For a given polyhedron in 3D space the corresponding dual is obtained
by associating each vertex with a face, and vice versa.  The dual
polyhedra of the five Platonic solids are particularly simple, as they
are Platonic solids themselves \cite{Wenninger83}. As seen in
\fig{platonic_dual}, the octahedron and cube form a dual pair, and so
do the icosahedron and dodecahedron, while the tetrahedron is
self-dual, i.e. it is its own dual.

Interestingly, these duality relationships can also be spotted in the
Majorana representations of the corresponding symmetric quantum
states. For example, we have already seen that the 20 \acp{CPP} of the
icosahedron state $\ket{\Psi_{12}}$ form the vertices of a
dodecahedron. On the other hand, when considering the 20 qubit
dodecahedron state \eqref{dodecahedron_state}, it is easy to verify
that this state has 12 \acp{CPP}, one at the centre of each face, thus
forming the vertices of an icosahedron.  The Majorana representations
of the icosahedron and dodecahedron state are therefore dual to each
other with respect to an interchange of the \acp{MP} and \acp{CPP}.

\begin{figure}
  \centering
  \begin{overpic}[scale=1.5]{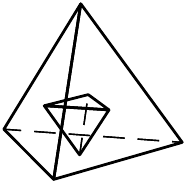}
  \end{overpic}
  \hspace{10mm}
  \begin{overpic}[scale=1.4]{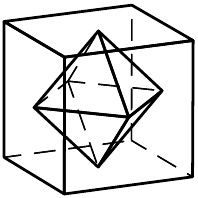}
  \end{overpic}
  \hspace{10mm}
  \begin{overpic}[scale=1.4]{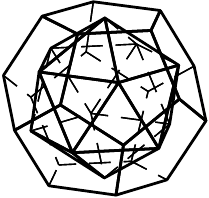}
  \end{overpic}
  \caption[Dual polyhedra of the Platonic
  solids]{\label{platonic_dual} The relationship between the Platonic
    solids and their duals.}
\end{figure}

\begin{figure}
  \centering
  \begin{overpic}[scale=0.65]{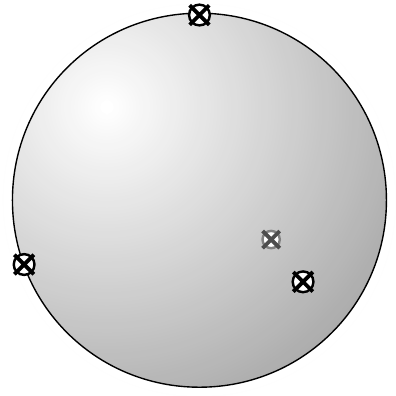}
  \end{overpic}
  \hspace{2mm}
  \begin{overpic}[scale=0.65]{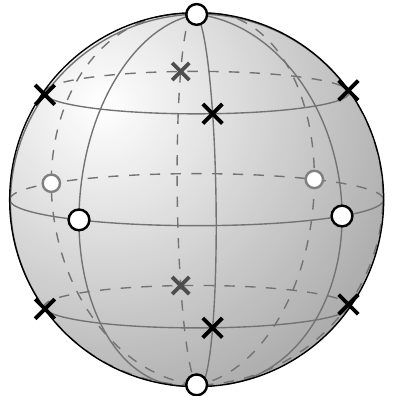}
  \end{overpic}
  \begin{overpic}[scale=0.65]{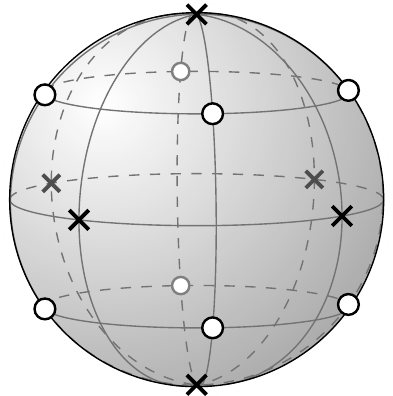}
  \end{overpic}
  \hspace{2mm}
  \begin{overpic}[scale=0.65]{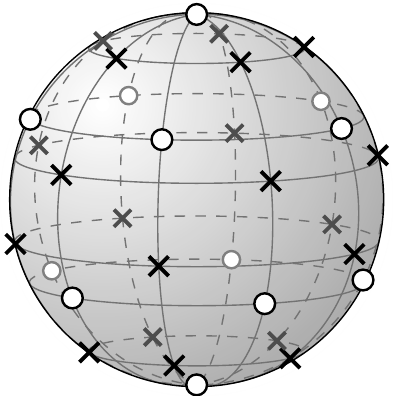}
  \end{overpic}
  \begin{overpic}[scale=0.65]{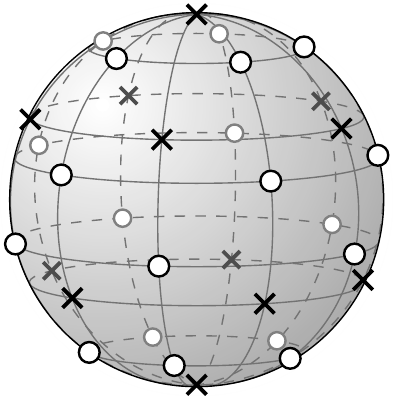}
  \end{overpic}
  \caption[MP and CPP distributions of the \quo{Platonic solids
    states}]{\label{platonic_dual_mpcpp} The \ac{MP} and \ac{CPP}
    distributions of the five \protect\quo{Platonic solid states}.}
\end{figure}

As shown in \fig{platonic_dual_mpcpp}, the same duality relationship
can be observed between the octahedron state and the cube
state. Furthermore, the tetrahedron state is its own dual, with the
\acp{MP} and \acp{CPP} being identical. Unlike the dual of the
Platonic solid, however, the dual state of the tetrahedron state is
not turned \quo{upside down}, but rather coincides with the original
state.

We thus find that the five \quo{Platonic solid states} exhibit the
same duality relationships as the classical Platonic solids, in the
sense that the vertex-face association is replaced with a
\ac{MP}-\ac{CPP} association.  We do not know the benefits of this
mathematical property for quantum information science and related
fields, but it is imaginable that uses can be found, since the five
\quo{Platonic solid states} are exceptional states.  They have been
coined the \quo{perfect states} \cite{Zimba06}, and it was found that
their property of being spherical designs directly implies their
anticoherence \cite{Crann10,Zimba06}. Furthermore, they were found to
be the optimal states for aligning Cartesian reference frames
\cite{Kolenderski08}.  With these recent discoveries in mind, it is
conceivable that the quantum analogue to the Platonic duals could come
in handy at some point in the future.

\section{Lipkin-Meshkov-Glick model}\label{LMG}

Quantum phase transitions are transitions between qualitatively
distinct phases of quantum many-body systems \cite{Sachdev}. Such
transitions play an important role in physical systems, and recent
studies have focused on analysing their phase diagrams in terms of
entanglement.  One-dimensional models such as Ising spins in a
magnetic field allow for exact solutions, but they do not exhibit a
particularly rich structure.  Higher-dimensional models are usually
accessible only through difficult numerical treatment, although
certain symmetries of the Hamiltonian can make the model exactly
solvable. One such integrable model is the \acf{LMG} model, whose
solutions can be derived from an algebraic Bethe ansatz \cite{Pan99},
but the model can also be efficiently treated numerically.  Originally
introduced for nuclear physics \cite{Lipkin65,Lipkin65b,Lipkin65c},
the \ac{LMG} model has since been employed to describe the quantum
tunnelling of bosons between two levels, and thus the Josephson effect
in two-mode Bose-Einstein-condensates \cite{Cirac98}.  It consists of
a system of $n$ mutually interacting spin-$\tfra{1}{2}$ particles
embedded in a transverse magnetic field $h$:
\begin{equation}\label{LMGhamiltonian}
  H = - \frac{1}{n} (\gamma_{x} S_{x}^{2} +
  \gamma_{y} S_{y}^{2}) - h S_z \ens .
\end{equation}
Due to the symmetries of the Hamiltonian it suffices to consider $h
\geq 0$ and $\abs{\gamma_{y}} \leq \gamma_{x}$.  This anisotropic
$X$-$Y$-system is known to undergo a second-order quantum phase
transition at $h = \gamma_{x}$ or $h = \gamma_{y}$.  Investigating the
zero-temperature phase diagram for the ground state reveals two
phases, namely a \textbf{symmetric phase} for $h > \gamma_{x}$ where
the ground state is unique, and a \textbf{broken phase} for $h <
\gamma_{x}$ where the ground state becomes two-fold degenerate in the
thermodynamic limit ($n \to \infty$).  The full spectrum is more
complicated, with four different zones arising in the phase diagram
\cite{Ribeiro07,Ribeiro08}.  The ground state always lies in the
maximum spin sector, and is therefore symmetric.  In the large-field
limit ($h \to \infty$) the ground state becomes separable $\ket{\psi}
= \ket{\!  \uparrow}^{\otimes N}$, and in the thermodynamic limit ($n
\to \infty$) the spectrum of $H$ remains discrete.

Among other entanglement measures, the von Neumann entropy $S(h)$,
which characterises the entanglement of a bipartite decomposition, has
been used to analyse the \ac{LMG} model \cite{Latorre05}. A maximum at
the critical point was found, which is consistent with a theoretical
conjecture in \cite{Hines05}.  Furthermore, the von Neumann entropy of
the ground state scales logarithmically with the block size $L$ of the
bipartite decomposition \cite{Latorre05}.

\subsection{Distribution of the MPs}

The symmetric eigenstates of the \ac{LMG} model
\cite{Ribeiro08,Ribeiro07} can be represented in the spin coherent
basis by their Majorana polynomial $\psi( \alpha ) \propto
\prod_{k=1}^{2 s} ( \alpha - \alpha_k )$, and the Majorana roots
$\alpha_k$ become the \acp{MP} $\ket{\phi_k}$ of the corresponding
Majorana representation by means of an inverse stereographic
projection.

As outlined in \cite{Ribeiro07}, the Schr\"{o}dinger equation of the
\ac{LMG} model has the form
\begin{equation}\label{schrodinger}
  \left[ \frac{P_2 (\alpha)}{(2s)^2} \partial^{2}_{\alpha} +
    \frac{P_1 (\alpha)}{2s} \partial_{\alpha} + P_{0} ( \alpha ) \right]
  \Psi ( \alpha ) =  \epsilon \, \Psi ( \alpha ) \enspace ,
\end{equation}
where $P_{0}, P_{1}$ and $P_{2}$ are polynomials in $\alpha \in
\mbbc$, and $\Psi ( \alpha )$ is the Majorana polynomial of the
eigenstate.  We easily verify that if \eq{schrodinger} is solved for a
tuple of roots $\{ \alpha_{1} , \ldots , \alpha_{2s} \}$, then this is
also the case for the tuples $\{ - \alpha_{1} , \ldots , - \alpha_{2s}
\}$ and $\{ \cc{\alpha}_1 , \ldots , \cc{\alpha}_{2s} \}$.  This
implies that if $\alpha_k$ is a Majorana root, then so are $-
\alpha_k$, $\cc{\alpha}_k$ and $- \cc{\alpha}_k$.  On the Majorana
sphere this leads to reflective symmetries of the \ac{MP} distribution
along the $X$-$Z$-plane and the $Y$-$Z$-plane, and combining these two
reflections, they give rise to a rotational symmetry along the
$Z$-axis with rotational angle $\varphi = \pi$.  These symmetries are
also visible in the examples of eigenstates shown in
\cite{Ribeiro07,Ribeiro08}.  From \lemref{rot_symm} and
\lemref{maj_real} we obtain the result that the coefficients of all
eigenstates $\ket{\psi} = \sum_{m= -s}^{s} a_{m} \ket{s,m}$ of the
\ac{LMG} model are constrained to $a_{m} \in \mathbb{R}$ for even $m$
and $a_{m} = 0$ for odd $m$.

It is known that for all eigenstates of the \ac{LMG} model the
\acp{MP} are distributed along two curves $\mathcal{C}_{0}$ and
$\mathcal{C}_{1}$ on the Majorana sphere \cite{Ribeiro07,Ribeiro08}.
The curve $\mathcal{C}_{0}$ always coincides with the imaginary great
circle, and while $\mathcal{C}_{1}$ coincides with the real great
circle in the simplest case, it is in general different.  The \acp{MP}
on the curves change with the parameters of the Hamiltonian
\eqref{LMGhamiltonian}.  For the ground state all \acp{MP} lie on
$\mathcal{C}_{0}$, as shown in \fig{lmg_figure}(a), and transitions
between neighbouring energy levels correspond to pairs of \acp{MP}
switching from one curve to the other.  This implies that the $k$-th
excited state has $2k$ \acp{MP} on $\mathcal{C}_{1}$ and $2(s-k)$
\acp{MP} on $\mathcal{C}_{0}$.  From our discussion in
\sect{determiningslocc} it is clear that this qualitative difference
of the \ac{MP} distributions renders these states
\ac{SLOCC}-inequivalent to each other, which means that the different
energy levels of the \ac{LMG} model correspond to different types of
entanglement.

In order to study the phase transition of the ground state in the
thermodynamic limit, it is convenient to simplify the Hamiltonian
\eqref{LMGhamiltonian} as follows
\begin{equation}\label{LMGhamiltoniansimple}
  H = - \frac{1}{n} \gamma S_y^2 - h S_z \ens .
\end{equation}
This simplified Hamiltonian does not give rise to the full phase
diagram anymore, but otherwise it is expected to exhibit the same
qualitative behaviour as the Hamiltonian \eqref{LMGhamiltonian}.  In
the thermodynamic limit the amount of \acp{MP} becomes infinite, so
the discrete distribution of \acp{MP} turns into a continuous
probability distribution.  From the simplified Hamiltonian
\eqref{LMGhamiltoniansimple} the following distribution has been
derived for the latitudinal angle of the \acp{MP} in the ground state
\cite{Ribeiro,Ribeiro08}:
\begin{equation}\label{mp_lmg}
  P_h (\theta) = \left\{ 
    \begin{array}{l l}
      \frac{\gamma + h \cos \theta}{2 \pi \gamma}&
      \quad \text{for $h \leq \gamma$ (broken phase)}
      \vspace{0.2cm} \\
      \frac{\sqrt{h (1 + \cos \theta) (2 \gamma - h +
          h \cos \theta)}}{2 \pi \gamma}&
      \quad \text{for $h \geq \gamma$ (symmetric phase)} \\
    \end{array} \right.
\end{equation}
These two expressions converge at $h = \gamma$, and since only the
ratio $\tfra{h}{\gamma}$ is physically significant, we can set $\gamma
= 1$ in the following.  The function $P_h (\theta)$ is normalised for
all values of $h$ within the respective areas where it is
well-defined, i.e.  $\theta \in [- \pi , \pi]$ for $h \leq 1$, and
$\theta \in [- \arccos \frac{h-2}{h} , \arccos \frac{h-2}{h}]$ for $h
\geq 1$.  Plots of $P_h (\theta)$ for three different values of $h$
are shown in \fig{lmg_figure}(b).

In the corresponding finite-spin case the \acp{MP} are distributed
along the imaginary circle in a pairwise fashion $\ket{\phi_i} =
\co_{\theta} \ket{0} \pm \I \si_{\theta} \ket{1}$, as shown in
\fig{lmg_figure}(a).  The total state $\ket{\Psi}$ is therefore
positive, and because of \lemref{lem_pos_cps} it suffices to determine
a positive \ac{CPP}.  The \acp{CPP} come in pairs too, $\ket{\sigma} =
\co_{\vartheta} \ket{0} \pm \si_{\vartheta} \ket{1}$, so the overlap
between any \ac{MP} and \ac{CPP} has the form
$\bracket{\phi_i}{\sigma} = \co_{\theta} \co_{\vartheta} \pm \I
\si_{\theta} \si_{\vartheta}$.

\begin{figure}[ht]
  \hspace{8mm}
  \begin{overpic}[scale=0.4]{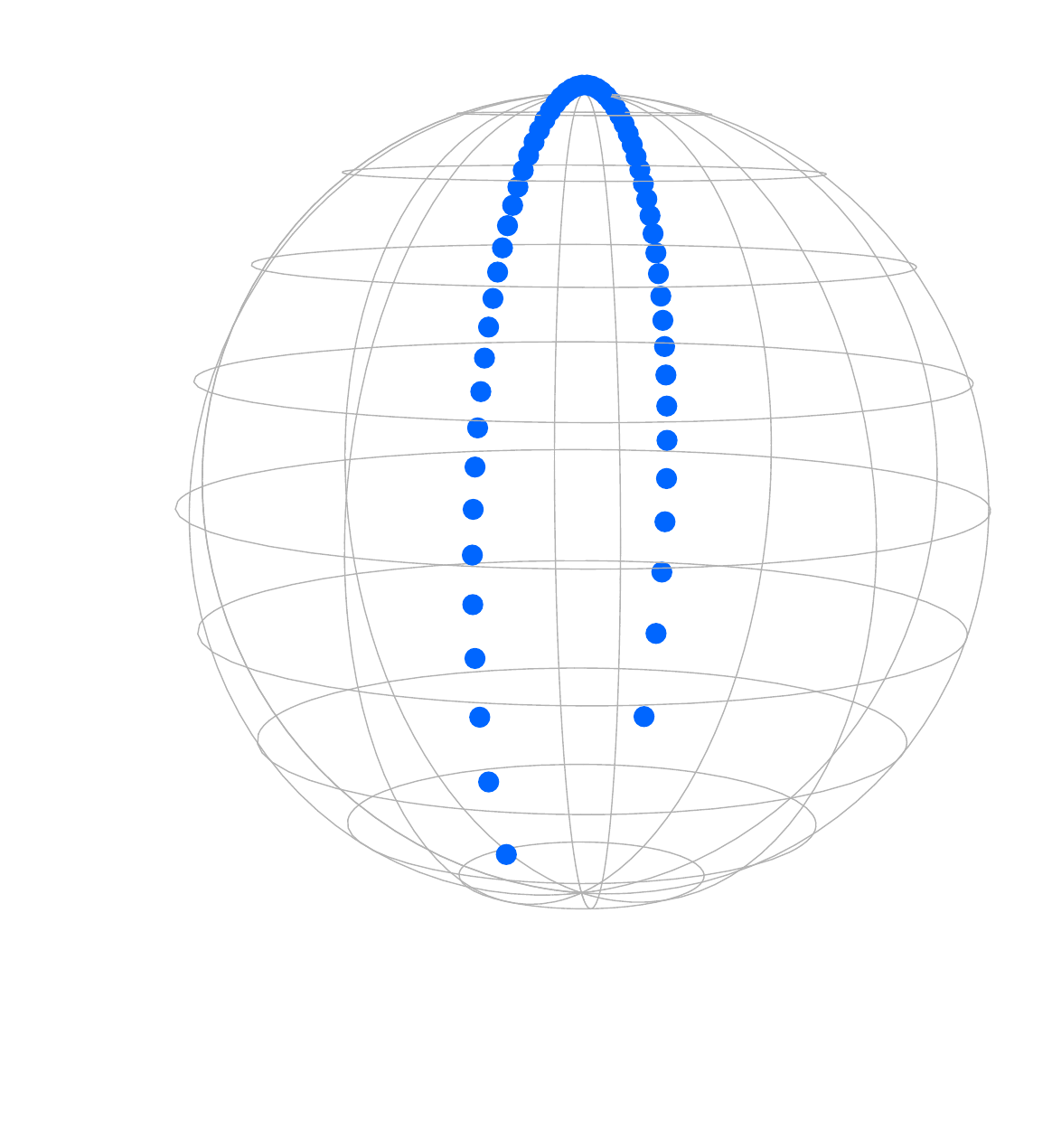}
    \put(5,5){(a)}
  \end{overpic}
  \hspace{12mm}
  \begin{overpic}[scale=0.65]{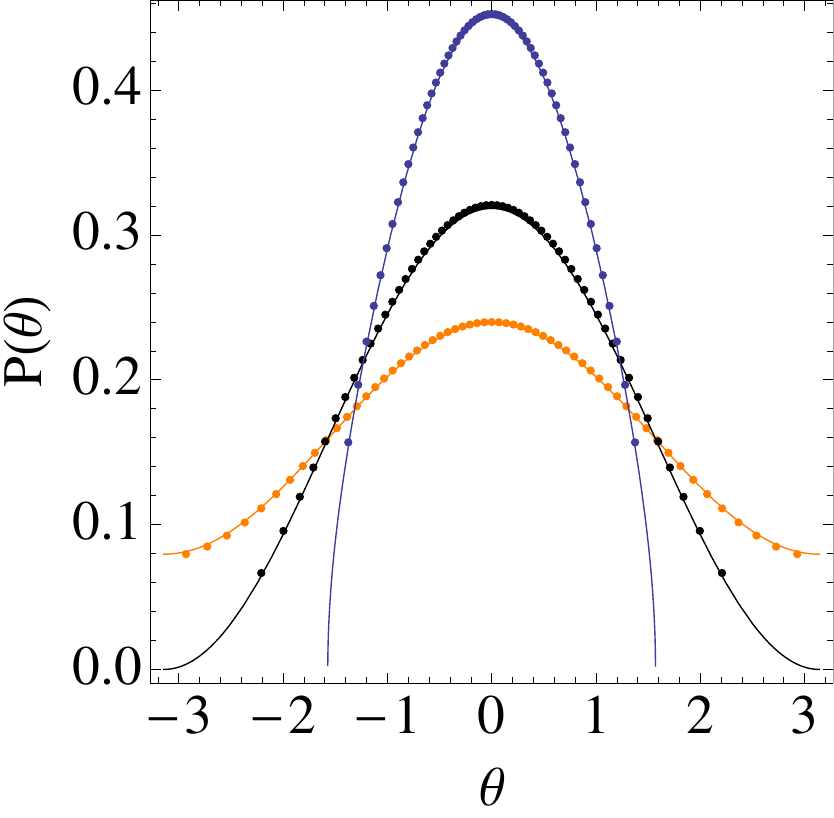}
    \put(-10,5){(b)}
  \end{overpic}
  \caption[MP distribution of the ground state in the LMG
  model]{\label{lmg_figure} The \ac{MP} distribution of the ground
    state of the \ac{LMG} model at the phase transition ($h=1$) is
    shown in (a) for total spin $s=30$. The $60$ \acp{MP} all lie on
    the imaginary great circle, and the corresponding eigenstate is
    positive.  The latitudinal probability distribution
    \protect\eqref{mp_lmg} of the \acp{MP} is shown in (b) for the
    three values $h = \frac{1}{2}$ (orange), $h = 1$ (black) and $h =
    2$ (blue). The finite-spin case and the thermodynamic limit are
    shown as point distributions and continuous lines,
    respectively. It is seen that the ring of \acp{MP} closes at the
    phase transition. [Figures generated from \textsc{Mathematica}
    code provided by Pedro Ribeiro.]}
\end{figure}

\subsection{Determination of the CPPs}

In the finite-spin case the location of the positive \ac{CPP} is
determined by the maximum of the spherical amplitude function $g^2 (
\vartheta ) = \abs{\bracket{\Psi}{\sigma (\vartheta ) }^{\otimes n}}^2
= \prod_{i = 1}^{n} \abs{\bracket{\phi_i}{\sigma (\vartheta )}}^2$
over all single-qubit states $\ket{\sigma ( \vartheta )} =
\co_{\vartheta} \ket{0} + \si_{\vartheta} \ket{1}$.  In order to find
a continuous variant of this quantity for the thermodynamic limit, we
consider its logarithm because it will turn the product into a sum,
which naturally becomes an integral over a probability distribution in
the infinite limit:
\begin{subequations}\label{limes}
  \begin{gather}
    \mathcal{G} (\vartheta ) := - \log_2 \, g^2 ( \vartheta ) = - \log_2
    \left[ \prod_{i = 1}^{n} \abs{\bracket{\phi_i}{ \sigma ( \vartheta )
        }}^2 \right] = - \sum_{i = 1}^{n} \log_2
    \abs{\bracket{\phi_i}{\sigma (\vartheta ) }}^2 \\
    \quad
    \xrightarrow{n \to \infty}
    \quad -
    \int\limits_{0}^{2 \pi} {P_{h} (\theta)} \log_2
    \abs{\bracket{\phi(\theta)}{\sigma (\vartheta ) }}^2 \,
    \D \theta \ens .
  \end{gather}
\end{subequations}
The monotonicity of the logarithm ensures that the maximum, and thus
the \ac{CPP}, remains the same. The integral runs over the imaginary
great circle where the \acp{MP} lie.  Note that $g^2 (\vartheta )$ is
missing the normalisation factor $K$ required for calculating the
geometric entanglement, and an analytic calculation of this factor is
believed to be intractable. For the determination of the \acp{CPP},
however, we do not need the normalisation factor, and our derivation
of the \acp{CPP} can prove helpful for further studies of the
geometric entanglement of the \ac{LMG} ground state.

In the following we present an analytic calculation of the positive
\ac{CPP} in the broken phase ($h \leq 1$).  The logarithmic amplitude
function along the real great circle is then
\begin{equation}\label{int}
  \begin{split}
    \mathcal{G} ( \vartheta )& =
    - \int\limits_{0}^{2 \pi} \frac{(1 + h \cos \theta )}{2 \pi}
    \log_2 \left[ \co_{\theta}^2 \co_{\vartheta}^2 +
      \si_{\theta}^2 \si_{\vartheta}^2 \right] \D \theta \\
    {}& = - \frac{1}{2 \pi} \int\limits_{0}^{2 \pi} (1 + h \cos \theta )
    \log_2 \left[ \tfrac{1}{2} ( 1 + \cos \vartheta \cos \theta ) \right]
    \D \theta \enspace .
  \end{split}
\end{equation}

This integral can be solved with the help of the following two
definite integrals (found with \textsc{Mathematica}) that hold for $a
\in [0,1]$:
\begin{align}\label{defints}
  {}& \int\limits_{0}^{2 \pi} \ln \left[ 1 + a \cos \theta \right] \D
  \theta = 2 \pi \ln \left[ \tfrac{1}{2}
    \left(1 + \sqrt{1-a^2} \right) \right] \ens , \\
  {}& \int\limits_{0}^{2 \pi} \cos \theta \ln
  \left[ 1 + a \cos \theta \right] \D \theta =
  \frac{2 \pi}{a} \left( 1 - \sqrt{1-a^2} \right) \ens .
\end{align}

With this we obtain
\begin{equation}\label{logamp}
  \mathcal{G} ( \vartheta ) =  2 - \log_2
  \left( 1 + \sqrt{ 1 - \cos^2 \vartheta } \right) - \frac{h}{\ln 2}
  \frac{1 - \sqrt{ 1 - \cos^2 \vartheta }}{\cos \vartheta } \enspace .
\end{equation}

A numerical integration of \eq{int} verifies that \eq{logamp} is
correct.  Setting the first derivative to zero yields the locations of
the \ac{CPP}:
\begin{equation}\label{integral}
  \frac{\partial \mathcal{G} (\vartheta)}{\partial \vartheta}
  = \frac{1}{\ln 2} \left( \frac{h -
      \cos \vartheta }{1+ \sin \vartheta} \right) \qquad , \qquad
  \frac{\partial \mathcal{G} (\vartheta)}{\partial \vartheta} = 0 \quad
  \Longleftrightarrow \quad h = \cos \vartheta \enspace .
\end{equation}

Thus the relationship between the parameter $h$ and the latitude
$\vartheta$ of the \acp{CPP} is $h = \cos \vartheta$, yielding the two
\acp{CPP} in the broken phase: $\ket{\sigma} = \sqrt{\frac{1+h}{2}}
\ket{0} \pm \sqrt{\frac{1-h}{2}} \ket{1}$.  This result has previously
been obtained via the mean-field approach in \cite{Dusuel05}, but our
method is different because instead of minimising the energy, we
maximised the spherical amplitude function of the geometric measure.

In the symmetric phase ($h \geq 1$) the logarithmic amplitude function
reads
\begin{equation}\label{intsymm}
  \mathcal{G} ( \vartheta ) = -
  \hspace{-6mm}
  \int\limits_{- \arccos \frac{h-2}{h}}^{\arccos \frac{h-2}{h}}
  \hspace{-6mm}
  \tfrac{1}{2 \pi}
  \sqrt{h (1 + \cos \theta ) (2 - h + h \cos \theta )}
  \log_2 \left[ \tfrac{1}{2} ( 1 +
    \cos \vartheta \cos \theta ) \right] \D \theta
  \enspace .
\end{equation}
An analytic solution of this integral is not known, but it is easily
verified numerically that the only \ac{CPP} is the north pole
$\ket{\sigma} = \ket{0}$ for all $h > 1$, something that has been
noted before \cite{Orus08}.

Surprisingly, the \acp{MP} and \acp{CPP} of the \ac{LMG} ground state
exhibit precisely the same qualitative behaviour under variation of
the magnetic field $h$ as observed in the simple 3 qubit model that we
investigated in \sect{two_three}.  In \fig{3_graph} the \ac{CPP}
initially stays at the north pole until the pair of \acp{MP} has moved
sufficiently far downwards.  At the distribution in \fig{3_graph}(c) a
\quo{phase transition} occurs, where the \ac{CPP} abruptly leaves the
north pole.  This is the same behaviour that the \ac{CPP} of the
\ac{LMG} ground state exhibits around the phase transition $h=1$ in
the thermodynamic limit, i.e. for infinitely many \acp{MP}.

\cleardoublepage

\chapter{Conclusions}\label{conclusion}

\begin{quotation}
  The aim of this final chapter is to provide a brief review of the
  main results presented in this thesis, and to give an outlook at
  promising further ideas and strategies that may be worthy of further
  investigation.
\end{quotation}

\section{Summary of  main results}\label{summary_results}

In this thesis permutation-symmetric quantum states were investigated
from various perspectives, such as entanglement classification,
extremal entanglement, invariants, and connections to related physical
phenomena.  For this the Majorana representation as well as the
geometric measure of entanglement were the essential tools.

\subsubsection*{Chapter 2: Geometric Measure of Entanglement}

A variety of analytical results about the \ac{GM} was derived in
\chap{geometric_measure}. \Theoref{numberofcps} predicts that the
maximally entangled pure state in terms of the \ac{GM} lies in the
span of its \acp{CPS}, and this theorem can be straightforwardly
adapted for the case of permutation-symmetric states.  Considering
arbitrary $n$ qubit states with $\ket{0}^{\otimes n}$ as a \ac{CPS},
the necessary conditions for the coefficients of these states were
presented in \theoref{coeff_theo}, and it was argued that the
conditions can be viewed as a standard form similar to the one of
Carteret \etal \cite{Carteret00}.  From
\theoref{theo_positive_entanglement} and \corref{max_pos_ent} we can
conclude that complex phases are in general an indispensable
ingredient for highly entangled multipartite states.  A new proof for
the upper bound $\Eg \leq \log_2 ( n+1 )$ on the maximal geometric
entanglement of $n$ qubit symmetric states was derived in
\theoref{const_integral}, and the proof of this theorem ignited the
idea of visualising symmetric states by spherical bodies of constant
volume.  Finally, it was shown in \theoref{mbqc_example} that W states
of an arbitrary, but fixed number of excitations cannot be useful for
\ac{MBQC}, not even in the approximate regime.

\subsubsection*{Chapter 3: Majorana Representation and Geometric
  Entanglement}

In \chap{majorana_representation} the Majorana representation was
employed to gain a better understanding of the geometric entanglement
of permutation-symmetric multiqubit states.  Effective visualisations
of the entire information about the entanglement of symmetric states
were presented in \sect{visualisation}, and in \sect{two_three} they
were applied to review the two and three qubit case.  With regard to
the concepts of totally invariant states and additivity, we discovered
in \sect{invariant_and_additivity} that states which are both positive
and totally invariant (e.g. the tetrahedron, octahedron and cube
state) are additive with respect to three distance-like entanglement
measures.  The relationship between the \quo{Majorana problem} of
determining the maximal symmetric entanglement and classical
optimisation problems on the sphere was elucidated in
\sect{extremal_point}.  In \sect{analytic} a variety of analytical
results were derived that link the analytic form of symmetric states
to the distribution of their \acp{MP} and \acp{CPP}.  For example,
\lemref{rot_symm} established that for a cyclic symmetry around the
$Z$-axis of the Majorana sphere to exist, many coefficients of the
underlying state need to vanish, and \lemref{maj_real} showed that
states are real \ac{iff} they exhibit an $X$-$Z$-plane reflective
symmetry.  In \theoref{theo_general_maj_rep} a \quo{generalised
  Majorana representation} was derived where the sum over all
permutations of the $n$ \acp{MP} is replaced with a sum over all
permutations of the subsets of an arbitrary, but fixed partition of
the \acp{MP}.  This generalisation has the advantage that the
analytical treatment of many \ac{MP} distributions can be simplified
by considering certain subsets of \acp{MP}, e.g. those \acp{MP} that
are equidistantly distributed over a circle.  In the concluding
\sect{pos_symm_states} some interesting results were obtained for the
Majorana representation of positive symmetric $n$ qubit states.  In
particular, they can have at most $2n - 4$ \acp{CPP}, with the
possible locations narrowly pinned down to either the positive
half-circle of the Majorana sphere, or to horizontal circles
corresponding to a cyclic $Z$-axis symmetry.

\subsubsection*{Chapter 4: Maximally Entangled Symmetric States}

Strong candidates for the maximally entangled symmetric states of up
to 12 qubits in terms of the \ac{GM} were found in \chap{solutions}.
The combination of analytical and numerical methods employed for the
search were outlined in \sect{methodology}.  Visualisations by means
of the \acp{MP} and \acp{CPP}, as well as the spherical amplitude and
volume functions proved to be very useful tools in the search for high
and maximal entanglement.  The \acp{CPP} of the positive-valued
Platonic states (tetrahedron, octahedron and cube) were found to
follow immediately from the rotation properties of the \acp{MP}.
Comparisons with the extremal distributions of \toths and Thomson's
problems show that, in some cases, the optimal solution to the
Majorana problem is the same, but in other cases it significantly
differs.  In \sect{discussion} the results obtained were interpreted
and discussed, and an overview of the properties of the investigated
states was presented in \tabref{table4.5}.  With regard to the
entanglement scaling it was found that our results are consistent with
the theory and similar studies, and that maximally entangled symmetric
states seem to admit a positive computational form only for $n<10$
qubits.  The distribution behaviour of the \acp{MP} and \acp{CPP}
could be appropriately explained with the spherical amplitude and
volume functions, which is due to the fact that the global maxima and
minima are manifestations of the \acp{CPP} and \acp{MP}, respectively.
Motivated by Euler's formula, \conref{conjecture_cpp} affirms our
belief in a deeper-lying geometric connection between the \ac{MP}
distribution and the maxima of the spherical amplitude function,
something that would yield a general upper bound of $2n-4$ \acp{CPP}.

\subsubsection*{Chapter 5: Classification of Symmetric Entanglement}

In \chap{classification} the entanglement of $n$ qubit symmetric
states was investigated from qualitative viewpoints.  The three
entanglement classification schemes \ac{LOCC}, \ac{SLOCC} and the
\acf{DC} were reviewed for symmetric multiqubit states in
\sect{overview_entclass}.  It was found in \sect{mobius} that the \mob
transformations from complex analysis do not only allow for a simple
and complete description of the freedoms present in symmetric
\ac{SLOCC} operations, but also provide a straightforward
visualisation of these freedoms by means of the Majorana sphere.  The
symmetric \ac{SLOCC} classes of up to 5 qubits were fully
characterised by representative states with simple \ac{MP}
distributions in \sect{representative}, and in the 4 qubit case these
representations are unique, unlike other classification schemes such
as the \acp{EF}.  Comparing the symmetric \ac{SLOCC} classes to the
\acp{EF}, it is found in \sect{families} that the partition into
symmetric \ac{SLOCC} classes is a refinement of the partition into
\acp{EF}.  In \sect{determiningslocc} it was seen how \quo{invariants}
of \mob transformations, such as circles, angles and cross-ratios,
allow one to check whether symmetric states are \ac{SLOCC}-equivalent
or not. In particular, the (S)LOCC-inequivalence of complex conjugate
states could be readily explained with geometric chirality.  The
different values of the cross-ratio under permutations was linked to
the generic \ac{DC} class of 4 qubits by means of \ac{SLOCC}
invariants in \sect{symm_inv}.  Global entanglement measures were
reviewed in \sect{global_ent}, and it was found that the tetrahedron
state and other symmetric states with high geometric entanglement or
symmetries in their \ac{MP} distribution play a prominent role in the
maximisation of these entanglement monotones.

\subsubsection*{Chapter 6: Links and Connections}

In the tripartite \chap{connections} several links between the
Majorana representation and related topics in mathematics and physics
were highlighted.

Two different definitions of maximally non-classical spin-$j$ states,
namely the \quo{anticoherent} spin states and the \quo{queens of
  quantum}, were linked to the corresponding symmetric states by means
of the Majorana representation in \sect{anticoherent_queens}.  It was
found that the \ac{MP} distribution of maximally entangled symmetric
$n$ qubit states in terms of the geometric measure does in general not
describe anticoherent spin states, but it coincides with the
\quo{Bures-queens of quantum}.

In \sect{dualpoly} it was discovered that an analogue to the dual
polyhedra of the five Platonic solids exists for the corresponding
symmetric states, in the sense that the sets of \acp{MP} and \acp{CPP}
are interchanged.  This deepens the relationship between the Majorana
representation and the polyhedra of classical geometry.

Finally, the permutation-symmetric ground state of the \ac{LMG} model
in the thermodynamic limit was investigated in light of the Majorana
representation in \sect{LMG}.  By making a suitable transition from a
discrete to a continuous \ac{MP} distribution for the thermodynamic
limit, we found a new method to prove the degeneracy and locations of
the \acp{CPP} in the broken phase, something that could be used for
investigating the geometric entanglement of the \ac{LMG} ground state.

\section{Outlook and new ideas}\label{outlook}

It is all too natural that some of the ideas and strategies that
spring up during a research degree cannot be investigated with the
rigour they deserve. This is especially true for major open questions
which may require a substantial amount of further literature review,
calculations or programming.  Here I will briefly outline several open
questions and promising new ideas related to topics of the present
thesis.

\subsection*{Quantum computing with liquid Helium}

Liquid Helium bubbles with stable electron patterns above their
surface were already mentioned in \sect{extremal_point}.  When liquid
$^4$He undergoes an electrohydrodynamic instability in a vacuum, the
system can emanate small bubbles of liquid Helium with millions to
billions of electrons hovering above the surface
\cite{Albrecht87,Leiderer95}.  The electrons are localised in a stable
potential well generated by the long-range positive mirror charge in
the dielectric Helium surface, and the (short-range) Pauli principle
which prevents the electrons from falling back into the liquid Helium
\cite{Albrecht87,Leiderer95}.  This leads to the creation of a nearly
ideal 2D electron gas described by a 1D hydrogenic spectrum
\cite{Albrecht87,Platzman99}.

The radius of such Helium drops is typically between 10$\mu$m and
100$\mu$m, with the electrons located around 100\text{\AA} above the
surface, and separated by 2000\text{\AA} or more from their nearest
neighbours \cite{Albrecht87}, see \fig{helium}. Typical electron
numbers are $10^5$ to $10^7$, but there can be as many as $10^9$.
Lifetimes for the Helium drop of more than $100$ms have been achieved
experimentally, and the feasibility of much longer times in a
quadrupole configuration has been proposed \cite{Albrecht87}.  Very
low decoherence rates can be achieved by cooling the system down into
the millikelvin range, yielding spin coherence times beyond 100s for
the electrons above the Helium \cite{Lyon06}.

\begin{figure}[ht]
  \centering
  \includegraphics[scale=0.55]{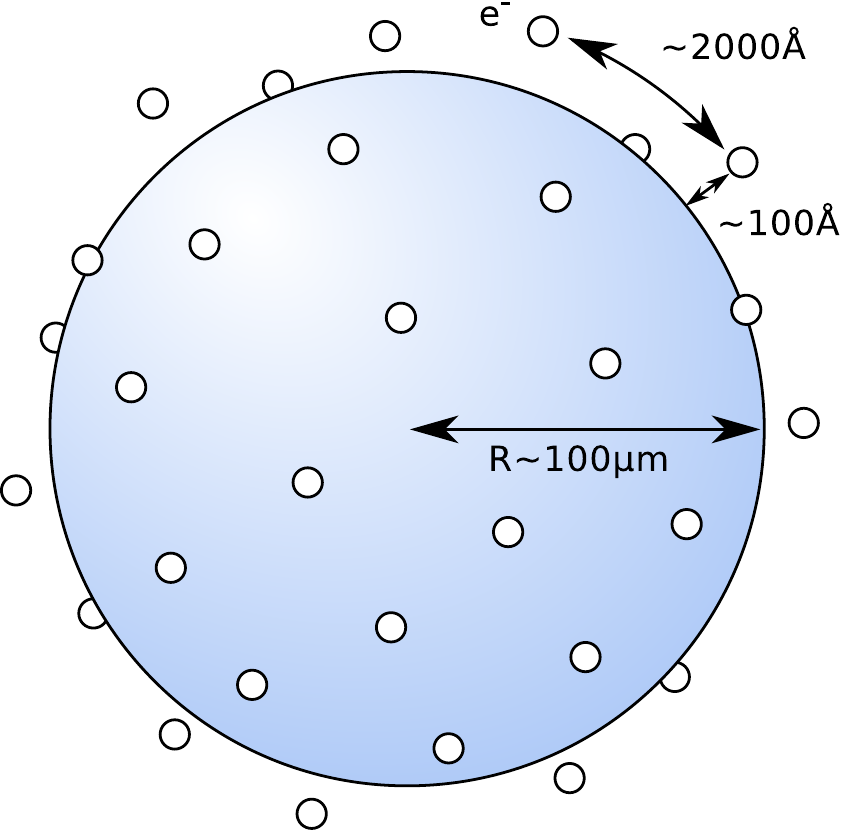}
  \caption[Electrons above the surface of liquid Helium
  drops]{\label{helium} Schematic diagram of electrons hovering above
    the surface of a liquid $^4$Helium drop. The number of electrons
    is typically in the range of $10^5$ to $10^7$, and when cooling
    the system down into the millikelvin range the electrons are
    expected to \protect\quo{freeze} into a pattern that solves
    Thomson problem \protect\eqref{thomson_def}.}
\end{figure}

For these reasons several schemes have been proposed for quantum
computing with liquid Helium, utilising the spin of electrons above a
flat Helium layer on a gate electrode arrangement \cite{Lyon06}, or by
employing the lowest two hydrogenic levels of individual electrons in
the potential well above a Helium film \cite{Platzman99}.  Here we
briefly sketch a scheme employing the spin of the electrons around a
Helium bubble, as seen in \fig{helium}.  The Coulomb energy between
pairs of neighbouring electrons is of the order 10K which is much
larger than $k_{\text{B}} T$ in a system that is cooled down to
millikelvins, and a phase transition of the electrons into a
\quo{frozen} 2D pattern has been observed experimentally
\cite{Grimes79}.  Considering that the electrons describe a nearly
ideal 2D Coulomb system, we expect this pattern to be a solution of
Thomson's problem \eqref{thomson_def}. Since the point distributions
of Thomson's problem have been verified to represent highly entangled
symmetric states, one can ask whether the geometry of the physical
setup could be employed to generate almost maximally entangled $n$
qubit symmetric states with possible values of $n$ spanning several
orders of magnitude.  To do so, one would need to find a way to
associate the spin of the electrons with the states of the
corresponding \acp{MP} in the Majorana representation of the symmetric
state.  It is not clear how exactly this could be achieved
experimentally, and in particular how the permutation present in the
Majorana representation \eqref{majorana_definition} translates to an
experimental setup, but any technique should take advantage of the
fact that the spatial distribution of the spin-$\tfrac{1}{2}$ systems
precisely matches the Majorana representation of the desired symmetric
state.  If it were possible to do this, one could easily generate
highly entangled symmetric quantum states that are subject to very
little decoherence and that could be used for applications in quantum
information theory that rely on symmetric states. To do so, it would
also be necessary to devise techniques for addressing and manipulating
the individual spins in a controlled way.

\subsection*{Maximally entangled states}

The question of which states of a given Hilbert space are maximally
entangled with respect to the geometric measure (or other entanglement
measures) is still not solved, although we were able make significant
progress for symmetric states.  For symmetric $n$ qubit states we
found the solutions for the first few $n$, and in the process of doing
so we gained a good understanding of how the \acp{MP} and \acp{CPP} of
highly entangled symmetric states are distributed.  For arbitrary
multipartite systems we found in \theoref{numberofcps} that the
maximally entangled state always lies in the span of its
\acp{CPS}. The dimension of this span remains an open problem,
although we expect it to be large in general, namely of the same order
as the Hilbert space itself.  Another open question is whether the set
of \acp{CPS} is in general discrete or continuous.

In the context of maximally 4-tangled states {\DJ}okovi\'{c} and
Osterloh recently stated that \quo{The account of the complementary
  set of non-graph states as a resource for quantum information
  processing is largely unexplored} \cite{Dokovic09}.  Being the
single non-graph state among the three maximally 4-tangled states, the
tetrahedron state should therefore be considered a prime focus of
future research.  A summary of the exceptional position that the
tetrahedron state holds in the 4 qubit Hilbert space was already given
in \sect{majorana_four}, and it is likely that further intriguing
properties can be found for this state or similar states (e.g. the
other Platonic states).

\subsection*{Morse theory}

There are still open questions with regard to the number and
distribution of the \acp{CPP} of symmetric states. Most notably, it
was shown in \sect{pos_symm_states} that with the exception of the
Dicke states (which are the only ones with a continuous rotational
symmetry of the \ac{MP} distribution) all positive symmetric $n$ qubit
states have at most $2n-4$ \acp{CPP}.  In \sect{number_cpp} the same
upper bound was conjectured to hold for general symmetric states as
well, with Euler's formula for convex polyhedra being a strong
indication in favour of this.  The conjectured relationship between
the surfaces of the polyhedron described by the \acp{MP} and the local
maxima in the spherical amplitude function is particularly apparent
for the five \quo{Platonic states} shown in \fig{platonic_dual_mpcpp}:
each \ac{CPP} of the tetrahedron state lies antipodal to the centre of
a face, and the \acp{CPP} of the other four Platonic states all lie at
the centre of a face.

One potential way to shed light on this relationship is Morse theory,
a branch of differential topology that investigates how the topology
of a manifold is related to the stationary points, such as maxima and
minima, of real-valued functions defined on that manifold
\cite{Milnor,Matsumoto}.  To give an example, the different topologies
of a line and a circle manifest themselves in the fact that a line
admits continuous functions with arbitrarily large values (e.g. $f(x)
= x^2$), while the codomain of continuous functions defined on a
circle is limited to finite values, thus assuming a maximum value
somewhere on the circle.  Recalling that the global maxima of the
spherical amplitude function are the \acp{CPP}, and that the global
minima -- the zeroes -- lie diametrically opposite to the \acp{MP}, it
is straightforward to expect that the topology of the Majorana sphere,
which is simply the 3D sphere $\mathcal{S}^2$, can tell us more about
the number of \acp{CPP}.

\subsection*{Extensions of the Majorana representation}

Finally we remark that questions still remain open as to how the
elegant Majorana representation of pure spin-$J$ states -- or
equivalently pure symmetric states of $2J$ qubits -- can be extended
to mixed states or to states of several spin-$J$ particles.  Efforts
have been made to find meaningful generalisations to $n$ spin-$J$
particles, most notably in \cite{Kolenderski10}, where an insightful
generalisation was derived with the help of the Schur-Weyl duality.  A
generalisation to mixed states is considered much more difficult,
because the number of parameters present in mixed states of a spin-$J$
particle scales much faster with the dimension than $2J$ Bloch vectors
in the interior of the Majorana sphere can account for
\cite{Kolenderski10}.  Nevertheless, an \ac{SLOCC} entanglement
classification of mixed symmetric $n$ qubit states has been achieved
very recently with the help of the \ac{DC} classes \cite{Bastin11}.

Judging by the attention that the Majorana representation has gathered
in recent years, and the multitude of recently discovered results, it
would not be surprising to see significant further progress in this
area.

\cleardoublepage
\singlespacing

\pagestyle{fancy} \fancyhead[LO,RE]{}
\fancyhead[LE]{\nouppercase{\bfseries \leftmark}}
\fancyhead[RO]{\nouppercase{\bfseries \rightmark}}
\phantomsection


\addcontentsline{toc}{chapter}{Bibliography}

\begin{thebibliography}{100}
\providecommand{\url}[1]{\texttt{#1}}
\providecommand{\urlprefix}{URL }
\expandafter\ifx\csname urlstyle\endcsname\relax
  \providecommand{\doi}[1]{doi:\discretionary{}{}{}#1}\else
  \providecommand{\doi}{doi:\discretionary{}{}{}\begingroup
  \urlstyle{rm}\Url}\fi
\providecommand{\eprint}[2][]{\url{#2}}

\bibitem{Wigner}
E.~P. Wigner.
\newblock \emph{Symmetries and Reflections} (Indiana University Press,
  Bloomington \& London, 1967).

\bibitem{Gross96}
D.~J. Gross.
\newblock \href{http://www.pnas.org/content/93/25/14256.abstract}{The role of
  symmetry in fundamental physics}.
\newblock Proc. Natl. Acad. Sci. USA \textbf{93}, 14256 (1996).

\bibitem{Prevedel09}
R.~Prevedel, G.~Cronenberg, M.~S. Tame, M.~Paternostro, P.~Walther, M.~S. Kim,
  and A.~Zeilinger.
\newblock \href{http://dx.doi.org/10.1103/PhysRevLett.103.020503}{Experimental
  Realization of Dicke States of up to Six Qubits for Multiparty Quantum
  Networking}.
\newblock Phys. Rev. Lett. \textbf{103}, 020503 (2009).

\bibitem{Wieczorek09}
W.~Wieczorek, R.~Krischek, N.~Kiesel, P.~Michelberger, G.~T\'{o}th, and
  H.~Weinfurter.
\newblock \href{http://dx.doi.org/10.1103/PhysRevLett.103.020504}{Experimental
  Entanglement of a Six-Photon Symmetric Dicke State}.
\newblock Phys. Rev. Lett. \textbf{103}, 020504 (2009).

\bibitem{Kiesel07}
N.~Kiesel, C.~Schmid, G.~T\'{o}th, E.~Solano, and H.~Weinfurter.
\newblock \href{http://dx.doi.org/10.1103/PhysRevLett.98.063604}{Experimental
  Observation of Four-Photon Entangled Dicke State with High Fidelity}.
\newblock Phys. Rev. Lett. \textbf{98}, 063604 (2007).

\bibitem{Bastin09b}
T.~Bastin, C.~Thiel, J.~von Zanthier, L.~Lamata, E.~Solano, and G.~S. Agarwal.
\newblock \href{http://dx.doi.org/10.1103/PhysRevLett.102.053601}{Operational
  Determination of Multiqubit Entanglement Classes via Tuning of Local
  Operations}.
\newblock Phys. Rev. Lett. \textbf{102}, 053601 (2009).

\bibitem{Nielsen}
M.~A. Nielsen and I.~L. Chuang.
\newblock \emph{Quantum Computation and Quantum Information} (Cambridge
  University Press, Cambridge, 2000).

\bibitem{Sarovar10}
M.~Sarovar, A.~Ishizaki, G.~R. Fleming, and K.~B. Whaley.
\newblock \href{http://dx.doi.org/10.1038/nphys1652}{Quantum entanglement in
  photosynthetic light-harvesting complexes}.
\newblock Nature Physics \textbf{6}, 462 (2010).

\bibitem{Gauger11}
E.~M. Gauger, E.~Rieper, J.~J.~L. Morton, S.~C. Benjamin, and V.~Vedral.
\newblock \href{http://dx.doi.org/10.1103/PhysRevLett.106.040503}{Sustained
  Quantum Coherence and Entanglement in the Avian Compass}.
\newblock Phys. Rev. Lett. \textbf{106}, 040503 (2011).

\bibitem{Feynman82}
R.~P. Feynman.
\newblock \href{http://dx.doi.org/10.1007/BF02650179}{Simulating physics with
  computers}.
\newblock Int. J. Theor. Phys. \textbf{21}, 467 (1982).

\bibitem{Shor94}
P.~W. Shor.
\newblock Algorithms for quantum computation: Discrete logarithms and
  factoring.
\newblock \emph{Proceedings of the 35th Annual Symposium on Foundations of
  Computer Science}, pp. 124--134 (IEEE Press, Los Alamitos, 1994).

\bibitem{Grover96}
L.~K. Grover.
\newblock A fast quantum mechanical algorithm for database search.
\newblock \emph{Proceedings of the 28th Annual ACM Symposium on Theory of
  Computing}, pp. 212--219 (Association for Computing Machinery, New York,
  1996).
\newblock {e}print:
  \href{http://arxiv.org/abs/quant-ph/9605043}{arXiv:quant-ph/9605043}.

\bibitem{Bennett93}
C.~H. Bennett, G.~Brassard, C.~Cr{\'e}peau, R.~Jozsa, A.~Peres, and W.~K.
  Wootters.
\newblock \href{http://dx.doi.org/10.1103/PhysRevLett.70.1895}{Teleporting an
  unknown quantum state via dual classical and Einstein-Podolsky-Rosen
  channels}.
\newblock Phys. Rev. Lett. \textbf{70}, 1895 (1993).

\bibitem{Bennett84}
C.~H. Bennett and G.~Brassard.
\newblock Quantum Cryptography: Public Key Distribution and Coin Tossing.
\newblock \emph{Proceedings of IEEE International Conference on Computers,
  Systems and Signal Processing}, pp. 175--179 (Bangalore, 1984).

\bibitem{Ekert91}
A.~K. Ekert.
\newblock \href{http://dx.doi.org/10.1103/PhysRevLett.67.661}{Quantum
  cryptography based on Bell's theorem}.
\newblock Phys. Rev. Lett. \textbf{67}, 661 (1991).

\bibitem{Bennett92}
C.~H. Bennett and S.~J. Wiesner.
\newblock \href{http://dx.doi.org/10.1103/PhysRevLett.69.2881}{Communication
  via one- and two-particle operators on Einstein-Podolsky-Rosen states}.
\newblock Phys. Rev. Lett. \textbf{69}, 2881 (1992).

\bibitem{Nest07}
M.~{V}an~den Nest, W.~D{\"u}r, A.~Miyake, and H.~J. Briegel.
\newblock \href{http://dx.doi.org/10.1088/1367-2630/9/6/204}{Fundamentals of
  universality in one-way quantum computation}.
\newblock New J. Phys. \textbf{9}, 204 (2007).

\bibitem{Dur00}
W.~D{\"u}r, G.~Vidal, and J.~I. Cirac.
\newblock \href{http://dx.doi.org/10.1103/PhysRevA.62.062314}{Three qubits can
  be entangled in two inequivalent ways}.
\newblock Phys. Rev. A \textbf{62}, 062314 (2000).

\bibitem{Horodecki09}
R.~Horodecki, P.~Horodecki, M.~Horodecki, and K.~Horodecki.
\newblock \href{http://dx.doi.org/10.1103/RevModPhys.81.865}{Quantum
  entanglement}.
\newblock Rev. Mod. Phys. \textbf{81}, 865 (2009).

\bibitem{Shimony95}
A.~Shimony.
\newblock \href{http://dx.doi.org/10.1111/j.1749-6632.1995.tb39008.x}{Degree of
  Entanglement}.
\newblock Ann. NY. Acad. Sci. \textbf{755}, 675 (1995).

\bibitem{Wei04}
T.-C. Wei, M.~Ericsson, P.~M. Goldbart, and W.~J. Munro.
\newblock Connections between relative entropy of entanglement and geometric
  measure of entanglement.
\newblock Quant. Inf. Comp. \textbf{4}, 252 (2004).
\newblock {e}print: \href{http://arxiv.org/abs/quant-ph/0405002}{
  arXiv:quant-ph/0405002}.

\bibitem{Ribeiro08}
P.~Ribeiro, J.~Vidal, and R.~Mosseri.
\newblock \href{http://dx.doi.org/10.1103/PhysRevE.78.021106}{Exact spectrum of
  the Lipkin-Meshkov-Glick model in the thermodynamic limit and finite-size
  corrections}.
\newblock Phys. Rev. E \textbf{78}, 021106 (2008).

\bibitem{Orus08}
R.~Or\'{u}s, S.~Dusuel, and J.~Vidal.
\newblock \href{http://dx.doi.org/10.1103/PhysRevLett.101.025701}{Equivalence
  of Critical Scaling Laws for Many-Body Entanglement in the
  Lipkin-Meshkov-Glick Model}.
\newblock Phys. Rev. Lett. \textbf{101}, 025701 (2008).

\bibitem{Dhondt06}
E.~D'Hondt and P.~Panangaden.
\newblock The computational power of the W and GHZ states.
\newblock Quant. Inf. Comp. \textbf{6}, 173 (2006).
\newblock {e}print: \href{http://arxiv.org/abs/quant-ph/0412177}{
  arXiv:quant-ph/0412177}.

\bibitem{Korbicz05}
J.~K. Korbicz, J.~I. Cirac, and M.~Lewenstein.
\newblock \href{http://dx.doi.org/10.1103/PhysRevLett.95.120502}{Spin Squeezing
  Inequalities and Entanglement of $N$ Qubit States}.
\newblock Phys. Rev. Lett. \textbf{95}, 120502 (2005).

\bibitem{Korbicz06}
J.~K. Korbicz, O.~G{\"u}hne, M.~Lewenstein, H.~H{\"a}ffner, C.~F. Roos, and
  R.~Blatt.
\newblock \href{http://dx.doi.org/10.1103/PhysRevA.74.052319}{Generalized
  spin-squeezing inequalities in $N$-qubit systems: Theory and experiment}.
\newblock Phys. Rev. A \textbf{74}, 052319 (2006).

\bibitem{Ivanov10}
S.~S. Ivanov, P.~A. Ivanov, I.~E. Linington, and N.~V. Vitanov.
\newblock \href{http://dx.doi.org/10.1103/PhysRevA.81.042328}{Scalable quantum
  search using trapped ions}.
\newblock Phys. Rev. A \textbf{81}, 042328 (2010).

\bibitem{Greenberger90}
D.~M. Greenberger, M.~A. Horne, A.~Shimony, and A.~Zeilinger.
\newblock \href{http://dx.doi.org/10.1119/1.16243}{Bell's theorem without
  inequalities}.
\newblock Am. J. Phys. \textbf{58}, 1131 (1990).

\bibitem{Dicke54}
R.~H. Dicke.
\newblock \href{http://dx.doi.org/10.1103/PhysRev.93.99}{Coherence in
  Spontaneous Radiation Processes}.
\newblock Phys. Rev. \textbf{93}, 99 (1954).

\bibitem{Smolin01}
J.~A. Smolin.
\newblock \href{http://dx.doi.org/10.1103/PhysRevA.63.032306}{Four-party
  unlockable bound entangled state}.
\newblock Phys. Rev. A \textbf{63}, 032306 (2001).

\bibitem{Majorana32}
E.~Majorana.
\newblock Atomi orientati in campo magnetico variabile.
\newblock Nuovo Cimento \textbf{9}, 43 (1932).

\bibitem{Chruscinski06}
D.~Chru\'{s}ci\'{n}ski.
\newblock \href{http://dx.doi.org/10.1088/1742-6596/30/1/002}{Geometric Aspects
  of Quantum Mechanics and Quantum Entanglement}.
\newblock J. Phys.: Conf. Ser. \textbf{30}, 9 (2006).

\bibitem{Wick52}
G.-C. Wick, A.~S. Wightman, and E.~P. Wigner.
\newblock \href{http://dx.doi.org/10.1103/PhysRev.88.101}{The Intrinsic Parity
  of Elementary Particles}.
\newblock Phys. Rev. \textbf{88}, 101 (1952).

\bibitem{Ashhab07}
S.~Ashhab, K.~Maruyama, and F.~Nori.
\newblock \href{http://dx.doi.org/10.1103/PhysRevA.76.052113}{Detecting mode
  entanglement: The role of coherent states, superselection rules, and particle
  statistics}.
\newblock Phys. Rev. A \textbf{76}, 052113 (2007).

\bibitem{Einstein35}
A.~Einstein, B.~Podolsky, and N.~Rosen.
\newblock \href{http://dx.doi.org/10.1103/PhysRev.47.777}{Can
  Quantum-Mechanical Description of Physical Reality Be Considered Complete?}
\newblock Phys. Rev. \textbf{47}, 777 (1935).

\bibitem{Aspect81}
A.~Aspect, P.~Grangier, and G.~Roger.
\newblock \href{http://dx.doi.org/10.1103/PhysRevLett.47.460}{Experimental
  Tests of Realistic Local Theories via Bell's Theorem}.
\newblock Phys. Rev. Lett. \textbf{47}, 460 (1981).

\bibitem{Bell64}
J.~S. Bell.
\newblock On the Einstein Podolski Rosen Paradox.
\newblock Physics \textbf{1}, 195 (1964).

\bibitem{Brody01}
D.~C. Brody and L.~P. Hughston.
\newblock \href{http://dx.doi.org/10.1016/S0393-0440(00)00052-8}{Geometric
  quantum mechanics}.
\newblock J. Geom. Phys. \textbf{38}, 19 (2001).

\bibitem{Acin00b}
A.~Ac\'\i{}n, A.~Andrianov, L.~Costa, E.~Jan{\'e}, J.~I. Latorre, and
  R.~Tarrach.
\newblock \href{http://dx.doi.org/10.1103/PhysRevLett.85.1560}{Generalized
  Schmidt Decomposition and Classification of Three-Quantum-Bit States}.
\newblock Phys. Rev. Lett. \textbf{85}, 1560 (2000).

\bibitem{Carteret00}
H.~A. Carteret, A.~Higuchi, and A.~Sudbery.
\newblock \href{http://dx.doi.org/10.1063/1.1319516}{Multipartite
  generalization of the Schmidt decomposition}.
\newblock J. Math. Phys. \textbf{41}, 7932 (2000).

\bibitem{Verstraete03}
F.~Verstraete, J.~Dehaene, and B.~{D}e Moor.
\newblock \href{http://dx.doi.org/10.1103/PhysRevA.68.012103}{Normal forms and
  entanglement measures for multipartite quantum states}.
\newblock Phys. Rev. A \textbf{68}, 012103 (2003).

\bibitem{Kraus10}
B.~Kraus.
\newblock \href{http://dx.doi.org/10.1103/PhysRevLett.104.020504}{Local Unitary
  Equivalence of Multipartite Pure States}.
\newblock Phys. Rev. Lett. \textbf{104}, 020504 (2010).

\bibitem{Kraus10b}
B.~Kraus.
\newblock \href{http://dx.doi.org/10.1103/PhysRevA.82.032121}{Local unitary
  equivalence and entanglement of multipartite pure states}.
\newblock Phys. Rev. A \textbf{82}, 032121 (2010).

\bibitem{Eisert01}
J.~Eisert and H.~J. Briegel.
\newblock \href{http://dx.doi.org/10.1103/PhysRevA.64.022306}{Schmidt measure
  as a tool for quantifying multiparticle entanglement}.
\newblock Phys. Rev. A \textbf{64}, 022306 (2001).

\bibitem{Chen10b}
L.~Chen, E.~Chitambar, R.~Duan, Z.~Ji, and A.~Winter.
\newblock \href{http://dx.doi.org/10.1103/PhysRevLett.105.200501}{Tensor Rank
  and Stochastic Entanglement Catalysis for Multipartite Pure States}.
\newblock Phys. Rev. Lett. \textbf{105}, 200501 (2010).

\bibitem{Werner01}
R.~F. Werner and M.~M. Wolf.
\newblock \href{http://dx.doi.org/10.1103/PhysRevA.64.032112}{All-multipartite
  Bell-correlation inequalities for two dichotomic observables per site}.
\newblock Phys. Rev. A \textbf{64}, 032112 (2001).

\bibitem{Vidal00}
G.~Vidal.
\newblock \href{http://dx.doi.org/10.1080/09500340008244048}{Entanglement
  monotones}.
\newblock J. Mod. Opt. \textbf{47}, 355 (2000).

\bibitem{Bennett00}
C.~H. Bennett, S.~Popescu, D.~Rohrlich, J.~A. Smolin, and A.~V. Thapliyal.
\newblock \href{http://dx.doi.org/10.1103/PhysRevA.63.012307}{Exact and
  asymptotic measures of multipartite pure-state entanglement}.
\newblock Phys. Rev. A \textbf{63}, 012307 (2000).

\bibitem{Nielsen99}
M.~A. Nielsen.
\newblock \href{http://dx.doi.org/10.1103/PhysRevLett.83.436}{Conditions for a
  Class of Entanglement Transformations}.
\newblock Phys. Rev. Lett. \textbf{83}, 436 (1999).

\bibitem{Coffman00}
V.~Coffman, J.~Kundu, and W.~K. Wootters.
\newblock \href{http://dx.doi.org/10.1103/PhysRevA.61.052306}{Distributed
  entanglement}.
\newblock Phys. Rev. A \textbf{61}, 052306 (2000).

\bibitem{Verstraete02}
F.~Verstraete, J.~Dehaene, B.~{D}e Moor, and H.~Verschelde.
\newblock \href{http://dx.doi.org/10.1103/PhysRevA.65.052112}{Four qubits can
  be entangled in nine different ways}.
\newblock Phys. Rev. A \textbf{65}, 052112 (2002).

\bibitem{Miyake03}
A.~Miyake.
\newblock \href{http://dx.doi.org/10.1103/PhysRevA.67.012108}{Classification of
  multipartite entangled states by multidimensional determinants}.
\newblock Phys. Rev. A \textbf{67}, 012108 (2003).

\bibitem{Lamata07}
L.~Lamata, J.~Le\'{o}n, D.~Salgado, and E.~Solano.
\newblock \href{http://dx.doi.org/10.1103/PhysRevA.75.022318}{Inductive
  entanglement classification of four qubits under stochastic local operations
  and classical communication}.
\newblock Phys. Rev. A \textbf{75}, 022318 (2007).

\bibitem{Viehmann11}
O.~Viehmann, C.~Eltschka, and J.~Siewert.
\newblock \href{http://dx.doi.org/10.1103/PhysRevA.83.052330}{Polynomial
  invariants for discrimination and classification of four-qubit entanglement}.
\newblock Phys. Rev. A \textbf{83}, 052330 (2011).

\bibitem{Borsten10}
L.~Borsten, D.~Dahanayake, M.~J. Duff, A.~Marrani, and W.~Rubens.
\newblock \href{http://dx.doi.org/10.1103/PhysRevLett.105.100507}{Four-Qubit
  Entanglement Classification from String Theory}.
\newblock Phys. Rev. Lett. \textbf{105}, 100507 (2010).

\bibitem{Wootters98}
W.~K. Wootters.
\newblock \href{http://dx.doi.org/10.1103/PhysRevLett.80.2245}{Entanglement of
  Formation of an Arbitrary State of Two Qubits}.
\newblock Phys. Rev. Lett. \textbf{80}, 2245 (1998).

\bibitem{Luque03}
J.-G. Luque and J.-Y. Thibon.
\newblock \href{http://dx.doi.org/10.1103/PhysRevA.67.042303}{Polynomial
  invariants of four qubits}.
\newblock Phys. Rev. A \textbf{67}, 042303 (2003).

\bibitem{Osterloh05}
A.~Osterloh and J.~Siewert.
\newblock \href{http://dx.doi.org/10.1103/PhysRevA.72.012337}{Constructing
  $N$-qubit entanglement monotones from antilinear operators}.
\newblock Phys. Rev. A \textbf{72}, 012337 (2005).

\bibitem{Osterloh06}
A.~Osterloh and J.~Siewert.
\newblock \href{http://dx.doi.org/10.1142/S0219749906001980}{Entanglement
  monotones and maximally entangled states in multipartite qubit systems}.
\newblock Int. J. Quantum Inf. \textbf{4}, 531 (2006).
\newblock {e}print: \href{http://arxiv.org/abs/quant-ph/0506073}{
  arXiv:quant-ph/0506073}.

\bibitem{Luque06}
J.-G. Luque and J.-Y. Thibon.
\newblock \href{http://dx.doi.org/10.1088/0305-4470/39/2/007}{Algebraic
  invariants of five qubits}.
\newblock J. Phys. A: Math. Gen. \textbf{39}, 371 (2006).

\bibitem{Levay06}
P.~L\'{e}vay.
\newblock \href{http://dx.doi.org/10.1088/0305-4470/39/30/009}{On the geometry
  of four-qubit invariants}.
\newblock J. Phys. A: Math. Gen. \textbf{39}, 9533 (2006).

\bibitem{Dokovic09}
D.~\v{Z}. {\DJ}okovi\'{c} and A.~Osterloh.
\newblock \href{http://dx.doi.org/10.1063/1.3075830}{On polynomial invariants
  of several qubits}.
\newblock J. Math. Phys. \textbf{50}, 033509 (2009).

\bibitem{Horodecki00}
M.~Horodecki, P.~Horodecki, and R.~Horodecki.
\newblock \href{http://dx.doi.org/10.1103/PhysRevLett.84.2014}{Limits for
  Entanglement Measures}.
\newblock Phys. Rev. Lett. \textbf{84}, 2014 (2000).

\bibitem{Morikoshi04}
F.~Morikoshi, M.~F. Santos, and V.~Vedral.
\newblock \href{http://dx.doi.org/10.1088/0305-4470/37/22/013}{Accessibility of
  physical states and non-uniqueness of entanglement measure}.
\newblock J. Phys. A: Math. Gen. \textbf{37}, 5887 (2004).

\bibitem{Wei03b}
T.-C. Wei, K.~Nemoto, P.~M. Goldbart, P.~G. Kwiat, W.~J. Munro, and
  F.~Verstraete.
\newblock \href{http://dx.doi.org/10.1103/PhysRevA.67.022110}{Maximal
  entanglement versus entropy for mixed quantum states}.
\newblock Phys. Rev. A \textbf{67}, 022110 (2003).

\bibitem{Vedral98}
V.~Vedral and M.~B. Plenio.
\newblock \href{http://dx.doi.org/10.1103/PhysRevA.57.1619}{Entanglement
  measures and purification procedures}.
\newblock Phys. Rev. A \textbf{57}, 1619 (1998).

\bibitem{Christandl04}
M.~Christandl and A.~Winter.
\newblock \href{http://dx.doi.org/10.1063/1.1643788}{\quo{Squashed
  entanglement}: An additive entanglement measure}.
\newblock J. Math. Phys. \textbf{45}, 829 (2004).

\bibitem{Hastings09}
M.~B. Hastings.
\newblock \href{http://dx.doi.org/10.1038/nphys1224}{Superadditivity of
  communication capacity using entangled inputs}.
\newblock Nature Physics \textbf{5}, 255 (2009).

\bibitem{Shor04}
P.~W. Shor.
\newblock \href{http://dx.doi.org/ 10.1007/s00220-003-0981-7}{Equivalence of
  Additivity Questions in Quantum Information Theory}.
\newblock Comm. Math. Phys. \textbf{246}, 453 (2004).

\bibitem{Plenio07}
M.~B. Plenio and S.~Virmani.
\newblock An introduction to entanglement measures.
\newblock Quant. Inf. Comp. \textbf{7}, 1 (2007).
\newblock {e}print: \href{http://arxiv.org/abs/quant-ph/0504163}{
  arXiv:quant-ph/0504163}.

\bibitem{Toth07}
G.~T\'{o}th.
\newblock \href{http://dx.doi.org/10.1364/JOSAB.24.000275}{Detection of
  multipartite entanglement in the vicinity of symmetric Dicke states}.
\newblock J. Opt. Soc. Am. B \textbf{24}, 275 (2007).

\bibitem{Stockton03}
J.~K. Stockton, J.~M. Geremia, A.~C. Doherty, and H.~Mabuchi.
\newblock \href{http://dx.doi.org/10.1103/PhysRevA.67.022112}{Characterizing
  the entanglement of symmetric many-particle spin-$\tfrac{1}{2}$ systems}.
\newblock Phys. Rev. A \textbf{67}, 022112 (2003).

\bibitem{Chiuri10}
A.~Chiuri, G.~Vallone, N.~Bruno, C.~Macchiavello, D.~Bru\ss{}, and P.~Mataloni.
\newblock
  \href{http://dx.doi.org/10.1103/PhysRevLett.105.250501}{Hyperentangled Mixed
  Phased Dicke States: Optical Design and Detection}.
\newblock Phys. Rev. Lett. \textbf{105}, 250501 (2010).

\bibitem{Thiel07}
C.~Thiel, J.~von Zanthier, T.~Bastin, E.~Solano, and G.~S. Agarwal.
\newblock \href{http://dx.doi.org/10.1103/PhysRevLett.99.193602}{Generation of
  Symmetric Dicke States of Remote Qubits with Linear Optics}.
\newblock Phys. Rev. Lett. \textbf{99}, 193602 (2007).

\bibitem{Krammer09}
P.~Krammer, H.~Kampermann, D.~Bru\ss{}, R.~A. Bertlmann, L.~C. Kwek, and
  C.~Macchiavello.
\newblock \href{http://dx.doi.org/10.1103/PhysRevLett.103.100502}{Multipartite
  Entanglement Detection via Structure Factors}.
\newblock Phys. Rev. Lett. \textbf{103}, 100502 (2009).

\bibitem{Hume09}
D.~B. Hume, C.~W. Chou, T.~Rosenband, and D.~J. Wineland.
\newblock \href{http://dx.doi.org/10.1103/PhysRevA.80.052302}{Preparation of
  Dicke states in an ion chain}.
\newblock Phys. Rev. A \textbf{80}, 052302 (2009).

\bibitem{Bastin11}
T.~Bastin, P.~Mathonet, and E.~Solano.
\newblock Operational Entanglement Families of Symmetric Mixed N-Qubit States
  (2010).
\newblock {p}reprint: \href{http://arxiv.org/abs/1011.1243}{arXiv:1011.1243}.

\bibitem{Nest06}
M.~{V}an~den Nest, A.~Miyake, W.~D{\"u}r, and H.~J. Briegel.
\newblock \href{http://dx.doi.org/10.1103/PhysRevLett.97.150504}{Universal
  Resources for Measurement-Based Quantum Computation}.
\newblock Phys. Rev. Lett. \textbf{97}, 150504 (2006).

\bibitem{Gross09}
D.~Gross, S.~T. Flammia, and J.~Eisert.
\newblock \href{http://dx.doi.org/10.1103/PhysRevLett.102.190501}{Most Quantum
  States Are Too Entangled To Be Useful As Computational Resources}.
\newblock Phys. Rev. Lett. \textbf{102}, 190501 (2009).

\bibitem{Schneider02}
S.~Schneider and G.~J. Milburn.
\newblock \href{http://dx.doi.org/10.1103/PhysRevA.65.042107}{Entanglement in
  the steady state of a collective-angular-momentum (Dicke) model}.
\newblock Phys. Rev. A \textbf{65}, 042107 (2002).

\bibitem{Vedral04}
V.~Vedral.
\newblock \href{http://dx.doi.org/10.1088/1367-2630/6/1/102}{High-temperature
  macroscopic entanglement}.
\newblock New J. Phys. \textbf{6}, 102 (2004).

\bibitem{Bastin09}
T.~Bastin, S.~Krins, P.~Mathonet, M.~Godefroid, L.~Lamata, and E.~Solano.
\newblock \href{http://dx.doi.org/10.1103/PhysRevLett.103.070503}{Operational
  Families of Entanglement Classes for Symmetric $N$-Qubit States}.
\newblock Phys. Rev. Lett. \textbf{103}, 070503 (2009).

\bibitem{Hayashi08}
M.~Hayashi, D.~Markham, M.~Murao, M.~Owari, and S.~Virmani.
\newblock \href{http://dx.doi.org/10.1103/PhysRevA.77.012104}{Entanglement of
  multiparty-stabilizer, symmetric, and antisymmetric states}.
\newblock Phys. Rev. A \textbf{77}, 012104 (2008).

\bibitem{Hubener09}
R.~H{\"u}bener, M.~Kleinmann, T.-C. Wei, C.~Gonz\'{a}lez-Guill\'{e}n, and
  O.~G{\"u}hne.
\newblock \href{http://dx.doi.org/10.1103/PhysRevA.80.032324}{Geometric measure
  of entanglement for symmetric states}.
\newblock Phys. Rev. A \textbf{80}, 032324 (2009).

\bibitem{Markham11}
D.~Markham.
\newblock \href{http://dx.doi.org/10.1103/PhysRevA.83.042332}{Entanglement and
  symmetry in permutation-symmetric states}.
\newblock Phys. Rev. A \textbf{83}, 042332 (2011).

\bibitem{Mathonet10}
P.~Mathonet, S.~Krins, M.~Godefroid, L.~Lamata, E.~Solano, and T.~Bastin.
\newblock \href{http://dx.doi.org/10.1103/PhysRevA.81.052315}{Entanglement
  equivalence of $N$-qubit symmetric states}.
\newblock Phys. Rev. A \textbf{81}, 052315 (2010).

\bibitem{Toth09}
G.~T\'{o}th and O.~G{\"u}hne.
\newblock \href{http://dx.doi.org/10.1103/PhysRevLett.102.170503}{Entanglement
  and Permutational Symmetry}.
\newblock Phys. Rev. Lett. \textbf{102}, 170503 (2009).

\bibitem{Aulbach11}
M.~Aulbach.
\newblock Symmetric entanglement classes for n qubits (2011).
\newblock {p}reprint: \href{http://arxiv.org/abs/1103.0271}{arXiv:1103.0271}.

\bibitem{Aulbach10}
M.~Aulbach, D.~Markham, and M.~Murao.
\newblock \href{http://dx.doi.org/10.1088/1367-2630/12/7/073025}{The maximally
  entangled symmetric state in terms of the geometric measure}.
\newblock New J. Phys. \textbf{12}, 073025 (2010).

\bibitem{Martin10}
J.~Martin, O.~Giraud, P.~A. Braun, D.~Braun, and T.~Bastin.
\newblock \href{http://dx.doi.org/10.1103/PhysRevA.81.062347}{Multiqubit
  symmetric states with high geometric entanglement}.
\newblock Phys. Rev. A \textbf{81}, 062347 (2010).

\bibitem{Aulbach10lncs}
M.~Aulbach, D.~Markham, and M.~Murao.
\newblock \href{http://dx.doi.org/10.1007/978-3-642-18073-6}{Geometric
  Entanglement of Symmetric States and the Majorana Representation}.
\newblock \emph{Proceedings of the 5th Conference on Theory of Quantum
  Computation, Communication and Cryptography}, edited by W.~van Dam, V.~M.
  Kendon, and S.~Severini, pp. 141--158 (LNCS, Berlin, 2010).
\newblock ISBN 978-3-642-18072-9.

\bibitem{Bacry74}
H.~Bacry.
\newblock \href{http://dx.doi.org/10.1063/1.1666525}{Orbits of the rotation
  group on spin states}.
\newblock J. Math. Phys. \textbf{15}, 1686 (1974).

\bibitem{WeissteinVieta}
E.~W. Weisstein.
\newblock Vieta's Formulas. Wolfram MathWorld.
\newblock \urlprefix\url{http://mathworld.wolfram.com/VietasFormulas.html}.

\bibitem{Leboeuf91}
P.~Leb{\oe}uf.
\newblock \href{http://dx.doi.org/10.1088/0305-4470/24/19/021}{Phase space
  approach to quantum dynamics}.
\newblock J. Phys. A: Math. Gen. \textbf{24}, 4575 (1991).

\bibitem{PenroseRindler}
R.~Penrose and W.~Rindler.
\newblock \emph{Spinors and Space-Time. Vol.1: Two-Spinor Calculus and
  Relativistic Fields} (Cambridge University Press, Cambridge, 1984).

\bibitem{Zimba93}
J.~Zimba and R.~Penrose.
\newblock \href{http://dx.doi.org/10.1016/0039-3681(93)90061-N}{On Bell
  Non-Locality Without Probabilities: More Curious Geometry}.
\newblock Stud. Hist. Phil. Sci. \textbf{24}, 697 (1993).

\bibitem{Zimba06}
J.~Zimba.
\newblock \href{http://www.ejtp.com/articles/index.html}{\quo{Anticoherent}
  Spin States via the Majorana Representation}.
\newblock Electron. J. Theor. Phys. \textbf{3}, 143 (2006).

\bibitem{Giraud10}
O.~Giraud, P.~A. Braun, and D.~Braun.
\newblock \href{http://dx.doi.org/10.1088/1367-2630/12/6/063005}{Quantifying
  quantumness and the quest for Queens of Quantum}.
\newblock New J. Phys. \textbf{12}, 063005 (2010).

\bibitem{Markham03}
D.~Markham and V.~Vedral.
\newblock \href{http://dx.doi.org/10.1103/PhysRevA.67.042113}{Classicality of
  spin-coherent states via entanglement and distinguishability}.
\newblock Phys. Rev. A \textbf{67}, 042113 (2003).

\bibitem{Hannay98}
J.~H. Hannay.
\newblock \href{http://dx.doi.org/10.1088/0305-4470/31/2/002}{The Berry phase
  for spin in the Majorana representation}.
\newblock J. Phys. A: Math. Gen. \textbf{31}, L53 (1998).

\bibitem{Hannay96}
J.~H. Hannay.
\newblock \href{http://dx.doi.org/10.1088/0305-4470/29/5/004}{Chaotic analytic
  zero points: exact statistics for those of a random spin state}.
\newblock J. Phys. A: Math. Gen. \textbf{29}, L101 (1996).

\bibitem{Barnett06}
R.~Barnett, A.~Turner, and E.~Demler.
\newblock \href{http://dx.doi.org/10.1103/PhysRevLett.97.180412}{Classifying
  Novel Phases of Spinor Atoms}.
\newblock Phys. Rev. Lett. \textbf{97}, 180412 (2006).

\bibitem{Barnett07}
R.~Barnett, A.~Turner, and E.~Demler.
\newblock \href{http://dx.doi.org/10.1103/PhysRevA.76.013605}{Classifying
  vortices in $S$=3 Bose-Einstein condensates}.
\newblock Phys. Rev. A \textbf{76}, 013605 (2007).

\bibitem{Barnett08}
R.~Barnett, S.~Mukerjee, and J.~E. Moore.
\newblock \href{http://dx.doi.org/10.1103/PhysRevLett.100.240405}{Vortex
  Lattice Transitions in Cyclic Spinor Condensates}.
\newblock Phys. Rev. Lett. \textbf{100}, 240405 (2008).

\bibitem{Makela07}
H.~M{\"a}kel{\"a} and K.-A. Suominen.
\newblock \href{http://dx.doi.org/10.1103/PhysRevLett.99.190408}{Inert States
  of Spin-$S$ Systems}.
\newblock Phys. Rev. Lett. \textbf{99}, 190408 (2007).

\bibitem{Kolenderski08}
P.~Kolenderski and R.~Demkowicz-Dobrzanski.
\newblock \href{http://dx.doi.org/10.1103/PhysRevA.78.052333}{Optimal state for
  keeping reference frames aligned and the platonic solids}.
\newblock Phys. Rev. A \textbf{78}, 052333 (2008).

\bibitem{Kolenderski10}
P.~Kolenderski.
\newblock \href{http://dx.doi.org/10.1142/S1230161210000084}{Geometry of Pure
  States of N Spin-J System}.
\newblock Open Systems \& Information Dynamics \textbf{17}, 107 (2010).
\newblock {e}print: \href{http://arxiv.org/abs/0910.3075}{ arXiv:0910.3075}.

\bibitem{Dennis04}
M.~R. Dennis.
\newblock \href{http://dx.doi.org/10.1088/0305-4470/37/40/011}{Canonical
  representation of spherical functions: Sylvester's theorem, Maxwell's
  multipoles and Majorana's sphere}.
\newblock J. Phys. A: Math. Gen. \textbf{37}, 9487 (2004).

\bibitem{Crann10}
J.~Crann, R.~Pereira, and D.~W. Kribs.
\newblock \href{http://dx.doi.org/10.1088/1751-8113/43/25/255307}{Spherical
  designs and anticoherent spin states}.
\newblock J. Phys. A: Math. Theor. \textbf{43}, 255307 (2010).

\bibitem{Chen11}
L.~Chen, H.~Zhu, and T.-C. Wei.
\newblock \href{http://dx.doi.org/10.1103/PhysRevA.83.012305}{Connections of
  geometric measure of entanglement of pure symmetric states to quantum state
  estimation}.
\newblock Phys. Rev. A \textbf{83}, 012305 (2011).

\bibitem{Forst}
W.~Forst and D.~Hoffmann.
\newblock \emph{Funktionentheorie erkunden mit Maple} (Springer Verlag, Berlin,
  2002).

\bibitem{Radcliffe71}
J.~M. Radcliffe.
\newblock \href{http://dx.doi.org/10.1088/0305-4470/4/3/009}{Some properties of
  coherent spin states}.
\newblock J. Phys. A: Math. Gen. \textbf{4}, 313 (1971).

\bibitem{Lipkin65}
H.~J. Lipkin, N.~Meshkov, and A.~J. Glick.
\newblock \href{http://dx.doi.org/10.1016/0029-5582(65)90862-X}{Validity of
  many-body approximation methods for a solvable model: (I). Exact solutions
  and perturbation theory}.
\newblock Nucl. Phys. \textbf{62}, 188 (1965).

\bibitem{Lipkin65b}
N.~Meshkov, A.~J. Glick, and H.~J. Lipkin.
\newblock \href{http://dx.doi.org/10.1016/0029-5582(65)90863-1}{Validity of
  many-body approximation methods for a solvable model: (II). Linearization
  procedures}.
\newblock Nucl. Phys. \textbf{62}, 199 (1965).

\bibitem{Lipkin65c}
A.~J. Glick, H.~J. Lipkin, and N.~Meshkov.
\newblock \href{http://dx.doi.org/10.1016/0029-5582(65)90864-3}{Validity of
  many-body approximation methods for a solvable model: (III). Diagram
  summations}.
\newblock Nucl. Phys. \textbf{62}, 211 (1965).

\bibitem{Wei03}
T.-C. Wei and P.~M. Goldbart.
\newblock \href{http://dx.doi.org/10.1103/PhysRevA.68.042307}{Geometric measure
  of entanglement and applications to bipartite and multipartite quantum
  states}.
\newblock Phys. Rev. A \textbf{68}, 042307 (2003).

\bibitem{Barnum01}
H.~Barnum and N.~Linden.
\newblock \href{http://dx.doi.org/10.1088/0305-4470/34/35/305}{Monotones and
  invariants for multi-particle quantum states}.
\newblock J. Phys. A: Math. Gen. \textbf{34}, 6787 (2001).

\bibitem{Plenio01}
M.~B. Plenio and V.~Vedral.
\newblock \href{http://dx.doi.org/10.1088/0305-4470/34/35/325}{Bounds on
  relative entropy of entanglement for multi-party systems}.
\newblock J. Phys. A: Math. Gen. \textbf{34}, 6997 (2001).

\bibitem{Vedral97}
V.~Vedral, M.~B. Plenio, M.~A. Rippin, and P.~L. Knight.
\newblock \href{http://dx.doi.org/10.1103/PhysRevLett.78.2275}{Quantifying
  Entanglement}.
\newblock Phys. Rev. Lett. \textbf{78}, 2275 (1997).

\bibitem{Lathauwer00}
L.~{D}e Lathauwer, B.~{D}e Moor, and J.~Vandewalle.
\newblock \href{http://dx.doi.org/10.1137/S0895479898346995}{On the Best Rank-1
  and Rank-($R_{1},R_{2},\ldots,R_{N}$) Approximation of Higher-Order
  Tensors)}.
\newblock SIAM J. Matrix Anal. Appl. \textbf{21}, 1324 (2000).

\bibitem{Zhang01}
T.~Zhang and G.~H. Golub.
\newblock \href{http://dx.doi.org/10.1137/S0895479899352045}{Rank-One
  Approximation to High Order Tensors}.
\newblock SIAM J. Matrix Anal. Appl. \textbf{23}, 534 (2001).

\bibitem{Kofidis02}
E.~Kofidis and P.~A. Regalia.
\newblock \href{http://dx.doi.org/10.1137/S0895479801387413}{On the Best Rank-1
  Approximation of Higher-Order Supersymmetric Tensors}.
\newblock SIAM J. Matrix Anal. Appl. \textbf{23}, 863 (2002).

\bibitem{Wang09}
H.~Wang and N.~Ahuja.
\newblock \href{http://dx.doi.org/10.1109/ICPR.2004.1334001}{Compact
  representation of multidimensional data using tensor rank-one decomposition}.
\newblock \emph{Proceedings of the 17th International Conference on Pattern
  Recognition}, edited by J.~Kittler, M.~Petrou, and M.~Nixon, volume~1, pp.
  44--47 (IEEE, 2004).
\newblock ISBN 0-7695-2128-2.

\bibitem{Ni07}
G.~Ni and Y.~Wang.
\newblock \href{http://dx.doi.org/10.1016/j.mcm.2007.01.008}{On the best rank-1
  approximation to higher-order symmetric tensors}.
\newblock Math. Comput. Modelling \textbf{46}, 1345 (2007).

\bibitem{Silva08}
V.~{D}e Silva and L.-H. Lim.
\newblock \href{http://dx.doi.org/10.1137/06066518X}{Tensor Rank and the
  III-Posedness of the Best Low-Rank Approximation Problem}.
\newblock SIAM J. Matrix Anal. Appl. \textbf{30}, 1084 (2008).

\bibitem{Orus08b}
R.~Or\'{u}s.
\newblock \href{http://dx.doi.org/10.1103/PhysRevLett.100.130502}{Universal
  Geometric Entanglement Close to Quantum Phase Transitions}.
\newblock Phys. Rev. Lett. \textbf{100}, 130502 (2008).

\bibitem{Wei05}
T.-C. Wei, D.~Das, S.~Mukhopadyay, S.~Vishveshwara, and P.~M. Goldbart.
\newblock \href{http://dx.doi.org/10.1103/PhysRevA.71.060305}{Global
  entanglement and quantum criticality in spin chains}.
\newblock Phys. Rev. A \textbf{71}, 060305(R) (2005).

\bibitem{Nakata09}
Y.~Nakata, D.~Markham, and M.~Murao.
\newblock \href{http://dx.doi.org/10.1103/PhysRevA.79.042313}{Thermal
  robustness of multipartite entanglement of the one-dimensional
  spin-$\frac{1}{2}$ $XY$ model}.
\newblock Phys. Rev. A \textbf{79}, 042313 (2009).

\bibitem{Markham08}
D.~Markham, J.~Anders, V.~Vedral, M.~Murao, and A.~Miyake.
\newblock \href{http://dx.doi.org/10.1209/0295-5075/81/40006}{Survival of
  entanglement in thermal states}.
\newblock Europhys. Lett. \textbf{81}, 40006 (2008).

\bibitem{Biham02}
O.~Biham, M.~A. Nielsen, and T.~J. Osborne.
\newblock \href{http://dx.doi.org/10.1103/PhysRevA.65.062312}{Entanglement
  monotone derived from Grover's algorithm}.
\newblock Phys. Rev. A \textbf{65}, 062312 (2002).

\bibitem{Shimoni04}
Y.~Shimoni, D.~Shapira, and O.~Biham.
\newblock \href{http://dx.doi.org/10.1103/PhysRevA.69.062303}{Characterization
  of pure quantum states of multiple qubits using the Groverian entanglement
  measure}.
\newblock Phys. Rev. A \textbf{69}, 062303 (2004).

\bibitem{Werner02}
R.~F. Werner and A.~S. Holevo.
\newblock \href{http://dx.doi.org/10.1063/1.1498491}{Counterexample to an
  additivity conjecture for output purity of quantum channels}.
\newblock J. Math. Phys. \textbf{43}, 4353 (2002).

\bibitem{Mora10}
C.~E. Mora, M.~Piani, A.~Miyake, M.~{V}an~den Nest, W.~D{\"u}r, and H.~J.
  Briegel.
\newblock \href{http://dx.doi.org/10.1103/PhysRevA.81.042315}{Universal
  resources for approximate and stochastic measurement-based quantum
  computation}.
\newblock Phys. Rev. A \textbf{81}, 042315 (2010).

\bibitem{Hayashi06}
M.~Hayashi, D.~Markham, M.~Murao, M.~Owari, and S.~Virmani.
\newblock \href{http://dx.doi.org/10.1103/PhysRevLett.96.040501}{Bounds on
  Multipartite Entangled Orthogonal State Discrimination Using Local Operations
  and Classical Communication}.
\newblock Phys. Rev. Lett. \textbf{96}, 040501 (2006).

\bibitem{Zhu10}
H.~Zhu, L.~Chen, and M.~Hayashi.
\newblock \href{http://dx.doi.org/10.1088/1367-2630/12/8/083002}{Additivity and
  non-additivity of multipartite entanglement measures}.
\newblock New J. Phys. \textbf{12}, 083002 (2010).

\bibitem{Hilling10}
J.~J. Hilling and A.~Sudbery.
\newblock \href{http://dx.doi.org/10.1063/1.3451264}{The geometric measure of
  multipartite entanglement and the singular values of a hypermatrix}.
\newblock J. Math. Phys. \textbf{51}, 072102 (2010).

\bibitem{Uhlmann76}
A.~Uhlmann.
\newblock \href{http://dx.doi.org/10.1016/0034-4877(76)90060-4}{The
  ``transition probability'' in the state space of a *-algebra}.
\newblock Rep. Math. Phys. \textbf{9}, 273 (1976).

\bibitem{Jozsa94}
R.~Jozsa.
\newblock \href{http://dx.doi.org/10.1080/09500349414552171}{Fidelity for mixed
  quantum states}.
\newblock J. Mod. Opt. \textbf{41}, 2315 (1994).

\bibitem{Cavalcanti06}
D.~Cavalcanti.
\newblock \href{http://dx.doi.org/10.1103/PhysRevA.73.044302}{Connecting the
  generalized robustness and the geometric measure of entanglement}.
\newblock Phys. Rev. A \textbf{73}, 044302 (2006).

\bibitem{Vidal99}
G.~Vidal and R.~Tarrach.
\newblock \href{http://dx.doi.org/10.1103/PhysRevA.59.141}{Robustness of
  entanglement}.
\newblock Phys. Rev. A \textbf{59}, 141 (1999).

\bibitem{Wei04b}
T.-C. Wei, J.~B. Altepeter, P.~M. Goldbart, and W.~J. Munro.
\newblock \href{http://dx.doi.org/10.1103/PhysRevA.70.022322}{Measures of
  entanglement in multipartite bound entangled states}.
\newblock Phys. Rev. A \textbf{70}, 022322 (2004).

\bibitem{Murao01}
M.~Murao and V.~Vedral.
\newblock \href{http://dx.doi.org/10.1103/PhysRevLett.86.352}{Remote
  Information Concentration Using a Bound Entangled State}.
\newblock Phys. Rev. Lett. \textbf{86}, 352 (2001).

\bibitem{Markham07}
D.~Markham, A.~Miyake, and S.~Virmani.
\newblock \href{http://dx.doi.org/10.1088/1367-2630/9/6/194}{Entanglement and
  local information access for graph states}.
\newblock New J. Phys. \textbf{9}, 194 (2007).

\bibitem{Jung08}
E.~Jung, M.-R. Hwang, H.~Kim, M.-S. Kim, D.~Park, J.-W. Son, and S.~Tamaryan.
\newblock \href{http://dx.doi.org/10.1103/PhysRevA.77.062317}{Reduced state
  uniquely defines the Groverian measure of the original pure state}.
\newblock Phys. Rev. A \textbf{77}, 062317 (2008).

\bibitem{Bremner09}
M.~J. Bremner, C.~Mora, and A.~Winter.
\newblock \href{http://dx.doi.org/10.1103/PhysRevLett.102.190502}{Are Random
  Pure States Useful for Quantum Computation?}
\newblock Phys. Rev. Lett. \textbf{102}, 190502 (2009).

\bibitem{Hayashi09}
M.~Hayashi, D.~Markham, M.~Murao, M.~Owari, and S.~Virmani.
\newblock \href{http://dx.doi.org/10.1063/1.3271041}{The geometric measure of
  entanglement for a symmetric pure state with non-negative amplitudes}.
\newblock J. Math. Phys. \textbf{50}, 122104 (2009).

\bibitem{Wei10}
T.-C. Wei and S.~Severini.
\newblock \href{http://dx.doi.org/10.1063/1.3464263}{Matrix permanent and
  quantum entanglement of permutation invariant states}.
\newblock J. Math. Phys. \textbf{51}, 092203 (2010).

\bibitem{Renner}
R.~Renner.
\newblock \emph{Security of Quantum Key Distribution}.
\newblock Ph.D. thesis, ETH Zurich (2005).
\newblock {e}print:
  \href{http://arxiv.org/abs/quant-ph/0512258}{arXiv:quant-ph/0512258}.

\bibitem{Tamaryan08}
L.~Tamaryan, D.~Park, and S.~Tamaryan.
\newblock \href{http://dx.doi.org/10.1103/PhysRevA.77.022325}{Analytic
  expressions for geometric measure of three-qubit states}.
\newblock Phys. Rev. A \textbf{77}, 022325 (2008).

\bibitem{Tamaryan10}
S.~Tamaryan, A.~Sudbery, and L.~Tamaryan.
\newblock \href{http://dx.doi.org/10.1103/PhysRevA.81.052319}{Duality and the
  geometric measure of entanglement of general multiqubit W states}.
\newblock Phys. Rev. A \textbf{81}, 052319 (2010).

\bibitem{Chen10}
L.~Chen, A.~Xu, and H.~Zhu.
\newblock \href{http://dx.doi.org/10.1103/PhysRevA.82.032301}{Computation of
  the geometric measure of entanglement for pure multiqubit states}.
\newblock Phys. Rev. A \textbf{82}, 032301 (2010).

\bibitem{Horn}
R.~A. Horn and C.~R. Johnson.
\newblock \emph{Matrix Analysis} (Cambridge University Press, Cambridge, 1990).

\bibitem{Barreiro05}
J.~T. Barreiro, N.~K. Langford, N.~A. Peters, and P.~G. Kwiat.
\newblock \href{http://dx.doi.org/10.1103/PhysRevLett.95.260501}{Generation of
  Hyperentangled Photon Pairs}.
\newblock Phys. Rev. Lett. \textbf{95}, 260501 (2005).

\bibitem{Abramowitz}
M.~Abramowitz and I.~A. Stegun.
\newblock \emph{Pocketbook of Mathematical Functions} (Verlag Harri Deutsch,
  Frankfurt, 1984).

\bibitem{Perelomov}
A.~Perelomov.
\newblock \emph{Generalized Coherent States and Their Applications} (Springer
  Verlag, Berlin, 1986).

\bibitem{Briegel01}
H.~J. Briegel and R.~Raussendorf.
\newblock \href{http://dx.doi.org/10.1103/PhysRevLett.86.910}{Persistent
  Entanglement in Arrays of Interacting Particles}.
\newblock Phys. Rev. Lett. \textbf{86}, 910 (2001).

\bibitem{Tamaryan09}
S.~Tamaryan, T.-C. Wei, and D.~Park.
\newblock \href{http://dx.doi.org/10.1103/PhysRevA.80.052315}{Maximally
  entangled three-qubit states via geometric measure of entanglement}.
\newblock Phys. Rev. A \textbf{80}, 052315 (2009).

\bibitem{Vollbrecht01}
K.~G.~H. Vollbrecht and R.~F. Werner.
\newblock \href{http://dx.doi.org/10.1103/PhysRevA.64.062307}{Entanglement
  measures under symmetry}.
\newblock Phys. Rev. A \textbf{64}, 062307 (2001).

\bibitem{Whyte52}
L.~L. Whyte.
\newblock \href{http://www.jstor.org/stable/2306764}{Unique Arrangements of
  Points on a Sphere}.
\newblock Amer. Math. Monthly \textbf{59}, 606 (1952).

\bibitem{Tammes30}
P.~M.~L. Tammes.
\newblock On the origin of number and arrangement of the places of exit on the
  surface of pollen grains.
\newblock Recueil des trav. bot. n\'{e}erlandais \textbf{27}, 1 (1930).

\bibitem{WeissteinSphere}
E.~W. Weisstein.
\newblock Spherical Code. Wolfram MathWorld.
\newblock \urlprefix\url{http://mathworld.wolfram.com/SphericalCode.html}.

\bibitem{Thomson04}
J.~J. Thomson.
\newblock \href{http://dx.doi.org/10.1080/14786440409463107}{On the structure
  of the atom: an investigation of the stability and periods of oscillation of
  a number of corpuscles arranged at equal intervals around the circumference
  of a circle; with application of the results to the theory of atomic
  structure}.
\newblock Phil. Mag. (Ser. 6) \textbf{7}, 237 (1904).

\bibitem{Albrecht87}
U.~Albrecht and P.~Leiderer.
\newblock \href{http://dx.doi.org/10.1209/0295-5075/3/6/009}{Multielectron
  Bubbles in Liquid Helium}.
\newblock Europhys. Lett. \textbf{3}, 705 (1987).

\bibitem{Leiderer95}
P.~Leiderer.
\newblock \href{http://dx.doi.org/10.1007/BF01338394}{Ions at helium
  interfaces}.
\newblock Z. Phys. B \textbf{98}, 303 (1995).

\bibitem{Davis97}
E.~J. Davis.
\newblock \href{http://dx.doi.org/10.1080/02786829708965426}{A History of
  Single Aerosol Particle Levitation}.
\newblock Aerosol Sci. Technol. \textbf{26}, 212 (1997).

\bibitem{Marzec93}
C.~J. Marzec and L.~A. Day.
\newblock \href{http://dx.doi.org/10.1016/S0006-3495(93)81313-4}{Pattern
  Formation in Icosahedral Virus Capsids: The Papova Viruses and Nudaurelia
  Capensis $\beta$ Virus}.
\newblock Biophys. J. \textbf{65}, 2559 (1993).

\bibitem{Dinsmore02}
A.~D. Dinsmore, M.~F. Hsu, M.~G. Nikolaides, M.~Marquez, A.~R. Bausch, and
  D.~A. Weitz.
\newblock \href{http://dx.doi.org/10.1126/science.1074868}{Colloidosomes:
  Selectively Permeable Capsules Composed of Colloidal Particles}.
\newblock Science \textbf{298}, 1006 (2002).

\bibitem{Kroto85}
H.~W. Kroto, J.~R. Heath, S.~C. O'Brien, R.~F. Curl, and R.~E. Smalley.
\newblock \href{http://dx.doi.org/10.1038/318162a0}{$\text{C}_{60}$:
  Buckminsterfullerene}.
\newblock Nature \textbf{318}, 162 (1985).

\bibitem{Dodgson97}
M.~J.~W. Dodgson and M.~A. Moore.
\newblock \href{http://dx.doi.org/10.1103/PhysRevB.55.3816}{Vortices in a
  thin-film superconductor with a spherical geometry}.
\newblock Phys. Rev. B \textbf{55}, 3816 (1997).

\bibitem{Aspden87}
H.~Aspden.
\newblock \href{http://dx.doi.org/10.1119/1.15206}{Earnshaw's theorem}.
\newblock Am. J. Phys. \textbf{55}, 199 (1987).

\bibitem{Platzman99}
P.~M. Platzman and M.~I. Dykman.
\newblock \href{http://dx.doi.org/10.1126/science.284.5422.1967}{Quantum
  Computing with Electrons Floating on Liquid Helium}.
\newblock Science \textbf{284}, 1967 (1999).

\bibitem{Leech57}
J.~Leech.
\newblock Equilibrium of sets of particles on a sphere.
\newblock Math. Gazette \textbf{41}, 81 (1957).

\bibitem{Erber91}
T.~Erber and G.~M. Hockney.
\newblock \href{http://dx.doi.org/10.1088/0305-4470/24/23/008}{Equilibrium
  configurations of N equal charges on a sphere}.
\newblock J. Phys. A: Math. Gen. \textbf{24}, L1369 (1991).

\bibitem{Ashby86}
N.~Ashby and W.~E. Brittin.
\newblock \href{http://dx.doi.org/10.1119/1.14440}{Thomson's problem}.
\newblock Am. J. Phys. \textbf{54}, 776 (1986).

\bibitem{Altschuler94}
E.~L. Altschuler, T.~J. Williams, E.~R. Ratner, F.~Dowla, and F.~Wooten.
\newblock \href{http://dx.doi.org/10.1103/PhysRevLett.72.2671}{Method of
  constrained global optimization}.
\newblock Phys. Rev. Lett. \textbf{72}, 2671 (1994).

\bibitem{Sloane}
N.~J.~A. Sloane, R.~H. Hardin, W.~D. Smith, \emph{et~al.}
\newblock Spherical Codes.
\newblock \urlprefix\url{http://www2.research.att.com/\~njas/packings/}.

\bibitem{Wales}
D.~J. Wales, J.~P.~K. Doye, \emph{et~al.}
\newblock The Cambridge Cluster Database.
\newblock \urlprefix\url{http://www-wales.ch.cam.ac.uk/CCD.html}.

\bibitem{Berezin85}
A.~A. Berezin.
\newblock \href{http://dx.doi.org/10.1119/1.14023}{Spontaneous symmetry
  breaking in classical systems}.
\newblock Amer. J. Phys. \textbf{53}, 1036 (1985).

\bibitem{Berezin85b}
A.~A. Berezin.
\newblock \href{http://dx.doi.org/10.1038/315104b0}{An unexpected result in
  classical electrostatics}.
\newblock Nature \textbf{315}, 104 (1985).

\bibitem{Berezin86}
A.~A. Berezin.
\newblock \href{http://dx.doi.org/10.1016/0009-2614(86)87015-4}{Electrostatic
  stability and instability of $N$ equal charges in a circle}.
\newblock Chem. Phys. Lett. \textbf{123}, 62 (1986).

\bibitem{Melnyk77}
T.~W. Melnyk, O.~Knop, and W.~R. Smith.
\newblock \href{http://dx.doi.org/10.1139/v77-246}{Extremal arrangements of
  points and unit charges on a sphere: equilibrium configurations revisited}.
\newblock Can. J. Chem. \textbf{55}, 1745 (1977).

\bibitem{Zauner11}
G.~Zauner.
\newblock \href{http://dx.doi.org/10.1142/S0219749911006776}{Quantum Designs:
  Foundations of a Noncommutative Design Theory} \emph{(translation of
  German-written Ph.D. thesis, University of Vienna, 1999)}.
\newblock Int. J. Quantum Inf. \textbf{9}, 445 (2011).

\bibitem{Renes04}
J.~M. Renes, R.~Blume-Kohout, A.~J. Scott, and C.~M. Caves.
\newblock \href{http://dx.doi.org/10.1063/1.1737053}{Symmetric informationally
  complete quantum measurements}.
\newblock J. Math. Phys. \textbf{45}, 2171 (2004).

\bibitem{Gisin98}
N.~Gisin and H.~Bechmann-Pasquinucci.
\newblock \href{http://dx.doi.org/10.1016/S0375-9601(98)00516-7}{Bell
  inequality, Bell states and maximally entangled states for $n$ qubits}.
\newblock Phys. Lett. A \textbf{246}, 1 (1998).

\bibitem{Marx70}
E.~Marx.
\newblock \href{http://dx.doi.org/10.1016/0016-0032(70)90061-X}{Five charges on
  a sphere}.
\newblock J. Franklin Inst. \textbf{290}, 71 (1970).

\bibitem{Toth43}
L.~{Fejes~T\'{o}th}.
\newblock
  \href{http://resolver.sub.uni-goettingen.de/purl?GDZPPN002133873}{{\"U}ber
  eine Absch{\"a}tzung des k{\"u}rzesten Abstandes zweier Punkte eines auf
  einer Kugelfl{\"a}che liegenden Punktsystems}.
\newblock Jahresber. Deutsch. Math.-Ver. \textbf{53}, 66 (1943).

\bibitem{Schutte51}
K.~Sch{\"u}tte and B.~L. van~der Waerden.
\newblock \href{http://dx.doi.org/10.1007/BF02054944}{Auf welcher Kugel haben
  5, 6, 7, 8 oder 9 Punkte mit Mindestabstand Eins Platz?}
\newblock Math. Ann. \textbf{123}, 96 (1951).

\bibitem{Tignol}
J.-P. Tignol.
\newblock \emph{Galois' Theory of Algebraic Equations} (World Scientific,
  Singapore, 2001).

\bibitem{Cenci11}
C.~D. Cenci, D.~W. Lyons, and S.~N. Walck.
\newblock Local unitary group stabilizers and entanglement for multiqubit
  symmetric states (2010).
\newblock {p}reprint: \href{http://arxiv.org/abs/1011.5229}{arXiv:1011.5229}.

\bibitem{Ribeiro11}
P.~Ribeiro and R.~Mosseri.
\newblock \href{http://dx.doi.org/10.1103/PhysRevLett.106.180502}{Entanglement
  in the Symmetric Sector of $n$ Qubits}.
\newblock Phys. Rev. Lett. \textbf{106}, 180502 (2011).

\bibitem{Knopp}
K.~Knopp.
\newblock \emph{Elements of the Theory of Functions} (Dover Publications, New
  York, 1952).

\bibitem{Arnold}
D.~N. Arnold and J.~Rogness.
\newblock M{\"o}bius Transformations Revealed.
\newblock \urlprefix\url{http://www.ima.umn.edu/\~arnold/moebius/index.html}.

\bibitem{Golub}
G.~H. Golub and C.~F. {V}an Loan.
\newblock \emph{Matrix Computations} (Johns Hopkins University Press,
  Baltimore, 1996).

\bibitem{Ichikawa08}
T.~Ichikawa, T.~Sasaki, I.~Tsutsui, and N.~Yonezawa.
\newblock \href{http://dx.doi.org/10.1103/PhysRevA.78.052105}{Exchange symmetry
  and multipartite entanglement}.
\newblock Phys. Rev. A \textbf{78}, 052105 (2008).

\bibitem{Acin00}
A.~Ac\'\i{}n, E.~Jan{\'e}, W.~D{\"u}r, and G.~Vidal.
\newblock \href{http://dx.doi.org/10.1103/PhysRevLett.85.4811}{Optimal
  Distillation of a Greenberger-Horne-Zeilinger State}.
\newblock Phys. Rev. Lett. \textbf{85}, 4811 (2000).

\bibitem{Vicente11}
J.~I. de~Vicente, T.~Carle, C.~Streitberger, and B.~Kraus.
\newblock Complete set of operational measures for the characterization of
  3-qubit entanglement (2011).
\newblock {p}reprint: \href{http://arxiv.org/abs/1106.4774}{arXiv:1106.4774}.

\bibitem{Kempe99}
J.~Kempe.
\newblock \href{http://dx.doi.org/10.1103/PhysRevA.60.910}{Multiparticle
  entanglement and its applications to cryptography}.
\newblock Phys. Rev. A \textbf{60}, 910 (1999).

\bibitem{Osterloh10}
A.~Osterloh and J.~Siewert.
\newblock \href{http://dx.doi.org/10.1088/1367-2630/12/7/075025}{The
  invariant-comb approach and its relation to the balancedness of multipartite
  entangled states}.
\newblock New J. Phys. \textbf{12}, 075025 (2010).

\bibitem{Ren08}
X.-J. Ren, W.~Jiang, X.~Zhou, Z.-W. Zhou, and G.-C. Guo.
\newblock
  \href{http://dx.doi.org/10.1103/PhysRevA.78.012343}{Permutation-invariant
  monotones for multipartite entanglement characterization}.
\newblock Phys. Rev. A \textbf{78}, 012343 (2008).

\bibitem{Love07}
P.~J. Love, A.~M. van~den Brink, A.~Y. Smirnov, M.~H.~S. Amin, M.~Grajcar,
  E.~Il'ichev, A.~Izmalkov, and A.~M. Zagoskin.
\newblock \href{http://dx.doi.org/10.1007/s11128-007-0052-7}{A Characterization
  of Global Entanglement}.
\newblock Quant. Inf. Proc. \textbf{6}, 187 (2007).
\newblock {e}print: \href{http://arxiv.org/abs/quant-ph/0602143}{
  arXiv:quant-ph/0602143}.

\bibitem{Hein04}
M.~Hein, J.~Eisert, and H.~J. Briegel.
\newblock \href{http://dx.doi.org/10.1103/PhysRevA.69.062311}{Multiparty
  entanglement in graph states}.
\newblock Phys. Rev. A \textbf{69}, 062311 (2004).

\bibitem{Bai08}
Y.-K. Bai and Z.~D. Wang.
\newblock \href{http://dx.doi.org/10.1103/PhysRevA.77.032313}{Multipartite
  entanglement in four-qubit cluster-class states}.
\newblock Phys. Rev. A \textbf{77}, 032313 (2008).

\bibitem{Wenninger83}
M.~J. Wenninger.
\newblock \emph{Dual Models} (Cambridge University Press, Cambridge, 1983).

\bibitem{Sachdev}
S.~Sachdev.
\newblock \emph{Quantum Phase Transitions} (Cambridge University Press,
  Cambridge, 1999).

\bibitem{Pan99}
F.~Pan and J.~P. Draayer.
\newblock \href{http://dx.doi.org/10.1016/S0370-2693(99)00191-4}{Analytical
  solutions for the LMG model}.
\newblock Phys. Lett. B \textbf{451}, 1 (1999).

\bibitem{Cirac98}
J.~I. Cirac, M.~Lewenstein, K.~M\o{}lmer, and P.~Zoller.
\newblock \href{http://dx.doi.org/10.1103/PhysRevA.57.1208}{Quantum
  superposition states of Bose-Einstein condensates}.
\newblock Phys. Rev. A \textbf{57}, 1208 (1998).

\bibitem{Ribeiro07}
P.~Ribeiro, J.~Vidal, and R.~Mosseri.
\newblock
  \href{http://dx.doi.org/10.1103/PhysRevLett.99.050402}{Thermodynamical Limit
  of the Lipkin-Meshkov-Glick Model}.
\newblock Phys. Rev. Lett. \textbf{99}, 050402 (2007).

\bibitem{Latorre05}
J.~I. Latorre, R.~Or\'{u}s, E.~Rico, and J.~Vidal.
\newblock \href{http://dx.doi.org/10.1103/PhysRevA.71.064101}{Entanglement
  entropy in the Lipkin-Meshkov-Glick model}.
\newblock Phys. Rev. A \textbf{71}, 064101 (2005).

\bibitem{Hines05}
A.~P. Hines, R.~H. McKenzie, and G.~J. Milburn.
\newblock \href{http://dx.doi.org/10.1103/PhysRevA.71.042303}{Quantum
  entanglement and fixed-point bifurcations}.
\newblock Phys. Rev. A \textbf{71}, 042303 (2005).

\bibitem{Ribeiro}
P.~Ribeiro.
\newblock \emph{Quantum Phase Transitions in Collective Spin Models.
  Applications to Adiabatic Quantum Computation}.
\newblock Ph.D. thesis, Universit\'{e} Pierre et Marie Curie (2008).

\bibitem{Dusuel05}
S.~Dusuel and J.~Vidal.
\newblock \href{http://dx.doi.org/10.1103/PhysRevB.71.224420}{Continuous
  unitary transformation and finite-size scaling exponents in the
  Lipkin-Meshkov-Glick model}.
\newblock Phys. Rev. B \textbf{71}, 224420 (2005).

\bibitem{Lyon06}
S.~A. Lyon.
\newblock \href{http://dx.doi.org/10.1103/PhysRevA.74.052338}{Spin-based
  quantum computing using electrons on liquid helium}.
\newblock Phys. Rev. A \textbf{74}, 052338 (2006).

\bibitem{Grimes79}
C.~C. Grimes and G.~Adams.
\newblock \href{http://dx.doi.org/10.1103/PhysRevLett.42.795}{Evidence for a
  Liquid-to-Crystal Phase Transition in a Classical, Two-Dimensional Sheet of
  Electrons}.
\newblock Phys. Rev. Lett. \textbf{42}, 795 (1979).

\bibitem{Milnor}
J.~W. Milnor.
\newblock \emph{Morse Theory} (Princeton University Press, Princeton, 1963).

\bibitem{Matsumoto}
Y.~Matsumoto.
\newblock \emph{An Introduction to Morse Theory} (American Mathematical
  Society, Providence, 2002).

\end{thebibliography}
\end{document}